{\let~\catcode~13=13\def^^M{^^J}~` 12~`\%12~`~12\xdef\asciiart{
                                                                                  
            ...                         .          ..       ..                    
         xH88"`~ .x8X      .uef^"      @88>  x .d88"  x .d88"    ..               
       :8888   .f"8888Hf :d88E         
      :8888>  X8L  ^""`  `888E          .     '888R    '888R   9888i   .dL        
      X8888  X888h        888E .z8k   .@88u    888R     888R   `Y888k:*888.       
      88888  !88888.      888E~?888L ''888E`   888R     888R     888E  888I       
      88888   
      88888 '> `8888>     888E  888E   888E    888R     888R     888E  888I       
      `8888L 
       `8888  `-*""   /   888E  888E   888&   .888B .  .888B .  x888N><888'       
         "888.      :"   m888N= 888>   R888"  ^*888
           `""***~"`      `Y"   888     ""      "
                               J88"                                   98"         
                               @
                             :"                                    ~`             
                                                                                  
}}\makeatletter

\documentclass[openany]{amsbook}
	\numberwithin{equation}{chapter}
	\numberwithin{figure}{chapter}
	\numberwithin{table}{chapter}

\usepackage{mathtools,amssymb,yhmath}
	\allowdisplaybreaks

\usepackage[english]{babel}
	\def\U@DEC#1{\bgroup\def\UTFviii@defined##1{\expandafter#1\string##1+}}
	\def\U@DEF#1:#2+{\U@LET:#2+{\@nameuse{意#2義}}\@namedef{意#2義}}
	\def\U@LET#1:#2+{\egroup\DeclareUnicodeCharacter{\UTFviii@hexnumber
	{\decode@UTFviii#2\relax}}}\U@DEC\U@LET令{\U@DEC\U@LET}令定{\U@DEC\U@DEF}
	
	令Γ\Gamma	令Δ\Delta	令Θ\Theta	令Λ\Lambda	令Ξ\Xi		令Π\Pi		令Σ\Sigma	
	令Υ\Upsilon	令Φ\Phi		令Ψ\Psi		令Ω\Omega	
	令α\alpha	令β\beta		令γ\gamma	令δ\delta	令ε\varepsilon			令ζ\zeta		
	令η\eta		令θ\theta	令ι\iota		令κ\kappa	令λ\lambda	令μ\mu		令ν\nu		
	令ξ\xi		令π\pi		令ρ\rho		令ς\varsigma	令σ\sigma	令τ\tau		令υ\upsilon	
	令φ\varphi	令χ\chi		令ψ\psi		令ω\omega	令ϑ\vartheta	令ϕ\phi		令ϖ\varpi	
	令Ϝ\Digamma	令ϝ\digamma	令ϰ\varkappa	令ϱ\varrho	令ϴ\varTheta	令ϵ\epsilon	
	令ℂ{\mathbb C}	令𝔻{\mathbb D}	令𝔼{\mathbb E}	令𝔽{\mathbb F}	令𝔾{\mathbb G}	
	令𝕀{\mathbb I}	令𝕂{\mathbb K}	令ℕ{\mathbb N}	令ℙ{\mathbb P}	令ℝ{\mathbb R}	
	令𝕌{\mathbb U}	令𝕍{\mathbb V}	令𝕏{\mathbb X}	令ℤ{\mathbb Z}	
	令𝒞{\mathcal C}	令𝒟{\mathcal D}	令ℰ{\mathcal E}	令𝒢{\mathcal G}	令𝒥{\mathcal J}	
	令ℳ{\mathcal M}	令𝒩{\mathcal N}	令𝒪{\mathcal O}	令𝒫{\mathcal P}	令𝒮{\mathcal S}	
	令𝒰{\mathcal U}	令𝒳{\mathcal X}	令𝒴{\mathcal Y}	
	令𝔈{\mathfrak E}	
	\DeclareMathAlphabet\mathsi{T1}\sfdefault\mddefault\sldefault
	令𝘈{\mathsi A}	令𝘉{\mathsi B}	令𝘊{\mathsi C}	令𝘋{\mathsi D}	令𝘌{\mathsi E}	
	令𝘍{\mathsi F}	令𝘏{\mathsi H}	令𝘐{\mathsi I}	令𝘑{\mathsi J}	令𝘒{\mathsi K}	
	令𝘓{\mathsi L}	令𝘔{\mathsi M}	令𝘗{\mathsi P}	令𝘘{\mathsi Q}	令𝘚{\mathsi S}	
	令𝘛{\mathsi T}	令𝘝{\mathsi V}	令𝘞{\mathsi W}	令𝘟{\mathsi X}	令𝘠{\mathsi Y}	
	令𝘡{\mathsi Z}	
	令𝘢{\mathsi a}	令𝘣{\mathsi b}	令𝘤{\mathsi c}	令𝘦{\mathsi e}	令𝘧{\mathsi f}	
	令𝘨{\mathsi g}	令𝘴{\mathsi s}
	令ℓ\ell		令∂\partial	令∇\nabla	
	令√\sqrt		
	\def\bigl@C#1{\bigl#1}					\def\bigr@C#1{\bigr#1}					
	\def\({\bigl@C(}	\def\){\bigr@C)}	令（{\Bigl(}			令）{\Bigr)}			
	令【{\biggl(}		令】{\biggr)}		令〖{\Biggl(}		令〗{\Biggr)}		
	令［{\bigl@C[}		令］{\bigr@C]}		令「{\Bigl[}			令」{\Bigr]}			
	令｛{\bigl@C\{}		令｝{\bigr@C\}}		令『{\Bigl\{}		令』{\Bigr\}}		
	令⌈\lceil	令⌉\rceil	令⌊\lfloor	令⌋\rfloor	令⟨\langle	令⟩\rangle			
	\def\|{\mathrel\Vert}	令‖{\mathrel\Big\Vert}	令｜{\mid\nobreak}		
	令；{\mathrel;\nobreak}	令＼{\mathord\setminus}	令、{\setminus\nobreak}	令：{\colon}	
	令ˆ\hat		令˜\tilde	令¯\bar		令˘\breve	令˙\dot		令¨\ddot		令°\mathring	
	令ˇ{^\vee}	
	令∏\prod		令∑\sum		令∫\int		令⋀\bigwedge	令⋁\bigvee	令⋂\bigcap	令⋃\bigcup	
	令⨁\bigoplus			令⨂\bigotimes			
	令±\pm		令·\cdot		令×\times	令÷\frac		令†\dagger	令•\bullet	令∘\circ		
	令∧\wedge	令∨\vee		令∩\cap		令∪\cup		令⊕\oplus	令⊗\otimes	令⋆\star		
	令¬\neg		令…\ldots	令∀\forall	令∁\complement			令∃\exists	令∞\infty	
	令⊤\top		令⊥\bot		令⋯\cdots	令★\bigstar	令♠\spadesuit			
	令♡\heartsuit			令♢\diamondsuit			令♣\clubsuit	令♭\flat		令♮\natural	
	令♯\sharp	
	令←\gets		令→\to		令↔\leftrightarrow		令↖\nwarrow	令↗\nearrow	令↘\searrow	
	令↙\swarrow	令↞\twoheadleftarrow		令↠\twoheadrightarrow	令↤\mapsfrom	令↦\mapsto	
	令↩\hookleftarrow		令↪\hookrightarrow		令↾{\mathbin\upharpoonright}			
	令∈\in		令∉\notin	令∋\ni		令∼\sim		令≅\cong		令≈\approx	令≔\coloneqq	
	令≕\eqqcolon	令≠\neq		令≡\equiv	令≤\leqslant	令≥\geqslant	令≪\ll		令≫\gg		
	令⊆\subseteq	令⋮\vdots	令⋰\adots	令⋱\ddots	令⟂\perp		令⟵\longleftarrow		
	令⟶\longrightarrow		令⟼\longmapsto			
	定†#1†{\text{#1}}		令…{,\penalty\binoppenalty\dotsc,\penalty\binoppenalty}		
	令⁰{^0}	令¹{^1}	令²{^2}	令³{^3}	令⁴{^4}	令⁵{^5}	令⁶{^6}	令⁷{^7}	令⁸{^8}	令⁹{^9}	
	令₀{_0}	令₁{_1}	令₂{_2}	令₃{_3}	令₄{_4}	令₅{_5}	令₆{_6}	令₇{_7}	令₈{_8}	令₉{_9}	
	令㏒\log
	令㏑\ln
	定Ｗ#1{^{(#1)}}
	令Ｐ{P_\mathrm e}
	令Ｒ{C-R}
	令Ｘ{¯X-μ}
	令Ｚ{Z_\mathrm{mxd}}
	令Ｓ{S_\mathrm{max}}
	\def\loll{\bsm{1&0\\1&1}}
	\def\locl{\bsm{1&0\\c&1}}
	\def\Git{G^{-⊤}}
	\def\Eo{E_\mathrm 0}
	\def\Er{E_\mathrm r}
	\def\Ec{\mathsi E_\mathrm c}
	\def\Vc{\mathsi V_\mathrm c}
	\def\GY{G_†Ye†}
	\def\GB{G_†Barg†}
	\def\tiltbox#1#2#3{\mkern-#1mu\hbox{\pdfsave
		\pdfsetmatrix{1 0 .2 1}\rlap{$#2$}\pdfrestore}\mkern#3mu}
	\def\Wt{\tiltbox{2.2}W{21.9}^t}
	\def\Xt{\tiltbox0X{18.8}^t}
	\def\Yt{\tiltbox{1}Y{16.9}^t}
	\def\diff{\mathrm d}
	\DeclareMathOperator\hdis{hdis}
	\DeclareMathOperator\hwt{hwt}
	\DeclareMathOperator\dist{dist}
	\DeclareMathOperator\tr{tr}
	\DeclareMathOperator\GL{GL}
	\DeclareMathOperator\poly{poly}
	\DeclareMathOperator*\argmax{arg\,max}
	
	\DeclarePairedDelimiter\abs\lvert\rvert
	
	\def\mat#1{\begin{matrix}#1\end{matrix}}
	\def\bma#1{\begin{bmatrix}#1\end{bmatrix}}
	\def\bsm#1{[\begin{smallmatrix}#1\end{smallmatrix}]}
	\def\cas#1{\begin{cases*}#1\end{cases*}}
	\def\glet{\global\let}
	\let\EA\expandafter
	\let\NE\noexpand
	\let\AB\allowbreak
	定彈#1{\advance#1\glueexpr0ptplus1ptminus1pt} 

\usepackage{tikz,tikz-cd,tikz-3dplot}
	\let\PMP\pgfmathparse
	\def\PMR{\pgfmathresult}
	\let\PMS\pgfmathsetmacro
	
	\let\PMD\pgfmathdeclarefunction
	\def\MD@three#1#2#3#4.{\PMP{0x#1#2#3}}\expandafter\MD@three\pdfmdfivesum{\jobname}.
	\pgfmathsetseed{\PMR-\day-31*(\month-12*(\year-2000))}
	\usetikzlibrary{decorations.pathreplacing,calc,graphs}
	\newdimen\tk
	\pgfmathsetlength\tk{(2.2*.4+1.6)*2.096774/2}
	\tikzset{
		every picture/.style={cap=round,join=round},
		every pin/.style={anchor=180+\tikz@label@angle,anchor/.code=},
		dot/.pic={\fill circle(.6pt);},
		pics/y-tick/.style={code={\draw[gray,very thin](0,0)--(-\tk,0)[left]node{#1};}},
		pics/x-tick/.style={code={\draw[gray,very thin](0,0)--(0,-\tk)[below]node{#1};}}
	}
	\PMD{g2o}1{\PMP{1/(5.714-4.714*h2o(#1))}}
	\PMD{h256}1{
		\PMS\comp{256-#1}%
		\PMP{(-#1*ln(#1)-\comp*ln(\comp))/177.445678+8}%
	}
	\PMD{h2o}1{
		\PMS\sin{sin(#1)}\PMS\cos{cos(#1)}%
		\ifdim\sin pt=0pt%
			\PMP{0}%
		\else\ifdim\cos pt=0pt%
			\PMP{0}%
		\else%
			\PMP{(-\sin*ln(\sin)*\sin-\cos*ln(\cos)*\cos)*2.88539}
		\fi\fi%
	}
	\PMD{h3}1{\PMP{(#1*ln(3*#1)+(1-#1)*ln((1-#1)*1.5))/ln(3)}}
	\PMD{h4}1{\PMS\xb{#1*1.58496}\PMP{(\xb*ln(\xb*1.5)+(1-\xb)*ln((1-\xb)*3))/1.09861}}
	\tikzcdset{}
	\def\cd{\begin{equation*}\catcode`\&=13\cd@aux}
	\newcommand\cd@aux[2][]{\begin{tikzcd}[#1]#2\end{tikzcd}\end{equation*}}

\usepackage{pgfplotstable,booktabs,colortbl}
	\pgfplotsset{
		compat/show suggested version=false,compat=1.17, 
	}

\usepackage{lettrine}
	\def\dropcap#1#2 {\lettrine[lraise=.2]#1{#2} }
	\def\[#1\]{\begin{equation*}{#1}\end{equation*}}

\usepackage{xurl}

\usepackage[unicode,pdfusetitle,linktocpage]{hyperref}
	\hypersetup{colorlinks,linkcolor=red!50!.,citecolor=green!50!.,urlcolor=blue!50!.}

\usepackage{cleveref}
	
	定名#1:#2~#3?#4 {\crefname{#1}{#2#3}{#2#4}}
	名chapter:Chapter~?s			名figure:Figure~?s			名table:Table~?s				
	定理#1:#2~#3?#4 {\newtheorem{#1}[thmall]{#2#3}名#1:#2~#3?#4 }
						理thm:Theorem~?s				
	理cor:Corollar~y?ies			理lem:Lemma~?s				理pro:Proposition~?s			
	理con:Conjecture~?s			
	\theoremstyle{definition}	理dfn:Definition~?s			理exa:Example~?s				
	\theoremstyle{remark}		理cla:Claim~?s				理rem:Remark~?s				
	定式#1:#2~#3?#4 {名#1:#2~#3?#4 \creflabelformat{#1}{\textup{(##2##1##3)}}}
	式equ:equation~?s			式ine:inequalit~y?ies		式ins:inequalities~?			
	式for:formula~?s				式app:approximation~?s		式sup:supremum~?s			
	式seq:sequence~?s			
	\def\Parse@TL#1:#2?{\label@in@display@optarg[#1]{#1:#2}}
	\AtBeginDocument{\def\label@in@display#1{\incr@eqnum\tag{\theequation}\Parse@TL#1?}}
	
	\def\tagcopy#1{\tag{\eqref{#1}'s copy}}

定行{\catcode13 }行13定讀#1{行13\def^^M{\\}\def\READ@ATTR{行5#1}\@ifnextchar[{選}{參}}
定選[^^M#1^^M]^^M#2^^M^^M{\READ@ATTR[#1]{#2}}定參^^M#1^^M^^M{\READ@ATTR{#1}}行5

讀
\title
					          Complexity and Second Moment          					
					                       of                       					
					    the Mathematical Theory of Communication

讀
\author
					                  Hsin-Po Wang

讀{\def
\@pdfsubject}
					              94A24; 68P30; math.IT

讀
\subjclass
					          Primary 94A24; Secondary 68P30

\begin{document}
\message{\asciiart}

\frontmatter

\noindent
\tikz[remember picture,overlay,line width=.1]{
	\tikzset{shift={(.5\textwidth,1em+\headsep+\headheight+\topmargin+1in)}}
	\large
	\draw[nodes={below,yshift=1em-.5ex,align=center}]
		(0,-2in)node{\noexpand\uppercase\expandafter{\@title}}
		(0,-3.5in)node{BY}
		(0,-4in)node{\noexpand\uppercase\expandafter{\authors}}
		(0,-5.5in)node{
			DISSERTATION\\
			~\\
			Submitted in partial fulfillment of the requirements\\
			for the degree of Doctor of Philosophy in Mathematics\\
			in the Graduate College of the\\
			University of Illinois Urbana-Champaign, 2021
		}
		(0,-7.5in)node{Urbana, Illinois}
		(0,-8in)node[text width=6.5in,align=left]{
			Doctoral Committee:\\
			~\\
			\hskip3.5em	Assistant Professor Partha S. Dey, Chair\\
			\hskip3.5em	Professor Iwan M. Duursma, Director of Research\\
			\hskip3.5em	Professor József Balogh\\
			\hskip3.5em	Professor Marius Junge
		}
}

\thispagestyle{empty}

\vbadness999\hbadness99\overfullrule1em彈\baselineskip/4彈\lineskip/4彈\parskip/2

\DeclareRobustCommand\gobblefive[5]{}\addtocontents{toc}\gobblefive
\chapter*{Abstract}

	The performance of an error correcting code is evaluated by
	its block error probability, code rate, and encoding and decoding complexity.
	The performance of a series of codes is evaluated by,
	as the block lengths approach infinity,
	whether their block error probabilities decay to zero,
	whether their code rates converge to channel capacity,
	and whether their growth in complexities stays under control.
	
	Over any discrete memoryless channel, I build codes such that:
	for one, their block error probabilities and code rates scale like random codes';
	and for two, their encoding and decoding complexities scale like polar codes'.
	Quantitatively, for any constants $\pi,\rho>0$ such that $\pi+2\rho<1$,
	I construct a series of error correcting codes with
	block length $N$ approaching infinity,
	block error probability $\exp(-N^\pi)$,
	code rate $N^{-\rho}$ less than the channel capacity,
	and encoding and decoding complexity $O(N\log N)$ per code block.
	
	Over any discrete memoryless channel, I also build codes such that:
	for one, they achieve channel capacity rapidly;
	and for two, their encoding and decoding complexities
	outperform all known codes over non-BEC channels.
	Quantitatively, for any constants $\tau,\rho>0$ such that $2\rho<1$,
	I construct a series of error correcting codes with
	block length $N$ approaching infinity,
	block error probability $\exp(-(\log N)^\tau)$,
	code rate $N^{-\rho}$ less than the channel capacity,
	and encoding and decoding complexity $O(N\log (\log N))$ per code block.
	
	The two aforementioned results are built upon two pillars---%
	a versatile framework that generates codes on the basis of channel polarization,
	and a calculus--probability machinery that evaluates the performances of codes.
	
	The framework that generates codes and the machinery that evaluates codes
	can be extended to many other scenarios in network information theory.
	To name a few:
	lossless compression with side information,
	lossy compression,
	Slepian--Wolf problem,
	Wyner--Ziv Problem,
	multiple access channel,
	wiretap channel of type I,
	and broadcast channel.
	In each scenario, the adapted notions of block error probability and code rate
	approach their limits at the same paces as specified above.

\addtocontents{toc}\gobblefive
\chapter*{}

	First there is Bo-Le%
	\footnote{
		The honorific name of Sun Yang, who is a horse tamer
		in Spring and Autumn period and renowned as a judge of horses;
		also refers to those who recognize (especially hidden) talent.
	}
	
	~
	
	Then can horses gallop hundreds of miles
	
	~
	
	Horses capable of galloping far are common
	
	~
	
	But Bo-Les are scarce
	
	~
	
	\hbox{}\hfill---Han Yu, \emph{On Horses}

\let\hold\hfil\let\cold\l@chapter\def\l@chapter{\def\hfil{\dotfill\glet\hfil\hold}\cold}
\tableofcontents

\mainmatter

\chapter{Introduction}

	\dropcap
	Seventy-three years ago, Claude E. Shannon founded the theory of information
	with an article titled \emph{A Mathematical Theory of Communication},
	which was later republished under the name
	\emph{The Mathematical Theory of Communication} to reflect its omnipotence.
	
	In the eternal work, Shannon explained how to measure the information content
	of a random variable $X$ and argued that, in the long term,
	the information content carried by $X$ costs $H(X｜Y)+ε$ bits to be remembered,
	given that we have free access to another random variable $Y$.
	Shannon also showed that, if sending $X$ results in the reception of $Y$,
	where $X→Y$ is called a communication channel, then the rate at which
	information can be transmitted is $I(X；Y)-ε$ bits per usage of channel.
	These results are now called [Shannon's] source coding theorem
	and noisy-channel coding theorem, respectively.
	
	The famous article left two loopholes.
	Loophole one:
	Shannon's proof involves the existence of certain mathematical objects
	which, in reality, are next to impossible to find constructively.
	As a consequence, Shannon's protocol is never utilized beyond academic interest.
	Loophole two:
	Whereas Shannon's bound on $ε$ is strong enough to conclude that $ε→0$,
	which evinces that $H(X｜Y)$ and $I(X；Y)$ are the limits we would like to achieve,
	we did not know how rapid $ε$ decays to $0$.
	That is, Shannon identified the first order limit of coding,
	but left the second order limit open.
	
	The current dissertation aims to patch the two said loopholes
	and succeeds in improving over existing patches.
	I will present an $ε→0$ coding scheme whose complexity is $O(N㏒(㏒N))$,
	where $N$ is the block length, whilst the best known result is $O(N㏒N)$.
	This in turn patches the first loophole further,
	and is referred to as the complexity paradigm of coding.
	I will also present an $O(N㏒N)$ coding scheme
	whose $ε$ decays to zero at the pace that is provably optimal,
	while earlier works only handle the binary case.
	This in turn fills the second loophole further,
	and is referred to as the second-moment paradigm of coding.
	I will then present a joint scheme that achieves
	both $O(N㏒(㏒N))$ and the optimal pace of $ε→0$.
	
	My codes are built upon two pillars.
	Pillar one:
	The overall code can be seen as a modification
	of a recently developed code---polar code.
	Depending on how we modify polar coding, we can inherit its
	$O(N㏒N)$ complexity or reduce it further down to $O(N㏒(㏒ N))$.
	Pillar two:
	The polarization kernel can be seen as a flag of the legacy codes---random codes.
	Since random coding is the only way to achieve the second-moment paradigm,
	I incorporate it to boost the performance of
	polar coding to the second-moment paradigm.
	The two pillars support a coding scheme whose
	complexity scales like polar coding but performance scales like random coding.
	
	After mastering the complexity and second-moment paradigms of
	the source coding theorem and the noisy-channel coding theorem,
	this dissertation moves forward to a network coding scenario
	called distributed lossless compression.
	I will then adapt the complexity and second-moment paradigms to these scenario.
	Using similar techniques, the same result generalizes to
	even more coding scenarios such as
	multiple access channels,
	wiretap channels of type I, and
	broadcast channels,
	and is left for future research.

\section{Organization of the Dissertation}

	The remaining sections of the current chapter map one-to-one to
	the remaining chapters of the current dissertation and serve as their summaries.
	
	Among the chapters:
	\Cref{cha:origin} was presented in the preprint \cite{ModerateDeviations18}.
	\Cref{cha:prune} was presented in the preprints \cite{LoglogTime18}
		and \cite{LoglogTime19}, the latter of which was later published in
		IEEE Transactions on Information Theory \cite{LoglogTime21}.
	\Cref{cha:general} was presented in the preprint \cite{LargeDeviations18}.
	\Cref{cha:random} was presented in the preprint \cite{Hypotenuse19}, which was later
		published in IEEE Transactions on Information Theory \cite{Hypotenuse21}.

\section{Original Channel Polarization}

	In \cref{cha:origin}, we will revisit Arıkan's original proposal
	of polar coding that is dedicated to binary-input channels
	and uses $\loll$ as the polarization kernel.
	This chapter serves three purposes:
	One, the construction by $\loll$ is simple yet powerful enough to achieve
	capacity (the first order limit), and is of historical significance.
	Two, I will provide a complete proof to the strongest version
	of the statements available in the literature that unifies
	the techniques spanning across several state-of-the-art works.
	Three, the main statement and its proof are the starting point of at least
	three generalizations, and will be referred to as the prototype every now and then.
	
	In the seminal paper, Arıkan started with a symmetric binary-input
	discrete-output memoryless channel $W(y｜x)$ and synthesized
	its children $WＷ1(y₁y₂｜u₁)$ and $WＷ2(y₁y₂u₁｜u₂)$ via
	\begin{gather*}
		WＷ1(y₁y₂｜u₁)≔∑_{u₂∈𝔽₂}÷12W(y₁｜u₁+u₂)W(y₂｜u₂),	\\
		WＷ2(y₁y₂u₁｜u₂)≔÷12W(y₁｜u₁+u₂)W(y₂｜u₂).
	\end{gather*}
	Treating $•Ｗ1$ and $•Ｗ2$ as transformations applied to channels,
	Arıkan continued synthesizing, in a recursive manner,
	$W$'s grandchildren $(WＷ1)Ｗ1$, $(WＷ1)Ｗ2$, $(WＷ2)Ｗ1$, and $(WＷ2)Ｗ2$,
	followed by $W$'s grand-grandchildren $\((WＷ1)Ｗ1\)Ｗ1$, $\((WＷ1)Ｗ1\)Ｗ2$, etc,
	followed by their descendants ad infinitum.
	
	Arıkan observed that a code can be established by
	selecting a subset of reliable synthetic channels.
	To evaluate the performance of codes established this way,
	we proceed to examine the stochastic process $\{𝘞_n\}$ defined by
	$𝘞₀≔W$ and $𝘞_{n+1}≔(𝘞_n)Ｗ{1† or †2† with equal chance†}$.
	The evolution of the synthetic channels can be controlled by
	$Z(WＷ2)=Z(W)²$ and $Z(W)√{2-Z(W)²}≤Z(WＷ1)≤2Z(W)-Z(W)²$.
	From that I will prove
	\[𝘗｛Z(𝘞_n)<e^{-2^{πn}}｝>𝘗\{Z(𝘞_n)→0\}-2^{-ρn},\label{ine:teaser-Z}\]
	where $(π,ρ)$ is any pair of constants that
	lies in the shaded area in \cref{fig:lol-bdmc}.
	This inequality tells us how reliable $𝘞_n$ can be.
	Thus we learn how the original polar coding performs.

\section{Asymmetric Channels}

	In \cref{cha:dual}, I will bring up a “dual picture” of \cref{cha:origin}.
	The dual picture consists of three elements:
	One, through examining the behavior of $𝘞_n$ when it becomes noisy,
	i.e., when $Z(𝘞_n)≈1$, we have a more complete
	understanding of $\{𝘞_n\}$ as a stochastic process.
	Two, the proof of said behavior is the mirror image of the one
	given in \cref{cha:origin}, reinforcing the duality.
	Three, this result is pivotal to source coding for lossy compressions,
	to noisy-channel coding over asymmetric channels,
	and to a pruning technique that will be covered in upcoming chapters.
	
	While \cref{ine:teaser-Z} addresses the behavior of $𝘞_n$ at the reliable end,
	another parameter $T$ and a stochastic process $\{T(𝘞_n)\}$
	are defined to examine the behavior of $𝘞_n$ at the noisy end.
	It satisfies $T(WＷ1)=T(W)²$ and $T(WＷ2)≤2T(W)-T(W)²$,
	and can be used to show that
	\[𝘗｛T(𝘞_n)<e^{-2^{πn}}｝>𝘗\{T(𝘞_n)→0\}-2^{-ρn},\label{ine:teaser-T}\]
	where $(π,ρ)$ is any pair of constants that
	lies in the \emph{same} shaded area in \cref{fig:lol-bdmc}.
	Note that $T(𝘞_n)→0$ iff $Z(𝘞_n)→1$, so \cref{ine:teaser-T} is the
	mirror image of \cref{ine:teaser-Z}, telling us how noisy $𝘞_n$ can be.
	
	\Cref{ine:teaser-Z} alone implies that polar coding is good for error correction
	over symmetric channels and lossless compression (with or without side information).
	\Cref{ine:teaser-T} alone implies that polar coding is good for lossy compression.
	Together, \cref{ine:teaser-Z,ine:teaser-T} imply that
	(a)	polar coding is good over asymmetric channels, and
	(b)	polar coding can be modified to attain an even lower complexity.
	More on (b) in \cref{cha:prune}.
	
	The proofs of \cref{ine:teaser-Z,ine:teaser-T} involve monitoring
	a “watch list” subset $𝘈_m$ of moderately reliable channels
	for $m=√n,2√n…n-√n$ for some large perfect square $n$.
	When $𝘞_m∈𝘈_m$ is moderately good and $𝘞_{m+√n}$ becomes extraordinary reliable,
	$𝘞_{m+√n}$ is moved from $𝘈_{m+√n}$ to another “trustworthy” subset $𝘌_{m+√n}$.
	On the other hand, when $𝘞_{m+√n}$ becomes noisy, it is temporary removed
	from $𝘈_{m+√n}$, but has some chance to be added back to $𝘈_{m+l√n}$
	(for some $l≥2$) once $𝘞_{m+l√n}$ becomes moderately reliable again.
	The rest of the proof is to quantify moderate and extraordinary
	reliabilities and the corresponding frequencies.

\section{Pruning Channel Tree}

	In \cref{cha:prune}, I will introduce a novel technique called pruning.
	Pruning reduces the complexity of encoder and decoder
	while retaining in part the performance of polar codes.
	And it is reported prior that pruning reduces the complexity by a constant factor.
	In this chapter, I will show that pruned polar codes can
	achieve capacity with encoding and decoding complexity $O(N㏒(㏒N))$,
	transcending the old $O(N㏒N)$ complexity.
	We will see in later chapters that pruning is a special case of
	dynamic kerneling and can be applied to more general polar codes.
	
	To explain pruning, I will establish the trinitarian correspondence among
	the encoder/decoder, the channel tree, and the channel process $\{𝘞_n\}$.
	Pruning the channel tree corresponds to trimming the unnecessary part
	of the encoder and decoder, which reduces complexity.
	Viewed from the a different perspective, pruning the channel tree
	corresponds to declaring a stopping time $𝘴$ that is adapted to
	$\{𝘞_n\}$ so that the stochastic process $\{𝘞_{n∧𝘴}\}$ stabilizes
	whenever channel transformation becomes ineffective.
	
	Now $𝘞_𝘴$ becomes a random variable associated to code performance.
	For all intents and purposes, it suffices to declare
	the stopping time $𝘴$ properly and prove inequalities of the form
	\begin{gather*}
		𝘗｛Z(𝘞_𝘴)<4^{-n}｝>1-H(W)-2^{-ρn},	\\
		𝘗｛T(𝘞_𝘴)<4^{-n}｝>H(W)-2^{-ρn}
	\end{gather*}
	and another inequality of the form
	\[N𝘌[𝘴]≤O(N㏒(㏒N)).\]
	The first two inequalities imply that the code achieves capacity;
	the third inequality confirms the complexity being $O(N㏒(㏒N))$.

\section{General Alphabet and Kernel}

	In \cref{cha:general}, I will investigate thoroughly
	two known ways to generalize polar coding.
	One is allowing channels with arbitrary finite input alphabet.
	This extends polar coding to all channels Shannon had considered in 1948.
	The other is utilizing arbitrary matrices as polarizing kernels.
	Doing so provably improves the performance in the long run,
	and is reportedly improving the performance for moderate block length.
	
	To begin, we will go over four regimes that connect
	probability theory, random coding theory, and polar coding theory:
	\begin{itemize}
		\item	Law of large numbers (LLN) and achieving capacity;
				this regime concerns whether block error probability $Ｐ$
				decays to $0$ while code rate $R$ converges to capacity.
		\item	Large deviation principle (LDP) and error exponent;
				this regime concerns how fast $Ｐ$ decays to $0$ when an $R$ is fixed.
		\item	Central limit theorem (CLT) and scaling exponent; this regime
				concerns how fast $R$ approaches capacity a when $Ｐ$ is fixed.
		\item	Moderate deviation principle (MDP);
				this regime concerns the general trade-off between $Ｐ$ and $R$.
	\end{itemize}
	
	Next, I will go back to prove results regarding polar coding.
	For any matrix $G$ over any finite field $𝔽_q$,
	the LDP data of $G$ include coset distances $D_ZＷj≔\hdis(r_j,R_j)$,
	where $\hdis$ is the Hamming distance, $r_j$ is the $j$th row of $G$,
	and $R_j$ is the subspace spanned by the rows below $r_j$.
	Coset distances are such that
	\[Z(WＷj)≈Z(W)^{D_ZＷj}.\]
	This approximation is used to control small $Z(𝘞_n)$,
	which eventually proves a generalization of \cref{ine:teaser-Z}.
	For the dual picture, there is a parameter $S$ generalizing $T$ and satisfying
	\[S(WＷj)≈S(W)^{D_SＷj},\]
	where $D_SＷj≔\hdis(c_j,C_j)$ is the Hamming distance from $c_j$ the $j$th column
	of $G^{-1}$ to $C_j$ the subspace spanned by the columns to the left of $c_j$.
	This eventually proves a generalization of \cref{ine:teaser-T}.
	
	The CLT data of an $ℓ×ℓ$ matrix $G$ consist of a choice of parameter $H$, a concave
	function $h:[0,1]→[0,1]$ such that $h(0)=h(1)=0$, and a number $ϱ$ such that
	\[÷1{ℓ}∑_{i=1}^ℓh(H(WＷj))≤ℓ^{-ϱ}h(H(W)),\]
	where $H$ could be the conditional entropy 
	or any other handy parameter that maximizes $ϱ$.
	
	The contribution of \cref{cha:general} is a calculus--probability machinery
	that predicts the MDP behavior of polar codes given the LDP and CLT data.
	The prediction is of the form
	\[𝘗｛Z(𝘞_n)<e^{-ℓ^{πn}}｝>1-H(W)-ℓ^{-ρn},\]
	where $(π,ρ)$ lies to the left of the convex envelope of $(0,ϱ)$
	and the convex conjugate of $t↦㏒_ℓ\(ℓ^{-1}∑_{j=1}^ℓ(D_ZＷj)^t\)$.
	This generalizes \cref{ine:teaser-Z}.
	Similarly, the generalization of \cref{ine:teaser-T} reads
	\[𝘗｛S(𝘞_n)<e^{-ℓ^{πn}}｝>H(W)-ℓ^{-ρn},\]
	where $(π,ρ)∈[0,1]²$ to the left of the convex envelope of $(0,ϱ)$
	and the convex conjugate of $t↦㏒_ℓ\(ℓ^{-1}∑_{j=1}^ℓ(D_SＷj)^t\)$.

\section{Random dynamic Kerneling}

	In \cref{cha:random}, two novel ideas are invoked to help
	achieve the optimal MDP behavior with low complexity.
	First, the matrix $G$ that induces polarization is not fixed
	but varying on a channel-by-channel basis.
	Second, since it is difficult to prove that a specific $G$ is good,
	a random variable $𝔾$ is to replace $G$ and I will investigate
	the typical behavior of $𝔾$ as a polarizing kernel.
	
	The MDP behavior of random coding, which is provably optimal, reads
	\[÷{-㏑Ｐ}{N(C-R)²}→÷1{2V},\]
	where $C$ is channel capacity and $V$ is another
	intrinsic parameter called channel dispersion or varentropy.
	Our target behavior is less impressive,
	yet it is asymptotically optimal in the logarithmic scale:
	\[÷{㏑(-㏑Ｐ)}{㏑(N(C-R)²)}≈1.\]
	Or equivalently, for any $π+2ρ<1$, there are codes
	such that $Ｐ<\exp(-N^π)$ and $C-R<N^{-ρ}$.
	
	For the typical LDP behavior of $𝔾$, we need to understand,
	for each $j$, the typical Hamming distance $D_ZＷj$ from
	its $j$th row to the subspace spanned by the rows below.
	This step is essentially the Gilbert--Varshamov bound with
	slight modifications so that it is easier to manipulate in later steps.
	
	For the typical CLT behavior of $𝔾$, I choose the concave function
	$h(x)≔\min(x,1-x)^{ℓ/㏑ℓ}$.
	Now we need to understand the typical behavior of
	\[÷1{ℓ}∑_{i=1}^ℓh(H(WＷj)),\]
	where $WＷj$ is a random variable depending on $𝔾$.
	This boils down to showing that the first few $H(WＷj)$ are close to $1$,
	while the last few $H(WＷj)$ are close to $0$.
	To show that $H(WＷj)≈1$ and to quantify the approximation,
	I reduce this to a reliability analysis of noisy-channel coding.
	To show that $H(WＷj)≈0$ and to quantify the approximation,
	I reduce this to a secrecy analysis of wiretap-channel coding.

\section{Joint Pruning and Kerneling}

	In \cref{cha:joint}, I will combine the techniques
	in \cref{cha:prune,cha:general} to depict a trade-off between
	the complexity---ranging from $O(N㏒N)$ to $O(N㏒(㏒N))$---%
	and the decay of $Ｐ$---ranging from $\exp(-N^π)$ to $\exp(-(㏒N)^τ)$.
	
	The main idea is to apply the stopping time analysis to
	any channel process $\{𝘞_n\}$ whose MDP behavior is known.
	It could be a process generated by $\locl$,
	for which we know that it is guaranteed to have a positive $ϱ$.
	It could be a process generated by a large kernel
	whose $ϱ$ is bounded by some other method.
	It could also be generated by random dynamic kerneling, for which we know $ϱ→1/2$.
	
	The result is that, for any kernel, polar codes have
	the same gap to capacity before and after pruning;
	and depending on how aggressively one wants to prune, the complexity per bit
	is approximately the logarithm of the logarithm of the block error probability.
	
	For example, if the targeted block error probability is $\exp(-N^π)$,
	then the predicted complexity is $O(N㏒N)$.
	This recovers the result of the previous chapter.
	On the other hand, if the targeted block error probability is $\exp(-(㏒N)^τ)$,
	then the predicted complexity is $O(N㏒(㏒N))$.
	This is, by far, the lowest complexity for
	capacity-achieving codes over generic channels.
	Plus the gap to capacity decay to $0$ optimally fast.

\section{Distributed Lossless Compression}

	In \cref{cha:dislession}, I will extend the theorems established
	in the previous chapters to distributed lossless compression problems.
	A distributed lossless compression problem
	is a network coding problem where there are $m$ sources,
	each to be compressed by a compressor that do not talk to each other,
	and a decompressor that attempt to reconstruct all sources.
	I will go over the two-source case as a warm up,
	the three-source case to demonstrate the difficulty,
	and finally the $m$-source case for a general result.
	
	I will explain that, modulo the previous chapters,
	the main challenge is to reduce a multiple-sender problems
	to several one-sender problems.
	The reduction consists of two steps.
	The first step is to demonstrate that a random source $X$ can be “split”
	into two random fragments $X⟨1⟩$ and $X⟨2⟩$ such that there is a bijection
	$X↔(X⟨1⟩,X⟨2⟩)$ and hence they carry the same amount of information.
	The second step is to show that, by interleaving the fragments of sources in a way
	that is related to Gray codes, we can fine-tune the workloads of every sender.
	That helps us achieve every possible distribution of workloads.
	
	A key to the second step is degree theory, an algebraic topology machinery
	that determines the surjectivity of a continuous map.
	The degree theory offers a sufficient condition on whether
	a map is onto the dominant face of a contra-polymatroid.
	Here, it is the rate region of a distributed lossless compression problem
	that is a contra-polymatroid.
	Dually, the capacity region of a multiple access channel
	is a polymatroid and a similar argument applies.
	This fact indicates that, for both distributed lossless compression and
	multiple access channels, splitting coupled with polar coding achieves
	the optimal block error probability and the optimal gap to boundary
	at the cost of $O(N㏒N)$ complexity.
	Or, following the complexity paradigm, one prunes the complexity to $O(N㏒(㏒N))$
	if a slightly higher block error probability is acceptable.

\chapter{Original Channel Polarization}\label{cha:origin}

	\dropcap
	Fifteen years ago, Erdal Arıkan developed a technique,
	called \emph{channel combining and splitting}, to combine two identical channels
	and then split them into two distinct channels \cite{Arikan06}.
	At the cost of having to prepare different codes to deal with
	distinct channels, the two new channels enjoy better metrics.
	More precisely, the average of the cutoff rates rises.
	Arıkan then argued that, by recursively synthesizing
	the children of the children of $\dotso$ of the channel,
	the rise in cutoff rates eventually pushes them towards the channel capacity.
	
	As it turns out, after a sufficient amount of recursion, one does not need
	$2^n$ different coding schemes to deal with the $2^n$ descendants of $W$.
	This is because most synthetic channels are either
	satisfactorily reliable---so we just transmit plain messages through these---%
	or desperately noisy---so we just ignore those.
	The phenomenon is named \emph{channel polarization}
	and the corresponding coding scheme \emph{polar coding}.
	
	Arıkan showed that this original polar coding achieves channel capacity.
	That is, if you follow Arıkan's instruction to construct codes, then $Ｐ→0$ and
	$R→I(W)$ over any symmetric binary-input discrete-output memoryless channels.
	This is done via proving, for some functions $θ(n)$ and $γ(n)$,
	(a)	that
	\[𝘗｛Z(𝘞_n)<θ(n)｝>I(W)-γ(n),\]
	(b)	that $2^nθ(n)$ is an upper bound on $Ｐ$, and
	(c)	that $γ(n)$ is an upper bound on $I(W)-R$.
	
	In this chapter, I will characterize the pace of achieving capacity.
	We will see that
	\[𝘗｛Z(𝘞_n)<e^{-2^{πn}}｝>I(W)-2^{-ρn}\]
	if $(π,ρ)∈[0,1]²$ lies to the left of the convex envelope of
	$(0,1/4.714)$ and $1-($the binary entropy function$)$.
	Prior to my work, the largest region of achievable $(π,ρ)$
	is considerably smaller and reaches only $(0,1/5.714)$ \cite{MHU16}.
	See \cref{fig:lol-bdmc} for plots.
	
	\begin{figure}
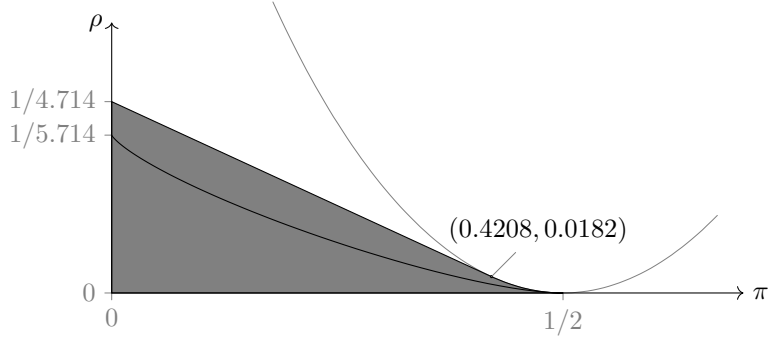

		\tikz[scale=12]{
			\draw[gray,very thin]
				(0,1/4.714)pic{y-tick=$1/4.714$}(0,1/5.714)pic{y-tick=$1/5.714$}
				(0,0)pic{y-tick=$0$}pic{x-tick=$0$}(1/2,0)pic{x-tick=$1/2$}
				plot[domain=25:55,samples=90]({sin(\x)^2},{1-h2o(\x)});
			\draw[fill=gray]
				(0,1/4.714)--plot[domain=40.443:45,samples=30]({sin(\x)^2},{1-h2o(\x)})-|cycle;
			\draw
				plot[domain=0:45,samples=90]({g2o(\x)*sin(\x)^2},{(1-g2o(\x))/4.714});
			\draw[->](0,0)--(0,.3)node[left]{$ρ$};
			\draw[->](0,0)--(.7,0)node[right]{$π$};
			\draw(.4208,.01818)pic{dot}coordinate[pin=45:{$(0.4208,0.0182)$}];
		}
		\caption{
			The achievable region of $(㏒(-㏒₂Ｐ),㏒(C-R))$ is shaded.
			The curve part is $1$ minus the binary entropy function.
			The straight part is the tangent line from $(0,1/4.714)$ to the curve,
			the tangent point being $(0.4208,0,0182)$.
			The lower left curve is the previous result \cite{MHU16},
			which attains $(0,1/5.714)$.
		}\label{fig:lol-bdmc}
	\end{figure}

\section{Problem Setup and Primary Definitions}

	After introducing some definitions,
	this section describes the main problem to be solved in this chapter.
	First goes the definition of the channels we want to attack.
	
	\begin{dfn}
		A \emph{symmetric binary-input discrete-output memoryless channels} (SBDMCs)
		is a Markov chain $W：𝔽₂→𝒴$, where
		\begin{itemize}
			\item	$𝔽₂$ is the finite field or order $2$;
			\item	$𝒴$ is a finite set;
			\item	$W(y｜x)∈[0,1]$, for $x∈𝔽₂$ and $y∈𝒴$, is an array of transition	\\
					probabilities satisfying $∑_{y∈𝒴}W(y｜x)=1$ for both $x∈𝔽₂$; and
			\item	there exists an involution $σ：𝒴→𝒴$
					such that $W(y｜0)=W(σ(y)｜1)$ for all $y∈𝒴$.
		\end{itemize}
	\end{dfn}
	
	Denote by $Q$ the uniform distribution on $𝔽₂$;
	treat this as the input distribution of the channel $W$.
	Denote by $W(x,y)$ the joint probability $Q(x)W(y｜x)$.
	Denote by $W(x｜y)$ the posterior probability;
	note that $W(•｜•)$ assumes two interpretations,
	depending on whether we want to predict $y$ from $x$ or the other way around.
	When it is necessary, $W(y)≔W(y｜0)+W(y｜1)$ denotes the output probability.
	Capital variables $X$ and $Y$, usually with indices,
	denote the input and output of the channel governed by $Q$ and $W$.
	
	The definitions of some channel parameters follow.
	
	\begin{dfn}
		The \emph{conditional entropy} of $W$ is
		\[H(W)≔-∑_{x∈𝔽₂}∑_{y∈𝒴}W(x,y)㏒₂W(x｜y),\]
		which is the amount of noise/equivocation/ambiguity/fuzziness caused by $W$.
	\end{dfn}
	
	\begin{dfn}
		The \emph{mutual information} of $W$ is
		\[I(W)≔H(Q)-H(W)=∑_{x∈𝔽₂}∑_{y∈𝒴}W(x,y)㏒₂÷{W(x｜y)}{Q(x)},\]
		which is also the channel capacity of $W$.
	\end{dfn}
	
	\begin{dfn}\label{dfn:bin-Z}
		The \emph{Bhattacharyya parameter} of $W$ is
		\[Z(W)≔2∑_{y∈𝒴}√{W(0,y)W(1,y)},\]
		which is twice the Bhattacharyya coefficient between
		the joint distributions $W(0,•)$ and $W(1,•)$.
	\end{dfn}
	
	The overall goal is to construct, for some large $N$, an encoder $ℰ：𝔽₂^{RN}→𝔽₂^N$
	and a decoder $𝒟：𝒴^N→𝔽₂^{RN}$ such that the composition
	\cd[every arrow/.append style=mapsto]{
		U₁^{RN}\rar{ℰ}	&	X₁^N\rar{W^N}	&	Y₁^N\rar{𝒟}	&	ˆU₁^{RN}
	}
	is the identity map as frequently as possible,
	and $R$ as close to the channel capacity $I(W)$ as possible.
	
	\begin{dfn}
		Call $N$ the block length.
		Call $R$ the code rate.
		Denote by $Ｐ$, called the block error probability,
		the probability that $ˆU₁^{RN}≠U₁^{RN}$.
	\end{dfn}
	
	To reach the overall goal of constructing good error correcting codes,
	I will introduce the building block of all techniques
	we are to utilize--channel transformation.

\section{Channel Transformation and Tree}

	This section motivates and defines the channel transformation.
	For the precise connection between the transformation and the actual encoder/decoder
	design, please refer to Arıkan's original work \cite{Arikan09}.
	
	Let $G∈𝔽₂^{2×2}$ be the matrix
	\[\bma{1&0\\1&1}.\]
	Let $U₁,U₂∈𝔽₂$ be two uniform random variables.
	Let $X₁²∈𝔽₂$ be the vector
	\[\bma{X₁&X₂}≔\bma{U₁&U₂}\bma{1&0\\1&1},\]
	or $X₁²≔U₁²G$ for short.
	Let $Y₁,Y₂∈𝒴$ be the outputs of two i.i.d.\ copies
	of $W$ given the inputs $X₁$ and $X₂$, respectively.
	Then the combination of the two $W$'s is the channel
	with input $U₁²$ and output $Y₁²$.
	
	To split the combination of the channels, consider a two-step guessing job:
	\begin{itemize}
		\item	Guess $U₁$ given $Y₁²$.
		\item	Guess $U₂$ given $Y₁²$,
				assuming that the guess $ˆU₁$ of $U₁$ is correct.
	\end{itemize}
	Pretend that there is a channel $WＷ1$ with input $U₁$ and output $Y₁²$;
	this channel captures the difficulty of the first step.
	Pretend also that there is a channel $WＷ2$ with input $U₂$ and output $Y₁²U₁$;
	this channel captures the difficulty of the second step.
	The precise definitions follows.
	
	\begin{dfn}
		Define synthetic channels
		\begin{gather*}
			WＷ1(y₁²｜u₁)≔∑_{u₂∈𝔽₂}÷12W(y₁｜u₁+u₂)W(y₂｜u₂),	\\
			WＷ2(y₁²u₁｜u₂)≔÷12W(y₁｜u₁+u₂)W(y₂｜u₂).
		\end{gather*}
	\end{dfn}
	
	Clearly $WＷ1$ and $WＷ2$ are of binary input and discrete output.
	It can be shown that they are symmetric, hence are SBDMCs.
	Therefore, the channel transformation $W↦(WＷ1,WＷ2)$
	maps the set of SBDMCs to the Cartesian square thereof.
	
	Once we accept the idea that the guessing jobs can be modeled as channels,
	we can talk about manipulating the channels
	as if they were actual objects instead of describing, abstractly,
	the change in ways we are guessing the random variables.
	Particularly, we can easily imagine that the channel transformation
	applies recursively and gives birth to descendants
	$WＷ{j₁}$, $(WＷ{j₁})Ｗ{j₂}$, $((WＷ{j₁})Ｗ{j₂})Ｗ{j₃}$, and so on and so forth.
	This family tree of synthetic channels
	rooted at $W$ is called the \emph{channel tree}.
	
	To construct a code, choose a large integer $n$ and synthesize the depth-$n$
	descendants of $W$, which are of the form $\(\dotsb((WＷ{j₁})Ｗ{j₂})\dotsb\)Ｗ{j_n}$.
	Select a subset of those channels, which is equivalent to
	selecting a subset of indices $(j₁,j₂…j_n)∈\{1,2\}^n$.
	Call this subset $𝒥$.
	Then by transmitting messages through
	the synthetic channels in $𝒥$, a code is established.
	This code has block length $N=2^n$, code rate $R=\abs{𝒥}/2^n$,
	and block error probability upper bounded by
	\[Ｐ≤∑_{j₁^n∈𝒥}Z\(\(\dotsb((WＷ{j₁})Ｗ{j₂})\dotsb\)Ｗ{j_n}\).\label{ine:P<sum}\]
	In all papers I have seen, no upper bound on $Ｐ$
	other than \cref{ine:P<sum} was used.
	So we may pretend that the right-hand side of
	\cref{ine:P<sum} is the design block error probability of $𝒥$.
	
	To construct good codes, it suffices to collect
	in $𝒥$ synthetic channels with small $Z$.
	But the more we collect, the higher the sum of $Z$'s.
	This induces a trade-off between $Ｐ$ and $R$,
	which is the subject of the current chapter.
	Let $θ$ be the collecting threshold;
	that is, $𝒥$ collects synthetic channels whose $Z$ falls below $θ$.
	Then $θ$ parametrizes the trade-off in the sense that $Ｐ<Nθ$
	and $R$ is the density of the synthetic channels whose $Z$ falls below $θ$.
	
	In the next section, I will introduce some stochastic processes
	that help us comprehend the trade-off between $R$ and $Ｐ$.

\section{Channel and Parameter Processes}

	We are to define some stochastic processes whose sample space
	is independent of those of the channels and user messages.
	To help distinguish the new source of randomness,
	I typeset the relevant symbols (such as $𝘗,𝘌$) in sans serif font.
	
	\begin{dfn}
		Let $𝘑₁,𝘑₂,\dotsc$ be i.i.d.\ tosses of a fair coin with sides $\{1,2\}$.
		That is,
		\[𝘑_n≔\cas{
			1	&	w.p. $1/2$,	\\
			2	&	w.p. $1/2$.	
		}\]
		Let $𝘞₀,𝘞₁,𝘞₂,\dotsc$, or $\{𝘞_n\}$ in short,
		be a stochastic process of SBDMCs defined as follows:
		\begin{itemize}
			\item	$𝘞₀≔W$; and
			\item	$𝘞_{n+1}≔𝘞_nＷ{𝘑_{n+1}}$.
		\end{itemize}
		This is called the \emph{channel process}.
	\end{dfn}
	
	\begin{dfn}
		Let $\{𝘏_n\}$ be the stochastic process obtained by applying $H$ to $\{𝘞_n\}$.
		That is, $𝘏_n≔H(𝘞_n)$.
		It is called \emph{Arıkan's martingale}.
	\end{dfn}
	
	\begin{dfn}
		Let $\{𝘡_n\}$ be the stochastic process obtained by applying $Z$ to $\{𝘞_n\}$.
		That is, $𝘡_n≔Z(𝘞_n)$.
		It is called \emph{Bhattacharyya's supermartingale}.
	\end{dfn}
	
	The remainder of this section is devoted to explaining that Arıkan's martingale
	is a martingale and Bhattacharyya's supermartingale is a supermartingale,
	as well as other relations among $H$ and $Z$.
	It will show that questions regarding the code performance
	can be passed to questions regarding the processes $\{𝘏_n\}$ and $\{𝘡_n\}$.
	
	\begin{pro}\label{pro:martin}
		Arıkan's martingale $\{𝘏_n\}$ is a martingale.
	\end{pro}
	
	\begin{proof}
		It suffices to check if $H(WＷ1)+H(WＷ2)=2H(W)$.
		Recall the inputs and outputs of $WＷ1$ and $WＷ2$;
		we have
		\begin{align*}
			H(WＷ1)+H(WＷ2)
			&	=H(U₁｜Y₁²)+H(U₂｜U₁Y₁²)=H(U₁²｜Y₁²)	\\
			&	=H(X₁²｜Y₁²)=2H(X｜Y)=2H(W).
		\end{align*}
		That finishes the proof.
	\end{proof}
	
	\begin{pro}\label{pro:superm}
		Bhattacharyya's supermartingale $\{𝘡_n\}$ is a supermartingale.
	\end{pro}
	
	\begin{proof}
		It suffices to check if $Z(WＷ1)+Z(WＷ2)≤2Z(W)$.
		But that is the sum of \cref{ine:2Z-Z^2,ine:Z^2} below.
	\end{proof}
	
	\begin{lem}[Evolution of $Z$]\label{lem:squares}
		The following hold for all SBDMCs $W$:
		\begin{gather*}
			Z(WＷ2)=Z(W)²,				\label{ine:Z^2}\\
			Z(WＷ1)≤2Z(W)-Z(W)²,		\label{ine:2Z-Z^2}\\
			Z(WＷ1)≥Z(W)√{2-Z(W)²}.	\label{ine:2-Z^2}
		\end{gather*}
	\end{lem}
	
	For a proof of the first two inequalities, see \cite[Proposition~5]{Arikan09}.
	Regarding the third inequality, it is used in \cite[inequality~(5)]{MHU16},
	wherein the authors cited \cite[Exercise~4.62]{RU08}.
	The proofs consist of elementary manipulations of summations and square roots.
	
	The next lemma relates $H$ and $Z$.
	Note that any relation automatically applies to $\{𝘏_n\}$ and $\{𝘡_n\}$.
	
	\begin{lem}[$Z$ vs $H$]\label{lem:ZvsH}
		The following hold for all SBDMCs $W$:
		\begin{gather*}
			Z(W)≥H(W),			\\
			Z(W)²≤H(W),			\label{ine:Z^2<H}\\
			1-Z(W)≥(1-H(W))㏑2.	
		\end{gather*}
	\end{lem}
	
	For proofs, see \cite[Corollary~5]{JA18};
	note that the last two inequalities are specializations
	of $ϕ(Z(W))≤H(W)$ for a smooth function $ϕ(z)≔h₂((1-√{1-z²})/2)$.
	For a visualization of the region where $(H(W),Z(W))$
	could possibly be, see \cref{fig:ZvsH}.
	
	\begin{figure}
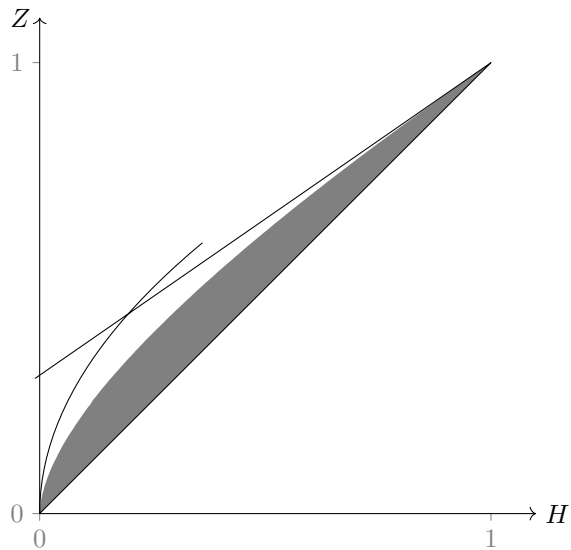

		\tikz[scale=6]{
			\draw
				(0,1)pic{y-tick=$1$}(0,0)pic{y-tick=$0$}
				(0,0)pic{x-tick=$0$}(1,0)pic{x-tick=$1$};
			\fill[gray]
				plot[domain=0:45,samples=180]({h2o(\x)},{sin(2*\x)})--cycle;
			\draw
				plot[domain=.6:0,samples=25](\x^2,\x)--
				plot[domain=1:.3,samples=3]({\x*1.4427-0.4427},\x);
			\draw[->](0,0)--(0,1.1)node[left]{$Z$};
			\draw[->](0,0)--(1.1,0)node[right]{$H$};
		}
		\caption{
			The possible region where $(H(W),Z(W))$ could lie in.
			The dark region is the exact one, whereas the outer boundaries
			are loosened to two pieces so that they are easier to
			describe---they are a parabola and a line of slope $㏑2$.
		}\label{fig:ZvsH}
	\end{figure}
	
	Remark:
	Imagine that we write down the parameter processes in an infinite array
	\[\bma{
		𝘏₀	&	𝘏₁	&	𝘏₂	&	𝘏₃	&	⋯\quad{}	\\
		𝘡₀	&	𝘡₁	&	𝘡₂	&	𝘡₃	&	⋯\quad{}
	}\]
	then \cref{pro:martin,pro:superm,lem:squares} are some horizontal relations,
	and \cref{lem:ZvsH} is some vertical relations.
	These lemmas help us predict where the processes are going.
	For example, if we happen to know $𝘏_n→0$,
	then $𝘡_n→0$ because \cref{ine:Z^2<H} says $𝘡_n≤√{𝘏_n}$.
	I call this the \emph{common-fate property}.
	
	The next lemma justifies why we want to predict the processes---%
	because it helps us evaluate code performance.
	
	\begin{lem}\label{lem:R=PZ}
		Fix an $n$.
		Declare a code by letting $𝒥$ collect
		synthetic channels with $Z$ less than the threshold $θ$.
		Then the code rate $R$ is $𝘗\{𝘡_n<θ\}$.
	\end{lem}
	
	\begin{proof}
		Recall that $R=\abs{𝒥}/2^n$, where $\abs{𝒥}$ is
		the number of depth-$n$ synthetic channels of the form
		$\(\dotsb((WＷ{j₁})Ｗ{j₂})\dotsb\)Ｗ{j_n}$ in $𝒥$.
		Since $𝘞_n$ assumes each depth-$n$ channel with probability $1/2^n$,
		the code rate $R$ is the probability that $𝘞_n$ is in $𝒥$.
		This quantity, by the definition of $𝒥$, is the probability that $Z(𝘞_n)<θ$;
		and $Z(𝘞_n)$ is just $𝘡_n$.
		This finishes the proof.
	\end{proof}
	Recap:
	When declaring a code by letting $𝒥$ collect synthetic channels with $Z$
	less than $θ$, the block error probability has an upper bound $Ｐ<Nθ$,
	and the code rate has an easy expression $R=𝘗\{𝘡_n<θ\}$.
	In summary, the following formula depicts the trade-off between $Ｐ$ and $R$:
	\[𝘗\{𝘡_n<Ｐ/N\}≈R.\]
	To rephrase it, we are interested in the cdf of $𝘡_n$,
	especially how close it is to the $y$-axis.
	
	Our end-of-chapter goal is to characterize the pairs $(π,ρ)∈[0,1]²$ such that
	\[𝘗｛𝘡_n<e^{-2^{πn}}｝>I(W)-2^{-ρn}.\]
	That immediately implies the existence of codes with $Ｐ$ on the order of
	$\exp(-2^{πn})$ and the gap to capacity $I(W)-R$ on the order of $2^{-ρn}$.
	In the next section, we approach this goal via first showing
	\[𝘗｛𝘡_n<e^{-n^{2/3}}｝>I(W)-2^{-ϱn+o(n)}\label{ine:Z-en23}\]
	for some $ϱ>0$, where $1/ϱ$ is sometimes called the \emph{scaling exponent}.

\section{Scaling Exponent Regime}

	As far as I can tell, the only way to show \cref{ine:Z-en23}
	is through the eigen behavior of $𝘡_n$.
	More precisely, I will first declare a concave function $h：[0,1]→[0,1]$
	and estimate the supremum
	\[2^{-ϱ}≔\sup_{W：†SBDMC†}÷{h(Z(WＷ1))+h(Z(WＷ2))}{2h(Z(W))}.\]
	This is called the \emph{eigen behavior} of $𝘡_n$.
	From that we can infer $𝘌[𝘡_n]≤𝘡₀2^{-ϱ}$,
	and then $𝘗\{\exp(-n^{2/3})≤𝘡_n≤1-\exp(-n^{2/3})\}<2^{-ϱn+o(n)}$,
	followed by $𝘗\{𝘡_n→0\}=I(W)$,
	and finally the \emph{en23 behavior} $𝘗\{𝘡_n<\exp(-n^{2/3})\}>I(W)-2^{-ϱn+o(n)}$.
	
	Let us walk through a toy example before we dive into the general case.
	Consider any binary erasure channel (BEC) $W$.
	Then \cref{ine:2Z-Z^2} assumes equality;
	that is, $Z(WＷ1)=2Z(W)-Z(W)²$.
	Declare an “eigenfunction” $h(z)≔√{z(1-z)}$.
	Then
	\[\sup_{W：†BEC†}÷{h(Z(WＷ1))+h(Z(WＷ2))}{2h(Z(W))}
		=\sup_{0<z<1}÷{h(2z-z²)+h(z²)}{2h(z)}=÷{√3}2.\label{sup:bec}\]
	This means that $√3/2$ is the “eigenvalue” corresponding to $h$,
	hence the name eigen behavior.
	
	We now deduce that $𝘌[h(𝘡_{n+1})]=𝘌[𝘌[h(𝘡_{n+1})｜𝘡_n]]≤𝘌[h(𝘡_n)√3/2]$.
	Applying this iteratively, we arrive at $𝘌[h(𝘡_n)]≤h(𝘡₀)(√3/2)^n$.
	From there we further deduce that, by Markov's inequality,
	\begin{align*}
		\qquad&\kern-2em
		𝘗｛e^{-n^{2/3}}≤𝘡_n≤1-e^{-n^{2/3}}｝=𝘗｛h(𝘡_n)≥h\(e^{-n^{2/3}}\)｝	\\
		&	≤÷{𝘌[h(𝘡_n)]}{h(\exp(-n^{2/3}))}≤÷{h(𝘡₀)(√3/2)^n}{h(\exp(-n^{2/3}))}
			<÷{(√3/2)^n}{\exp(-n^{2/3})}<（÷{√3}2）^{n-o(n)}.
			\label{ine:bec-expel}
	\end{align*}
	What we see here is that $𝘡_n$ refuses to stay around the middle of
	the interval $[0,1]$, and we can quantify how unwilling $𝘡_n$ is.
	
	Next, we run the following analysis nonsense to derive that $𝘗\{𝘡_n→0\}=I(W)$:
	Since $𝘌[h(𝘡_n)]$ decays exponentially fast in $n$
	while $h$ is bounded, $h(𝘡_n)→0$ almost surely.
	This implies that $𝘡_n$ might
	(a)	converge to $0$,
	(b)	converge to $1$, or
	(c)	jump back and forth between $0$ and $1$.
	The last one cannot happen to a supermartingale,
	so there exists a random variable $𝘡_∞$ such that $𝘡_n→𝘡_∞∈\{0,1\}$.
	By \cref{lem:ZvsH}, $\{𝘏_n\}$ obeys the same law---there exists
	a random variable $𝘏_∞$ such that $𝘏_n→𝘏_∞∈\{0,1\}$.
	Now $𝘗\{𝘡_n→1\}=𝘗\{𝘏_n→1\}=𝘗\{𝘏_∞=1\}=𝘌[H_∞]=𝘏₀=H(W)$,
	so the complement probability is $𝘗\{𝘡_n→0\}=I(W)$.
	
	\Cref{ine:bec-expel} and $𝘗\{𝘡_n→0\}=I(W)$ are what we need to derive
	the en23 behavior of $𝘡_n$---the former expels $𝘡_n$ to the ends of the interval
	$[0,1]$, and the latter predicates how much $𝘡_n$ goes to which end.
	In detail, consider the bad event $𝘉_n≔\{𝘡_n≥\exp(-n^{2/3})$ but $𝘡_m→0\}$.
	This event collects the samples where $𝘡_n$
	used to be moderate or bad but turns out to be good.
	By \cref{lem:squares}, there must be some $m≥n$ such that $𝘡_m$
	“visits the middle”, i.e., $\exp(-m^{2/3})≤𝘡_m≤1-\exp(-m^{2/3})$.
	But since the upper bound by \cref{ine:bec-expel} is a geometric series in $m$,
	the probability that some $𝘡_m$ (for some $m≥n$)
	visit the middle is at most $(√3/2)^{n-o(n)}$.
	We conclude that $𝘗\{𝘡_n<\exp(-n^{2/3})\}≥𝘗\{𝘡→0\}-𝘗(𝘉_n)>I(W)-(√3/2)^{n-o(n)}$.
	
	Remark on the last four paragraphs:
	We used an eigen-pair $(h,√3/2)$ to quantify the pace of decay of $𝘌[h(𝘡_n)]$
	and further computed how much $𝘡_n$ visits the middle or close to $0$.
	The general SBDMC version of this argument comes with two modifications.
	First, the eigenvalue $√3/2$ could be improved if we use a better eigenfunction $h$.
	Second, $Z(WＷ1)$ is not exactly $2Z(W)-Z(W)²$
	but can be as low as $Z(W)√{2-Z(W)²}$ (\cref{lem:squares}).
	This has to be taken into consideration
	for the generalization of \cref{sup:bec}.
	
	The remainder of this section deals with the general SBDMC case.
	
	\begin{thm}[SBDMC scaling exponent]\label{thm:bdmc-mu}
		Assume SBDMCs.
		There exists a concave function $h:[0,1]→[0,1]$ such that $h(0)=h(1)=0$ and
		\[\sup_{W：†SBDMC†}÷{h(Z(WＷ1))+h(Z(WＷ2))}{2h(Z(W))}>2^{-1/4.714}.\]
	\end{thm}
	
	\begin{proof}
		Thanks to \cref{lem:squares}, it suffices to find a function $h$ that minimizes
		\[\sup_{0<z<1}\sup_{z√{2-z²}≤z'≤2z-z²}÷{h(z')+h(z²)}{2h(z)}.
			\label{sup:sbdmc}\]
		It is not known yet if there is an analytic form for such $h$, but
		numerical computation in \cite[Theorem~2]{MHU16} suggests that $2^{-1/4.717}$
		is achievable by some spline $h$, and I will stick to this number.
		
		Note one:
		By a compactness argument, \cref{sup:sbdmc} is strictly less than $1$.
		As long as it is less than $1$, the arguments in the remainder
		of this section apply with $2^{-1/4.717}$ replaced by the weaker supremum.
		Therefore, readers should not worry about
		whether $2^{-1/4.714}$ is mathematically sound.
		
		Note two:
		Despite of potential rounding errors, there is another reason
		why I think $2^{-1/4.714}$ is not the final value.
		Recall that we took the supremum over $z√{2-z²}≤z'≤2z-z²$;
		this is a pessimistic estimate and chances are that
		we will have tighter inequalities to bound $Z(WＷ1)$.
	\end{proof}
	
	Now, let us start deriving the en23 behavior of SBDMCs.
	
	\begin{lem}[From eigen to en23]\label{lem:Z-en23}
		Fix $ϱ≔1/4.714$.
		Assume \cref{thm:bdmc-mu}.
		Then
		\[𝘗｛𝘡_n<e^{-n^{2/3}}｝>I(W)-2^{-ϱn+o(n)}.\tagcopy{ine:Z-en23}\]
	\end{lem}
	
	\begin{proof}
		The proof was sketched when we walked through the toy example.
		First, \cref{thm:bdmc-mu} yields that $𝘌[h(𝘡_{n+1})]≤𝘌[h(𝘡_n)2^{-ϱ}]$.
		Telescoping, we obtain that $𝘌[h(𝘡_n)]≤𝘡₀2^{-ϱn}<2^{-ϱn+o(n)}$.
		By Markov's inequality, we see
		\[𝘗｛e^{-n^{2/3}}≤𝘡_n≤1-e^{-n^{2/3}}｝
			≤÷{𝘌[h(𝘡_n)]}{h(\exp(-n^{2/3}))}<2^{-ϱn+o(n)}.\label{ine:bdmc-expel}\]
		
		Next, we recall why $𝘗\{𝘡_n→0\}=I(W)$:
		By that $\{h(𝘡_n)\}$ is a bounded supermartingale and decays by a constant
		factor every time $n$ increases, it converges to $0$ almost surely.
		Since $h$ is concave, $h(0)$ and $h(1)$ are
		the only places that evaluate to $0$, which means that
		$𝘡_n$ is getting closer and closer to either $0$ or $1$.
		By \cref{lem:squares}, or that $\{𝘡_n\}$ is a supermartingale,
		it cannot jump from the neighborhood of $0$ to the neighborhood of $1$,
		so each realization of $\{𝘡_n\}$ must choose
		either $0$ or $1$ and converge to it.
		By \cref{lem:ZvsH}, $\{𝘏_n\}$ must converge,
		and is converging to the same limit $\{𝘡_n\}$ does.
		But as $\{𝘏_n\}$ is a bounded martingale,
		we know that $𝘌[\lim_{n→∞}𝘏_n]=\lim_{n→∞}𝘌[𝘏_n]=𝘏₀=H(W)$.
		Hence $H(W)$ is the probability that $\{𝘏_n\}$ and $\{𝘡_n\}$ converge to $1$.
		Its complement is that $\{𝘏_n\}$ and $\{𝘡_n\}$
		converge to $0$ with probability $1-H(W)=I(W)$.
		
		Lastly, I explain why $𝘗\{𝘡_n<\exp(-n^{2/3})\}>𝘗\{𝘡_n→0\}-2^{-ϱn+o(n)}$:
		In general, we would like to believe that if a realization of $\{𝘡_n\}$
		converges to $0$, its prefix (the first few terms) would be rather small.
		But there are exceptions:
		Let $𝘉_n$ be the exceptional event $\{𝘡_m→0$ but $𝘡_n≥\exp(-n^{2/3})\}$.
		We would like to estimate $𝘗(𝘉_n)$.
		To do so, realize that if $𝘡_n≥\exp(-n^{2/3})$, then either
		\begin{itemize}
			\item	$\exp(-n^{2/3})≤𝘡_n≤1-\exp(-n^{2/3})$, or
			\item	$1-\exp(-n^{2/3})≤𝘡_n$ and $𝘡_m$ will visit the closed interval	\\
					$[\exp(-m^{2/3}),\exp(-m^{2/3})]$ for some later $m>n$.
		\end{itemize}
		Either case, $𝘡_n$ or its descendants will step in
		the left-hand side of \cref{ine:bdmc-expel}.
		Sum \cref{ine:bdmc-expel} over $m≥n$ and apply the union bound over $m≥n$;
		we get an upper bound $𝘗(𝘉_n)<∑_{m≥n}2^{-ϱm+o(m)}<2^{-ϱn+o(n)}$.
		Consequently, $𝘗\{𝘡_n<\exp(-n^{2/3})\}≥𝘗\{𝘡_n→0\}-𝘗(𝘉_n)>I(W)-2^{-ϱn+o(n)}$.
		That closes the proof.
	\end{proof}
	
	So far I have proved $𝘗\{𝘡_n<\exp(-n^{2/3})\}>I(W)-2^{-n/4.714+o(n)}$.
	In the next section, I will prove the same inequality
	with $\exp\(-e^{n^{1/3}}\)$ in place of $\exp(-n^{2/3})$.
	After that, I will prove the final goal---$𝘗\{𝘡_n<\exp(-ℓ^{πn})\}>I(W)-2^{-ρn+o(n)}$
	for some pairs $(π,ρ)$ lying in the shaded area in \cref{fig:lol-bdmc}.

\section{Stepping Stone Regime}

	In this section, we will verify the \emph{een13 behavior}
	$𝘗｛𝘡_n<\exp\(-e^{n^{1/3}}\)｝>I(W)-2^{-n/4.714+o(n)}$
	on top of the en23 behavior proved in the last section.
	The idea behind the proof is to keep track of how many $𝘡_m$ are
	moderately small, i.e., $𝘡_m<\exp(-m^{2/3})$ and how many descendants $𝘡_n$
	thereof become even smaller, i.e., $𝘡_n<\exp\(-e^{n^{1/3}}\)$ for some $n>m$.
	
	To understand the idea better, pretend that we have a $𝘡_n$
	that is moderately small---about $\exp(-n^{2/3})$ small.
	Then the punchline here is that squaring ($𝘡_n²$)
	will scale $𝘡_n$ down rapidly, while doubling ($2𝘡_n-𝘡_n²$)
	barely does anything to the order of magnitude of $𝘡_n$.
	So the problem boils down to counting how many times
	a trajectory of $\{𝘡_n\}$ undergoes the squaring branches.
	This number obeys a binomial distribution whose limiting behavior is well-known.
	
	Before the actual proof, let me walk through
	a tentative strategy to demonstrate what could go wrong.
	Let $m<n$ be two large numbers, then \cref{lem:Z-en23} yields
	\[𝘗｛𝘡_m<e^{-m^{2/3}}｝>I(W)-2^{-ϱm+o(m)}.\label{ine:try-em23}\]
	In order to end up with $𝘡_n≈\exp\(-e^{n^{1/3}}\)$ from $𝘡_m≈\exp(-m^{2/3})$, it
	requires, among the remaining $n-m$ Bernoulli trials, $n^{1/3}$ squaring branches.
	By Hoeffding's inequality, it would not meet the requirement with probability
	\[\exp（-Ω（÷{n-m}2-n^{1/3}））.\label{ine:hoeffding}\]
	Now we see the dilemma:
	If $m$ is too small compared to $n$, i.e., $m<n-Ω(n)$,
	the right-hand side of \cref{ine:try-em23} is $I(W)-2^{-ϱn+Ω(n)}$,
	which is too far away from the expected code rate $I(W)-2^{-ϱn+o(n)}$.
	If, otherwise, $m$ is comparable to $n$, i.e., $m=n-o(n)$,
	the right-hand side of \cref{ine:hoeffding} is $\exp(-o(n))$,
	which exceeds the expected gap to capacity $2^{-ϱn+o(n)}$.
	
	The preceding hand-waving argument demonstrates that
	no $m$ can settle the argument down once and for all.
	So the second punchline is to use multiple $m$'s.
	In my case, I choose $m=√n,2√n…n-√n$ for some perfect square $n$
	to apply Hoeffding's inequality $√n$ times.
	There are flexibilities in choosing $m$'s;
	for instance, when $n$ is not a perfect square, using $m=⌈√n⌉,2⌈√n⌉…⌊n/⌈√n⌉⌋⌈√n⌉$
	would not alter the proof up to some little-$o$ terms.
	So let us presume that $n$ is always a perfect square.
	
	That could be the end of the proof if it were not for the falling of the assumption:
	For an $m$, if the first few branches after $𝘡_m$ are doubling it,
	$𝘡_m$ will soon become so large that the $2$ in $2𝘡_m$ is not negligible.
	For this concern, Hoeffding's inequality does not help---%
	it does not control whether the $n^{1/3}$ squaring branches
	take place within the last few branches or spread out evenly.
	To resolve that, we need to keep an eye on the entire trajectory $𝘡_m,𝘡_{m+1}…𝘡_n$
	to make sure that it stays in the range where doubling is negligible.
	
	The actual proof provided below aims to mimic the tentative argument for multiple
	$m$'s at once while resolving the issue that doubling too much could break things.
	I will keep monitoring two conditions---whether a trajectory of $𝘡_n$
	undergoes sufficiently many squaring branches and whether
	that trajectory of $𝘡_n$ stays low enough such that doubling is negligible.
	
	\begin{lem}[From en23 to een13]\label{lem:Z-een13}
		Given \cref{lem:Z-en23}, that is, given
		\[𝘗｛𝘡_n<e^{-n^{2/3}}｝>I(W)-2^{-ϱn+o(n)},\]
		we have
		\[𝘗｛𝘡_n<\exp\(-e^{n^{1/3}}\)｝>I(W)-2^{-ϱn+o(n)}.\label{ine:Z-een13}\]
	\end{lem}
	
	\begin{proof}
		(Select constants.)
		Consider the stochastic process $\{19𝘑_n/20\}$.
		Since $2z-z²≤2z≤z^{19/20}$ whenever $z<2^{-20}$,
		we have $𝘡_{n+1}≤𝘡_n^{19𝘑_{n+1}/20}$ whenever $𝘡_n<2^{-20}$.
		Also notice that, numerically, $𝘌[𝘑_n^{-1}]6^{1/20}≈0.8202<2^{-1/3.6}$.
		
		(Define events.)
		Let $n$ be a perfect square.
		Let $𝘌₀⁰$ be the empty event.
		For every $m=√n,2√n…n-√n$, we define five series of events
		$𝘈_m$, $𝘉_m$, $𝘊_m$, $𝘌_m$, and $𝘌₀^m$ inductively as below:
		Let $𝘈_m$ be $\{𝘡_m<\exp(-m^{2/3})\}、𝘌₀^{m-√n}$.
		Let $𝘉_m$ be a subevent of $𝘈_m$ where $𝘡_l≥2^{-20}$ for some $l≥m$.
		Let $𝘊_m$ a subevent of $𝘈_m$ where
		\[𝘑_{m+1}𝘑_{m+2}\dotsm𝘑_{m+√n}≤6^{√n/20}.\label{ine:mgf-6n20}\]
		Let $𝘌_m$ be $𝘈_m、(𝘉_m∪𝘊_m)$.
		Let $𝘌₀^m$ be $𝘌₀^{m-√n}∪𝘌_m$.
		Let $𝘢_m$, $𝘣_m$, $𝘤_m$, $𝘦_m$, and $𝘦₀^m$
		be the probability measures of the corresponding capital letter events.
		Moreover, let $𝘨_m$ be $I(W)-𝘦₀^m$.
		
		(Bound $𝘣_m/𝘢_m$ from above.)
		Conditioning on $𝘈_m$, we want to estimate
		the probability that $𝘡_l≥2^{-20}$ for some $l≥m$.
		Recall that $\{𝘡_l\}$ is a supermartingale.
		Hence by Ville's inequality \cite[Exercise~4.8.2]{Durrett19},
		$𝘗\{𝘡_l≥2^{-20}$ for some $l≥m｜𝘈_m\}≤2^{20}𝘡_m<2^{20}\exp(-m^{2/3})$.
		This is an upper bound on $𝘣_m/𝘢_m$
		and will be summoned in \cref{ine:engine-een13}.
		
		(Bound $𝘤_m/𝘢_m$ from above.)
		We want to estimate how often \cref{ine:mgf-6n20} happens.
		That is the probability that $(𝘑_{m+1}𝘑_{m+2}\dotsm𝘑_{m+√n})^{-1}≥6^{-√n/20}$.
		This probability cannot exceed
		$𝘌[(𝘑_{m+1}𝘑_{m+2}\dotsm𝘑_{m+√n})^{-1}]6^{√n/20}
			=𝘌[𝘑₁^{-1}]^{√n}6^{√n/20}=(𝘌[𝘑₁^{-1}]6^{1/20})^{√n}<2^{-√n/3.6}$
		by Markov's inequality.
		This is an upper bound on $𝘤_m/𝘢_m$
		and will be summoned in \cref{ine:engine-een13}.
		
		(Bound $(𝘨_{m-√n}-𝘢_m)^+$ from above.)
		By definitions, $𝘨_{m-√n}-𝘢_m=I(W)-(𝘦₀^{m-√n}+𝘢_m)$.
		The definition of $𝘈_m$ forces it to be disjoint from $𝘌₀^{m-√n}$,
		thus $𝘦₀^{m-√n}+𝘢_m$ is the probability measure of $𝘌₀^{m-√n}∪𝘈_m$.
		This union event must contain the event
		$\{𝘡_m<\exp(-m^{2/3})\}$ by how $𝘈_m$ was defined.
		Recall the en23 behavior $𝘗\{𝘡_m<\exp(-m^{2/3})\}>I(W)-ℓ^{-ϱm+o(m)}$.
		Chaining all inequalities together, we deduce $𝘨_{m-√n}-𝘢_m<ℓ^{-ϱm+o(m)}$.
		Let $(𝘨_{m-√n}-𝘢_m)^+$ be $\max(0,𝘨_{m-√n}-𝘢_m)$
		so we can write $(𝘨_{m-√n}-𝘢_m)^+<ℓ^{-ϱm+o(m)}$.
		This upper bound will be summoned in \cref{ine:engine-een13}.
		
		(Bound $𝘦₀^{n-√n}$ from below.)
		We start rewriting $𝘨_m$ with $𝘨_m^+$ being $\max(0,𝘨_m)$:
		\begin{align*}
			𝘨_m
			&	=I(W)-𝘦₀^m=I(W)-(𝘦₀^{m-√n}+𝘦_m)=I(W)-𝘦₀^{m-√n}-𝘦_m	\\
			&	=𝘨_{m-√n}-𝘦_m=𝘨_{m-√n}（1-÷{𝘦_m}{𝘢_m}）+÷{𝘦_m}{𝘢_m}(𝘨_{m-√n}-𝘢_m)	\\
			&	≤𝘨_{m-√n}^+（1-÷{𝘦_m}{𝘢_m}）+÷{𝘦_m}{𝘢_m}(𝘨_{m-√n}-𝘢_m)^+	\\
			&	≤𝘨_{m-√n}^+（1-÷{𝘦_m}{𝘢_m}）+(𝘨_{m-√n}-𝘢_m)^+	\\
			&	≤𝘨_{m-√n}^+（÷{𝘣_m}{𝘢_m}+÷{𝘤_m}{𝘢_m}）+(𝘨_{m-√n}-𝘢_m)^+	\\
			&	<𝘨_{m-√n}^+（2^{20}e^{-m^{2/3}}+ℓ^{-√n/3.6}）+ℓ^{-ϱm+o(m)}.
				\label{ine:engine-een13}
		\end{align*}
		The first four equalities are by the definitions of $𝘨_m$ and $𝘌₀^m$.
		The next equality is simple algebra.
		The next two inequalities are by $0≤𝘦_m/𝘢_m≤1$.
		The next inequality is by the definition of $𝘌_m$.
		The last inequality summons upper bounds derived in the last three paragraphs.
		The last line contains two terms in the big parentheses;
		between them, $2^{-√n/3.6}$ dominates $2^{20}\exp(-m^{2/3})$
		once $m$ is greater than $O(n^{3/4})$.
		Subsequently, we obtain this recurrence relation:
		\[\cas{
			𝘨_{O(n^{3/4})}≤1,	\\
			𝘨_m≤2𝘨_{m-√n}^+ℓ^{-√n/3.6}+ℓ^{-ϱm+o(m)}.
		}\]
		Solve it (cf.\ the master theorem);
		we get that $𝘨_{n-√n}<ℓ^{-ϱn+o(n)}$.
		By the definition of $𝘨_{n-√n}$,
		we immediately get $𝘦₀^{n-√n}>I(W)-ℓ^{-ϱn+o(n)}$.
		
		(Analyze $𝘌₀^{n-√n}$.)
		We want to estimate $𝘏_n$ when $𝘌₀^{n-√n}$ happens.
		To be precise, for each $m=√n,2√n…n-√n$,
		we attempt to bound $𝘡_{m+√n}$ when $𝘌_m$ happens.
		Fix an $m$.
		When $𝘌_m$ happens, its superevent $𝘈_m$ happens,
		so we know that $𝘡_m<\exp(-m^{2/3})$.
		But $𝘉_m$ does not happen, so $𝘡_l<2^{-20}$ for all $l≥m$.
		This implies that $𝘡_{l+1}≤𝘡_l^{19𝘑_{l+1}/20}$ for those $l$.
		Telescope;
		$𝘡_{m+√n}$ is less than or equal to $𝘡_m$ raised to
		the power of $𝘑_{m+1}𝘑_{m+2}\dotsm𝘑_{m+√n}(19/20)^{√n}$.
		But $𝘊_m$ does not happen, so the product is greater than
		$6^{√n/20}(19/20)^{√n}=(6(19/20)^{20})^{√n/20}>2^{√n/20}$.
		Jointly we have $𝘡_{m+√n}≤𝘡_m^{2^{√n/20}}<\exp(-m^{2/3}2^{√n/20})$.
		Recall that $𝘡_{l+1}≤2𝘡_l$ for all $l≥m+√n$.
		Then telescope again;
		$𝘡_n≤2^{n-m-√n}𝘡_{m+√n}<2^n\exp(-m^{2/3}2^{√n/20})<\exp\(-e^{n^{1/3}}\)$
		provided that $n$ is sufficiently large.
		In other words, $𝘌₀^{n-√n}$ implies $𝘡_n<\exp\(-e^{n^{1/3}}\)$.
		
		(Summary.)
		Now we may conclude, for all perfect squares $n$, that
		$𝘗｛𝘡_n<\exp\(-e^{n^{1/3}}\)｝≥𝘗(𝘌₀^{n-√n})=𝘦₀^n>I(W)-ℓ^{-ϱn+o(n)}$.
		For non-squares, round $n$ down to the nearest square
		and rerun the whole argument above.
		We will get $𝘡_n<2^n\exp(-m^{2/3}2^{⌊√n⌋/20})$
		with probability  at least $I(W)-ℓ^{-ϱ⌊√n⌋²+o(n)}$,
		which still leads to $𝘗｛𝘡_n<\exp\(-e^{n^{1/3}}\)｝>I(W)-ℓ^{-ϱn+o(n)}$.
		And hence the proof of the een13 behavior,
		\cref{ine:Z-een13}, is sound for all $n$.
	\end{proof}
	
	This section is parallel to \cite[section~V]{LargeDeviations18},
	to \cite[appendix~C.C]{Hypotenuse19}, and to \cite[section~10.2]{GRY19}.
	Do not confuse this section with the next section.
	The subtlety is explained in \cite[section~III]{LargeDeviations18}.
	
	Now we know $𝘗｛𝘡_n<\exp\(-e^{n^{1/3}}\)｝>I(W)-2^{-ϱn+o(n)}$.
	We are ready to learn what $(π,ρ)$ pairs satisfy
	$𝘗｛𝘡_n<\exp\(-2^{πn}\)｝>I(W)-2^{-ρn+o(n)}$.

\section{Moderate Deviations Regime}

	I will build upon the een13 behavior and utilize a technique
	similar to before to answer the following main question:
	Knowing $𝘗\{𝘡_n<\exp(-2^{0.499n})\}>I(W)-1/999$ \cite{AT09}
	and $𝘗\{𝘡_n<1/999\}>I(W)-2^{-ϱn+o(n)}$ \cite{MHU16},
	can we find a interpolating result between these two results?
	This section, finally, offers an answer by characterizing the region
	of pairs $(π,ρ)$ that satisfy $𝘗｛𝘡_n<\exp\(-2^{πn}\)｝>I(W)-2^{-ρn+o(n)}$.
	
	Recall $ϱ≔1/4.714$.
	Let $h₂(p)≔-p㏒₂p-(1-p)㏒₂(1-p)$ be the binary entropy function.
	Let $𝒪⊆[0,1/2]×[0,ϱ]$ be an open region defined by the following criterion:
	for any $(π,ρ)∈𝒪$, the ray shooting from $(π,ρ)$ toward the opposite direction
	of $(0,ϱ)$ does not intersect the function graph of $1-h₂$.
	See \cref{fig:lol-bdmc};
	this criterion is equivalent to that $(π,ρ)$ lies to the left
	of the convex envelope of $(0,ϱ)$ and $1-h₂$.
	This criterion is also equivalent to
	\[1-h₂（÷{πn}{n-m}）>÷{ρn-ϱm}{n-m}\label{ine:ray-raw}\]
	for all $0<m<n$.
	The last criterion is what will be used in the proof.
	
	\begin{thm}[From een13 to e2pin]\label{thm:Z-e2pin}
		Fix a pair $(π,ρ)∈𝒪$.
		Given the conclusion of \cref{lem:Z-een13}, that is, given
		\[𝘗｛𝘡_n<\exp\(-e^{n^{1/3}}\)｝>I(W)-2^{-ϱn+o(n)},\tagcopy{ine:Z-een13}\]
		then
		\[𝘗｛𝘡_n<e^{-2^{πn}}｝>I(W)-2^{-ρn+o(n)}.\label{ine:Z-e2pin}\]
	\end{thm}
	
	\begin{proof}
		(Select constants.)
		Since \cref{ine:ray-raw} holds, there exists a small constant $ε>0$ such that
		\[1-h₂（÷{πn}{n-m}+2ε）>÷{ρn-ϱm}{n-m}\label{ine:ray-gap}\]
		by the compactness argument.
		Fix this $ε$.
		There exists a small constant $δ>0$ such that
		$𝘡_{n+1}≤𝘡_n^{𝘑_{n+1}(1-ε)}$ whenever $𝘡_n<δ$.
		
		(Define events.)
		Let $n$ be a perfect square.
		Let $𝘈₀⁰$ and $𝘌₀⁰$ be the empty event.
		For every $m=√n,2√n…n-√n$, we define six series of events
		$𝘈_m$, $𝘈₀^m$, $𝘉_m$, $𝘊_m$, $𝘌_m$, and $𝘌₀^m$
		inductively as follows:
		Let $𝘈_m$ be $｛𝘡_m<\exp(-e^{m^{1/3}})｝、𝘈₀^{m-√n}$.
		Let $𝘈₀^m$ be $𝘈₀^{m-√n}∪𝘈_m$.
		Let $𝘉_m$ be a subevent of $𝘈_m$ where $𝘡_l≥δ$ for some $l≥m$.
		Let $𝘊_m$ a subevent of $𝘈_m$ where
		\[𝘑_{m+1}𝘑_{m+2}\dotsm𝘑_n≤2^{πn+2ε(n-m)}.\label{ine:mgf-2pin}\]
		Let $𝘌_m$ be $𝘈_m、(𝘉_m∪𝘊_m)$.
		Let $𝘌₀^m$ be $𝘌₀^{m-√n}∪𝘌_m$.
		Let $𝘢_m$, $𝘢₀^m$, $𝘣_m$, $𝘤_m$, $𝘦_m$, and $𝘦₀^m$
		be the probability measures of the corresponding capital letter events.
		Moreover, let $𝘧_m$ be $I(W)-𝘢₀^m$ and let $𝘨_m$ be $I(W)-𝘦₀^m$.
		
		(Bound $𝘣_m/𝘢_m$ from above.)
		Conditioning on $𝘈_m$, we want to estimate
		the probability that $𝘡_l≥δ$ for some $l≥m$.
		Recall that $𝘡_l$ is a supermartingale.
		Hence by Ville's inequality (\cite[Exercise~4.8.2]{Durrett19}),
		$𝘗\{𝘡_l≥δ$ for some $l≥m｜𝘈_m\}≤𝘡_m/δ<\exp\(-e^{m^{1/3}}\)/δ$.
		This is an upper bound on $𝘣_m/𝘢_m$
		and will be summoned in \cref{ine:engine-e2pin}.
		
		(Bound $𝘤_m/𝘢_m$ from above.)
		We want to estimate how often \cref{ine:mgf-2pin} happens.
		This is equivalent to asking how often do $n-m$
		fair coin tosses end up with $πn+2ε(n-m)$ heads.
		By the large deviations theory,
		this probability is less than $2$ to the power of
		\[-(n-m)（1-h₂（÷{πn}{n-m}+2ε））.\]
		By \cref{ine:ray-gap}, this exponent is less than $ϱm-ρn$.
		Thus, the probability is less than $2^{ϱm-ρn}$.
		This is an upper bound on $𝘤_m/𝘢_m$
		and will be summoned in \cref{ine:engine-e2pin}.
		
		(Bound $𝘧_m^+$ from above.)
		The definition of $𝘧_m$ reads $I(W)-𝘢₀^m$.
		Here $𝘢₀^m$ is the probability measure of $𝘈₀^m$,
		and $𝘈₀^m$ is a superevent of $𝘈_m$ by how the former is defined.
		Event $𝘈₀^m$ must contain $｛𝘡_m<\exp\(-e^{m^{1/3}}\)｝$
		by how $𝘈_m$ was defined.
		By the een13 behavior, $𝘗｛𝘡_m<\exp\(-e^{m^{1/3}}\)｝>I(W)-ℓ^{-ϱm+o(m)}$.
		Chaining all inequalities together, we infer that $𝘧_m<ℓ^{-ϱm+o(m)}$.
		Let $𝘧_m^+$ be $\max(0,𝘧_m)$ so
		we can write $𝘧_m^+<ℓ^{-ϱm+o(m)}$.
		This upper bound will be summoned in \cref{ine:engine-e2pin}.
		
		(Bound $𝘦₀^{n-√n}$ from below.)
		We start rewriting $𝘨_m-𝘧_m^+$ with
		$(𝘧_{m-√n}-𝘢_m)^+$ being $\max(0,𝘧_{m-√n}-𝘢_m)$:
		\begin{align*}
			𝘨_m-𝘧_m^+
			&	=I(W)-𝘦₀^m-(I(W)-𝘢₀^m)^+	\\
			&	=I(W)-𝘦₀^{m-√n}-𝘦_m-(I(W)-𝘢₀^{m-√n}-𝘢_m)^+	\\
			&	=𝘨_{m-√n}-𝘦_m-(𝘧_{m-√n}-𝘢_m)^+	\\
			&	≤𝘨_{m-√n}-𝘦_m-÷{𝘦_m}{𝘢_m}(𝘧_{m-√n}-𝘢_m)^+	\\
			&	≤𝘨_{m-√n}-𝘦_m-÷{𝘦_m}{𝘢_m}(𝘧_{m-√n}^+-𝘢_m)	\\
			&	=𝘨_{m-√n}-𝘧_{m-√n}^++𝘧_{m-√n}^+（1-÷{𝘦_m}{𝘢_m}）	\\
			&	≤𝘨_{m-√n}-𝘧_{m-√n}^++𝘧_{m-√n}^+（÷{𝘣_m}{𝘢_m}+÷{𝘤_m}{𝘢_m}）	\\
			&	<𝘨_{m-√n}-𝘧_{m-√n}^++ℓ^{-ϱ(m-√n)+o(m-√n)}
				（\exp\(-e^{m^{1/3}}\)/δ+2^{ϱm-ρn}）.\label{ine:engine-e2pin}
		\end{align*}
		The first three equalities are by the definitions of $𝘨_m$ and $𝘧_m$.
		The next inequality is by $0≤𝘦_m/𝘢_m≤1$.
		The next inequality is by $\max(0,f-a)=\max(a,f)-a≥\max(0,f)-a$.
		The next equality is simple algebra.
		The next inequality is by the definition of $𝘌_m$.
		The last inequality summons upper bounds derived in the last three paragraphs.
		Now the last line contains two terms in the big parentheses;
		between them, $2^{ϱm-ρn}$ dominates $\exp\(-e^{m^{1/3}}\)/δ$ once $n→∞$.
		Subsequently, we obtain this recurrence relation
		\[\cas{
			𝘨₀-𝘧₀^+=0;	\\
			𝘨_m-𝘧_m^+≤𝘨_{m-√n}-𝘧_{m-√n}^++2ℓ^{-ρn+o(n)}.
		}\]
		Solve it (cf.\ the Cesàro summation);
		we get that $𝘨_{n-√n}-𝘧_{n-√n}^+<ℓ^{-ρn+o(n)}$.
		Once again we summon $𝘧_{n-√n}^+<ℓ^{-ϱ(n-√n)+o(n-√n)}<ℓ^{-ϱn+o(n)}$;
		therefore $𝘨_{n-√n}<ℓ^{-ρn+o(n)}$.
		Based on the definition of $𝘨_{n-√n}$
		we immediately get $𝘦₀^{n-√n}>I(W)-ℓ^{-ρn+o(n)}$.
		
		(Analyze $𝘌₀^{n-√n}$.)
		We want to estimate $𝘡_n$ when $𝘌₀^{n-√n}$ happens.
		To be precise, for each $m=√n,2√n…n-√n$,
		we attempt to bound $𝘡_n$ when $𝘌_m$ happens .
		Fix an $m$.
		When $𝘌_m$ happens, its superevent $𝘈_m$ happens,
		so we know that $𝘡_m<\exp\(-e^{m^{1/3}}\)$.
		But $𝘉_m$ does not happen, so $𝘡_l<δ$ for all $l≥m$.
		This implies $𝘡_{l+1}≤𝘡_l^{𝘑_{l+1}(1-ε)}$ for those $l$.
		Telescope;
		$𝘡_n$ is less than or equal to $𝘡_m$ raised to
		the power of $𝘑_{m+1}𝘑_{m+2}\dotsm𝘑_n(1-ε)^{n-m}$.
		But $𝘊_m$ does not happen, so the product
		is greater than $2^{πn+2ε(n-m)}(1-ε)^{n-m}$,
		which is greater than $2^{πn}$ granted that $ε<1/2$.
		Jointly we have
		$𝘡_n≤𝘡_m^{2^{πn}}<\exp\(-e^{m^{1/3}}2^{πn}\)<\exp(-2^{πn})$.
		In other words, $𝘌₀^{n-√n}$ implies $𝘡_n<\exp(-ℓ^{πn})$.
		
		(Summary.)
		Now we may conclude, for all perfect squares $n$, that
		$𝘗\{𝘡_n<\exp(-2^{πn})\}≥𝘗(𝘌₀^{n-√n})=𝘦₀^n>I(W)-2^{-ρn+o(n)}$.
		For non-squares, round $n$ down to the nearest square
		and rerun the whole argument above.
		We will get $𝘡_n<2^n\exp\(-e^{m^{1/3}}2^{π⌊√n⌋²}\)$
		with probability at least $I(W)-ℓ^{-ρ⌊√n⌋²+o(n)}$,
		which still leads to $𝘗\{𝘡_n<\exp(-2^{πn})\}>I(W)-2^{-ρn+o(n)}$.
		And hence the proof of the moderate deviations behavior,
		\cref{ine:Z-e2pin}, is sound for all $n$.
	\end{proof}
	
	This section is parallel
	to \cite[section~V]{ModerateDeviations18}, to \cite[section~VI]{LargeDeviations18},
	to \cite[appendix~C.D]{Hypotenuse19}, and to \cite[section~10.3]{GRY19}.
	Do not confuse this section with the previous.
	The subtlety is explained in \cite[section~III]{LargeDeviations18}.

\section{Chapter Summary}

	In this chapter, we defined SBDMC and set a goal---that we
	want to construct error correcting codes and characterize
	their block error probabilities $Ｐ$ and codes rates $R$.
	We succeed in proving that $𝘡_n$ is such that
	\[𝘗\{𝘡_n<\exp(-2^{πn})\}>I(W)-2^{-ρn+o(n)},\]
	which implies that there are codes with
	$Ｐ<2^n\exp(-2^{πn})$ and $I(W)-R<2^{-ρn+o(n)}$.
	By a topological argument that fluctuates $π$,
	we get $Ｐ<\exp(-2^{πn})$ for $n$ very large.
	By fluctuating $ρ$, similarly, we get $I(W)-R<2^{-ρn}$
	granted that $n$ is astronomically large.
	
	\begin{cor}[Good code for SBDMC]
		Over SBDMCs, polar coding as constructed by Arıkan enjoys, for any $(π,ρ)∈𝒪$,
		block error probability $\exp(-N^π)$ and code rate $I(W)-N^{-ρ}$ for large $N$.
	\end{cor}
	
	Over BEC, we have a better estimate of $ϱ=1/3.627$ \cite{HAU14}.
	In this case, the region $𝒪$ is with $(0,1/3.627)$ in place of $(0,1/4.714)$.
	See also \cref{fig:lol-bec}.
	
	\begin{cor}[Good code for BEC]
		Over BECs, polar coding as constructed by Arıkan enjoys, for any $(π,ρ)$
		lying to the left of the convex envelope of $(0,1/3.627)$ and $1-h₂$,
		block error probability $\exp(-N^π)$ and code rate $I(W)-N^{-ρ}$ for large $N$.
	\end{cor}
	
	\begin{figure}
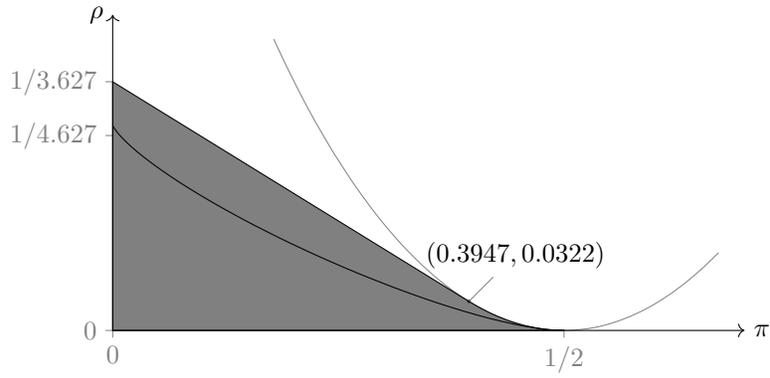

		\tikz[scale=12]{
			\draw[gray,very thin]
				(0,1/3.627)pic{y-tick=$1/3.627$}(0,1/4.627)pic{y-tick=$1/4.627$}
				(0,0)pic{y-tick=$0$}pic{x-tick=$0$}(1/2,0)pic{x-tick=$1/2$}
				plot[domain=25:55,samples=90]({sin(\x)^2},{1-h2o(\x)});
			\draw[fill=gray]
				(0,1/3.627)--plot[domain=38.922:45,samples=30]({sin(\x)^2},{1-h2o(\x)})-|cycle;
			\draw
				plot[domain=0:45,samples=90]({g2o(\x)*sin(\x)^2},{(1-g2o(\x))/3.627});
			\draw[->](0,0)--(0,.35)node[left]{$ρ$};
			\draw[->](0,0)--(.7,0)node[right]{$π$};
			\draw(.3947,.032234)pic{dot}coordinate[pin=45:{$(0.3947,0.0322)$}];
		}
		\caption{
			BEC special case of $(㏒(-㏒₂Ｐ),㏒(C-R))$.
			The tangent line is from $(0,1/3.6227)$ to $(0.4208,0,0182)$.
			The lower left curve is the previous result \cite{MHU16},
			which attains $(0,1/4.627)$.
		}\label{fig:lol-bec}
	\end{figure}
	
	In the next chapter, I will present the dual picture to this chapter.
	The duality stands for three aspects:
	The scenario is in duality,
	the theorem statement is in duality,
	and the proof technique is in duality.

\chapter{Asymmetric Channels}\label{cha:dual}

	\dropcap
	Researchers in the polar coding field, knowing the SBDMC results,
	had looked forward to applying polar coding to more channels or sources.
	As it turns out, polar coding also applies to source coding for lossy compression
	and noisy-channel coding over asymmetric binary-input discrete-output channels.
	But for polarization to work in those scenarios,
	the original analysis is subject to some modifications.
	
	First and foremost, let us review lossy compression.
	In lossy compression, the compressor is presented a random variable $Y$ and
	wants to send some messages to the decompressor so that the latter can generate
	a random variable $X$ that is close enough to $Y$ under a certain distance metric.
	A trade-off emerges as the compressor wants to send as few messages as possible
	while the $X$ generated from those messages should be as close to $Y$ as possible.
	This subject is usually referred to as the \emph{rate--distortion theory}.
	
	Polar coding applies to lossy compression by pretending that $X$ and $Y$
	are the input and output of an abstract channel, called the \emph{test channel}.
	Once there is a channel, it is polarized, a subset $𝒥$ of
	synthetic channels is selected, and synthetic channels are
	handled in two ways depending on whether they are in selected $𝒥$ or not.
	In this particular case, synthetic channels that are reliable correspond to the bits 
	that snapshot the essence of $Y$, and hence is worth messaging to the decompressor.
	On the other hand, synthetic channels that are noisy correspond to
	the randomness that separates $X$ from $Y$, which can be simulated
	by a pseudo random number generator on the decompressor side.
	By examining the pace the test channel polarizes,
	we gain control of the rate--distortion trade-off.
	
	Secondly, let me elaborate on asymmetric channels.
	Asymmetric channels differ from symmetric ones by the fact that
	the uniform input distribution does not necessarily achieve capacity.
	As a result, a code designer needs to spend extra resources on
	shaping the input distribution apart from the usual anti-error routine.
	This shaping component of coding shares common elements with generating $X$
	from the messages the compressor sends as in the lossy compression scenario.
	
	Polar coding applies to asymmetric channels
	by using a specialized decoder as an encoder.
	The new encoder polarizes the input distribution $Q$
	as if $Q$ were a channel with constant output.
	Now the reliable descendants of $Q$ are those who make $Q$ in the shape of $Q$;
	they are inflexible, deducible form the other descendants,
	and unable to carry new information.
	On the other hand, the noisy descendants of $Q$ are the source
	of the randomness of $Q$ and can carry user messages.
	Meanwhile, the actual channel $W$ through which
	we transmit messages is polarized as usual.
	It suffices to select a subset $𝒥$ of indices that correspond to, simultaneously,
	the noisy descendants of $Q$ and the reliable descendants of $W$.
	
	It happens that the performances of lossy compression and coding over
	asymmetric channels are both controlled by a stochastic process $\{T(𝘞_n)\}$.
	By the end of this chapter, I will prove
	\[𝘗｛T(𝘞_n)<e^{-2^{πn}}｝>H(W)-2^{-ρn}\]
	in order to describe the performances of polar coding in those scenarios.

\section{Problem Setup---Lossy Compression}

	Let $𝔽₂$ be the finite field of order $2$.
	Let $𝒴$ be any finite set equipped with a probability measure $W(y)$.
	Let $\dist：𝔽₂×𝒴→[0,1]$ be a bounded distortion function;
	this function quantifies how well an $x∈𝔽₂$ represents/approximates a $y∈𝒴$.
	For instance, we can have $𝒴⊆[0,1]²$ as a set of pairs and $\dist(x,(y₀,y₁))≔y_x$.
	
	The overall goal is to construct, for some large block length $N$, a compressor
	$𝒞：𝒴^N→𝔽₂^{RN}$ and a decompressor $𝒟：𝔽₂^{RN}→𝔽₂^N$ such that the composition
	\cd[every arrow/.append style=mapsto]{
		Y₁^N\rar{𝒞}	&	U₁^{RN}\rar{𝒟}	&	X₁^N
	}
	minimizes $D≔E［∑_{j=1}^N\dist(X_j,Y_j)/N］$, the long-term average of
	the point-wise distortion between $X₁^N$ and $Y₁^N$, for a given code rate $R$.
	Or conversely, we want to minimizes $R$ for a given $D$.
	The standard result follows.
	
	\begin{thm}[Rate--distortion trade-off]
		Allow arbitrarily large $N$.
		Then the infimum of code rates $R$ that are
		achievable within a fixed distortion $Δ$ is
		\[R(Δ)≔\min_{W(x｜y)}I(X；Y),\label{for:minIXY}\]
		where the minimum is taken over all transition arrays $W(x｜y)$ such that, when
		$X$ is governed by $W$, the expected distortion is bounded as $E[\dist(X,Y)]≤Δ$.
	\end{thm}
	
	For a proof, see standard textbooks, e.g., \cite{Blahut87}.
	See \cref{fig:RvsD} for an example of the function $R(Δ)$.
	
	To apply polar coding, we fix beforehand
	a transition array $W(x｜y)$ that minimizes \cref{for:minIXY}.
	Let $X$ and $Y$ be governed by $W(x｜y)$ and $W(y)$.
	And then pretend that $W$ is a channel as if
	$X$ were the input and $Y$ were the output.
	(Although in reality, it is $Y$ that is given
	to the compressor and the decompressor outputs $X$.)
	This $W$ is called the test channel and
	we can apply the channel transformation to it.
	It remains to specify what to do to the descendants of $W$ and explain
	how the behavior of $\{𝘞_n\}$ relates to the compression metrics $D$ and $R$.
	
	Note that a test channel is not a priori an SBDMC---it satisfies
	all conditions of being an SBDMC but not the symmetry one.
	Plus, we sometimes want to communicate over asymmetric channels.
	So we have to deal with asymmetric channels sooner or later.
	Or now.

\section{Problem Setup---Asymmetric Channel}

	A BDMC (binary-input discrete-output memoryless channel)
	is an SBDMC except that the involution $σ：𝒴→𝒴$ is not mandatory.
	The main aftermath after taking away the symmetry is that the uniform distribution
	is not guaranteed to be optimal in terms of achievable code rates.
	That being the case, there always exists another
	input distribution that achieves the optimal rate.
	
	\begin{thm}[Asymmetric chapacity]
		The channel capacity of a BDMC $W$, the supremum of
		codes rates at which reliable communication can happen, is
		\[\max_{Q(x)}I(X；Y),\]
		where the maximum is taken over all input distributions $Q(x)$ on $𝔽₂$.
	\end{thm}
	
	For a proof, see standard textbooks, e.g., \cite{Blahut87}.
	
	\begin{figure}
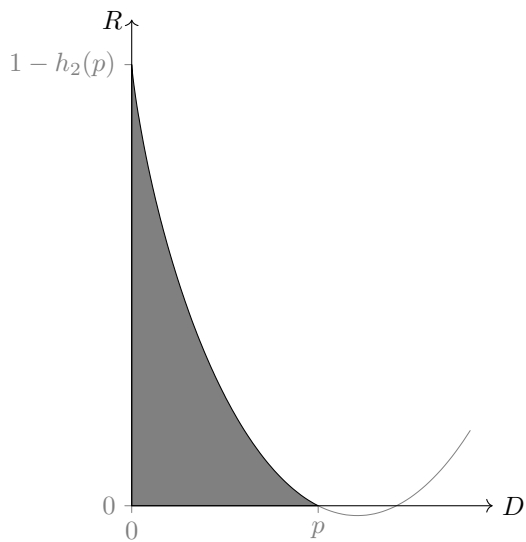

		\tikz[scale=6]{
			\PMS\a{sin(40)^2}
			\PMS\b{{1-h2o(40)}}
			\draw[gray,very thin]
				(0,1)pic{y-tick=$1-h₂(p)$}(0,\b)pic{y-tick=$0$}
				(0,\b)pic{x-tick=$0$}(\a,\b)pic{x-tick=$p$}
				plot[domain=40:60,samples=40]({sin(\x)^2},{1-h2o(\x)});
			\draw[fill=gray]
				plot[domain=0:40,samples=160]({sin(\x)^2},{1-h2o(\x)})-|cycle;
			\draw[->](0,\b)--(0,1.1)node[left]{$R$};
			\draw[->](0,\b)--(.8,\b)node[right]{$D$};
		}
		\caption{
			Assume lossy compression for a binary source of mean $p$
			and Hamming distortion function $\hdis(x,y)≔𝕀\{x≠y\}$.
			The shaded area is where $(R,D)$ is not possible.
			The curve part is $1-h₂$ shifted.
		}\label{fig:RvsD}
	\end{figure}
	
	Fix this $Q$ from now on.
	Then we can apply channel transformation following the same philosophy.
	In detail, we define a vector $U₁²∈𝔽₂²$ by
	\[\bma{U₁&U₂}≔\bma{X₁&X₂}\bma{1&0\\1&1},\label{for:U=X/G}\]
	or $U₁²≔X₁²G^{-1}$ for short.
	Note that this implicitly assigns a non-uniform, non-product
	distribution to $U₁²$ if $Q$ is not uniform to begin with.
	Ignoring that, we proceed to define $WＷ1$ to be an abstract channel with input $U₁$
	and output $Y₁²$, and $WＷ2$ a channel with input $U₂$ and output $Y₁²U₁$.
	
	The question to answer is,
	How do asymmetric channels evolve under channel transformation?
	It turns out that nothing really changes;
	the old theory extends to the new channels seamlessly.
	And it is due to a symmetrization technique.

\section{Channel Symmetrization}

	Given a BDMC $W$, we want to find an SBDMC $˜W$ such that any meaningful property
	concerning the descendants of $˜W$ automatically applies to those of $W$.
	To that end, a strategy is to define an equivalence relation $≅$ such that
	(a)	a BDMC is equivalent to at least one SBDMC,
	(b)	channel parameters such as $H$ and $Z$ are functions in classes,
		meaning that $H(W)=H(˜W)$ if $W≅˜W$, and
	(c)	the channel transformation respects the equivalence relation,
		meaning that $WＷj≅˜WＷj$ if $W≅˜W$.
	If such relation can be found, than almost all questions we want to ask about $W$
	have answers when we, instead, ask an SBDMC in the same class as $W$.
	
	\begin{dfn}
		Two BDMCs, $W$ and $˜W$, are said to be equivalent,
		denoted by $W≅˜W$, if $\{W(0｜Y),W(1｜Y)\}$ and $\{˜W(0｜˜Y),˜W(1｜˜Y)\}$
		obey the same distribution on the power set of $[0,1]$.
	\end{dfn}
	
	Let me briefly remark what $≅$ identifies.
	For one, the labeling on $𝒴$ is not important;
	after all, the decoder only care about the posterior probabilities.
	For two, if two outputs $y,y'$ have the same posterior probabilities, that is,
	$W(x｜y)=W(x｜y')$ for both $x∈𝔽₂$, the decoder might as well identify $y$ and $y'$.
	For three, relabeling the input $𝔽₂$ does not matter;
	the decoder just cares about how biased $\{W(0｜Y),W(1｜Y)\}$ is,
	but not about toward which way it biases.
	
	\begin{lem}[Reduction to symmetry]
		For any BDMC $W$, it is equivalent to at least one SBDMC.
	\end{lem}
	
	\begin{proof}
		Let $F∈𝔽₂$ be a “flag” that obeys an independent uniform distribution on $𝔽₂$.
		Let $˜W$  be a channel with input $X-F$ and output $(F,Y)$.
		Intuition:
		When the encoder attempts to input $X$ into a channel,
		it sees the flag $F$ and input $X-F$ instead;
		the decoder also sees the flag $F$ so it would
		simply add that back after all the decoding jobs.
		
		This $˜W$ is a SBDMC because adding $F$ to $X$
		will turn the input into a uniform random variable.
		This $˜W$ is equivalent to $W$ because
		\[\{˜W(0｜fy),˜W(1｜fy)\}=\{˜W(0-f｜fy),˜W(1-f｜fy)\}=\{W(0｜y),W(1｜y)\}.\]
		That is, the random sets $\{˜W(0｜fy),˜W(1｜fy)\}$ and $\{W(0｜y),W(1｜y)\}$
		coincide, and hence obey the same distribution.
	\end{proof}
	
	Before I state (b) that channel parameters are function in classes,
	one more parameter is defined to be utilized in the remainder of this chapter.
	Note that the definitions of $H$ and $Z$ automatically apply to the asymmetric case.
	
	\begin{dfn}\label{dfn:bin-T}
		Define the \emph{total variation norm} of $W$ to be
		\[T(W)≔∑_{y∈𝒴}W(y)∑_{x∈𝔽₂}\abs[\Big]{W(x｜y)-÷12},\]
		which is the total variation distance from $W(•｜y)$ to
		the uniform distribution, weighted by the frequency each $y∈𝒴$ appears.
	\end{dfn}
	
	\begin{lem}[Parameters in class]
		For any two equivalent BDMCs, $W$ and $˜W$,
		\[H(W)=H(˜W),\quad Z(W)=Z(˜W),\quad T(W)=T(˜W).\]
	\end{lem}
	
	\begin{proof}
		By definition.
	\end{proof}
	
	Remark:
	When $W$ is BDMC and $˜W$ is SBDMC, $I(W)=H(Q)-H(W)≠1-H(˜W)=I(˜W)$
	unless $H(Q)=1$ (i.e., $Q$ is uniform).
	What used to be $I(˜W)$ in the last chapter becomes $1-H(W)$ in this chapter,
	In particular, results of the form $𝘗\{𝘡_n<θ(n)\}>I(˜W)-γ(n)$
	are becoming $𝘗\{𝘡_n<θ(n)\}>1-H(W)-γ(n)$ when cited.
	
	\begin{lem}[Transformation in class]
		For any two equivalent BDMCs, $W$ and $˜W$,
		\[WＷ1≅˜WＷ1,\qquad WＷ2≅˜WＷ2.\]
	\end{lem}
	
	\begin{proof}
		By definition.
	\end{proof}
	
	We therefore conclude that, when it comes to channel transformation and channel
	processes, reasoning about $W$ is logically equivalent to reasoning about $˜W$.
	
	Being able to control the descendants of a BDMC, we will connect
	the performance of coding to channel parameters/processes in the next section.

\section{Problem Reduction to Process}

	Now we can reduce lossy compression to operations on synthetic channels.
	Let $W$ be a test channel of a lossy compression problem.
	Then, for any descendant $𝘞_n$ (including $W$ itself),
	we want a compressor to observe $𝘞_n$'s output and send a message
	to the decompressor so that the latter can generate $𝘞_n$'s input.
	By that $H(𝘞_n)→𝘏_∞∈\{0,1\}$, the descendants of $W$ are either
	extremely reliable or completely noisy, and hence assume easy treatments.
	
	When a realization of $𝘞_n$ is completely noisy,
	it means that its input $𝘟_n$ is almost irrelevant to its output $𝘠_n$.
	In this case, the decompressor does not need to know anything about $𝘠_n$
	and can query a pseudo random number generator to simulate $𝘟_n$.
	When a realization of $𝘞_n$ is (extremely) reliable, on the other hand,
	its input $𝘟_n$ depends (heavily) on its output $𝘠_n$.
	In this case, the compressor is suggested to send $𝘟_n$
	to the decompressor so that the latter can simply output $𝘟_n$.
	
	The last paragraph can actually be translated into
	a lossy compression scheme but I decided to omit the details
	as they can be found in past works, e.g., \cite{KU10o}.
	I claim without a proof that the excess of distortion of this coding scheme
	is bounded by the average of total variation norms of the noisy descendants.
	More symbolically,
	\[D-Δ≤÷1N∑_{j₁^n∈𝒥}T\(\(\dotsb((WＷ{j₁})Ｗ{j₂})\dotsb\)Ｗ{j_n}\).
		\label{ine:D<sum}\]
	It remains to define $𝒥$.
	Motivated by \cref{ine:D<sum}, we simply let $𝒥$ collect depth-$n$
	synthetic channels whose $T$ is less than a threshold $θ$.
	Then, similar to \cref{cha:origin}, we have that $D-Δ<θ$ and $R=1-𝘗\{T(𝘞_n)<θ\}$.
	To rephrase it, the trade-off between $Ｐ$ and $R$ for lossy compression is
	\[𝘗\{T(𝘞_n)<D-Δ\}≈1-R.\]
	Let $\{𝘛_n\}$ be the total variation process defined by $𝘛_n≔T(𝘞_n)$ .
	
	The last paragraph is not the only driving force
	for learning the total variation process $\{𝘛_n\}$;
	channel coding over asymmetric channels enjoys a similar reduction.
	
	Recall that $Q(x)$, presumably non-uniform,
	is the capacity-achieving input distribution w.r.t.\ an asymmetric $W$.
	Pretend that $Q(x,♣)$ models a channel with a constant output $♣$.
	Then we can polarize $Q$ and talk about its descendants.
	Details omitted, the block error probability of polar coding over $W$ is \cite{HY13}
	\[Ｐ≤∑_{j₁^n∈𝒥}Z\(\(\dotsb((WＷ{j₁})Ｗ{j₂})\dotsb\)Ｗ{j_n}\)
		+T\(\(\dotsb((QＷ{j₁})Ｗ{j₂})\dotsb\)Ｗ{j_n}\).\label{ine:P<sumsum}\]
	We know how to bound the sum of $Z$'s, so it remains to bound the sum of $T$'s.
	
	All in all, we now want to show
	\[𝘗｛𝘛_n<e^{-2^{πn}}｝>H(W)-2^{-ρn}\label{ine:T-e2pin}\]
	for $(π,ρ)$ lying in the same region as in \cref{cha:origin}.
	The following section argues that we can reuse the same the proof.

\section{Stochastic Process Nonsense}

	In this section, I will show that since the total variation process
	$\{𝘛_n\}$ satisfies almost all properties satisfied by $\{𝘡_n\}$,
	the majority of the proof of \cref{ine:Z-e2pin} applies to \cref{ine:T-e2pin}.
	
	Let us start with an counterpart to \cref{lem:squares}.
	
	\begin{lem}[Evolution of $T$]
		\cite[Lemma~3]{Muramatsu21}
		The following hold for all SBDMCs $W$:
		\begin{gather*}
			T(WＷ1)=T(W)²,			\label{ine:T^2}\\
			T(WＷ2)≤2T(W)-T(W).	\label{ine:2T-T^2}
		\end{gather*}
	\end{lem}
	
	This implies a nice consequence---that there is another supermartingale
	that will dominate the behavior of $𝘞_n$ at the noisy end.
	
	\begin{lem}
		The process of total variation norms, $\{𝘛_n\}$, is a supermartingale.
	\end{lem}
	
	Note that we do not have a counterpart to \cref{ine:2-Z^2}.
	This will affect how we deal with the en23 behavior of $\{𝘛_n\}$.
	The counterpart of \cref{lem:ZvsH} follows.
	
	\begin{lem}[$T$ vs $H$]\label{lem:TvsH}
		\cite[Lemma~4]{Muramatsu21}
		The following holds for all SBDMCs $W$:
		\[1-T(W)≤H(W)≤h₂（÷{1-T(W)}2）\]
		where $h₂$ is the binary entropy function.
	\end{lem}
	
	See \cref{fig:TvsH} for a visualization.
	
	\begin{figure}
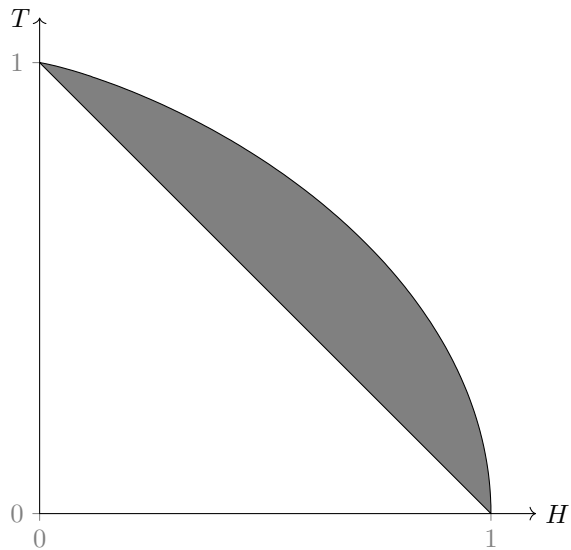

		\tikz[scale=6]{
			\draw
				(0,1)pic{y-tick=$1$}(0,0)pic{y-tick=$0$}
				(0,0)pic{x-tick=$0$}(1,0)pic{x-tick=$1$};
			\draw[fill=gray]
				plot[domain=0:45,samples=180]({h2o(\x)},{cos(2*\x)})--cycle;
			\draw[->](0,0)--(0,1.1)node[left]{$T$};
			\draw[->](0,0)--(1.1,0)node[right]{$H$};
		}
		\caption{
			The possible region where $(T(W),Z(W))$ could lie in.
			The curve part is half of $h₂$ after rotating and rescaling.
		}\label{fig:TvsH}
	\end{figure}
	
	The counterpart to \cref{lem:Z-en23} follows.
	Note that we skipped \cref{thm:bdmc-mu}, but the following proof
	shows that we need not duplicate \cref{thm:bdmc-mu}.
	
	\begin{lem}[From eigen to en23]\label{lem:T-en23}
		Fix $ϱ≔1/4.714$.
		Assume \cref{thm:bdmc-mu}.
		Then
		\[𝘗｛𝘛_n<e^{-n^{2/3}}｝>H(W)-2^{-ϱn+o(n)}.\]
	\end{lem}
	
	\begin{proof}
		By \cref{lem:ZvsH,lem:TvsH}, $T(W)→0$ iff $H(W)→1$ iff $Z(W)→1$ on the noisy end
		and $T(W)→1$ iff $H(W)→0$ iff $Z(W)→0$ on the reliable end.
		More strongly, each “iff\kern.1em” can be stated as
		two Hölder conditions in two directions.
		That is to say, there exist constants $c,d>0$ such that
		$T(W)≤c(1-H(W))^d$ and $1-H(W)≤cT(W)^d$, as well as
		$1-H(W)≤c(1-Z(W))^d$ and $1-Z(W)≤c(1-H(W))^d$,
		as well as the other four inequalities on the reliable end.
		
		This common-fate property implies that, if $n$ is sufficiently large,
		\begin{align*}
			\qquad&\kern-2em
			𝘗｛e^{-n^{2/3}}≤𝘛_n≤1-e^{-n^{2/3}}｝<𝘗｛e^{-n^{3/4}}≤𝘡_n≤1-e^{-n^{3/4}}｝	\\
			&	≤÷{𝘌[h(𝘡_n)]}{h(\exp(-n^{3/4}))}
				≤÷{h(𝘡₀)2^{-ϱ}}{\exp(-n^{3/4})}<2^{-ϱn-o(n)}.
		\end{align*}
		This is the counterpart to \cref{ine:bdmc-expel}.
		
		It remains to show that $𝘗\{𝘛_n→0\}=H(W)$ and that
		the bad event $𝘉_n≔\{𝘛_n→0$ but $𝘛_n≥\exp(-n^{2/3})\}$
		is exponentially rare, i.e., $2^{-ϱ(n)+o(n)}$-rare.
		The former is again by the common-fate property
		$𝘗\{𝘛_n→0\}=𝘗\{𝘏_n→1\}=𝘌[𝘏_∞]=𝘌[𝘏₀]=H(W)$.
		The latter is by that $\{𝘛_n\}$, a supermartingale, cannot jump back and forth
		between the neighborhood of $0$ and the neighborhood of $1$,
		so $𝘗(𝘉_n)$ is bounded from above by $∑_{m≥n}2^{-ϱm+o(m)}<2^{-ϱn+o(n)}$.
		
		We can finally conclude that
		\[𝘗｛𝘛_n<e^{-n^{2/3}}｝>𝘗\{𝘛_n→0\}-𝘗(𝘉_n)>H(W)-2^{-ϱn+o(n)}.\]
		This calls the end of the proof.
	\end{proof}
	
	The last lemma is the most technical one in this chapter.
	It uses the common-fate property to show that not only do
	$\{𝘡_n\}$, $\{𝘏_n\}$, and $\{𝘛_n\}$ control each other's limit,
	but they also control each other's pace of convergence.
	In particular, we can also show that $𝘗\{𝘏_n<\exp(-n^{2/3})\}>I(W)-2^{-ϱn+o(n)}$ and
	that $𝘗\{1-𝘏_n<\exp(-n^{2/3})\}>H(W)-2^{-ϱn+o(n)}$ although it is not useful here.
	
	The een13 behavior follows.
	
	\begin{lem}[From en23 to een13]\label{lem:T-een13}
		Given \cref{lem:T-en23}, that is, given
		\[𝘗｛𝘛_n<e^{-n^{2/3}}｝>H(W)-2^{-ϱn+o(n)},\]
		we have
		\[𝘗｛𝘛_n<\exp\(-e^{n^{1/3}}\)｝>H(W)-2^{-ϱn+o(n)}.\]
	\end{lem}
	
	\begin{proof}
		The proof is just a copy of that of \cref{lem:Z-een13}.
		All we did there was by that $𝘡_n$ is squared or doubled with equal probability.
		Since $𝘛_n$ is squared or doubled with equal chance, the conclusion follows.
	\end{proof}
	
	The ultimate behavior follows.
	
	\begin{thm}[From een13 to e2pin]\label{thm:T-e2pin}
		Fix a pair $(π,ρ)∈𝒪$.
		Given the conclusion of \cref{lem:T-een13}, that is, given
		\[𝘗｛𝘛_n<\exp\(-e^{n^{1/3}}\)｝>H(W)-2^{-ϱn+o(n)},\]
		then
		\[𝘗｛𝘛_n<e^{-2^{πn}}｝>H(W)-2^{-ρn+o(n)}.\]
	\end{thm}
	
	\begin{proof}
		The proof follows the same type of argument as in \cref{thm:Z-e2pin}.
		All we need is based on the fact/axiom that
		the process is squared or doubled with equal chance.
	\end{proof}

\section{Chapter Wrap up}

	In this chapter, we first see two coding problems---source coding for
	lossy compression and noisy channel coding over asymmetric channels.
	I then argued that polar coding applies to those scenarios
	with an extension to asymmetric (test) channels.
	To deal with asymmetric channels, we prove a series of lemmas whose takeaway is that
	we only have to consider the SBDMC $˜W$ that lies in the same equivalence class.
	To characterize the performance of polar coding,
	it remains to understand the behavior of $\{𝘛_n\}$.
	And then we move on to proving the behavior of $\{𝘛_n\}$
	using the same techniques we used to handle $\{𝘡_n\}$.
	
	Now that we have proved
	$𝘗\{𝘡_n<\exp(-2^{πn})\}>I(W)-2^{-ρn+o(n)}$ (from the previous chapter)
	and $𝘗\{𝘛_n<\exp(-2^{πn})\}>H(W)-2^{-ρn+o(n)}$ (in this chapter),
	plug them into \cref{ine:D<sum,ine:P<sumsum}.
	For lossy compression, my result implies that lossy compression via polar coding
	enjoys code rate $I(W)+2^{-ρn+o(n)}$ and distortion $Δ+\exp(-2^{πn})$.
	Eliminate the $o(n)$ term with a topological argument.
	
	\begin{cor}[good code for lossy compression]
		For any lossy compression problem whose test channel $W$ is a BDMC,
		using polar coding yields excess of distortion $D-Δ<\exp(-N^π)$
		and gap to capacity $R-I(W)<N^{-ρ}$ for any $(π,ρ)∈𝒪$ and big $N$.
	\end{cor}
	
	The asymmetric channel case is more involved.
	We actually have, and need, four inequalities
	\begin{gather*}
		𝘗\{Z(𝘞_n)<θ\}>1-H(W)+γ,	\\
		𝘗\{T(𝘞_n)<θ\}>H(W)+γ,		\\
		𝘗\{Z(𝘘_n)<θ\}>1-H(Q)+γ,	\\
		𝘗\{T(𝘘_n)<θ\}>H(Q)+γ,		
	\end{gather*}
	and basic facts
	\begin{gather*}
		\{Z(𝘞_n)<θ\}∩\{T(𝘞_n)<θ\}=0,	\\
		\{Z(𝘘_n)<θ\}∩\{T(𝘘_n)<θ\}=0,	\\
		\{Z(𝘘_n)<θ\}∩\{T(𝘞_n)<θ\}=0.
	\end{gather*}
	Here, $\{𝘘_n\}$ is the channel process grown from $Q$, and
	$θ$ (threshold) and $γ$ (gap) are the shorthands of the complicated functions.
	The basic facts are consequences of the common-fate property
	and the monotonicity property $H(W)=H(X｜Y)≤H(X)=H(X｜♣)=H(Q)$.
	
	Applying the inclusion--exclusion principle, we conclude that
	$𝘗\{T(𝘘_n)<θ$ and $Z(𝘞_n)<θ\}>H(Q)-H(W)-6γ=I(W)-6γ$.
	This implies the existence of codes with $Ｐ<\exp(-2^{πn})$ and $R>I(W)-2^{-ρn}$.
	
	\begin{cor}[Good code for BDMC]
		For any BDMC, polar coding yields block error probability $Ｐ<\exp(-N^π)$
		and gap to capacity $I(W)-R<N^{-ρ}$ for any $(π,ρ)∈𝒪$ and big $N$.
	\end{cor}
	
	The next chapter contains yet another application of \cref{thm:Z-e2pin,thm:T-e2pin}.
	In brief, the idea is to prune the channel process
	when $𝘞_n$ becomes too reliable or too noisy.
	And it is only after learning how reliable/noisy
	$𝘞_n$ is that we know where to prune.

\section{A Side Note on Lossless Compression}

	Lossless compression with side information at the decoder side can be solved
	by polar coding as well but assumes a different format \cite{Arikan10,CK10}.
	In this scenario, the random variable to be compressed is denoted by $X$,
	and will be treated as the input of an abstract channel $W$;
	the side information accessible by the decoder
	is denoted by $Y$ and treated as the output of $W$.
	
	$W$ is then polarized.
	For if a descendant of $W$ is reliable, it means that
	its input is mostly determined by its output, and hence needs no extra action.
	If, otherwise, a descendant of $W$ is noisy, its input is
	largely independent from its output, and we should record this input.
	The code rate is thus the density of the noisy descendants.
	There will be a block error if the decoder fails to
	recover the input of a descendant that was not recorded;
	the block error probability is thus the sum of
	the $Z$'s of the reliable descendants.
	
	This come down to the following corollary.
	
	\begin{cor}[Good code for lossless compression]
		For any binary source $X$ to be compressed losslessly and side information $Y$,
		polar coding yields block error probability $Ｐ<\exp(-N^π)$
		and gap to entropy $R-H(X｜Y)<N^{-ρ}$, for any $(π,ρ)∈𝒪$ and big $N$.
	\end{cor}

\chapter{Pruning Channel Tree}\label{cha:prune}

	\dropcap
	Complexity of encoder and decoder influences,
	sometimes dominates, practicality of a code.
	Throughout the history of coding---Hamming, Reed--Muller, turbo, LDPC, polar,
	et seq---real world codes are always easy to implement in the first place
	regardless of achieving capacity or not (usually not).
	This is why, apart from $N$, $P$, and $R$,
	we should care about the complexities of codes.
	
	Polar code, an outlier in said list, is the first code
	to achieve capacity in combination with a very low complexity.
	Considering that the decoder should at least read in
	all $N$ symbols in a code block, which costs $O(N)$ resources,
	polar coding's complexity of $O(N㏒N)$ is impressive, if not surprising.
	
	In this chapter, a modification is made to polar coding.
	It aims to eliminate the inefficient components
	in the encoder/decoder to reduce the complexity even further.
	By doing so carefully, the new complexity is $O(N㏒(㏒N))$,
	the block error probability decays quasi-polynomially fast to $0$,
	and the gap to capacity decays as fast as before.
	For comparison, existing works like \cite{AK11,EMFLK17,MHCG20}
	did not break the $O(N㏒N)$ barrier, and some latest work \cite{HMFVCG20}
	achieves $N㏒(㏒N)$ \emph{latency} only in the fully-parallel mode.
	
	Being called pruning, this technique is inspired by a trinitarian correspondence
	among the encoder/decoder, the channel tree, and the channel process.
	The argument start from the channel process side:
	Since $𝘏_n→𝘏_∞∈\{0,1\}$, the increments $𝘏_{n+1}-𝘏_n$ converge to $0$.
	By the correspondence, that means that on the channel tree side,
	branches that are deep enough are purposeless---%
	applying channel transformation barely changes anything.
	Now turn our focus to the encoder/decoder side;
	the said observation implies that some components in the encoder/decoder
	are consuming resources without helping the code become better,
	and we should have removed them.
	
	The goal of the current chapter is to make precise the last paragraph.

\section{Encoder/Decoder vs Tree vs Process}

	The design of the encoder and decoder of polar coding is best described by figures.
	There is a device called EU (encoding unit)
	and another device called DU (decoding unit);
	their actual implementations are not of interest here.
	But the devices are such that, when you wrap two copies of $W$ with one EU--DU pair,
	like \cref{fig:transform} does, the top pair of pins (A and B) behaves like $WＷ1$,
	and the bottom pair of pins (C and D) behaves like $WＷ2$.
	We say that \cref{fig:transform} corresponds to a tree with three vertices:
	$W$ the root and $WＷ1$ and $WＷ2$ its children.
	
	\tikzset{
		E/.pic={
			\pgftransformscale{1/6}\def~{coordinate}
			\draw[very thick](-6,3)~(=o0)--(6,3)~(=i0)(-6,-3)~(=o1)--(6,-3)~(=i1);
			\draw[fill=white](-5,-5)rectangle(5,5)(0,0)node{EU};
		},
		D/.pic={
			\pgftransformscale{1/6}\def~{coordinate}
			\draw[very thick](-6,3)~(=i0)--(6,3)~(=o0)(-6,-3)~(=i1)--(6,-3)~(=o1);
			\draw[fill=white](-5,-5)rectangle(5,5)(0,0)node{DU};
		},
		C/.pic={
			\pgftransformscale{1/6}\def\c{coordinate}
			\draw[very thick](-4,0)\c(=e)--(4,0)\c(=d);
			\draw[fill=white](-3,-2)rectangle(3,2)(0,0)node{CH};
		}
	}
	\begin{figure}
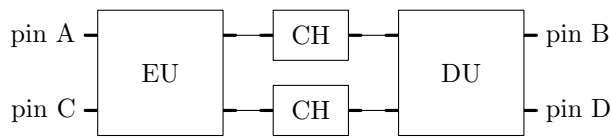

		\tikz[scale=1/6]{
			\draw
									(0,12)pic(c0){C}
				(-12,9)pic(e){E}						(12,9)pic(d){D}
									(0,6)pic(c1){C}
			;
			\draw
				\foreach\f in{e,d}{
					\foreach\j in{0,1}{
						(c\j=\f)--(\f=i\j)
					}
				}
				(e=o0)node[left]{pin A}(d=o0)node[right]{pin B}
				(e=o1)node[left]{pin C}(d=o1)node[right]{pin D}
			;
		}
		\caption{
			The design of encoder and decoder---level $1$.
		}\label{fig:transform}
	\end{figure}
	
	To construct the grandchildren of $W$, wrap more
	EU--DU pairs around two copies of \cref{fig:transform}.
	For instance, in \cref{fig:transform(1)}, the two copies of pin A
	are connected to another EU, the two copies of pin B to another DU.
	Recall that pin A--pin B behaves like the input--output of $WＷ1$;
	so wrapping one more layer of EU--DU transforms it further into
	$(WＷ1)Ｗ1$ (from pin E to pin F) and $(WＷ1)Ｗ2$ (from pin G to pin H).
	Note that the two copies of pin C and pin D are naked (not connecting to anything);
	this represents the fact that there are two copies of $WＷ2$
	that are not (yet) transformed into $(WＷ2)Ｗ1$ and $(WＷ2)Ｗ2$.
	Now \cref{fig:transform(1)} corresponds to a channel tree with five vertices:
	$W$ the root, $WＷ1$ and $WＷ2$ its children, plus
	$(WＷ1)Ｗ1$ and $(WＷ1)Ｗ2$ the children of the elder sibling.
	
	\begin{figure}
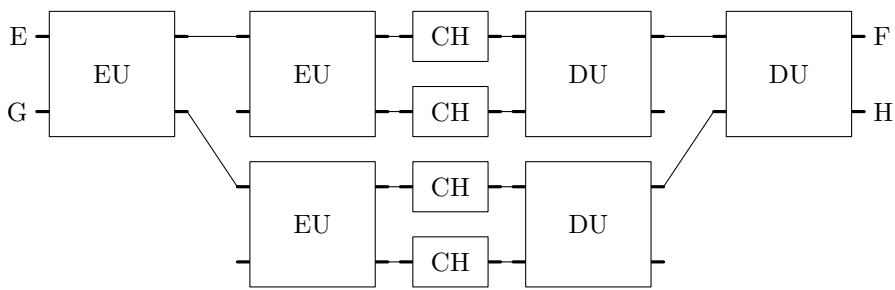

		\tikz[scale=1/6]{
			\draw
											(0,24)pic(c00){C}
			(-27,21)pic(2e0){E}(-11,21)pic(1e0){E}	(11,21)pic(1d0){D}(27,21)pic(2d0){D}
											(0,18)pic(c01){C}
											
											(0,12)pic(c10){C}
								(-11,9)pic(1e1){E}		(11,9)pic(1d1){D}
											(0,6)pic(c11){C}
			;
			\draw
				\foreach\f in{e,d}{
					\foreach\k in{0,1}{
						\foreach\j in{0,1}{
							(c\k\j=\f)--(1\f\k=i\j)
						}
						(1\f\k=o0)--(2\f0=i\k)
					}
				}
				(2e0=o0)node[left]{E}(2d0=o0)node[right]{F}
				(2e0=o1)node[left]{G}(2d0=o1)node[right]{H}
			;
		}
		\caption{
			The design of encoder and decoder---%
			transforming $WＷ1$ further at level $2$.
		}\label{fig:transform(1)}
	\end{figure}
	
	Duplicate \cref{fig:transform(1)}.
	This time, wrap around the two copies of pin G and H
	as shown in \cref{fig:transform(1)(2)}.
	Then we are effectively transforming $(WＷ1)Ｗ2$
	into $((WＷ1)Ｗ2)Ｗ1$ and $((WＷ1)Ｗ2)Ｗ2$.
	Now \cref{fig:transform(1)(2)} corresponds to a channel tree with seven vertices:
	$W$, $WＷ1$, $WＷ2$, $(WＷ1)Ｗ1$, and $(WＷ1)Ｗ2$ and its children.
	
	\begin{figure}
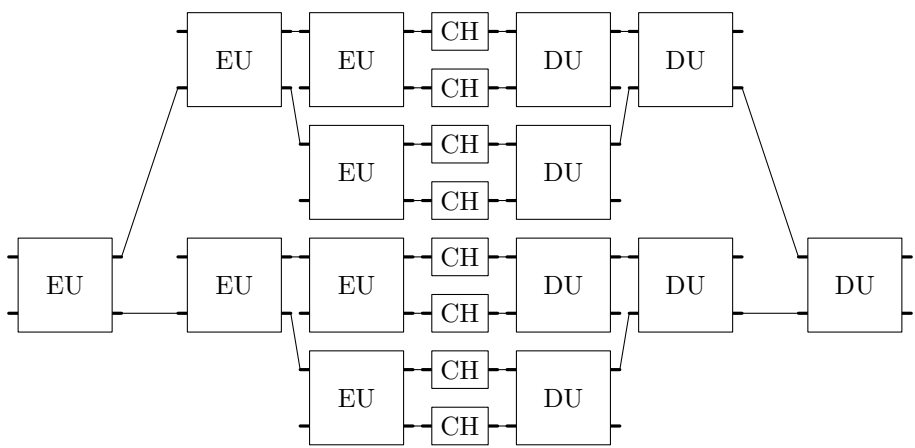

		\tikz[scale=1/6,scale=3/4,every pic/.style={scale=3/4}]{
			\def\pi{pic}
			\draw
														(0,48)π(c000){C}
						(-24,45)π(2e00){E}(-11,45)π(1e00){E}	(11,45)π(1d00){D}(24,45)π(2d00){D}
														(0,42)π(c001){C}
														
														(0,36)π(c010){C}
											(-11,33)π(1e01){E}	(11,33)π(1d01){D}
														(0,30)π(c011){C}
														
														(0,24)π(c100){C}
		(-42,21)π(3e00){E}(-24,21)π(2e10){E}(-11,21)π(1e10){E}	(11,21)π(1d10){D}(24,21)π(2d10){D}(42,21)π(3d00){D}
														(0,18)π(c101){C}
														
														(0,12)π(c110){C}
											(-11,9)π(1e11){E}	(11,9)π(1d11){D}
														(0,6)π(c111){C}
			;
			\draw
				\foreach\f in{e,d}{
					\foreach\l in{0,1}{
						\foreach\k in{0,1}{
							\foreach\j in{0,1}{
								(c\l\k\j=\f)--(1\f\l\k=i\j)
							}
							(1\f\l\k=o0)--(2\f\l0=i\k)
						}
						(2\f\l0=o1)--(3\f00=i\l)
					}
				}
			;
		}
		\caption{
			The design of encoder and decoder---%
			transforming $(WＷ1)Ｗ2$ further at level $3$.
		}\label{fig:transform(1)(2)}
	\end{figure}
	
	Lesson:
	Each pair of pins corresponds to a (synthetic) channel;
	one channel may have multiple copies that each corresponds to a pair.
	For any synthetic channel, we wrap another layer of EU--DU pair
	around the corresponding pins if we want to transform it further.
	We leave the pins corresponding to a channel naked
	if we do not want to transform it anymore.
	
	Rephrase in terms of channel process:
	The channel process $\{𝘞_n\}$ is defined prior to
	whether we want to transform channels or not.
	But we can choose “not to look at it”.
	More rigorously, we look at the endless sequence $𝘑₁,𝘑₂,\dotsc$
	and determine the least $n$ such that $(𝘑₁,𝘑₂…𝘑_n)$
	points to a channel we do not want to transform anymore.
	Let $𝘴$ be that $n$ in this case.
	That is, $𝘴$ is a random variable depending on $\{𝘑_n\}$;
	moreover, whether or not $𝘴=n$ is determined by the first $n$ terms of $\{𝘑_n\}$.
	This makes $𝘴$ a stopping time adapted to $\{𝘑_n\}$.
	Consequently, $\{𝘞_{n∧𝘴}\}$ is a stopped process that evolves like $\{𝘞_n\}$
	at the beginning but halts at $𝘞_𝘴$, a channel we are satisfied with.
	
	It remains to do the following three things:
	(a)	Show how $𝘴$ is related to the complexity,
	(b)	show how $𝘞_𝘴$ is related to the code metrics $Ｐ$ and $R$, and
	(c)	define a good stopping time $𝘴$ to optimize (a) and (b).

\section{Stopping Time vs Complexity}

	Fix an SBDMC $W$.
	Fix an $n$;
	we are to construct a low-complexity code of block length $N=2^n$.
	Pretend that the complexity is the number of EU and DU devices used.
	Doing so is backed up by the fact that, in reality,
	EU and DU cost bounded amount of arithmetic and memory.
	
	Let $𝘴$ be any stopping time adapted to $\{𝘑_n\}$, and assume $𝘴≤n$.
	Having $𝘴≤n$ is to make sure that we do not transform any depth-$n$ channel further.
	We have the following lemma concerning $𝘴$ and the complexity.
	
	\begin{lem}[Complexity in terms of $𝘴$]\label{lem:prune-C}
		Use $𝘴$ to generate a code.
		Then the encoding and decoding complexity is
		$O(N𝘌[𝘴])$ per code block, or $O(𝘌[𝘴])$ per channel usage.
	\end{lem}
	
	\begin{proof}
		Let us begin with $N=2^n$ copies of the channel $W$.
		If $𝘴>0$, then there are $2^{n-1}$ EU--DU pairs that
		wrap around $W$'s to synthesize $2^{n-1}$ copies of $WＷ1$ and $WＷ2$.
		
		Next, let us consider the $𝘑₁=1$ case.
		If $𝘴>1$ in this case, it means that we want to transform $WＷ1$.
		We need $2^{n-2}$ EU--DU pairs to wrap around the $2^{n-1}$ copies of $WＷ1$.
		Similarly, if $𝘴>1$ in the $𝘑₁=2$ case, then we are transforming
		$2^{n-1}$ copies of $WＷ2$ using another $2^{n-2}$ EU--DU pairs.
		
		Next, consider the depth-$2$ branching of $(𝘑₁,𝘑₂)$.
		For each of the four possible realizations of the pair $(𝘑₁,𝘑₂)=(j₁,j₂)$,
		if $𝘴>2$, then we are to transform the $2^{n-2}$ copies
		of $(WＷ{j₁})Ｗ{j₂}$ with $2^{n-3}$ EU--DU pairs.
		
		Finally, consider the general case at arbitrary depth:
		If $𝘴>m$ for a certain tuple $(𝘑₁,𝘑₂…𝘑_m)=(j₁,j₂…j_m)$,
		then we will transform $2^{n-m}$ copies of
		$𝘞_m=\(\dotsb((WＷ{j₁})Ｗ{j₂})\dotsb\)Ｗ{j_m}$
		further with $2^{n-m-1}$ EU--DU pairs.
		The total number of EU--DU pairs is thus
		\[∑_{m=0}^n∑_{j₁^m}2^{n-m-1}·𝘐\{𝘑₁^m=j₁^m† and †𝘴>m\}
			=∑_{m=0}^n2^{n-1}𝘗\{𝘴>m\}=÷N2𝘌[𝘴].\]
		Here $𝘐$ is the indicator of an event.
		
		Recall that I claimed without a proof that
		EU and DU devices cost constant resources.
		Thus, the complexity of a code is proportional
		to the number of such devices used.
		So a code generated by $𝘴$ has complexity
		$O(N𝘌[𝘴])$ per code block or $O(𝘌[𝘴])$ per channel usage.
		The end.
	\end{proof}
	
	The last lemma connects a stopping time $𝘴$ to the cost of its code.
	Intuitively speaking, it attributes to the fact that
	preparing a synthetic channel at depth $m$ costs $m$ layers of EU--DU.
	Hence preparing $𝘞_𝘴$ costs $𝘴$ layers,
	and preparing all instances of $𝘞_𝘴$ costs $N𝘌[𝘴]$.
	
	In the next section, we will see how $𝘞_𝘴$ connects to the code performance
	in a very similar way, i.e., $Ｐ≤N𝘌[𝘡_𝘴·𝘐\{𝘡_𝘴<θ\}]$ and $R=𝘗\{𝘡_𝘴<θ\}$,
	due to very similar reasons.

\section{Stopped Process vs Performance}

	Thanks to the stopping time $𝘴$, the channel tree is pruned,
	or “harvested”, before it reaches depth $n$.
	This implies that the leafs $𝘞_𝘴$ are not as polarized as before.
	We therefore are liable to rebound the block error probability and code rate.
	
	Let $𝒥$ be some set of indices $(𝘑₁,𝘑₂…𝘑_𝘴)$ that point to
	the synthetic channels we are to use to send plain messages.
	Lemmas follow.
	
	\begin{lem}[$R$ in terms of $𝒥$]\label{lem:prune-R}
		The code rate
		\[R=𝘗\{(𝘑₁,𝘑₂…𝘑_𝘴)∈𝒥\}\]
		is the probability that the prefix $𝘑₁^𝘴$ is selected in $𝒥$.
	\end{lem}
	
	\begin{proof}
		I claim without a proof that the code rate is the density
		of the naked pins that correspond to channels in $𝒥$.
		
		Assuming that, we see that each $𝘞_𝘴$
		corresponds to $2^{n-𝘴}$ pairs of naked pins.
		That is to say, each $𝘞_𝘴$ assumes $2^{n-𝘴}$ copies in the EU--DU circuit.
		Since there are always $2^n$ pairs of naked pins
		(adding ED/DU does not alter the number of naked pins),
		a pair of naked pins possesses probability measure $2^{-n}$.
		Thus each $𝘞_𝘴$ possesses probability measure $2^{-𝘴}$,
		the same probability measure possessed by $(𝘑₁,𝘑₂…𝘑_𝘴)$.
		This finishes the proof.
	\end{proof}

	\begin{lem}[$Ｐ$ in terms of $𝒥$]\label{lem:prune-P}
		The block error probability
		\[Ｐ≤N𝘌[𝘡_𝘴·𝘐\{(𝘑₁,𝘑₂…𝘑_𝘴)∈𝒥\}]\]
		is bounded by the sum of $𝘡_𝘴$ that are
		selected in $𝒥$, weighted by multiplicity.
	\end{lem}
	
	\begin{proof}
		I claim without a proof that the block error probability
		is bounded from above by the sum of $Z$'s of
		the synthetic channels selected in $𝒥$, multiplicity included.
		
		Assuming that, and knowing that each $𝘞_𝘴$ assumes $2^{n-𝘴}$ copies
		and possesses probability measure $2^{-𝘴}$, we infer that
		each $𝘞_𝘴$ contributes $N𝘗\{𝘑₁^𝘴=j₁^𝘴\}·Z(𝘞_𝘴)$ to the upper bound
		if it is selected, otherwise it would have contributed $0$.
		Each $𝘞_𝘴$ contributing $N𝘗\{𝘑₁^𝘴=j₁^𝘴∈𝒥\}·Z(𝘞_𝘴)$,
		their sum is clearly $N𝘌[𝘡_𝘴·𝘐\{𝘑₁^𝘴∈𝒥\}]$.
		This is the upper bound we want to prove.
	\end{proof}

	It is time to declare an $𝘴$ and compute the induced $R$, $Ｐ$, and complexity.
	The basic strategy, like what we did in \cref{cha:origin}, is to set a threshold $θ$
	and hope that $𝘞_𝘴$ is “at least $θ$-good”, or $𝘡_𝘴≤θ$ to be rigorous.
	In contrast to \cref{cha:origin}, we can now control $𝘴$,
	so it is actually more efficient if we incorporate $θ$ into the declaration of $𝘴$.

\section{Actual Code Construction}

	Let $θ$ be $4^{-n}$.
	Define the stopping time $𝘴$ as below
	\[𝘴≔n∧\min\{m:𝘛_m<θ† or †𝘡_m<θ\}.\label{for:stop}\]
	That is, we look for the first $m$ such that either $𝘛_m$ or $𝘡_m$ are small enough;
	but if we cannot find such $m$ before reaching depth $n$, let $𝘴$ default to $n$.
	Let $𝒥$ be the indices with $𝘡_𝘴<θ$.
	That is, if $𝘴$ is set to be $m$ because $𝘡_m<θ$,
	then the corresponding $𝘑₁^𝘴$ are included in $𝒥$.
	If $𝘴$ is set to be $m$ because $𝘛_m<θ$ or because we reach $m=n$,
	then the corresponding $𝘑₁^𝘴$ are not included in $𝒥$.
	
	The block error probability incurred by this choice of $𝘴$ and $𝒥$
	is bounded by the weighted sum of $Z$'s that are all smaller than $θ≔4^{-n}$.
	Hence $Ｐ<Nθ=1/N$ by \cref{lem:prune-P}.
	
	The complexity and code rate are more complicated.
	See the next two theorems.
	
	\begin{thm}\label{thm:actual-C}
		Given that $𝘴$ is defined as in the first paragraph of this section,
		an upper bound on the complexity is
		\[O(N𝘌[𝘴])=O(N㏒(㏒N)).\]
	\end{thm}
	
	\begin{proof}
		By how $𝘴$ is defined, we have
		\[𝘌[𝘴]=∑_{m=0}^{n-1}𝘗\{𝘴>m\}=∑_{m=0}^{n-1}𝘗\{𝘛_m≥4^{-n}† and †𝘡_m≥4^{-n}\}.
			\label{ine:fubini}\]
		It remains to understand what is going on in $\{𝘛_m≥4^{-n}$ and $𝘡_m≥4^{-n}\}$.
		For that, recall the last two chapters:
		\begin{gather*}
			𝘗\{𝘡_m<\exp(-2^{m/40})\}>I(W)+2^{-m/5+o(m)},	\\
			𝘗\{𝘛_m<\exp(-2^{m/40})\}>H(W)+2^{-m/5+o(m)}.	
		\end{gather*}
		Here I choose $(π,ρ)=(1/40,1/5)∈𝒪$.
		In words, I showed that either $𝘡_m$ or $𝘛_m$,
		exclusively, becomes small with high probability.
		Hence it is unlikely that $𝘛_m≥4^{-n}$ and $𝘡_m≥4^{-n}$ at once
		unless $m$ is small, in which case the proved threshold
		$\exp(-2^{m/40})$ is less than the demanded $4^{-n}$.
		
		We now classify $m$ into two classes:
		Those such that $\exp(-2^{m/40})≥4^{-n}$ are called small $m$.
		Those such that $\exp(-2^{m/40})<4^{-n}$ are called large $m$.
		For small $m$, we have nothing but $𝘗\{$both $𝘛_m,𝘡_m≥4^{-n}\}≤1$.
		For large $m$,
		\[𝘗\{†both †𝘛_m,𝘡_m≥4^{-n}\}
			≤𝘗\{†both †𝘛_m,𝘡_m≥\exp(-2^{m/40})\}<2^{-m/5+o(m)}.\]
		Continue bounding \cref{ine:fubini}:
		\begin{align*}
			\qquad&\kern-2em
			∑_{m=0}^{n-1}𝘗\{†both †𝘛_m,𝘡_m≥4^{-n}\}	\\
			&	=∑_{†small †m}𝘗\{†both †𝘛_m,𝘡_m≥4^{-n}\}
				+∑_{†large †m}𝘗\{†both †𝘛_m,𝘡_m≥4^{-n}\}	\\
			&	=∑_{†small †m}1+∑_{†large †m}2^{-m/5+o(m)}=\#\{†small †m†'s†\}+O(1).
		\end{align*}
		The number of small $m$'s is the root of the euqation $\exp(-2^{m/40})=4^{-n}$.
		The root is $m=O(㏒n)$.
		Hence the complexity $𝘌[𝘴]<O(㏒n)=O(㏒(㏒N))$, as desired.
	\end{proof}
	
	\begin{thm}\label{thm:actual-R}
		Given that $𝘴$ and $𝒥$ are defined as in the first paragraph of this section,
		a lower bound on the code rate $R$ is
		\[R≥I(W)-2^{-n/5+o(n)}.\]
	\end{thm}
	
	\begin{proof}
		Just to clarify the whole picture, let me claim the following trichotomy:
		\begin{itemize}
			\item	The frequency that $𝘴$ is set to $m$
					due to $𝘛_m<θ$ is $H(W)-2^{-n/5+o(n)}$.
			\item	The frequency that $𝘴$ is set to $m$
					due to reaching $n$ is $2^{-n/5+o(n)}$.
			\item	The frequency that $𝘴$ is set to $m$
					due to $𝘡_m<θ$ is $I(W)-2^{-n/5+o(n)}$.
		\end{itemize}
		$𝒥$ collects those $𝘑₁^𝘴$ when $𝘡_m<θ$ is the case.
		So the theorem statement is implied by the last bullet point.
		I will show the last bullet point;
		the proof thereof applies to the first bullet point,
		so I am effectively showing all three bullet points at once.
		
		In order to show the last bullet point,
		it suffices to understand the bad event $𝘉≔\{𝘡_𝘴≥θ$ but $𝘡_m→0\}$.
		Once we know how to bound $𝘗(𝘉)$, we will conclude that,
		with \cref{lem:prune-R}, the code rate is $R=𝘗\{𝘡_𝘴<θ\}≥𝘗\{𝘡_m→0\}-𝘗(𝘉)$,
		which will be $I(W)-2^{-n/5+o(n)}$.
		To bound $𝘗(𝘉)$, two cases will be discussed---$𝘡_𝘴≥θ$ because
		$𝘛_m<θ$ happens first, and $𝘡_𝘴≥θ$ because we reach $m=n$.
		
		Here goes the rigorous bound on $𝘗(𝘉)$:
		If $𝘡_𝘴≥θ$, then $𝘴$ is not set to the current value because $𝘡_m<θ$;
		it must be the case that the other criterion $𝘛_𝘴<θ$ holds,
		or that we reach $𝘴=n$.
		For the former case, $𝘗\{𝘛_𝘴<θ$ but $𝘡_m→0\}=𝘗\{𝘛_𝘴<θ$ but $𝘛_m→1\}≤θ=2^{-n}$
		by the common-fate property and $\{𝘛_n\}$ being a martingale.
		For the latter case,
		$𝘗\{𝘴=n$ and $𝘡_m→0\}<𝘗\{𝘴=n\}=𝘗\{$both $𝘛_m,𝘡_m≥θ$ for all $m<n\}
			<𝘗\{$both $𝘛_{n-1},𝘡_{n-1}≥θ\}=2^{-(n-1)/5+o(n-1)}$.
		Sum the upper bounds of the two cases;
		it is $2^{-n/5+o(n)}$.
		
		Now that we get $𝘗(𝘉)<2^{-n/5+o(n)}$, deduce the code rate
		$R=𝘗\{𝘡_𝘴<θ\}=𝘗\{𝘡_m→0\}-𝘗(𝘉)=I(W)-2^{-n/5+o(n)}$;
		that is what we want to prove.
	\end{proof}
	
	Recap:
	So far we had seen that, when a code is defined by my choice of $𝘴$ and $𝒥$,
	(a)	the block error probability is $1/N$,
	(b)	the complexity is $O(N㏒(㏒N))$ per code block
		or $O(㏒(㏒N))$ per channel usage, and
	(c)	the code rate is $I(W)-2^{-n/5+o(n)}$.
	The minor term $o(n)$ within $R$ can be eliminated by a topological argument.
	We summarize the chapter with the following corollary.
	
	\begin{cor}[Log-log code for SBDMC]
		Over any SBDMC, there exist capacity-achieving codes with encoding and decoding
		complexity $O(N㏒(㏒N))$ per code block or $O(㏒(㏒N))$ per channel usage.
		Moreover, the said codes have gaps to capacity $I(W)-R<N^{-1/5}$.
	\end{cor}
	
	In one sentence, \emph{log-logarithmic complexity achieves capacity}.

\section{A Note on Better Pace}

	The only place we used the pair $(1/40,1/5)∈𝒪$ explicitly is when
	we were solving $\exp(-2^{m/40})=4^{-n}$ for $m$ and obtain that $m=O(㏒n)$.
	We could, instead, choose another pair such as $(1/150,1/4.8)∈𝒪$.
	Then, we still get to keep $m=O(㏒n)$
	but the gap to capacity is $I(W)-R<N^{-1/4.8}$.
	By a topological argument, any $ρ<1/4.714$ works.
	This is why I claimed that the pruned polar code
	has the same gap to capacity as the original version.
	
	Concerning the equation $\exp(-2^{m/40})=4^{-n}$,
	we can also replace $4^{-n}$ with $\exp(-n^τ)$ for arbitrarily large $τ$.
	The root is then $m=O(τ㏒n)$, which is $O(㏒n)$ if $τ$ is fixed.
	That means the block error probability can be as low as $\exp(-n^τ)$
	while keeping the log-logarithmic complexity.
	This asymptote lies below $1/\poly(N)$ and belongs to $\exp(-1\poly(㏒N))$.
	This is why I claimed that the pruned polar code
	has quasi-polynomial block error probability.

\section{A Note on Asymmetric Case}

	For if the concerned channel $W$ is an asymmetric BDMC, we will need to apply
	a similar argument to $\{𝘘_n\}$, the channel process grown from
	the non-uniform input distribution $Q$ treated as a channel with constant output.
	
	In the asymmetric case, the stopping time will be defined as
	\[𝘴≔n∧\min\{m:\min(Z(𝘞_m),T(𝘞_m))<4^{-n}† and †\min(Z(𝘘_m),T(𝘘_m))<4^{-n}\}.\]
	To put in another way, we are looking for the minimal $m$
	that satisfies any of the following:
	\begin{itemize}
		\item	Both $Z(𝘞_m),Z(𝘘_m)$ are small.
		\item	Both $T(𝘞_m),T(𝘘_m)$ are small.
		\item	Or, both $Z(𝘞_m),T(𝘘_m)$ are small.
		\item	(Both $T(𝘞_m),Z(𝘘_m)$ being small is not possible.)
	\end{itemize}
	And if no $m$ meets the requirement before depth $n$,
	the stopping time defaults to $n$.
	Furthermore, $𝒥$ will be the indices $𝘑₁^𝘴$ with small $Z(𝘞_m)$ and small $T(𝘘_𝘴)$:
	that is, the indices where $𝘴$ is set to $m$ due to the third bullet point.
	
	The block error probability is upper bounded similarly to \cref{ine:P<sumsum}:
	\[Ｐ≤N𝔼[Z(W_𝘴)𝘐(𝒥)]+N𝔼[T(𝘘_𝘴)𝘐(𝒥)]≤2N·4^{-n}=2^{-n+1}.\]
	The complexity is upper bounded similarly to the symmetric case
	\begin{align*}
		𝘌[𝘴]
		&	=∑_{m=0}^{n-1}𝘗\{𝘴>m\}	\\
		&	=∑_{m=0}^{n-1}𝘗\{\max(Z(𝘞_m),Z(𝘞_m))≥4^{-n}
			† or †\max(Z(𝘘_m),Z(𝘘_m))≥4^{-n}\}	\\
		&	≤∑_{m=0}^{n-1}𝘗\{\max(Z(𝘞_m),Z(𝘞_m))≥4^{-n}\}
			+𝘗\{\max(Z(𝘘_m),Z(𝘘_m))≥4^{-n}\}	\\
		&	=O(㏒n)+O(1)+O(㏒n)+O(1)=O(㏒n).
	\end{align*}
	The code rate is lower bounded similarly to the symmetric case
	\begin{align*}
		R
		&	=𝘗\{𝘑₁^𝘴∈𝒥\}=𝘗\{†both †Z(𝘞_𝘴),T(𝘘_𝘴)<4^{-n}\}	\\
		&	≥𝘗\{Z(𝘞_m),T(𝘘_m)→0\}-𝘗\{Z(𝘞_m)→0† but †Z(𝘞_𝘴)≥4^{-n}\}	\\*
		&\kern12em	-𝘗\{T(𝘘_m)→0† but †T(𝘘_𝘴)≥4^{-n}\}	\\
		&	≥𝘗\{Z(𝘞_m),T(𝘘_m)→0\}-𝘗\{Z(𝘞_m)→0† but †T(𝘞_𝘴)≤4^{-n}\}	\\*
		&\kern12em	-𝘗\{T(𝘘_m)→0† but †Z(𝘘_𝘴)≤4^{-n}\}-𝘗\{𝘴=n\}	\\
		&	≥I(W)-4^{-n}-4^{-n}-2^{-ρn+o(n)}=I(W)-2^{-ρn+o(n)}.
	\end{align*}
	Hence the following corollary.
	
	\begin{cor}[Log-log code for BDMC]
		Over any BDMC, there exist capacity-achieving codes with encoding and decoding
		complexity $O(N㏒(㏒N))$ per code block or $O(㏒(㏒N))$ per channel usage.
		Moreover, the said codes have gaps to capacity $I(W)-R<N^{-1/5}$.
	\end{cor}
	
	The same statement can be made to lossless and lossy compression.
	
	\begin{thm}[Log-log code for lossless compression]
		If $X$ is a binary source and $Y$ is the side information,
		then lossless compression can be done with encoding and decoding complexity
		$O(N㏒(㏒N))$ per code block or $O(㏒(㏒N))$ per source observation,
		block error probability $Ｐ<1/N$, and gap to entropy $H(X｜Y)-R<N^{-1/5}$.
	\end{thm}
	
	\begin{cor}[Log-log code for lossy compression]
		If $Y∈𝒴$ is a random source and $\dist：𝔽₂×𝒴→[0,1]$ is the distortion function,
		then lossy compression can be performed with encoding and decoding complexity
		$O(N㏒(㏒N))$ per code block or $O(㏒(㏒N))$ per source observation.
		excess of distortion $D-Δ<1/N²$, and gap to capacity $R-I(W)<N^{-1/5}$.
	\end{cor}
	
	The discussion in the previous section applies.
	That is to say, the gap to capacity can be made as low as $N^{-1/4.8}$,
	the block error probability (or the excess of distortion)
	as low as $\exp(-n^τ)$, while retaining the same log-logarithmic complexity.

\section{Prospective Pruning}

	The pruning technique generalizes to arbitrary input (non-binary)
	and arbitrary matrix (no longer $\loll$).
	The next two chapters establish the theory of $\{𝘞_n\}$ for the general scenario.
	After that we can prune, yielding $𝘞_𝘴$, etc.

\chapter{General Alphabet and Kernel}\label{cha:general}

	\dropcap
	Binary channel models, be it BEC, BSC, BDMC, or BI-AWGN, are favored
	in theory for their simplicity and for real world communications.
	But the real world is not always binary;
	we might want to compress losslessly a non-binary source (yes, no, or unanswered);
	or we might want to approximate colors with a $256$-color palette.
	This is one of the two generalizations I want to make
	in this chapter---to enable all input alphabets of finite sizes \cite{STA09}.
	
	Enabling more input alphabets creates new challenges.
	The root of all challenges is that, to synthesize $WＷ1$ and $WＷ2$,
	we need to talk about the random variables $U₁²≔X₁²G^{-1}$ (c.f.\ \cref{for:U=X/G}).
	This does not a priori make sense as the addition structure
	on the input alphabet $𝒳$ does not uniquely exist.
	Plus, even if we equip $𝒳$ with some group structure (abelian or not),
	there are cases where the descendants $𝘞_n$ are not polarized.
	This chapter addresses this challenge---how to equip $𝒳$ with
	a proper algebraic structure to facilitate polarization,
	even if it means adding dummy symbols into $𝒳$.
	
	Specking of the root of the definition $U₁²≔X₁²G^{-1}$,
	it is not hard to imagine that, if we substitute $G$ with a general $ℓ×ℓ$ matrix,
	the overall performance of polar coding may differ \cite{KSU10}.
	To tell if it becomes better, worse, or not working at all,
	we need a machinery that predicts the performance of each matrix.
	This is the second generalization I want to make
	in this chapter---to judge all matrices as the polarizing kernel.
	
	The judgement of a matrix will be
	very similar to those in \cref{cha:origin,cha:dual}.
	We will need to prove (or assume) an eigen behavior.
	From that we can derive the en23 behavior (named after $\exp(-n^{2/3})$),
	the een13 behavior (named after $\exp\(-e^{n^{1/3}}\)$),
	and finally the elpin behavior (named after $\exp(-ℓ^{πn})$).
	Along the way, I will explain how these behaviors relate to
	LLN, LDP, CLT, and MDP in probability theory.
	
	This is a long chapter.
	Here goes the actual plan for this chapter.

\section{Chapter Organization}

	I will define the most general channel (\cref{sec:setup}).
	Following that is a reduction of input alphabet
	to prime size or prime-power size (\cref{sec:six}).
	Even more definitions are then made, including channel parameters
	$Ｐ$, $Z$, $Ｚ$, $T$, $S$, and $Ｓ$ (\cref{sec:param}) and how to use
	an invertible $ℓ×ℓ$ matrix to synthesize $WＷ1,WＷ2…WＷ{ℓ}$ (\cref{sec:matrix}).
	Inserted here is a briff review of LLN, LDP, CLT, and MDP in probability theory
	and how they relate to coding theory (\cref{sec:regime}).
	I will then clarify that only a subset of matrices are
	qualified to polarize channels (\cref{sec:ergodic}).
	After that is a compactness argument that shows
	$ϱ>0$ for those qualified matrices (\cref{sec:en23}).
	Finally, the region of $(π,ρ)$ where $𝘗\{𝘡_n<\exp(-ℓ^{πn})\}>1-H(W)+ℓ^{-ρn}$
	will be pictured in two steps (\cref{sec:een13,sec:elpin}).

\section{General Problem Setup}\label{sec:setup}

	The following channel is what was considered by Shannon in the eternal paper,
	followed by other big scholars who generalized Shannon.
	It is of historical interest to prove theorems over the same class of channels.
	
	\begin{dfn}
		A \emph{discrete memoryless channel} (DMC) is a Markov chain $W：𝒳→𝒴$ such that
		\begin{itemize}
			\item	the input alphabet $𝒳$ is a finite set,
			\item	the output alphabet $𝒴$ is a finite set, and
			\item	$W(y｜x)$ is the array of transition probabilities satisfying	\\
					$∑_{y∈𝒴}W(y｜x)=1$ for all $x∈𝒳$.
		\end{itemize}
	\end{dfn}
	
	Denote by $Q(x)$ an input distribution, by $W(x,y)$ the joint distribution,
	by $W(x｜y)$ the posterior distribution, and by $W(y)$ the output distribution.
	
	\begin{dfn}
		The \emph{conditional entropy} of $W$ is
		\[H(W)≔-∑_{x∈𝒳}∑_{y∈𝒴}W(x,y)㏒_{\abs{𝒳}}W(x｜y),\]
		which is the amount of noise/equivocation/ambiguity/fuzziness caused by $W$.
	\end{dfn}
	
	\begin{dfn}
		The \emph{mutual information} of $W$ is
		\[I(W)≔H(Q)-H(W)=∑_{x∈𝒳}∑_{y∈𝒴}W(x,y)㏒_{\abs{𝒳}}÷{W(x｜y)}{Q(x)},\]
		Fix a $Q$ that maximizes $I(W)$.
		Call this $Q$ a \emph{capacity-achieving input distribution}.
		Call this maximal $I(W)$ the \emph{channel capacity} of $W$.
	\end{dfn}
	
	Let $𝒰$ be a finite set called the user alphabet.
	Messages to be transmitted over $W$ will be
	pre-encoded into $𝒰$ and distributed uniformly.
	The overall goal is to construct, for some large $N$, an encoder $ℰ：𝒰^{RN}→𝒳^N$
	and a decoder $𝒟：𝒴^N→𝒰^{RN}$ such that the composition
	\cd[every arrow/.append style=mapsto]{
		U₁^{RN}\rar{ℰ}	&	X₁^N\rar{W^N}	&	Y₁^N\rar{𝒟}	&	ˆU₁^{RN}
	}
	is the identity map as frequently as possible,
	and $R$ as close to the channel capacity $I(W)$ as possible.
	
	\begin{dfn}
		Call $N$ the block length.
		Call $R$ the code rate.
		Denote by $Ｐ$, called the block error probability,
		the probability that $ˆU₁^{RN}≠U₁^{RN}$.
	\end{dfn}
	
	In the next section, I will argue that we can reduce the problem of noisy-channel
	coding over arbitrary DMCs to $𝒰$ and $𝒳$ having finite field structure.

\section{Pay Asymmetry to Buy Non-Field}\label{sec:six}

	Reducing the size of user alphabet is easier and
	will be handled before the reduction of input alphabet:
	Consider integer factorization $|𝒰|=p₁^{k₁}p₂^{k₂}\dotsm p_ω^{k_ω}$.
	We can then choose a bijection
	\[𝒰≅𝔽_{p₁^{k₁}}×𝔽_{p₂^{k₂}}×\dotsb×𝔽_{p_ω^{k_ω}}.\]
	To transmit a symbol $u∈𝒰$, it suffices to transmit its projections
	onto each of the constituent fields $𝔽_{p₁^{k₁}}…𝔽_{p_ω^{k_ω}}$.
	It remains to show, for each finite filed $𝔽_{p^k}$,
	how to design encoder/decoder for the user alphabet $𝒰=𝔽_{p^k}$.
	
	Moreover, each finite field admits a vector space structure
	\[𝔽_{p^k}≅𝔽_p^k.\]
	Therefore, when necessary, each $u∈𝔽_{p^k}$ is considered as a vector
	in $𝔽_p^k$, with each of its coordinates treated as a standalone symbol in $𝔽_p$.
	That is to say, when necessary, we can even assume $𝒰=𝔽_p$ is of prime order.
	
	Reducing the size of user alphabet takes effort and will be handled in the sequel:
	Recall that $\abs{𝒰}$ is either a prime or a prime power.
	Let $q$ be any power of $\abs{𝒰}$ greater than or equal to $\abs{𝒳}$.
	Consider an embedding $𝒳⊆𝔽_q$;
	call an $x∈𝒳$ an actual symbol;
	call a $v∈𝔽_q、𝒳$ a virtual symbol.
	Define a preprocessor channel $♮：𝔽_q→𝒳$ that sends an actual symbol to itself,
	and a virtual symbol to a fixed actual symbol $x₁∈𝒳$.
	In terms of transition probabilities:
	\[♮(y｜x)=\cas{
		1	&	if $y=x∈𝒳$,	\\
		1	&	if $y=x₁$ and $x∉𝒳$,	\\
		0	&	otherwise.
	}\]
	That is to say, all virtual symbols are semantic copies of $x₁$;
	they make no difference to the decoder whatsoever.
	
	Now build the following degraded channel
	\cd{
		𝔽_q\rar{♮}	&	𝒳\rar{W}	&	𝒴.
	}
	By the data processing inequality, the channel capacity of $W∘♮$
	is at most the channel capacity of $W$.
	Hence using $𝔽_q$ as the input alphabet does no harm at best,
	and handicaps ourselves at worst.
	The former is the case:
	If we take any capacity-achieving input distribution $Q$ of $W$ and
	apply it directly to $W∘♮$, then we can transmit information at the same rate.
	In conclusion, $W$ and $W∘♮$ share the same channel capacity.
	It suffices to consider $W∘♮$---channels with prime-power sized
	input alphabet---for the remainder of this chapter.
	
	Remark:
	$♮∘W$ is not a symmetric channel;
	in particular, the uniform distribution does not achieve the best transmission rate.
	(Because that means $x₁$ is input more frequently than it should have been.)
	As a consequence, this technique of allowing dummy symbols, albeit trivial,
	does not apply until \cite{HY13} taught us how to deal with asymmetric channels.
	Hence the section title \emph{pay asymmetry to buy non-field}.
	
	We hereby assume that $𝒳=𝔽_q$ is a finite field.
	If necessary, $p$ is the characteristic of $𝔽_q$
	and $𝔽_p$ denotes the ground field of $𝔽_q$.
	I am going to define more channel parameters in the next section;
	some of them require the finite field structure on $𝒳$
	in a way that is not required/obvious when $q=2$.

\section{Parameters and Hölder Tolls}\label{sec:param}

	Let $W$ be any DMC with input alphabet $𝒳=𝔽_q$ for some prime power $q$.
	To emphasize that the input alphabet was reduced to $q$,
	we call $W$ a \emph{$q$-ary channel}.
	The parameters below generalize $H$, $I$, and $Z$ defined in \cref{cha:origin}
	and $T$ defined in \cref{cha:dual}.
	
	\begin{dfn}
		The \emph{bit error probability} of a $q$-ary channel $W$ is
		\[Ｐ(W)≔∑_{y∈𝒴}W(y)（1-\max_{x∈𝒳}W(x｜y)）,\]
		which is how often the maximum a posteriori estimator
		$y↦\argmax_{x∈𝒳}W(x｜y)$ makes a mistake.
	\end{dfn}
	
	\begin{dfn}
		The \emph{Bhattacharyya parameter} of a $q$-ary channel $W$ is
		\[Z(W)≔÷1{q-1}∑_{\substack{x,x'∈𝔽_q\\x≠x'}}
			∑_{y∈𝒴}√{W(x,y)W(x',y)}.\]
		In addition, define
		\[Ｚ(W)≔\max_{0≠d∈𝔽_q}∑_{x∈𝔽_q}∑_{y∈𝒴}√{W(x,y)W(x+d,y)}.\]
		The parameter below is not used directly in this work,
		but plays a central role in a theorem I cite.
		I include it for completeness:
		\[Z_d(W)≔∑_{x∈𝔽_q}∑_{y∈𝒴}√{W(x,y)W(x+d,y)}.\]
	\end{dfn}
	
	Remarks:
	$Z$ is the average of $Z_d$ over $d∈𝔽_q^×$ and $Ｚ$ is the maximum thereof.
	The rescaling is such that $0≤Z≤Ｚ≤(q-1)Z≤q-1$ and that $0≤Z_d≤1$.
	When $q=2$, i.e., for the binary case, $Z=Ｚ=Z₁=$
	the $Z$-parameter defined in \cref{dfn:bin-Z}.
	Notice how the summation over distinct $x,x'∈𝔽₂$
	turns into twice the case of $(x,x')=(0,1)$.
	
	\begin{dfn}
		The \emph{total variation norm} of a $q$-ary channel $W$ is
		\[T(W)≔∑_{y∈𝒴}W(y)∑_{x∈𝒳}\abs[\Big]{W(x｜y)-÷1q}.\]
		which is the total variation distance from $W(•｜y)$ to
		the uniform distribution, weighted by the frequency each $y∈𝒴$ appears.
	\end{dfn}
	
	\begin{dfn}
		The \emph{Fourier coefficients} of a $q$-ary channel $W$ is
		\[M(w｜y)≔∑_{z∈𝔽_q}W(z｜y)χ(wz),\]
		where $χ：𝔽_q→ℂ$ is an additive character defined as $χ(x)≔\exp(2πi\tr(x)/p)$,
		and $\tr：𝔽_q→𝔽_p$ is the field trace onto the ground field.
	\end{dfn}
	
	\begin{dfn}
		The \emph{Fourier $ℓ¹$ norm} of a $q$-ary channel $W$ are
		\[S(W)≔÷1{q-1}∑_{0≠w∈𝔽_q}∑_{y∈𝒴}W(y)·\abs[\Big]{M(w｜y)}\]
		In addition, define
		\[Ｓ(W)≔\max_{0≠w∈𝔽_q}∑_{y∈𝒴}W(y)·\abs[\Big]{M(w｜y)}.\]
	\end{dfn}
	
	Remarks:
	The rescaling is such that $0≤S≤Ｓ≤(q-1)S≤q-1$.
	When $q=2$, the parameters collapse---$S=Ｓ=T=1-2Ｐ=$
	the $T$-parameter defined in \cref{dfn:bin-T}.
	
	I borrow some lemmas to relate these parameters.
	
	\begin{lem}[$P$ vs $Z$]\label{lem:PvsZ}
		\cite[Lemma~22 with $k=1$]{MT14}
		For any $q$-ary channel $W$,
		\[÷{q-1}{q²}（√{1+(q-1)Z(W)}-√{1-Z(W)}）²≤Ｐ(W)≤÷{q-1}2Z(W).\]
	\end{lem}
	
	\begin{lem}[$P$ vs $T$]\label{lem:PvsT}
		\cite[Lemma~23 with $k=q-1$]{MT14}
		For any $q$-ary channel $W$,
		\[÷{q-1}q-Ｐ(W)≤÷{T(W)}2≤÷{q-1}q-÷1q（(q-1)qＰ(W)-(q-1)(q-2)）.\]
	\end{lem}
	
	\begin{lem}[$P$ vs $S$]\label{lem:PvsS}
		\cite[Lemma~26 with $k=q-1$]{MT14}
		For any $q$-ary channel $W$,
		\[1-÷q{q-1}Ｐ(W)≤S(W)≤(q-1)q（÷{q-1}q-Ｐ(W)）√{1-÷q{q-1}÷{q-2}{q-1}}.\]
	\end{lem}
	
	\begin{lem}[$P$ vs $H$]\label{lem:PvsH}
		\cite[Theorem~1]{FM94}
		For any $q$-ary channel $W$,
		\begin{gather*}
			h₂(Ｐ(W))+Ｐ(W)㏒₂(q-1)≥H(W)㏒₂q≥2Ｐ(W),	\\
			H(W)㏒₂q≥(q-1)q㏒₂÷q{q-1}（Ｐ(W)-÷{q-2}{q-1}）+㏒₂(q-1).
		\end{gather*}
		Here, $h₂$ is the binary entropy function.
		The upper bound is Fano's inequality.
		The first lower bound is useful when $H(W)$ and $Ｐ(W)$ are small;
		the second lower bound is useful when $H(W)$ and $Ｐ(W)$ are close to $1$.
	\end{lem}
	
	The phenomenon described by the cited lemmas
	is distilled to form the following definition.
	It is the quantified version of the common-fate property.
	
	\begin{dfn}
		Fix a $q$.
		Let $A$ and $B$ be two channel parameters.
		Say $A$ and $B$ are \emph{bi-Hölder at $(a,b)$}
		if there exist constants $c,d>0$ such that, for all $q$-ary channels,
		$\abs{A(W)-a}≤c·\abs{B(W)-b}^d$ and $\abs{B(W)-b}≤c·\abs{A(W)-a}^d$.
	\end{dfn}
	
	Bi-Hölder-ness is clearly an equivalence relation.
	In particular, we will make use of its transitivity property---%
	if $A$ and $B$ are bi-Hölder at $(a,b)$ and $B$ and $C$ are bi-Hölder at $(b,c)$,
	then $A$ and $C$ are bi-Hölder at $(a,c)$.
	In this case, we say $A$, $B$, and $C$ are bi-Hölder at $(a,b,c)$.
	This notion generalizes to arbitrarily many parameters.
	What \cref{lem:PvsZ,lem:PvsT,lem:PvsS,lem:PvsH} tell us
	is where the parameters are bi-Hölder at.
	
	\begin{pro}[Implicit bi-Hölder tolls]\label{pro:im-toll}
		Fix a prime power $q$.
		Then channel parameters $H$, $Ｐ$, $Z$, and $Ｚ$ are bi-Hölder at $(0,0,0,0)$.
		Channel parameters $H$, $Ｐ$, $T$, $S$, and $Ｓ$
		are bi-Hölder at $(1,1-1/q,0,0,0)$.
	\end{pro}
	
	\begin{proof}
		As for the first statement:
		$Z,Ｚ$ are bi-Hölder at $(0,0)$ since $Z≤Ｚ≤(q-1)Z$.
		\Cref{lem:PvsZ} implies that $Ｐ$ and $Z$ are bi-Hölder at $(0,0)$.
		\Cref{lem:PvsH} (with the first lower bound)
		implies that $Ｐ$ and $H$ are bi-Hölder at $(0,0)$.
		Now apply the transitivity to conclude the first statement.
		
		As for the second statement:
		$S$ and $Ｓ$ are bi-Hölder at $(0,0)$ since $S≤Ｓ≤(q-1)S$.
		\Cref{lem:PvsS} implies that $Ｐ$ and $S$ are bi-Hölder at $(1-1/q,0)$.
		\Cref{lem:PvsT} implies that $Ｐ$ and $T$ are bi-Hölder at $(1-1/q,0)$.
		\Cref{lem:PvsH} (with the second lower bound)
		implies that $Ｐ$ and $H$ are bi-Hölder at $(1-1/q,1)$.
		Now apply the transitivity to conclude the second statement.
	\end{proof}
	
	In short, when a channel is reliable, $H,Ｐ,Z,Ｚ$ are small and $T,S,Ｓ$ are large.
	When a channel is noisy, $H,Ｐ,Z,Ｚ$ are large and $T,S,Ｓ$ are small.
	See \cref{fig:holder} for an example of the possible region.
	
	\begin{figure}
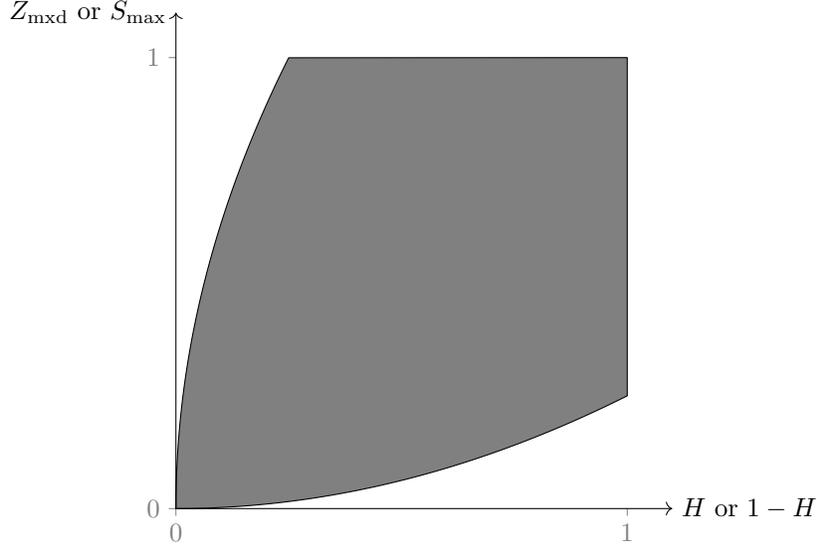

		\tikz[scale=6]{
			\draw
				(0,1)pic{y-tick=$1$}(0,0)pic{y-tick=$0$}
				(0,0)pic{x-tick=$0$}(1,0)pic{x-tick=$1$};
			\draw[fill=gray]
				plot[domain=0:1,samples=180](\x/2*\x/2,\x)--(1,1)--
				plot[domain=1:0,samples=180](\x,\x/2*\x/2)--cycle;
			\draw[->](0,0)--(0,1.1)node[left]{$Ｚ$ or $Ｓ$};
			\draw[->](0,0)--(1.1,0)node[right]{$H$ or $1-H$};
		}
		\caption{
			A figurative example of bi-Hölder relation.
			This plot assumes $c=2$ and $d=1/2$.
			But $c$ is in general greater, which makes the dark region larger.
		}\label{fig:holder}
	\end{figure}
	
	Sometimes, explicit bounds are desired as it saves some epsilon--delta notations.
	The bounds do not have to be tight,
	but they should be as easy to manipulate as possible.
	The next proposition serves exactly that purpose.
	
	\begin{pro}[Explicit Hölder tolls]\label{pro:ex-toll}
		For all $q$-ary channels $W$, it holds that
		\begin{gather*}
			Ｚ(W)≤q³√{H(W)},		\label{ine:Z<H}\\
			H(W)≤q³√{Ｚ(W)},		\label{ine:H<Z}\\
			Ｓ(W)≤q³√{1-H(W)},	\label{ine:S<1-H}\\
			1-H(W)≤q³√{Ｓ(W)}.	\label{ine:1-H<S}
		\end{gather*}
	\end{pro}
	
	\begin{proof}
		The proof is nothing but working out
		\cref{lem:PvsZ,lem:PvsT,lem:PvsS,lem:PvsH} very carefully.
		By “working out”, I mean to Taylor expand every inequality at a proper point
		and keep track of how the constants $c,d$ accumulate as transitivity applies.
		
		In the upcoming arguments, $H$, $Ｐ$, $Z$, $Ｚ$, $S$, and $Ｓ$ mean
		$H(W)$, $Ｐ(W)$, $Z(W)$, $Ｚ(W)$, $S(W)$, and $Ｓ(W)$, respectively.
		Also $q'$ means $q-1$, and $q''$ means $q-2$.
		Furthermore, $\lg$ means the base-$2$ logarithm;
		this is handy when we jump back and forth between nats, bits, and $q$-bits.
		
		First we show \cref{ine:Z<H}.
		Start from $Ｚ$:
		By definition $Ｚ≤q'Z$.
		Move on to $Z$:
		By \cref{lem:PvsZ}, $q'q^{-2}(√{1+q'Z}-√{1-Z})²≤Ｐ$,
		hence $√{1+q'Z}-√{1-Z}≤q√{Ｐ/q'}$.
		Multiplying both sides by the conjugate yields
		$(1+q'Z)-(1-Z)≤q√{Ｐ/q'}(√{1+q'Z}+√{1-Z})$.
		The left-hand side is $qZ$;
		in the right-hand side $√{1+q'z}+√{1-z}$ assumes
		the maximum $q/√{q'}$ at $z=q''/q'$ by taking derivative.
		So $Z≤√{Ｐ/q'}(q/√{q'})=q√{Ｐ}/q'$.
		Move on to $Ｐ$:
		By \cref{lem:PvsH} (the first lower bound),
		$2Ｐ≤H\lg q$ or equivalently $Ｐ≤H㏒₄q$.
		Now we chain the inequalities $Ｚ≤q'Z≤q√{Ｐ}≤q√{H㏒₄q}$.
		This completes \cref{ine:Z<H} as $q√{㏒₄q}<q³$.
		
		Second we show \cref{ine:H<Z}.
		Start from $H$:
		By \cref{lem:PvsH} (the upper bound, Fano's inequality),
		$H\lg q≤h₂(Ｐ)+Ｐ\lg q'$.
		By \cref{fig:rootep}, $h₂(Ｐ)+Ｐ\lg q'≤√{eＰ}+Ｐ\lg q'=√{Ｐ}(√e+√{Ｐ}\lg q')$.
		What is inside parentheses is less than $√e+√{q'/q}\lg q'$.
		Hence $H≤√{Ｐ}(√e+√{q'/q}\lg q')/\lg q$.
		Focus on the scalar---$(√e+√{q'/q}\lg q')/\lg q$
		has maximum $√e$ at $q=2$ (remember that $q≥2$).
		So $H≤√{eＰ}$.
		Move on to $Ｐ$:
		By \cref{lem:PvsZ}, $Ｐ≤q'Z/2$.
		Move on to $Z$:
		By definition $Z≤Ｚ$.
		Now we chain the inequalities $H≤√{eＰ}≤√{eq'Z/2}≤√{eq'Ｚ/2}$.
		This completes \cref{ine:H<Z} as $√{eq'/2}<q³$.
		
		Third we show \cref{ine:S<1-H}.
		Start from $Ｓ$:
		By definition $Ｓ≤q'S$.
		Move on to $S$:
		By \cref{lem:PvsS}, $S≤q'q(q'/q-Ｐ)√{1-÷q{q'}÷{q''}{q'}}$.
		The square root simplifies to $√{1/(q')²}=1/q'$ as $qq''=(q')²-1$.
		So $S≤q'-qＰ$.
		Move on to $q'-qＰ$:
		By \cref{lem:PvsH} (the upper bound, Fano's inequality),
		$H\lg q≤h₂(Ｐ)+Ｐ\lg q'$.
		We claim that $h₂(Ｐ)+Ｐ\lg q'≤\lg q-2(q'/q-Ｐ)²/㏑2$.
		To prove the claim, Taylor expand both sides at $Ｐ=q'/q$.
		Verify that both evaluate to $\lg q$ at $Ｐ=q'/q$;
		verify that both have derivative $0$ at $Ｐ=q'/q$; and
		verify that the acceleration of the left-hand side, $-1/(Ｐ(1-Ｐ)㏑2)$,
		is more negative than the acceleration of the right-hand side, $-4/㏑2$.
		By Taylor's theorem, mean value theorem, or the Euler method,
		the function with greater acceleration is greater;
		hence the claim.
		See also \cite[Fig.~1]{FM94};
		the $Φ$-curve seems parabolic at the upper right corner.
		Now we have $H\lg q≤\lg q-2(q'/q-Ｐ)²/㏑2$, which is equivalent to
		$2(q'/q-Ｐ)²/㏑q≤1-H$ and to $q'-qＰ≤q√{(1-H)㏑(q)/2}$.
		Now we chain the inequalities $Ｓ≤q'S≤q'(q'-qＰ)≤q'q√{(1-H)\ln(q)/2}$.
		This completes \cref{ine:S<1-H} as $q'q√{\ln(q)/2}<q³$.
		
		Fourth we show \cref{ine:1-H<S}.
		Start from $1-H$:
		By \cref{lem:PvsH} (the second lower bound),
		$H\lg q≥q'q\lg(q/q')(Ｐ-q''/q')+\lg q'$.
		The right-hand side is $\lg q-q'\lg(q/q')(q'-qＰ)$
		by matching the (rational) coefficients of
		$Ｐ\lg q$, $Ｐ\lg q'$, $\lg q$, and $\lg q'$, respectively.
		As $H\lg q≥\lg q-q'\lg(q/q')(q'-qＰ)$ we bound
		$\lg(q/q')=\lg(1+1/q')≤1/q'$ by the tangent line at $1/q'=0$.
		So $H\lg q≥\lg q-(q'-qＰ)$ and hence $1-H≤(q'-qＰ)/\lg q$.
		Move on to $q'-qＰ$:
		By \cref{lem:PvsH}, $1-qＰ/q'≤S$ so $q'-qＰ≤q'S$.
		Move on to $S$:
		By definition $S≤Ｓ$.
		Now we chain the inequalities
		$1-H≤(q'-qＰ)/\lg q≤q'S/\lg q≤q'Ｓ/\lg q$.
		This completes \cref{ine:1-H<S} as $q'/\lg q<q³$.
	\end{proof}
	
	\begin{figure}
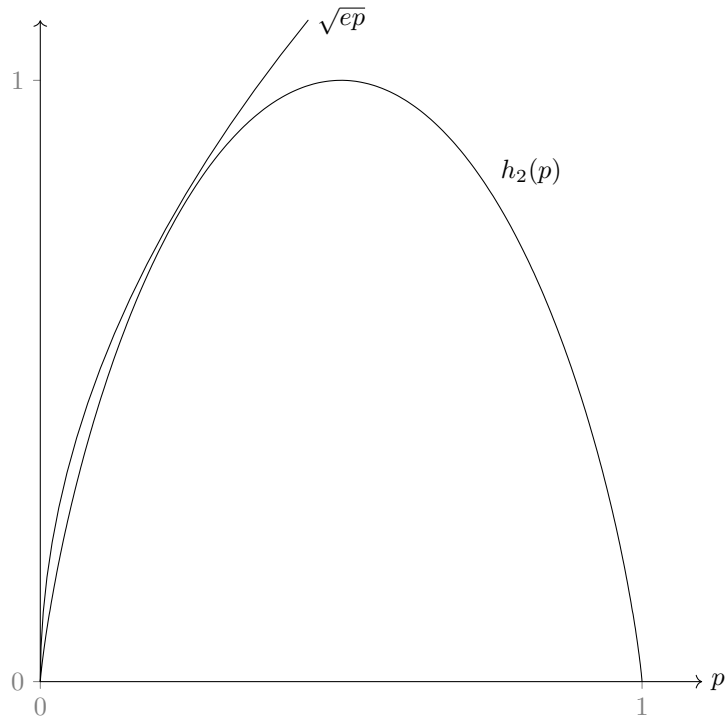

		\tikz[scale=8]{
			\draw
				(0,1)pic{y-tick=$1$}(0,0)pic{y-tick=$0$}
				(0,0)pic{x-tick=$0$}(1,0)pic{x-tick=$1$};
			\draw
				plot[domain=0:90,samples=360]({sin(\x)^2},{h2o(\x)})
				({sin(60)^2},{h2o(60)})node[above right]{$h₂(p)$}
				plot[domain=0:1.1,samples=20](\x^2/e,\x)node[right]{$√{ep}$};
			\draw[->](0,0)--(0,1.1);
			\draw[->](0,0)--(1.1,0)node[right]{$p$};
		}
		\caption{
			The binary entropy function and an upper bound of $√{ep}$.
		}\label{fig:rootep}
	\end{figure}
	
	Hölder-ness is a “toll” because, while \cref{pro:im-toll,pro:ex-toll}
	are the bridges that connect channel parameters,
	it feels like we pay fees, or taxes, to translate bounds on channels;
	and the fee/tax is charged regardless we go one way or another.
	For example, say we start with $H(𝘞_n)<\exp(-n^{2/3})$,
	then $Ｚ(𝘞_n)<q³\exp(-n^{2/3}/2)$,
	which contains annoying minor terms that harden the proof.
	Once we succeed in proving $Ｚ(𝘞_n)<\exp\(-e^{n^{1/3}}\)$, we pay the price again
	when translating it back to $H(𝘞_n)<q³\exp\(-e^{n^{1/3}}/2\)$.
	
	The next section synthesizes channels on the basis of general matrices,
	and shows how synthetic channels relate to parameters defined here.

\section{Kernel and Fundamental Theorems}\label{sec:matrix}

	Fix a prime power $q$.
	Let $ℓ≥2$.
	Let $G∈𝔽_q^{ℓ×ℓ}$ be an invertible $ℓ×ℓ$ matrix over $𝔽_q$.
	This matrix is to replace $\loll$ that was used to polarize channels
	in \cref{cha:origin,cha:dual,cha:prune}, and is called a \emph{kernel}.
	See cref{fig:nine} for an illustration of the EU--DU pairs that corresponds to
	a $3×3$ kernel and how to wrap them around (synthetic) channels.
	
	\tikzset{
		E/.pic={
			\pgftransformscale{1/6}\def~{coordinate}
			\draw[very thick](-6,6)~(=o0)--(6,6)~(=i0)
				(-6,0)~(=o1)--(6,0)~(=i1)(-6,-6)~(=o2)--(6,-6)~(=i2);
			\draw[fill=white](-5,-8)rectangle(5,8)(0,0)node[align=center]{$3×3$\\EU};
		},
		D/.pic={
			\pgftransformscale{1/6}\def~{coordinate}
			\draw[very thick](-6,6)~(=i0)--(6,6)~(=o0)
				(-6,0)~(=i1)--(6,0)~(=o1)(-6,-6)~(=i2)--(6,-6)~(=o2);
			\draw[fill=white](-5,-8)rectangle(5,8)(0,0)node[align=center]{$3×3$\\DU};
		}
	}
	\begin{figure}
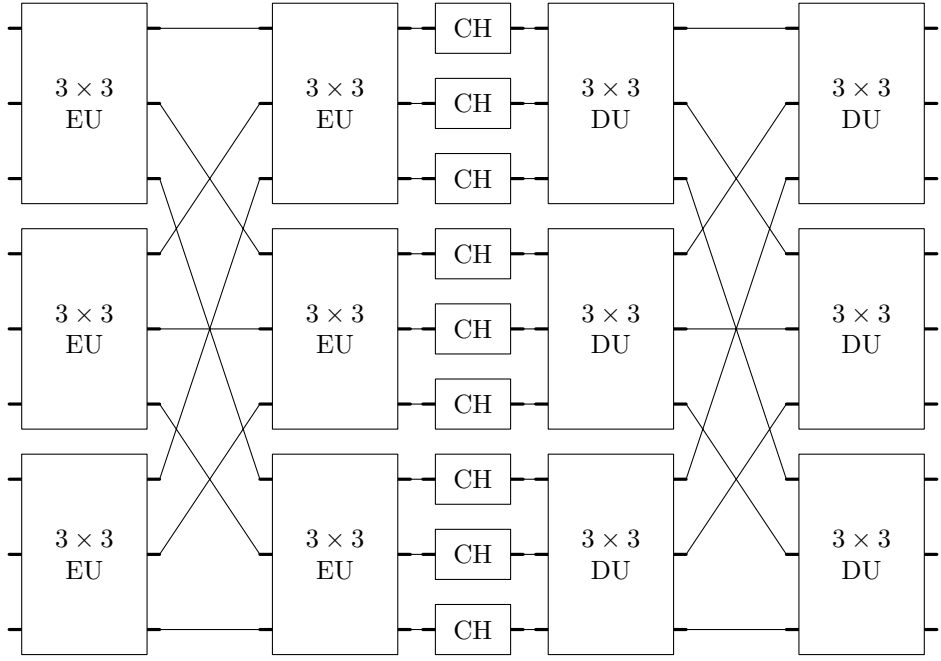

		\tikz[scale=1/6]{
			\def\pi{pic}
			\draw
											  (0,54)π(c00){C}
			(-31,48)π(2e0){E}(-11,48)π(1e0){E}(0,48)π(c01){C}(11,48)π(1d0){D}(31,48)π(2d0){D}
											  (0,42)π(c02){C}
											  (0,36)π(c10){C}
			(-31,30)π(2e1){E}(-11,30)π(1e1){E}(0,30)π(c11){C}(11,30)π(1d1){D}(31,30)π(2d1){D}
											  (0,24)π(c12){C}
											  (0,18)π(c20){C}
			(-31,12)π(2e2){E}(-11,12)π(1e2){E}(0,12)π(c21){C}(11,12)π(1d2){D}(31,12)π(2d2){D}
											  (0,6)π(c22){C}
			;
			\draw
				\foreach\f in{e,d}{
					\foreach\k in{0,1,2}{
						\foreach\j in{0,1,2}{
							(c\k\j=\f)--(1\f\k=i\j)
							(1\f\k=o\j)--(2\f\j=i\k)
						}
					}
				}
			;
		}
		\caption{
			The design of encoder and decoder when $ℓ=3$.
		}\label{fig:nine}
	\end{figure}
	
	Let $X₁^ℓ$ be $ℓ$ i.i.d.\ random variables that follow $Q$,
	the capacity-achieving input distribution of $W$.
	Define $U₁^ℓ≔X₁^ℓG^{-1}$.
	Let $Y₁^ℓ$ be the outputs of $ℓ$ i.i.d.\ copies of $W$ given the inputs $X₁^ℓ$.
	That is to say, each $(X_j,Y_j)$ is governed by $Q(x)$ and $W(y｜x)$.
	Now consider the following guessing job:
	\begin{itemize}
		\item	Guess $U₁$ given $Y₁^ℓ$.
		\item	Guess $U₂$ given $Y₁^ℓ$,
				assuming that the guess $ˆU₁$ of $U₁$ is correct.
		\item	Guess $U₃$ given $Y₁^ℓ$,
				assuming that the guesses $ˆU₁²$ of $U₁²$ are correct.
				\[⋮\]
		\item	Guess $U_ℓ$ given $Y₁^ℓ$,
				assuming that the guesses $ˆU₁^{ℓ-1}$ of $U₁^{ℓ-1}$ are correct.
	\end{itemize}
	For each $j=1,2…ℓ$, pretend that $WＷj$ is a channel
	whose input is $U_j$ and output is $Y₁^ℓU₁^{j-1}$.
	This channel captures the difficulty of the $j$th guessing job.
	The precise definition follows.
	
	\begin{dfn}
		For each $j=1,2…ℓ$, define a synthetic channel
		\[WＷj(y₁^ℓu₁^{j-1}｜u_j)≔÷{P\{Y₁^ℓU₁^{j-1}=y₁^ℓu₁^{j-1}\}}{P\{U_j=u_j\}}
			=÷{∑\limits_{u_{j+1}^ℓ∈𝔽_q^ℓ}W(u₁^ℓG,y_k)}
				{∑\limits_{v₁^{j-1}u_{j+1}^ℓ∈𝔽_q^{ℓ-1}}Q^ℓ(v₁^{j-1}u_j^ℓG)},\]
		where $Q^ℓ$ is the product measure and $W^ℓ$ is the product channel, i.e.,
		\[Q^ℓ(x₁^ℓ)=∏_{j=1}^ℓQ(x_j),\qquad W^ℓ(x₁^ℓ,y₁^ℓ)≔∏_{j=1}^ℓW(x_j,y_j),\]
		and $v₁^{j-1}u_j^ℓG$ is vector concatenation
		before vector-matrix multiplication.
	\end{dfn}
	
	There are three results that control $WＷj$ via
	the parameters $H,Ｚ,Ｓ$ defined in the last section.
	These results play central roles in the theory of polar coding
	and is bestowed the title of fundamental theorems.
	
	\begin{thm}[Fundamental theorem of polar coding---$H$ version (FTPC$H$)]
		\label{thm:ftpcH}
		For any $q$-ary channel $W$ and any invertible matrix kernel $G∈𝔽_q^{ℓ×ℓ}$,
		\[∑_{j=1}^ℓH(WＷj)=ℓH(W).\]
	\end{thm}
	
	\begin{proof}
		We derive that
		\[∑_{j=1}^ℓH(WＷj)=∑_{j=1}^ℓH(U_j｜Y₁^ℓU₁^{j-1})
			=H(U₁^ℓ｜Y₁^ℓ)=H(X₁^ℓ｜Y₁^ℓ)=ℓH(X｜Y).\]
		The first equality is by definition.
		The next equality is by the chain rule of conditional entropy.
		The next equality is by $G$ being invertible.
		The last equality is because $(X₁^ℓ,Y₁^ℓ)$ are i.i.d.\ copies of $(X,Y)$.
	\end{proof}
	
	Fundamental theorem stated and proved,
	the other two fundamental theorems require more details about $G$ to be stated.
	
	Let $0₁^{j-1}1_ju_{j+1}^ℓ∈𝔽_q^ℓ$ be a vector of
	$0$ repeated $j-1$ times followed by a $1$ and $ℓ-j$ arbitrary symbols.
	A \emph{coset code} is a subset of codewords of the form
	$\{0₁^{j-1}1_ju_{j+1}^ℓG:u_{j+1}^ℓ∈𝔽_q^{ℓ-j}\}⊆𝔽_q^ℓ$.
	The coset codes have weight distributions just like every other code does.
	Let $\hwt(x₁^ℓ)$ be the Hamming weight of $x₁^ℓ$.
	The weight enumerator of the $j$th coset code is
	this one-variable polynomial over the integers
	\[f_ZＷj(z)≔∑_{u_{j+1}^ℓ}z^{\hwt(0₁^{j-1}1_ju_{j+1}^ℓG)}∈ℤ[z].\]
	Note that this coincides with the distance enumerator
	from the $j$th row of $G$ to the span of the rows beneath.
	
	\begin{thm}[Fundamental theorem of polar coding---$Z$ end (FTPC$Z$)]
		\label{thm:ftpcZ}
		For any $q$-ary channel $W$ and any invertible matrix kernel $G∈𝔽_q^{ℓ×ℓ}$,
		\[Ｚ(WＷj)≤f_ZＷj(Ｚ(W)).\]
	\end{thm}
	
	\begin{proof}
		By the definition of the synthetic channel $WＷj$ and
		that of the Bhattacharyya parameter, $Ｚ(WＷj)$ is
		\[\max_{0≠d_j∈𝔽_q}∑_{u_j∈𝔽_q}∑_{u₁^{j-1}y₁^ℓ∈𝔽_q^j×𝒴^ℓ}
			√{WＷj(u_j,y₁^ℓu₁^{j-1})WＷj(u_j+d_j,y₁^ℓu₁^{j-1})}.\]
		By the nature of $\max_{0≠d_j∈𝔽_q}$, it suffices to show that
		the double sum within is at most $f_ZＷj(Ｚ(W))$ for arbitrary nonzero $d_j$.
		
		In the upcoming argument, vector concatenation takes precedence
		over vector-matrix multiplication and vector addition.
		Fix a $d_j∈𝔽_q^×$, we argue that
		\begin{align*}
			\qquad&\kern-2em
			∑_{u_j∈𝔽_q}∑_{y₁^ℓu₁^{j-1}∈𝒴^ℓ×𝔽_q^j}
				√{WＷj(u_j,y₁^ℓu₁^{j-1})WＷj(u_j+d_j,y₁^ℓu₁^{j-1})}	\\
			&	=∑_{u₁^jy₁^ℓ}√{WＷj(u_j,y₁^ℓu₁^{j-1})WＷj(u_j+d_j,y₁^ℓu₁^{j-1})}	\\
			&	=∑_{u₁^jy₁^ℓ}√{\vphantom{∑}\smash{
					∑_{u_{j+1}^ℓ∈𝔽_q^{ℓ-j}}W^ℓ(u₁^ju_{j+1}^ℓG,y₁^ℓ)
					∑_{v_{j+1}^ℓ∈𝔽_q^{ℓ-j}}W^ℓ(u₁^{j-1}(u_j+d_j)v_{j+1}^ℓG,y₁^ℓ)
				}}	\\
			&	≤∑_{u₁^jy₁^ℓ}∑_{u_{j+1}^ℓ}∑_{v_{j+1}^ℓ}√{W^ℓ(u₁^ju_{j+1}^ℓG,y₁^ℓ)
					W^ℓ(u₁^{j-1}(u_j+d_j)v_{j+1}^ℓG,y₁^ℓ)}	\\
			&	=∑_{y₁^ℓ}∑_{u₁^ℓ}∑_{d_{j+1}^ℓ∈𝔽_q^{ℓ-j}}
				√{W^ℓ(u₁^ℓG,y₁^ℓ)W^ℓ(u₁^{j-1}(u_j^ℓ+d_j^ℓ)G,y₁^ℓ)}	\\
			&	=∑_{y₁^ℓ}∑_{x₁^ℓ∈𝔽_q^ℓ}∑_{d_{j+1}^ℓ}
				√{W^ℓ(x₁^ℓ,y₁^ℓ)W^ℓ(x₁^ℓ+0₁^{j-1}d_j^ℓG,y₁^ℓ)}	\\
			&	=∑_{d_{j+1}^ℓ}∑_{y₁^ℓ}∑_{x₁^ℓ}
				√{W^ℓ(x₁^ℓ,y₁^ℓ)W^ℓ(x₁^ℓ+e₁^ℓ,y₁^ℓ)}	\\
			&	=∑_{d_{j+1}^ℓ}∑_{y₁^ℓ}∑_{x₁^ℓ}
				∏_{k∈[ℓ]}√{W(x_k,y_k)W(x_k+e_k,y_k)}	\\
			&	=∑_{d_{j+1}^ℓ}∑_{y₁^ℓ}∑_{x₁^ℓ}∏_{k∈K}
				√{W(x_k,y_k)W(x_k+e_k,y_k)}∏_{k∉K}W(x_k,y_k)	\\
			&	=∑_{d_{j+1}^ℓ}∏_{k∈K}（∑_{x_ky_k}
				√{W(x_k,y_k)W(x_k+e_k,y_k)}）∏_{k∉K}（∑_{x_ky_k}W(x_k,y_k)）	\\
			&	=∑_{d_{j+1}^ℓ}∏_{k∈K}（∑_{x_ky_k}√{W(x_k,y_k)W(x_k+e_k,y_k)}）	\\
			&	≤∑_{d_{j+1}^ℓ}∏_{k∈K}\max_{0≠e_k∈𝔽_q}
				（∑_{x_ky_k}√{W(x_k,y_k)W(x_k+e_k,y_k)}）	\\
			&	=∑_{d_{j+1}^ℓ}∏_{k∈k}Ｚ(W)=∑_{d_{j+1}^ℓ}Ｚ(W)^{\abs K}
				=∑_{d_{j+1}^ℓ}Ｚ(W)^{\hwt(0₁^{j-1}d_jd_{j+1}^ℓG)}	\\
			&	=∑_{d_{j+1}^ℓ}Ｚ(W)^{\hwt(0₁^{j-1}1_jd_{j+1}^ℓG)}=f_ZＷj(Ｚ(W)).
		\end{align*}
		The first equality abbreviates the summation.
		The next equality expands $WＷj$ by the very definition,
		where $u_{j+1}^ℓ$ and $v_{j+1}^ℓ$ are free variables in $𝔽_q$.
		The next inequality is by sub-additivity of square root.
		In the next equality we define $d_{j+1}^ℓ≔v_{j+1}^ℓ-u_{j+1}^ℓ$;
		so summing over $v_{j+1}^ℓ$ is equivalent to summing over $d_{j+1}^ℓ$.
		In the next equality we define $x₁^ℓ≔u₁^ℓG$;
		so summing over $u₁^ℓ$ is equivalent to summing over $x₁^ℓ$
		as $G$ is invertible.
		In the next equality we substitute $e₁^ℓ≔0₁^{j-1}d_j^ℓG$
		and reorder the summation.
		The next equality expands the product of the memoryless channels.
		The next equality classifies the indices into two classes---%
		$k∈K$ are those such that $e_k≠0$ and $k∉K$ are those such that $e_k=0$.
		The next equality is the distributive law $ax+ay+bx+by=(a+b)(x+y)$.
		The next equality uses the fact that $W(x,y)$ sum to $1$.
		In the next inequality we replace $e_k$ by a nonzero element
		that maximizes the sum in the parentheses.
		In the next equality we realize that
		the maximum is the Bhattacharyya parameter (surprisingly).
		The second last equality uses the fact that
		multiplying a vector by a scalar preserves its Hamming weight.
		And quod erat demonstrandum.
	\end{proof}
	
	The last fundamental theorem lives in the dual picture.
	Let $u₁^{j-1}1_j0_{j+1}^ℓ∈𝔽_q^ℓ$ be a vector of
	$j-1$ arbitrary symbols followed by a $1$ and $0$ repeated $ℓ-j$ times.
	Let $\Git$ be the inverse transpose of $G$.
	The $j$th dual coset code of $G$ is a subset of codewords of the form
	$\{u₁^{j-1}1_j0_{j+1}^ℓ\Git:u₁^{j-1}∈𝔽_q^{j-1}\}⊆𝔽_q^ℓ$.
	The weight enumerator of the $j$th dual coset code is
	defined to be this one-variable polynomial over the integers
	\[f_SＷj(s)≔∑_{u₁^{j-1}}s^{\hwt(u₁^{j-1}1_j0_{j+1}^ℓ\Git)}∈ℤ[s].\]
	We can now state the dual of the second fundamental theorem.
	
	\begin{thm}[Fundamental theorem of polar coding---$S$-end (FTPC$S$)]
		\label{thm:ftpcS}
		For any $q$-ary channel $W$ and any invertible matrix kernel $G∈𝔽_q^{ℓ×ℓ}$,
		\[Ｓ(WＷj)≤f_ZＷj(Ｓ(W)).\]
	\end{thm}
	
	\begin{proof}
		(Reminder: 
		For those who attempt to skip the proof, the proof spans about $2.2$ pages.)
		
		Recall the character $χ(x)≔\exp(2πi\tr(x)/p)$.
		We need these properties of $χ$:
		(xa)	$χ(0)=1$;
		(xb)	$\abs{χ(x)}=1$ for all $x∈𝔽_q$;
		(xc)	$χ(x)χ(z)=χ(x+z)$ for all $x,z∈𝔽_q$; and
		(xd)	$∑_{x∈𝔽_q}χ(x)=0$.
		See also \cite[Definition~24]{MT14} or a dedicated book \cite{Terras99}.
		To prove the theorem,
		we first verify that the Fourier coefficients recover the origin:
		Let $M(w,y)≔W(y)M(w｜y)=∑_{z∈𝔽_q}W(z,y)χ(wz)$, then
		\begin{align*}
			\qquad&\kern-2em
			∑_{w∈𝔽_q}M(w,y)χ(-xw) = ∑_{w∈𝔽_q}∑_{z∈𝔽_q}W(z,y)χ(wz)χ(-xw)	\\
			&	=∑_{z∈𝔽_q}W(z,y)∑_{w∈𝔽_q}χ(w(z-x))=∑_{z∈𝔽_q}W(z,y)q𝕀\{z-x=0\}=qW(x,y).
		\end{align*}
		The first equality expands $M(w,y)$ by the definition.
		The next equality uses that $χ$ is an additive character (xc),
		and reorders the summation.
		The next equality uses $∑_{w∈𝔽_q}χ(w)=0$ (xd) and $∑_{w∈𝔽_q}χ(0)=q$ (xa);
		and $𝕀$ is the indicator function.
		
		Knowing $W(x_j,y_j)=q^{-1}∑_{w_j∈𝔽_q}M(w_j,y_j)χ(-x_jw_j)$,
		we proceed to
		\begin{align*}
			&
			WＷj(u_j,y₁^ℓu₁^{j-1})
				=∑_{u_{j+1}^ℓ}W^ℓ(u₁^ℓG,y₁^ℓ)=∑_{u_{j+1}^ℓ∈𝔽_q^{ℓ-j}}W^ℓ(x₁^ℓ,y₁^ℓ)
				=∑_{u_{j+1}^ℓ}∏_{k∈[ℓ]}W(x_k,y_k)	\\
			&	=∑_{u_{j+1}^ℓ}∏_{k∈[ℓ]}（÷1q∑_{w_k∈𝔽_q}M(w_k,y_k)χ(-x_kw_k)）
				=÷1{q^ℓ}∑_{u_{j+1}^ℓ}∑_{w₁^ℓ}∏_{k∈[ℓ]}M(w_k,y_k)χ(-x_kw_k)	\\
			&	=÷1{q^ℓ}∑_{u_{j+1}^ℓ}∑_{w₁^ℓ}χ(-x₁^ℓ(w₁^ℓ)^⊤)∏_{k∈[ℓ]}M(w_k,y_k)
				=÷1{q^ℓ}∑_{u_{j+1}^ℓ}∑_{w₁^ℓ}χ(-x₁^ℓ(w₁^ℓ)^⊤)M^ℓ(w₁^ℓ,y₁^ℓ)	\\
			&	=÷1{q^ℓ}∑_{u_{j+1}^ℓ}∑_{w₁^ℓ}χ(-u₁^ℓG(w₁^ℓ)^⊤)M^ℓ(w₁^ℓ,y₁^ℓ)
				=÷1{q^ℓ}∑_{u_{j+1}^ℓ}∑_{w₁^ℓ}χ(-u₁^ℓ(w₁^ℓG^⊤)^⊤)M^ℓ(w₁^ℓ,y₁^ℓ)	\\
			&	=÷1{q^ℓ}∑_{u_{j+1}^ℓ}∑_{v₁^ℓ}χ(-u₁^ℓ(v₁^ℓ)^⊤)
				M^ℓ(v₁^ℓ\Git,y₁^ℓ)	\\
			&	=÷1{q^ℓ}∑_{v₁^ℓ}χ(-u₁^j(v₁^j)^⊤)
				M^ℓ(v₁^ℓ\Git,y₁^ℓ)∑_{u_{j+1}^ℓ}χ(-u_{j+1}^ℓ(v_{j+1}^ℓ)^⊤)	\\
			&	=÷1{q^ℓ}∑_{v₁^ℓ}χ(-u₁^j(v₁^j)^⊤)
				M^ℓ(v₁^ℓ\Git,y₁^ℓ)q^{ℓ-j}𝕀\{v_{j+1}^ℓ=0\}	\\
			&	=÷1{q^j}∑_{v₁^j}χ(-u₁^j(v₁^j)^⊤)M^ℓ(v₁^j0_{j+1}^ℓ\Git,y₁^ℓ).
		\end{align*}
		The first equality expands the definition of $WＷj$.
		In the next equality, we substitute $x₁^ℓ≔u₁^ℓG$.
		The next equality expands the definition of $W^ℓ$ down to $W$.
		The next two equalities Fourier expand $W$ and reorder the operators.
		The next equality merges all $χ(-x_kw_k)$
		into one inner product by additivity (xc).
		In the next equality we define $M^ℓ(w₁^ℓ,y₁^ℓ)$
		to be the product of all $M(w_k,y_k)$.
		The next two equalities use
		$x₁^ℓ(w₁^ℓ)^⊤=u₁^ℓG(w₁^ℓ)^⊤=u₁^ℓ(w₁^ℓG^⊤)^⊤$.
		In the next equality we define $v₁^ℓ≔w₁^ℓG^⊤$;
		so summing over $w₁^ℓ$ is equivalent to summing over $v₁^ℓ$.
		The last three equalities sum over $u_{j+1}^ℓ$ to force $v_{j+1}^ℓ=0$.
		
		Having that
		$WＷj(u_j,y₁^ℓu₁^{j-1})=q^{-j}∑_{v₁^j}
			χ(-u₁^j(v₁^j)^⊤)M^ℓ(v₁^j0_{j+1}^ℓ\Git,y₁^ℓ)$
		in mind, we move on to
		\begin{align*}
			\qquad&\kern-2em
			MＷj(ω_j,u₁^{j-1}y₁^ℓ)≔∑_{z_j∈𝔽_q}WＷj(z_j,y₁^ℓu₁^{j-1})χ(ω_jz_j)	\\
			&	=∑_{z_j∈𝔽_q}÷1{q^j}∑_{v₁^j}χ(-u₁^{j-1}z_j(v₁^j)^⊤)
				M^ℓ(v₁^j0_{j+1}^ℓ\Git,y₁^ℓ)χ(ω_jz_j)	\\
			&	=÷1{q^j}∑_{v₁^j}χ(-u₁^{j-1}(v₁^{j-1})^{⊤})
				M^ℓ(v₁^j0_{j+1}^ℓ\Git,y₁^ℓ)∑_{z_j∈𝔽_q}χ(z_j(ω_j-v_j))	\\
			&	=÷1{q^j}∑_{v₁^j}χ(-u₁^{j-1}(v₁^{j-1})^⊤)
				M^ℓ(v₁^j0_{j+1}^ℓ\Git,y₁^ℓ)q𝕀\{ω_j=v_j\}	\\
			&	=÷q{q^j}∑_{v₁^{j-1}}χ(-u₁^{j-1}(v₁^{j-1})^⊤)
				M^ℓ(v₁^{j-1}ω_j0_{j+1}^ℓ\Git,y₁^ℓ).
		\end{align*}
		In the first line we let $MＷj$ be the Fourier coefficients of $WＷj$.
		The next equality plugs in what we have in mind about $WＷj$.
		The next three equalities sum over $z_j$ to force $v_j=ω_j$.
		
		With the fact that $MＷj(ω_j,u₁^{j-1}y₁^ℓ)$ is equal to
		$q^{1-j}∑_{v₁^{j-1}}χ(-u₁^{j-1}(v₁^{j-1})^⊤)\*
			M^ℓ(v₁^{j-1}ω_j0_{j+1}^ℓ\Git,y₁^ℓ)$,
		we obtain that with arbitrary $0≠ω_j∈𝔽_q$,
		\begin{align*}
			\qquad&\kern-2em
			∑_{u₁^{j-1}y₁^ℓ∈𝔽^{j-1}×𝒴^ℓ}\abs{MＷj(ω_j,u₁^{j-1}y₁^ℓ)}	\label{for:sum-Mj}\\
			&	=∑_{u₁^{j-1}y₁^ℓ}\abs[\Big]{÷q{q^j}∑_{v₁^{j-1}}
				χ(-u₁^{j-1}(v₁^{j-1})^⊤)M^ℓ(v₁^{j-1}ω_j0_{j+1}^ℓ\Git,y₁^ℓ)}	\\
			&	≤∑_{u₁^{j-1}y₁^ℓ}÷q{q^j}∑_{v₁^{j-1}}
				\abs{M^ℓ(v₁^{j-1}ω_j0_{j+1}^ℓ\Git,y₁^ℓ)}	\\
			&	=∑_{y₁^ℓ}∑_{v₁^{j-1}}\abs{M^ℓ(v₁^{j-1}ω_j0_{j+1}^ℓ\Git,y₁^ℓ)}
				=∑_{y₁^ℓ}∑_{v₁^{j-1}}∏_{k∈[ℓ]}\abs{M(w_k,y_k)}	\\
			&	=∑_{y₁^ℓ}∑_{v₁^{j-1}}∏_{k∈K}\abs{M(w_k,y_k)}∏_{k∉K}\abs{M(w_k,y_k)}	\\
			&	=∑_{v₁^{j-1}}∏_{k∈K}（∑_{y_k}\abs{M(w_k,y_k)}）
				∏_{k∉K}（∑_{y_k}\abs{M(w_k,y_k)}）	\\
			&	=∑_{v₁^{j-1}}∏_{k∈K}（∑_{y_k}\abs{M(w_k,y_k)}）
				≤∑_{v₁^{j-1}}∏_{k∈K}Ｓ(W)	\\
			&	=∑_{v₁^{j-1}}Ｓ(W)^{\abs K}
				=∑_{v₁^{j-1}}Ｓ(W)^{\hwt(v₁^{j-1}ω_j0_{j+1}^ℓ\Git)}	\\
			&	=∑_{v₁^{j-1}}Ｓ(W)^{\hwt(v₁^{j-1}1_j0_{j+1}^ℓ\Git)}=f_SＷj(Ｓ(W)).
		\end{align*}
		The first equality expands the Fourier coefficients.
		The next inequality is triangle inequality plus (xb).
		The next equality cancels the summation over $u₁^{j-1}$ with $q^{1-j}$.
		In the next equality we substitute $w₁^ℓ≔v₁^{j-1}ω_j0_{j+1}^ℓ\Git$;
		slightly different from the free $w₁^ℓ$ before,
		they are now restricted to a proper subspace.
		The next equality classifies the indices into two classes---%
		$j∈K$ are those such that $w_j≠0$ and $k∉K$ are such that $w_k=0$.
		The next two equalities reorder the operators and
		simplify $∑_{y_k}\abs{M(0,y_k)}=∑_{y_k}W(y_k)=1$.
		The next inequality replaces $w_j$ by
		one that maximizes $∑_{y_k}\abs{M(w_j,y_k)}$.
		The rest is trivial.
		
		\Cref{thm:ftpcS} claims that $Ｓ(WＷj)≤f_SＷj(Ｓ(W))$.
		Since $Ｓ(WＷj)$ is merely the maximum of \cref{for:sum-Mj}
		over $0≠ω_j∈𝔽_q$, we arrive at $Ｓ(WＷj)≤f_SＷj(Ｓ(W))$.
		And quod erat demonstrandum.
	\end{proof}
	
	Sometimes, a simpler bound is enough for the analysis---instead of the exact
	weight enumerator, we use $(†\#codewords†)z^†minimum distance†$ as an upper bound.
	The number of codewords is easy to predict (it is some power of $2$).
	So one only needs to record the minimum distances.
	
	\begin{dfn}
		For any $G∈𝔽_q^{ℓ×ℓ}$, define \emph{coset distance}
		\[D_ZＷj≔\hdis(r_j,R_j),\]
		where $\hdis$ is the Hamming distance, $r_j$ is the $j$th row of $G$,
		and $R_j$ is the subspace spanned by the rows beneath $r_j$.
	\end{dfn}
	
	Then FTPC$Z$ reads $Ｚ(WＷj)≤q^{ℓ-j}Ｚ(W)^{D_ZＷj}$.
	The dual picture obeys the same logic.
	
	\begin{dfn}
		For any $G∈𝔽_q^{ℓ×ℓ}$, define \emph{dual coset distance}
		\[D_SＷj≔\hdis(c_j,C_j),\]
		where $\hdis$ is the Hamming distance, $c_j$ is the $j$th column of $G^{-1}$,
		and $C_j$ is the subspace spanned by the columns before $c_j$.
	\end{dfn}
	
	Then FTPC$S$ reads $Ｓ(WＷj)≤q^{j-1}Ｓ(W)^{D_ZＷj}$.
	
	The fundamental theorems introduced here are the “local relations” that
	control how one iteration of the channel transformation manipulates $W$.
	But in the end, we want to talk about all of its descendants at once,
	hence the introduction of the stochastic processes.
	
	\begin{dfn}
		Let $𝘑₁,𝘑₂,\dotsc$ be i.i.d.\ uniform
		random variables taking values in $\{1,2…ℓ\}$.
		Fix a $q$-ary channel $W$ and an invertible matrix $G∈𝔽_q^{ℓ×ℓ}$.
		Let $𝘞₉,𝘞₁,𝘞₂,\dotsc$, or $\{𝘞_n\}$ in short,
		be a stochastic process of DMCs defined as follows:
		\begin{itemize}
			\item	$𝘞₀≔W$; and
			\item	$𝘞_{n+1}≔𝘞_nＷ{𝘑_{n+1}}$.
		\end{itemize}
		This is called the \emph{channel process}.
	\end{dfn}
	
	Imagine the following infinite array:
	\[\bma{
		Ｓ(𝘞₀)	&	Ｓ(𝘞₁)	&	Ｓ(𝘞₂)	&	Ｓ(𝘞₃)	&	⋯\quad{}	\\
		T(𝘞₀)	&	T(𝘞₁)	&	T(𝘞₂)	&	T(𝘞₃)	&	⋯\quad{}	\\
		H(𝘞₀)	&	H(𝘞₁)	&	H(𝘞₂)	&	H(𝘞₃)	&	⋯\quad{}	\\
		Ｐ(𝘞₀)	&	Ｐ(𝘞₁)	&	Ｐ(𝘞₂)	&	Ｐ(𝘞₃)	&	⋯\quad{}	\\
		Ｚ(𝘞₀)	&	Ｚ(𝘞₁)	&	Ｚ(𝘞₂)	&	Ｚ(𝘞₃)	&	⋯\quad{}	
	}\]
	The Hölder tolls, \cref{lem:PvsZ,lem:PvsT,lem:PvsS,lem:PvsH,%
	pro:im-toll,pro:ex-toll}, are vertical relations.
	The fundamental theorems, \cref{thm:ftpcH,thm:ftpcZ,thm:ftpcS},
	are horizontal relations.
	The top two rows are related to the distortion of
	lossy compression and shaping the input distributions for asymmetric channels.
	The bottom two rows are related to the block error probability of
	lossless compression and the anti-error part of noisy-channel coding.
	The middle row controls the code rate.
	In particular, \cref{thm:ftpcH} implies
	the following generalization of \cref{pro:martin}.
	
	\begin{pro}
		$\{H(𝘞_n)\}$ is a martingale.
	\end{pro}
	
	\Cref{pro:im-toll,pro:ex-toll,thm:ftpcH,thm:ftpcZ,thm:ftpcS}
	are all we need to control the behavior of $\{𝘞_n\}$.
	But before we make use of these tools to examine the performance of polar coding,
	let us walk through some terminologies to see the big picture and what to expect.

\section{Probability Theory Regimes}\label{sec:regime}

	There is an analogy between coding theory and probability theory
	that connects the results from both sides and the proofs thereof.
	This analogy constituents the picture of the expected performance of coding.
	This section is a brief introduction to that and is inspired by \cite{AW14}.
	
	Consider i.i.d.\ copies of some bounded random variable
	$X₁,X₂,\dotsc,X_N$ and their average $¯X_N$.
	We want to understand the distribution of $¯X_N$,
	i.e., we want to understand $P\{¯X_N≤x\}$ for various $x$.
	The key is to identify the following
	\begin{itemize}
		\item	the mean $μ≔E[X₁]$ with channel capacity $C$,
		\item	the number of copies $N$ with the block length $N$,
		\item	the cutoff $x$ with the code rate $R$
		\item	the cumulative probability $P\{¯X_N≤x\}$
				with the block error probability $Ｐ$, and
		\item	the variance $σ²$ with another channel parameter called $V$.
	\end{itemize}
	
	For example, when Shannon said there exist error correcting codes
	with code rate $R<C$ close to channel capacity and block error probability $Ｐ→0$,
	this translates into when the cutoff $x<μ$ is close to
	the mean, the cumulative probability $P\{¯X_N≤x\}$ converges to $0$ as $N→∞$.
	The latter is the law of large numbers (LLN).
	That establishes the first analogue.
	
	Later, Gallager said that the block error probability
	scales like $\exp(-\Er(R)N)$ for a fixed $R<C$.
	It translates into that the cumulative probability scales like
	$P\{¯X_N≤x\}≈\exp(-L(x)N)$ for a fixed cutoff $z<μ$ \cite{Gallager68}.
	The former is called the error exponent regime;
	the latter is called the \emph{large deviation principle} (LDP).
	Moreover, Gallager's error exponent function $\Er$
	is analogous to Cramér's rate function $L(x)$.
	(To avoid abusing the word rate, I will referred to this as the Cramér function.)
	The former is the convex conjugate of Gallager's $\Eo$ function;
	the latter is the convex conjugate of
	the cumulant generating function $t↦㏑E[\exp(tX₁)]$.
	That establishes the second analogue.
	
	On a parallel track, Strassen said that the code rate
	scales like $C+Φ^{-1}(Ｐ)\*√{V/N}$, where $Φ$ is the cdf of
	the standard normal distribution, and $V$ is another intrinsic parameter
	called channel dispersion or varentropy \cite{Strassen62}.
	It translates into that the cutoff scales like $x≈μ+Φ^{-1}(p)√{σ²/N}$
	in order to attain a certain cumulative probability $p$.
	The former is called the finite block length regime
	(although, in fact, both this and the error exponent regime has finite $N$);
	the latter is called the \emph{central limit theorem} (CLT).
	That establishes the third analogue.
	Moreover $V$ is the analog of variance in coding theory.
	
	It is clear that LDP and CLT generalize LLN in ways that
	fix a variable and inspect the asymptote of the other variable.
	A third regime varies both.
	By \cite{AW14,PV10},
	\[÷{-㏑Ｐ}{N(C-R)²}⟶÷1{2V}.\]
	And it specializes to $Ｐ≈\exp(-Ω(N))$ for a fixed $R<C$
	and $R≈C-O(1)/√N$ for a fixed $Ｐ$.
	Similarly, in probability theory, there is
	\[÷{㏑P\{¯X-μ>γ(N)x\}}{Nγ(N)²}⟶L(x),\]
	where $γ(N)$ is some appropriate re-scalars,
	and $L(x)=1/σ²x²$ is another Cramér function .
	This is called the \emph{moderate deviation principle} (MDP).
	That establishes the fourth analogue.

\begin{table}
	\caption{
		An analogy among probability theory,
		random coding theory, and polar coding theory.
		All $δ>0$ can be made arbitrarily close to $0$.
	}\label{tab:analog}
	\pgfplotstableread{
		par	{Random variables}				{Random codes}			{Polar codes}		
		LLN	¯X→μ							(Ｐ,R)→(0,C)				(Ｐ,R)→(0,C)		
		LDP	ℙ\{Ｘ>x\}≈e^{-NL(x)}				Ｐ≈e^{-\Er(R)N}			Ｐ≈e^{-N^{1-δ}}		
		CLT	Ｘ∼𝒩(0,÷{σ}{√N})				Ｒ≈÷{Q^{-1}(Ｐ)}{√{VN}}	Ｒ≈N^{-1/2+δ}		
		MDP	\ƒ{P\{Ｘ>γ(N)x\}}{γ(N)²}≈NL(x)	\ƒ{Ｐ}{(Ｒ)²}≈÷N{2V}		\ƒ{Ｐ}{(Ｒ)²}≈N^{1-δ}
	}\tableTrinity
	\def\arraystretch{1.44}
	\def\ƒ#1{÷{-㏑#1}}
	\def÷#1#2{\frac{#1}{#2\rule[-1ex]{0pt}{0ex}}}
	\pgfplotstabletypeset[
		columns/par/.style={column name=,string type},
		every head row/.style={before row=\toprule,after row=\midrule},
		every last row/.style={after row=\bottomrule},
		assign cell content/.code={
			\pgfkeyssetvalue{/pgfplots/table/@cell content}{$#1$}}
	]\tableTrinity
\end{table}

	For their achievability bounds, the aforementioned results use
	random coding whose complexity is out of control.
	Beside random coding, polar coding is the only low-complexity code
	that is strong enough to approach the LDP, CLT, and MDP regimes.
	Some history is briefed below.
	See also \cref{tab:analog} for a comparison among
	probability theory and random and polar coding.
	
	Arıkan's original works on channel polarization \cite{Arikan09} established
	the foundation of polar coding, placing polar coding in the LLN paradigm.
	Arıkan--Telatar \cite{AT09} characterized the LDP behavior of polar coding,
	showing that $Ｐ$ scales like $\exp(-√N)$ when an $R<I$ is fixed.
	Later, Korada--Şaşoğlu--Urbanke \cite{KSU10} generalized polar codes
	from Arıkan's kernel $\loll$ to any invertible $ℓ×ℓ$ matrix $G$ over $𝔽₂$,
	granted that $ℓ≥2$ and $G$ is not column-equivalent to a lower triangular matrix.
	And then they showed that the LDP behavior is $Ｐ≈\exp(-N^{\Ec(G)})$
	where $\Ec(G)$ is a constant depending on the kernel matrix $G$.
	In fact, $\Ec(G)$ is the expectation of $-㏒_ℓ𝘋₁$.
	The notation $\Ec(G)$ is meant to resemble Gallager's error exponent $\Er(R)$;
	but be aware of that the former is inserted at $\exp(-N^†this†)$ place
	while the latter is inserted at $\exp(-†this†N)$ place.
	The LDP behavior of polar codes is then refined in \cite{HMTU13}.
	Therein, $Ｐ$ is approximated by $\exp(-ℓ^𝔈)$ where
	$𝔈≔\Ec(G)n-√{\Vc(G)n}Q^{-1}(R/I)+o(√n)$ is a more accurate exponent,
	$ℓ$ is the matrix dimension, $n$ is the depth of the code,
	and $\Vc(G)$ is another constant depending on $G$.
	The notation $\Vc(G)$ is meant to resemble the channel dispersion $V$.
	Appearing to be a CLT behavior, this result lies in
	the corner of the LDP paradigm that touches the MDP paradigm.
	Finally, Mori--Tanaka \cite{MT14} generalized everything above
	from binary input to channels with prime-power input.
	Over arbitrary input alphabets, \cite{STA09,Sasoglu11} 
	showed the equivalence of \cite{Arikan09,AT09}.
	Over binary but asymmetric channels, \cite{SRDR12,HY13}
	showed the counterpart of \cite{Arikan09,AT09}
	with the channel capacity in place of the symmetric capacity $I$.
	No further result on the LDP side,
	e.g.\ over non-binary asymmetric channels, is known.
	The technique in \cite{HY13} (\cref{ine:P<sumsum}) and \cref{sec:six} fill the gap.
	
	The CLT behavior of polar codes turns out to be difficult to characterize.
	It was Korada--Montanari--Telatar--Urbanke \cite{KMTU10} who came up with the idea
	that approximating an eigenfunction tightly bounds the eigenvalue $ℓ^{-ϱ}$.
	Here $ϱ>0$ will become a number such that
	$R$ scales like $I-N^{-ϱ}$ with a fixed $Ｐ$.
	It is called the scaling exponent because it controls
	the scaling of $N$ as a function of the gap to capacity $I-R$.
	(Although the LDP regime can also be rephrased as the scaling of $N$
	as a function of $㏑Ｐ$, it was named error exponent regime beforehand.
	Thus the name scaling exponent [regime] is dedicated to the CLT regime).
	They had $0.2669≤ϱ≤0.2841$ over binary erasure channels (BECs).
	The upper bound was brought down to $3.553ϱ≤1$ over
	binary-input discrete-output memoryless channels (BDMCs) \cite{GHU12}.
	Hassini--Alishahi--Urbanke \cite{HAU14} moved down the upper bound
	to $3.627ϱ≤1$ over BECs and $3.579ϱ≤1$ over BDMCs.
	They also proved a lower bound $1≤6ϱ$ over BDMCs.
	The latter is suboptimal and \cite{GB14,MHU16}
	improved the bound to $1≤5.702ϱ$ and to $1≤4.714ϱ$.
	Additive white Gaussian noise channles (AWGNCs) have continuous
	output alphabet, but \cite{FT17} show that they have $1≤4.714ϱ$ too.
	Over BECs particularly, \cite{FV14,YFV19} examined a series of larger kernels;
	the current record is a $64×64$ kernel believed to have $1≤2.9ϱ$.
	Near the end of the road to $2ϱ<1$, \cite{PU16} showed that by allowing $q→∞$,
	Reed--Solomon kernels achieve $2ϱ<1$ over $q$-ary channels.
	This does not really prove that polar codes achieve $2ϱ<1$
	over any specific channel, but gave hopes.
	Fazeli--Hassani--Mondelli--Vardy \cite{FHMV17,FHMV18,FHMV20}, eventually,
	showed that large random kernels achieve $2ϱ<1$ over BECs, breaking the barrier.
	Guruswami--Riazanov--Ye \cite{GRY19,GRY20} extended their result
	to all BDMCs utilizing the dynamic kernel technique.
	Our paper \cite{Hypotenuse21} fills the gap.
	
	Between LDP and CLT is polar coding's MDP behavior.
	First, Guruswami--Xia \cite{GX13,GX15} showed that there exists $ρ>0$ such that
	$Ｐ$ scales like $\exp(-N^{0.49})$ while $R$ scales like $I-N^{-ρ}$ over BDMCs.
	This raised a question about what are the possible pairs $(π,ρ)$
	such that $(Ｐ,R)$ scales like $(\exp(-N^π),I-N^{-ρ})$.
	Mondelli--Hassani--Urbanke \cite{MHU16} answered this,
	partially, in the same paper they bounded $ρ$.
	They showed that under a certain curve connecting $(0,1/5.714)$ and
	$(1/2,0)$ all $(π,ρ)$ are achievable over BDMCs.
	For BECs the upper left corner is $(0,1/4.627)$.
	A straightforward generalization to AWGNCs was also given in \cite{FT17}.
	We in \cite{ModerateDeviations18,LargeDeviations18} improved their result,
	suggesting that via a combinatorial trick the upper left corner
	of the curve is $(0,ϱ)$ for any $ϱ$ that is valid in the CLT regime.
	The same trick also implicated that over BECs all $(π,ρ)$
	such that $π+2ρ<1$ are achievable, which is mainly owing to
	\cite{FHMV17}'s result that $2ρ<1$ over BECs is achievable.
	Meanwhile, \cite{BGNRS18} made the first step to investigate
	the general kernel matrices over general prime-ary channels.
	They showed that it is possible to achieve $ρ>0$ with $Ｐ≈N^{-Ω(1)}$.
	This is, strictly speaking, ``only'' a CLT behavior as
	the desired block error probability in the MDP world is $\exp(-N^π)$.
	Later, Błasiok--Guruswami--Sudan \cite{BGS18} were able to show that
	for all $π<\Ec(G)$ there exists $ρ>0$ such that $(π,ρ)$ is achievable.
	This makes it a direct generalization of \cite{GX13}
	to all polarizing kernel matrices $G$ over all prime-ary channels.
	The preprint \cite{GRY20} contains a section that pushes the conference
	abstract \cite{GRY19} to positive $π$ while maintaining $ρ≈1/2$.
	Over the general DMCs, our \cite{Hypotenuse21} fills the gap.
	
	See \cref{fig:triangle} for a comprehensive plot of all these results.
	
	\begin{figure}
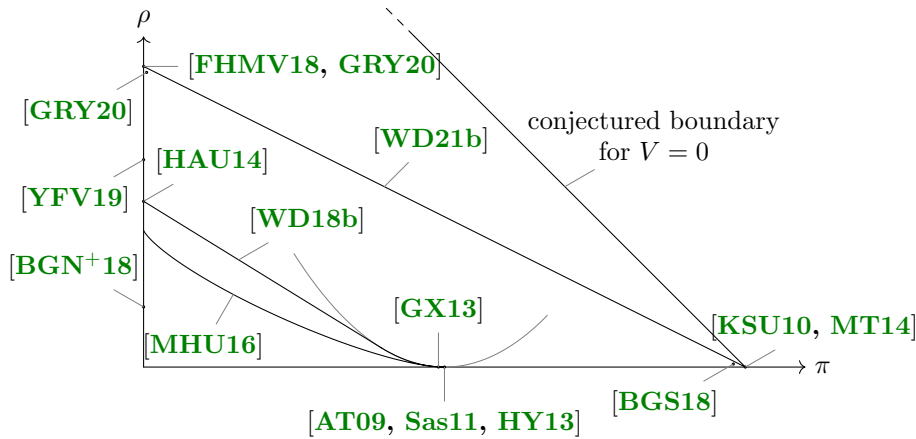

		\def\cp#1:#2:{coordinate[pin=#1:{#2}](X)}
		\tikz[scale=8,nodes={black,align=center}]{
			\draw[->](0,0)--(1.1,0)node[right]{$π$};
			\draw[->](0,0)--(0,.55)node[above]{$ρ$};
			\draw[gray,very thin]
				plot[domain=30:55,samples=75]({sin(\x)^2},{1-h2o(\x)})
			;
			\draw
				(.45,.55)edge[dashed](.4,.6)--(1,0)
				(.7,.3)\cp30:conjectured boundary\\for $V=0$:
				(0,.5)pic{dot}\cp0:\cite{FHMV18,GRY20}:
				(.005,.49)pic{dot}\cp210:\cite{GRY20}:
				(0,.5)--(1,0)(.4,.3)\cp45:\cite{Hypotenuse21}:
				(0,1/2.9)pic{dot}\cp210:\cite{YFV19}:
				(0,1/3.627)pic{dot}\cp30:\cite{HAU14}:
				(0,1/3.627)--node[pos=.4](X){}(.3947,.03223)
				(X)\cp30:\cite{ModerateDeviations18}:
				plot[domain=38.922:45,samples=30]({sin(\x)^2},{1-h2o(\x)})
				plot[domain=0:45,samples=90]({g2o(\x)*sin(\x)^2},{(1-g2o(\x))/3.627})
				({g2o(31)*sin(31)^2},{(1-g2o(31))/3.627})\cp240:\cite{MHU16}:
				(1,0)pic{dot}\cp30:\cite{KSU10,MT14}:
				(.49,.0001)pic{dot}\cp90:\cite{GX13}:
				(0,.1)pic{dot}\cp150:\cite{BGNRS18}:
				(0.98,.005)pic{dot}\cp210:\cite{BGS18}:
				(.5,0)pic{dot}\cp270:\cite{AT09,Sasoglu11,HY13}:;
			;
		}
		\caption{
			Recent works on polar coding arranged on a $ρ$--$π$ plot.
			Note that results utilizing different kernels
			over various channels are mixed.
			The higher $ρ,π$, the better performance.
		}\label{fig:triangle}
	\end{figure}
	
	In \Cref{cha:origin}, I presented the main contribution of
	\cite{ModerateDeviations18} which, to be more precise,
	is an interpolating result $(Ｐ,R)≈(\exp(-N^π),C-N^{-ρ})$ for pairs $(π,ρ)$
	lying in the region $𝒪$ that touches $(0,1/4.714)$ and $(1/2,0)$.
	This is the result I want to generalize in this chapter.
	Id est, I want to characterize the region of $π,ρ$
	for any invertible matrix $G$ over any finite field $𝔽_q$.
	And this region will be determined by the best $ϱ>0$ one can find (or believe in)
	plus the coset distance profile $D_ZＷj$ and $D_SＷj$.
	
	I will do this step-by-step.
	First, I will show that most kernels enjoy a (very weak) CLT behavior.
	More precisely, kernels that satisfy a certain ergodic precondition
	will enjoy an eigen behavior with positive $ϱ$.
	After that, we either stick to the weak but provable $ϱ$ or assume
	a higher $ϱ$ based on experiments, simulations, and/or heuristics.
	And then we go through an ergodic--eigen--en23--een13--elpin that resembles
	the eigen--en23--een13--elpin chain in \cref{cha:origin,cha:dual}.

\section{An Ergodicity Precondition}\label{sec:ergodic}

	\def\nabla{\tikz\draw(2em/3,0)|-(0,2em/3)--cycle;}
	
	Before we board the long proof train eigen--en23--een13--elpin,
	I want to recall the classification of matrix kernels into two groups.
	The bad group consists of matrices that do nothing to the channels, and hence get
	no chance to polarize channels, let alone enjoying any LDP, CLT, or MDP behavior.
	The good group consists of matrices that can polarize channels in a calculus regard.
	And then in the next section I will show that all of them
	enjoy some LDP, CLT, and MDP behaviors.
	This strengthens the dichotomy even further---a matrix is either not altering
	the channels at all, or polarizing the channels exponentially fast.
	
	To motivate the classification, recall that channels $WＷj$
	are synthesized based on the matrix-multiplication $X₁^ℓ=U₁^ℓG$.
	The kernel $G∈𝔽_q^{ℓ×ℓ}$ is said to be ergodic if it mixes/blends the content $U₁^ℓ$
	such that there are nontrivial relations between each $U_j$ and all of $Y₁^ℓ$.
	The following two counterexamples demonstrate the necessity of this condition.
	
	\begin{exa}
		Consider any prime power $q$ and any $ℓ≥2$.
		If $G∈𝔽_q^{ℓ×ℓ}$ is an upper triangular matrix with $1$'s on the diagonal,
		then $X₁^ℓ=U₁^ℓG$ is such that $X_j-U_j$ is a linear combination of $U₁^{j-1}$.
		Therefore, when we want to guess $U_j$ given $Y₁^ℓU₁^{j-1}$,
		it suffices to guess $X_j$ based on $Y_j$ and then
		subtract the correction term $X_j-U_j=U₁^ℓG=U₁^{j-1}0_j^ℓG$.
		That is to say, we are facing essentially
		the same guessing job as if we were facing $W$;
		the channel transformation does no benefit to us.
		In this case, $H(𝘞_n)$ stays where $H(W)$ is
		and does not polarize to $\{0,1\}$.
		Lesson:
		$G$ must not be upper triangular.
		In fact, it can be shown that if two matrices $G$ and $˜G$
		differ by some upper triangular row-operations, i.e., $˜G=∇G$,
		then they share the same polarizing ability.
	\end{exa}
	
	\begin{exa}
		Consider $q=4$ and $G=\loll$.
		Let $W$ be the independent product of BEC$(1/3)$ and BEC$(2/3)$.
		That is, $W$ takes a pair $(x',x'')∈𝔽₂²$ as an input and then outputs
		\[\cas{
			(x',x'')	&	w.p.\ $2/9$,	\\
			(?,x'')		&	w.p.\ $1/9$,	\\
			(x',?)		&	w.p.\ $4/9$,	\\
			(?,?)		&	w.p.\ $2/9$.	
		}\]
		Let $c∈𝔽₄、𝔽₂$ be a non-binary element,
		and identify each $(x',x'')∈𝔽₂²$ with $x≔x'+cx''∈𝔽₄$, then
		$W$ behaves like a channel with input alphabet $𝔽₄$, and
		\[\bma{x₁'&x₂'}+c\bma{x₁''&x₂''}=\bma{x₁&x₂}
			=\bma{u₁&u₂}G=\bma{u₁'&u₂'}G+c\bma{u₁''&u₂''}G.\]
		In plain English, $G$ multiplies the prime component
		and the double-prime component separately.
		Now we attempt to use $G$ to polarize $W$.
		Doing that is equivalent to polarizing BEC$(1/3)$ and BEC$(2/3)$ separately.
		Then, with probability $1/3$, the entropy process $\{H(𝘞_n)\}$ converges to
		$1/2$ because the prime component converges to the entirely noisy channel
		but the double-prime component converges to the completely reliable channel.
		Lesson: $G$ needs to bring interaction to
		the vector space substructure within a prime-power finite field.
	\end{exa}
	
	The two examples motivate the following definition and theorem
	for classifying and prejudging matrices.
	
	\begin{dfn}
		For any invertible matrix $G∈𝔽_q^{ℓ×ℓ}$,
		the \emph{lowered form} of $G$ is the lower triangular matrix $˜G$
		with $1$'s on the diagonal and such that $˜GG^{-1}$ is upper triangular.
		A matrix $G$ is said to be \emph{ergodic} if
		the off-diagonal entries of $˜G$ generate $𝔽_q$ as an $𝔽_p$-algebra;
		or $𝔽_q=𝔽_p[˜G]$ for short.
	\end{dfn}
	
	\begin{thm}[Ergodic kernel polarizes]\cite[Theorem 14]{MT14}\label{thm:ergodic}
		Let $W$ be a $q$-ary channel.
		The matrix kernel $G∈𝔽_q^{ℓ×ℓ}$ is ergodic iff
		$\{H(𝘞_n)\}$ converges to $0$ or $1$ almost surely.
	\end{thm}
	
	\begin{proof}[Sketch of the proof in \cite{MT14}]
		The stated theorem is strong and general, handling all edge cases.
		The majority of its proof involves reducing some general situation
		(e.g., $ℓ≥2$) to a special case (e.g., $ℓ=2$) to ease notational burdens.
		I pick out what I think is the key part of the proof.
		
		First, we know that a bounded martingale will almost always converge,
		which implies $H(𝘞_n)-H(𝘞_{n+1})→0$.
		So all we need to show is that when $H(𝘞_n)$ starts slowing down,
		i.e., when $\abs{H(𝘞_n)-H(𝘞_{n+1})}$ is small,
		$𝘞_n$ will be either very reliable or very noisy.
		To put in another way, our goal is to prove
		$H(𝘞_n)$ only slows down when it is reaching $0$ or $1$.
		Then we can conclude that the limit of $H(𝘞_n)$ is either $0$ or $1$.
		
		To show that $𝘞_n$ is not mediocre when $\abs{H(𝘞_n)-H(𝘞_{n+1})}$ is small,
		the subscript Bhattacharyya parameter $Z_d(W)$ is introduced to inspect
		the relation between $W(x｜y)$ and $W(x+d｜y)$, for any $d∈𝔽_q^×$.
		Now follow the recipe below to show $𝘞_n$ is extreme:
		\begin{itemize}
			\item	Show that if $\abs{H(W)-H(WＷj)}$ is small, then
					$\abs{H(W)-H(W^{[c]})}$ is small for some
					simpler channel transformation $•^{[c]}$,
					where $c$ is any nonzero entry of $˜G$.
			\item	Show that if $\abs{H(W)-H(W^{[c]})}$ is small, then
					$Z_{cd}(W)(1-Z_d(W))$ is small for all $d∈𝔽_q^×$.
			\item	Show that if $Z_d(W)$ is small, then
					$Z_{cd}(W)$ is small for any nonzero entry $c$ of $˜G$.
					Or, if $Z_d(W)$ is close to $1$, then
					$Z_{cd}(W)$ is close to $1$ for any nonzero entry $c$ of $˜G$.
					This step is the scalar-multiplication part
					of $𝔽_p[˜G]$ as an $𝔽_p$-algebra.
			\item	Show that if both $Z_c(W)$ and $Z_d(W)$ are small, then
					$Z_{c+d}(W)$ is small for any entries $c$ and $d$ of $˜G$.
					Similarly, if both $Z_c(W)$ and $Z_d(W)$ are close to $1$, then
					$Z_{c+d}(W)$ is close to $1$ for any $c,d∈𝔽_q^×$.
					This step is the addition part of $𝔽_p[˜G]$ as an $𝔽_p$-algebra.
			\item	Show that if $\{Z_d(W):d∈𝔽_q^×\}$ are all small,
					then $W$ is very reliable.
					Otherwise, if $\{Z_d(W):d∈𝔽_q^×\}$ are all close to $1$,
					then $W$ is very noisy.
		\end{itemize}
		
		The key to the first $•$ is to simplify the given premise
		$\abs{H(W)-H(WＷj)}<ε$ concerning $˜G$ to a condition
		$\abs{H(W)-H(W^{[c]})}<ℓε$ concerning some $2×2$ matrix, where
		\[W^{[c]}(y₁y₂u₁｜u₂)≔÷12W(y₁｜u₁+cu₂)W(y₂｜u₂),\]
		for any nonzero entry $c$ of $˜G$.
		Note that $W^{[c]}$ represents the guessing job
		of $U₂$ given $Y₁Y₂U₁$ and the matrix kernel $\locl$.
		To simplify the premise, let $(j,i)$ points to a non-zero entry $c$ of $˜G$.
		By how $U_i,U_j,X_i,X_j$ are related by $c$, one can prove
		\[0≤H(W)-H(W^{[c]})≤(ℓ-j+1)H(W)-∑_{k=j}^ℓH(WＷk).\]
		In layman's terms, the extra amount of information $W^{[c]}$ can steal from $W$
		is at most the amount that was stolen by $WＷj,WＷ{j+1}…WＷ{ℓ}$ from $W^{ℓ-j+1}$.
		Therefore, the difference $\abs{H(𝘞_n)-H(𝘞_n^{[c]})}$ converges to $0$.
		
		The key to the second and third $•$ is to show
		$f(Z_{cd}(W)(1-Z_d(W)))≤H(W)-H(W^{[c]})$ for all $d∈𝔽_q^×$,
		where $f$ is some monotonically increasing function that passes $(0,0)$.
		Once this is done, we have that
		either $Z_{cd}(𝘞_n)$ is small or $Z_d(𝘞_n)$ is close to $1$.
		If small $Z_{cd}(𝘞_n)$ is the case,
		then we obtain small $Z_{c²d}(𝘞_n)$ when plugging in $d=cd$;
		we further obtain small $Z_{c³d}(𝘞_n)$ when plugging in $d=c²d$;
		and so on.
		A similar argument applies if we choose close-to-$1$ $Z_d(𝘞_n)$---%
		we will obtain close-to-$1$ $Z_{d/c}(𝘞_n)$, $Z_{d/c²}(𝘞_n)$, etc.
		
		The fourth $•$ is by an independent inequality.
		And it implies that “being small” and “being close to $1$”
		are properties that can propagate among $Z_d(𝘞_n)$ for distinct $d$'s.
		The first four $•$'s together imply that either
		all of $\{Z_d(W):d∈𝔽_q^×\}$ are small or all of them are close to $1$.
		The last $•$ is just some Hölder tolls that connect the fact that
		all $Z_d(W)$ are small with the fact that $H(W)$ is small, and vice versa.
		And the Hölder tolls imply that $𝘞_n$ is extreme.
		This finishes the sketch of the proof.
	\end{proof}
	
	Remark on the theorem:
	Knowing that $\{H(𝘞_n)\}$ converges to $0$ or $1$ says
	nothing about the pace of convergence.
	In particular, we do not even know if at least one of $H(WＷj)$
	is not equal to $H(W)$---it could be that $H(𝘞_n)$ stays unchanged
	for many $n$'s and then moves a little bit before another long relaxing.
	For instance, $1/⌊㏑n⌋$ eventually converges to $0$ but it moves occasionally.
	
	Our next goal in this section is to extract, from the proof of \cref{thm:ergodic},
	a lemma that some $H(WＷj)$ is different from $H(W)$.
	This lemma will evince that the pace of convergence
	is exponential in $n$ in the next section.
	
	\begin{lem}[Ergodic kernel perturbs]\label{lem:moving}
		Let $W$ be a $q$-ary channel.
		Let $G∈𝔽_q^{ℓ×ℓ}$ be an ergodic kernel.
		Then $H(WＷj)≠H(W)$ for some $1≤j≤ℓ$ unless $H(W)∈\{0,1\}$.
	\end{lem}
	
	\begin{proof}
		First and foremost, assume that $W$ is a symmetric channel with uniform input
		and that $G=˜G$ is a lower triangular matrix.
		The former is due to a symmetrization technique that identifies
		the conditional entropy of a $q$-ary channel $W$ and its symmetric sibling $˜W$.
		The latter is by that upper triangular row-operations do not alter
		the synthetic channels up to some equivalence relation.
		See \cite{MT14} for more details about these two reductions.
		
		Now we assume the opposite of the conclusion,
		that $H(WＷj)=H(W)$ for all $1≤j≤ℓ$.
		Then, for any $1≤j≤ℓ$,
		\begin{align*}
			(ℓ-j+1)H(W)
			&	=∑_{k=j}^ℓH(WＷk)=∑_{k=j}^ℓH(U_k｜Y₁^ℓU₁^{k-1})
				=H(U_j^ℓ｜Y₁^ℓU₁^{j-1})	\\
			&	=∑_{k=j}^ℓH(U_k｜Y₁^ℓU₁^{j-1}U_{k+1}^ℓ)≤∑_{k=j}^ℓH(W)=(ℓ-j+1)H(W).
		\end{align*}
		The equality that starts the second line
		changes the order we guess $U_k$---the new order is $k=1,2…j-1,ℓ,ℓ-1…j$.
		The inequality in the second line is by that $˜G$ is lower triangular,
		and hence guessing $U_k$ is no harder than guessing $X_k$ from $Y_k$
		follow by the subtraction of the correction term $X_k-U_k=U_{k+1}^ℓ˜G$.
		Now the inequality squeezes, so all
		$H(U_k｜Y₁^ℓU₁^{j-1}U_{k+1}^ℓ)$ are equal to $H(W)$.
		
		Plugging in $k=j$ yields, particularly, $H(U_j｜Y₁^ℓU₁^{j-1}U_{j+1}^ℓ)=H(W)$.
		That is, the last guessing job, which is supposedly the easiest,
		turns out to be as difficult as the others.
		Fix any pair $i<j$ such that the $(j,i)$th entry of $˜G$ is $c≠0$.
		Then
		\[H(W)=H(U_j｜Y₁^ℓU₁^{j-1}U_{j+1}^ℓ)≤H(U_j｜Y_iY_jU₁^{j-1}U_{j+1}^ℓ)≤H(W).\]
		The first inequality is by monotonicity.
		The second inequality is by that guessing $U_j$ is no harder than
		guessing $X_j$ up to the subtraction of some correction term.
		Now the inequalities squeeze and $H(U_j｜Y_iY_jU₁^{j-1}U_{j+1}^ℓ)=H(W)$.
		
		Now look closer at the “channel” $U_j｜Y_iY_jU₁^{j-1}U_{j+1}^ℓ$.
		We get a second chance to learn more about $U_j$---since $X_i$ is a linear
		combination of $U_i^ℓ$, wherein the coefficient of $U_j$ is $c$, the output
		$Y_i$ corresponding to the input $X_i$ also carries information about $U_j$.
		The way this information is carried is the same as
		the outputs corresponding to the inputs $U_i$ and $cU_i$.
		This comes down to the second synthetic channel
		of the $2×2$ kernel $\locl$, i.e.
		\[W^{[c]}(y₁y₂u₁｜u₂)≔÷12W(y₁｜u₁+cu₂)W(y₂｜u₂).\]
		Hence $H(U_j｜Y_iY_jU₁^{j-1}U_{j+1}^ℓ)≤H(W^{[c]})≤H(W)$.
		The inequalities squeeze again, $H(W^{[c]})=H(W)$.
		
		Now we cite a difficult inequality in \cite[Appendix~A]{MT14}:
		\[H(W)-H(W^{[c]})≥-㏑（1-÷1q∑_{d∈𝔽^×}Z_{cd}(W)²(1-Z_d(W))）≥0.\]
		Squeeze one more time;
		we end up with $Z_{cd}²(1-Z_d)=0$, for all $d∈𝔽_q^×$.
		(We omit the argument “$(W)$” from now on.)
		If it is the case that $Z_{cd}=0$, then the next factor $1-Z_{cd}$
		around the corner is $1$, which forces $Z_{c²d}$ to be $0$.
		The latter in turns forces $Z_{c³d}$ to be $0$.
		This argument propagates throughout the multiplicative orbit $⟨c⟩d⊆𝔽_q^×$.
		Similarly, if it is the case that $Z_d=1$,
		then $Z_{d/c},Z_{d/c²}…Z_d$ are, forcably, all $1$.
		In summary, all $Z_d$ in the orbit $⟨c⟩d$ share a common fate.
		In fact, since $c$ can be any nonzero off-diagonal entry of $˜G$,
		we deduce that all $Z_d$ in the big orbit $⟨˜G⟩d$ share a common fate.
		
		The last piece of the jigsaw puzzle is to show that all $Z_d$,
		no matter which orbit they lie in, share the same common fate.
		We hereby cite \cite[Lemma~21]{MT14}:
		If $Z_d=Z_e=1$ and $d+e≠0$, then $Z_{d+e}=1$.
		Therefore, if $Z_d=1$ for some $d∈𝔽_p$ in the ground field,
		then $Z_d=1$ for all $d∈𝔽_p$.
		By that off-diagonal entries of $˜G$ generate $𝔽_q$ as an $𝔽_p$-algebra and
		that $Z_d$'s lying in the same $˜G$-orbit share the common fate,
		we conclude that $Z_d=1$ for all $d∈𝔽_q$.
		If, otherwise, no $d$ in the ground field has $Z_d=1$,
		then all $d$ have $Z_d=0$.
		
		Now we have that all $d∈𝔽_q$ share the common fate---either
		all $Z_d(W)$ are $1$ or all $Z_d(W)$ are $0$.
		So $Z(W)$ is either $1$ or $0$.
		Thus $H(W)$ is either $1$ or $0$.
		This completes the proof.
	\end{proof}
	
	We just see that either $G$ does nothing to the channel $W$
	or $G$ will synthesize a $WＷj$ that has a distinct conditional entropy.
	In other words, $H(𝘞₁)$ is not a constant but a true random variable.
	In the next section, I will leverage the fact that $H(𝘞₁)$ is not constant,
	regardless how small its variance is, to show that
	$H(𝘞₀),H(𝘞₁),H(𝘞₂),\dotsc$ polarizes in an exponential pace.

\section{Eigen Behavior}\label{sec:eigen}

	What constitutes this section is a compactness argument that aims to show that
	every ergodic kernel has a positive $ϱ$ as in the eigenvalue formula
	\[ℓ^{-ϱ}≔\sup_{W：q†-ary†}÷{h(H(WＷ1))+h(H(WＷ2))+\dotsb+h(H(WＷ{ℓ}))}{ℓh(H(W))},
		\label{sup:dmc}\]
	where $h$ is an easy-to-handle eigenfunciton.
	From that we can conclude $𝘌[h(𝘞_n)]<ℓ^{-ϱn}$
	and then move on to the en23 behavior.
	
	We can almost see how \cref{sup:dmc} can be estimated:
	Take a (strictly) concave $h$.
	By the last section,
	$H(WＷj)$ are not equal to each other whilst their average is $H(W)$.
	By Jensen's inequality, $𝘌[h(H(𝘞₁))]≤h(H(𝘞₀))$, and the equality cannot hold.
	This means that the fraction within the supremum is strictly less than $1$.
	Should the supremum be less than $1$, a positive $ϱ$ exists.
	Now we see why \cref{sup:dmc} is nontrivial:
	The supremum of some less-than-$1$ numbers can be $1$,
	especially when the domain is not compact.
	
	The remainder of this section finds a compact subset of $q$-ary channels,
	on which the supremum is strictly less than $1$,
	and handles the supremum over the complement set separately.
	
	\begin{lem}[Eigen for mediocre channels]\label{lem:eigen-H}
		Fix a $G∈𝔽_q^{ℓ×ℓ}$.
		Fix an $h(z)≔√{\min(z,1-z)}$.
		For any $δ>0$,
		\[\sup_{δ≤H(W)≤1-δ}÷{h(H(WＷ1))+h(H(WＷ2))+\dotsb+h(H(WＷ{ℓ}))}{ℓh(H(W))}<1,
			\label{sup:mediocre}\]
		where the supremum is taken over all $q$-ary $(Q,W)$-pairs
		whose conditional entropy lies within $[δ,1-δ]$.
	\end{lem}
	
	\begin{proof}
		Approach one:
		All delta--epsilon arguments in \cref{lem:moving}
		(which is inspired by \cref{thm:ergodic}) can be made explicit.
		That will give an upper bound on \cref{sup:mediocre}.
		
		Approach two:
		It suffices to show that the space of $q$-ary channels
		with mediocre conditional entropy is sequentially compact,
		plus $•Ｗj$ and $H$ are continuous w.r.t.\ the same topology.
		Once this is done, any sequence $W₁,W₂,\dotsc$ whose corresponding fractions
		converge to $1$ must converge to the corresponding fraction of some $W_∞$,
		which is strictly less than $1$ and leads to a contradiction.
		
		Let the simplex $Δ^𝒳$ be the closed subset of $[0,1]^𝒳$
		constrained by that the sum of coordinates is $1$.
		This is the set of all probability distributions on $𝒳$.
		Let $𝒫(Δ^𝒳)$ be the set of probability distributions on $Δ^𝒳$.
		A pair $(Q,W)$ of a DMC $W$ with an input distribution $Q$
		corresponds to a distribution on $Δ^𝒳$, i.e. an element of $𝒫(Δ^𝒳)$,
		through the posterior probabilities seen by the decoder.
		In details, whenever the decoder sees $Y=y$,
		it looks up the symbol $y$ in the table of
		posterior probabilities and learns a tuple
		\[(W(x₁｜y),W(x₂｜y)…W(x_q｜y)),\]
		where $x₁,x₂…x_q$ enumerate the symbols of $𝒳$.
		This tuple is an element of $Δ^𝒳$.
		Now that the channel output $Y$ is a random variable,
		the tuple of posterior probabilities
		\[(W(x₁｜Y),W(x₂｜Y)…W(x_q｜Y))\]
		is itself random.
		This tuple is a $Δ^𝒳$-valued random variable
		and obeys some distribution in $𝒫(Δ^𝒳)$.
		This distribution of a random tuple is
		the representative of $(Q,W)$ in $𝒫(Δ^𝒳)$.
		
		Here comes the measure theory nonsense:
		Since $𝒳$ is finite, $[0,1]^𝒳$ and $Δ^𝒳$
		are compact w.r.t.\ the Euclidean topology.
		So $𝒫(Δ^𝒳)$, the set of all distributions on $Δ^𝒳$, is tight.
		By Prokhorov's theorem, $𝒫(Δ^𝒳)$ is sequentially compact
		w.r.t.\ the topology of weak convergence.
		Now notice that $H(W)$ is just the expectation/integral
		\[E「∑_{x∈𝒳}-W(x｜Y)㏒_qW(x｜Y)」=∫∑_{x∈𝒳}-W(x｜y)㏒_qW(x｜y)\,\diff y.\]
		Note that the “integratee” $∑_{x∈𝒳}-W(x｜y)㏒_qW(x｜y)≤q$
		is bounded and continuous in the tuple.
		So $H$ can be extended to a continuous map from $𝒫(Δ^𝒳)$ to $[0,1]$.
		By extension I mean that $H$ is now defined over all $q$-ary input channels
		regardless of whether the output alphabet is discrete or continuous.
		Similarly, all $H(WＷj)$ can be written as
		more complicated expectations/integrals of $W(x｜Y)$.
		They are all continuous w.r.t.\ the topology of weak convergence.
		
		Finally, this is what happens if there exists a sequence of channels
		$W₁,W₂…$ (notice the font) whose corresponding fractions
		\[÷{h(H(•Ｗ1))+h(H(•Ｗ2))+\dotsb+h(H(•Ｗ{ℓ}))}{ℓh(H(•))}\]
		converge to $1$:
		Map these channels to $𝒫(Δ^𝒳)$.
		A sequence in $𝒫(Δ^𝒳)$ must contain a convergent subsequence.
		Let $W_∞$ be the limit of any such subsequence.
		Then, as $W_∞$ also satisfies \cref{lem:moving},
		\[÷{h(H(W_∞Ｗ1))+h(H(W_∞Ｗ2))+\dotsb+h(H(W_∞Ｗ{ℓ}))}{ℓh(H(W_∞))}<1.\]
		(In particular, the denominator is $≥h(δ)=√δ$.)
		A contradiction.
		This means the supremum is strictly less than $1$ and is what I calimed.
		
		This proof is partially inspired by \cite{Nasser18c,Nasser18t}.
	\end{proof}
	
	The case of mediocre channels is done.
	It remains to upper bound the supremum for
	channels in the neighborhoods of $H=0$ and $H=1$.
	First goes the neighborhood of $H=0$.
	
	\begin{lem}[Eigen for reliable channels]\label{lem:eigen-Z}
		Fix a $G∈𝔽_q^{ℓ×ℓ}$ with coset distance $D_ZＷj≥5$ for some $j$,
		that is, the Hamming distance from some row to
		the subspace spanned by the rows below is $5$ or farther.
		Fix an $h(z)≔√{\min(z,1-z)}$.
		Then there exists an $δ$ such that
		\[\sup_{0<H(W)<δ}÷{h(H(WＷ1))+h(H(WＷ2))+\dotsb+h(H(WＷ{ℓ}))}{ℓh(H(W))}<1,
			\label{sup:good}\]
		where the supremum is taken over all $q$-ary $(Q,W)$-pairs
		whose conditional entropy lies within the interval $(0,δ)$.
	\end{lem}
	
	\begin{proof}
		Let us assume that it is the last row of $G$ that has Hamming weight $5$.
		For if $\hdis(r_j,R_j)≥5$ is satisfied with other index $j<ℓ$,
		the following proof works with minor modifications.
		
		Temporarily let $δ$ be $H(W)$ and assume that this is really small.
		This paragraph bounds $ε≕H(WＷ{ℓ})$ from above:
		Pay the explicit Hölder toll (\cref{pro:ex-toll}),
		we obtain $Ｚ(W)<q³√{H(W)}=q³√δ$.
		Apply FTPC$Z$ (\cref{thm:ftpcZ}), we see that
		$Ｚ(WＷ{ℓ})≤f_ZＷ{ℓ}(Ｚ(W))=Ｚ(W)⁵<q^{15}δ^{2.5}$.
		Pay the Hölder toll for the return-trip (\cref{pro:ex-toll}),
		we arrive at $ε=H(WＷ{ℓ})<q³√{Ｚ(WＷ{ℓ})}<q^{10.5}δ^{1.25}$,
		or $ε<q^{10.5}δ^{1.25}$ for short.
		This is the only $H(WＷj)$ I know how to estimate.
		
		Look at the fraction in \cref{sup:good}.
		The last term in the numerator is $h(ε)≤√ε<√{q^{10.5}δ^{1.25}}$.
		We do not know about the other terms in the numerator;
		but at least we can apply Jenson's inequality
		(i.e., we assume they are all equal to minimize the extent of polarization)
		\begin{align*}
			\qquad&\kern-2em
			h(H(WＷ1))+h(H(WＷ2))+\dotsb+h(H(WＷ{ℓ-1}))	\\
			&	≤(ℓ-1)h（÷{H(WＷ1)+H(WＷ2)+\dotsb+H(WＷ{ℓ-1})}{ℓ-1}）\\
			&	=(ℓ-1)h（÷{ℓδ-ε}{ℓ-1}）≤(ℓ-1)√{÷{ℓδ-ε}{ℓ-1}}=√{ℓ-1}√{ℓδ-ε}<√{(ℓ-1)ℓδ}.
		\end{align*}
		The first equality is by FTPC$H$ (\cref{thm:ftpcH}).
		
		Now the fraction in \cref{sup:good} has
		its numerator simplified down to two terms:
		\[†\cref{sup:good}†≤÷{√{(ℓ-1)ℓδ}+√ε}{ℓ√δ}
			=√{÷{ℓ-1}{ℓ}}+÷1{ℓ}√{÷{ε}{δ}}\label{ine:1small}\]
		If we let $δ→0$, then $ε/δ≤q^{10.5}δ^{1.25}/δ=q^{10.5}δ^{0.25}→0$.
		So there is a positive $δ$ such that
		the right-hand side of \cref{ine:1small} is strictly less than $1$.
		This completes the proof.
	\end{proof}
	
	Next goes the neighborhood of $H=1$.
	But it is just the dual of the previous lemma.
	
	\begin{lem}[Eigen for noisy channels]\label{lem:eigen-S}
		Fix a $G∈𝔽_q^{ℓ×ℓ}$ with dual coset distance $D_SＷj≥5$ for some $j$,
		that is, the Hamming distance from some column of $G^{-1}$ to
		the subspace spanned by the column to the left is $5$ or farther.
		Fix $h(z)≔√{\min(z,1-z)}$.
		Then there exists an $δ$ such that
		\[\sup_{1-δ<H(W)<1}÷{h(H(WＷ1))+h(H(WＷ2))+\dotsb+h(H(WＷ{ℓ}))}{ℓh(H(W))}<1,\]
		where the supremum is taken over all $q$-ary $(Q,W)$-pairs
		whose conditional entropy lies within the interval $(1-δ,1)$.
	\end{lem}
	
	\begin{proof}
		The proof is merely the dual of the proof of \cref{lem:eigen-Z}.
		We shall not repeat.
	\end{proof}
	
	\Cref{lem:eigen-H,lem:eigen-Z,lem:eigen-S} together imply that,
	if $G$ and $G^{-1}$ have large Hamming distances among their rows and columns,
	respectively, then \cref{sup:dmc} is strictly less than $1$.
	However, this does not imply anything about matrices with shorter distances,
	especially when $G$ is $4×4$ or smaller.
	Thankfully, Kronecker product (tensor product) leverages the Hamming distances.
	
	\begin{lem}[Thress steps as one big step]\label{lem:cube}
		Fix any ergodic matrix $K∈𝔽_q^{ℓ×ℓ}$.
		Its cubic Kronecker power $G≔K^{⊗3}$
		(a)	is ergodic,
		(b)	has some $D_ZＷj≥8$, and
		(c)	has some $D_SＷj≥8$.
	\end{lem}
	
	\begin{proof}
		For (a):
		It suffices to consider the Kronecker power of the lowered form $˜K$.
		Since the diagonal of $˜K$ is all $1$, any Kronecker power $˜K^{⊗n}$
		keeps a copy of the original $˜K$ around the diagonal.
		Use the original copies to generate $𝔽_q$ as an $𝔽_p$-algebra.
		
		For (b):
		The proof is the combination of two simple facts:
		\begin{itemize}
			\item	An ergodic matrix must have some coset distance $≥2$.
			\item	Let $(u,U)$ be a vector--subspace pair of some ambient vector space
					$𝕌$ and $(v,V)$ a vector--subspace pair of another ambient
					vector space $𝕍$, then $\hdis(u⊗v,u⊗V+U⊗𝕍)≥\hdis(u,U)\*\hdis(v,V)$.
		\end{itemize}
		The first bullet point is again a consequence of \cref{thm:ergodic}.
		To elaborate, note that $K$ and and its lowered form $˜K$
		share the same coset distance profile.
		Then note that $˜K$ has at least one off-diagonal entry that is nonzero.
		Let $j$ be the last row with nonzero off-diagonal entry, then $D_ZＷj≥2$.
		Hence $K$ has some coset distance $≥2$.
		The second bullet point is worth a standalone lemma
		and will be proved as \cref{lem:tensor}.
		The two bullet points imply that
		the $(jℓ²+jℓ+j)$th coset distance of $G≔K^{⊗3}$ is at least $8$.
		
		For (c):
		It is the dual of (b).
		The proof ends here.
	\end{proof}
	
	\begin{lem}\label{lem:tensor}
		Let $(u,U)$ be a vector--subspace pair of some ambient vector space $𝕌$
		and $(v,V)$ a vector--subspace pair of another ambient vector space $𝕍$,
		then it holds that $\hdis(u⊗v,u⊗V+U⊗𝕍)≥\hdis(u,U)\hdis(v,V)$.
	\end{lem}
	
	\begin{proof}
		View $u$, $v$ and $u⊗v$ as a
		column vector, a row vector, and a rank-$1$ matrix, respectively.
		Now we want to compute the Hamming distance
		from the matrix $u⊗v$ to a subset of matrices $u⊗V+U⊗𝕍$.
		This is equal to the Hamming distance from $u⊗(v+V)$ to $U⊗𝕍$.
		Let $v'∈v+V$ be any row vector (whose Hamming weight is at least $\hdis(v,V)$).
		It suffices to prove $\hdis(u⊗v',U⊗𝕍)≥\hdis(u,U)\hwt(v')$.
		
		Without loss of generality, assume that the first coordinate of $v'$ is nonzero.
		Then the first column of $u⊗v'$ is a nonzero multiple of $u$,
		whilst the first column of any matrix in $U⊗𝕍$ is a column vector in $U$.
		The number of mismatching entries in the first column is thus $≥\hdis(u,U)$.
		Repeat this argument for all columns where $v'$ is nonzero,
		then the total number of mismatching entries is
		$≥\hdis(u,U)\hwt(v')≥\hdis(u,U)\hdis(v,V)$.
		This is what we want to show.
	\end{proof}
	
	Remark:
	There is no practical reason to use a large matrix
	with poor coset distance profile to polarize channels.
	What \cref{lem:cube} is good for is to prove that
	easy-to-implement $2×2$ matrices such as
	\[\bma{
		1	&	0	\\
		c	&	1	
	}\]
	polarize channels, given that $c$ generates $𝔽_q$ over $𝔽_p$.
	
	Recap of this section:
	Up to this point, we have seen that
	an ergodic kernel $G$ will make $H(𝘞₁)$ non-constant,
	that the eigenvalue is $<1$ for mediocre channels,
	that the eigenvalue is $<1$ for reliable channels given $D_ZＷj≥5$,
	that the eigenvalue is $<1$ for noisy channels given $D_SＷj≥5$,
	and that the third Kronecker power $K^{⊗3}$ has sufficient coset distances.
	These result in $K^{⊗3}$ having eigenvalue (i.e., \cref{sup:dmc}) $<1$.
	The next theorem summarizes the eigen behavior of any ergodic kernel.
	
	\begin{thm}[From ergodic to eigen]\label{thm:dmc-eigen}
		If $K∈𝔽_q^{ℓ×ℓ}$ is an ergodic matrix, then $G≔K^{⊗3}$ is a kernel such that
		\[\sup_{0<H(W)<1}÷{h(H(WＷ1))+h(H(WＷ2))+\dotsb+h(H(WＷ{ℓ³}))}{ℓ³h(H(W))}<1,\]
		where the supremum is taken over all $q$-ary channels
		that are not completely noisy or entirely reliable.
	\end{thm}
	
	The next section will take advantage of the eigen theorem to show that
	every ergodic kernel enjoys an en23 behavior with positive $ϱ$.

\section{En23 Behavior}\label{sec:en23}

	Fix a $q$-ary $W$ and an ergodic $G∈𝔽_q^{ℓ×ℓ}$.
	The latter could be a large matrix with sufficient coset distances
	or the cubic power of any ergodic matrix.
	Let the entropy process $\{𝘏_n\}$ be defined by $𝘏_n≔H(𝘞_n)$,
	where $\{𝘞_n\}$ is the channel process grown from $W$ via $G$.
	Let the Bhattacharyya process $\{𝘡_n\}$ be defined by $𝘡_n≔Ｚ(𝘞_n)$,
	We want to understand the asymptotic behavior of $\{𝘏_n\}$ and $\{𝘡_n\}$.
	
	To get the asymptotic behavior, \cref{lem:Z-en23} provides
	a template to translate an eigen behavior into an en23 behavior.
	The only difference is that there, the eigen behavior was given in terms of $Z$,
	and here, the eigen behavior is given in terms of $H$.
	This turns out to be ineffective to the proof;
	in fact, an eigen behavior can be given in terms of
	any parameter that is bi-Hölder to $H$ and $Z$.
	Beyond this, there is nothing new to comment on.
	Let us go straight toward the lemma.
	
	\begin{lem}[From eigen to en23]\label{lem:dmc-en23}
		If a kernel $G∈𝔽_q^{ℓ×ℓ}$ and a concave function $h$
		are such that $h(0)=h(1)=0$ and
		\[\sup_{0<H(W)<1}÷{h(H(WＷ1))+h(H(WＷ2))+\dotsb+h(H(WＷ{ℓ}))}{ℓh(H(W))}
			=ℓ^{-ϱn},\]
		then
		\[𝘗｛𝘡_n<e^{-n^{2/3}}｝>1-H(W)-ℓ^{-ϱn+o(n)}.\]
	\end{lem}
	
	\begin{proof}
		Telescope
		$𝘌[h(𝘏_n)]≤𝘌[𝘌[h(𝘏_n)｜𝘑₁𝘑₂\dotsm𝘑_{n-1}]]≤𝘌[h(𝘏_{n-1})]ℓ^{-ϱ}≤ℓ^{-ϱn}$.
		So $𝘏_n$ refuses to stay around the middle
		\begin{align*}
			\qquad&\kern-2em
			𝘗｛e^{-n^{3/4}}≤𝘏_n≤1-e^{-n^{3/4}}｝=𝘗｛h(𝘏_n)≥h\(e^{-n^{3/4}}\)｝	\\
			&	≤÷{𝘌[h(𝘏_n)]}{h(\exp(-n^{3/4}))}≤÷{h(𝘏₀)ℓ^{-ϱn}}{h(\exp(-n^{3/4}))}
				<÷{ℓ^{-ϱn}}{\exp(-n^{3/4})}<ℓ^{-ϱn+o(n)}.
		\end{align*}
		Next, recall that $𝘏_n→𝘏_∞∈\{0,1\}$ and $𝘗\{𝘏_n→0\}=1-𝘏_0=1-H(W)$.
		So
		\begin{align*}
			\qquad&\kern-2em
			𝘗｛𝘏_n<e^{-n^{3/4}}｝≥𝘗\{𝘏_n→0\}-𝘗｛𝘏_m→0† but †𝘏_n≥e^{-n^{3/4}}｝	\\
			&	=1-H(W)-𝘗｛𝘏_m† will visit †［e^{-m^{3/4}},1-e^{-m^{3/4}}］†
				for some †m≥n｝	\\
			&	>1-H(W)-∑_{m=n}^∞ℓ^{-ϱm+o(m)}=1-H(W)-ℓ^{-ϱn+o(n)}.
		\end{align*}
		Now pay the Hölder toll to translate the inequality into one about $𝘡_n$.
		That is, $𝘗\{𝘡_n<\exp(-n^{2/3})\}>1-H(W)-ℓ^{ϱn-o(n)}$.
		That is exactly the inequality we want.
	\end{proof}
	
	Before the section ends, I want to comment on
	possible sources of the eigen behavior.
	As one can see in \cref{cha:origin},
	the eigen behavior for BDMC and $\loll$ was given in terms of $Z$
	because it is easier to bound $Z$ back then (cf.\ \cref{lem:squares}).
	Here, the eigen behavior is given in terms of $H$ because it forms a martingale,
	and Jensen's inequality indicates that the eigenvalue is $≤1$.
	In the next chapter, readers will see $H$ again for the same reason.
	In general, eigen behavior can be stated in any parameter such that
	$\textsf{YourParameter}_n≤\exp(-n^{3/4})$ implies $𝘡_n≤\exp(-n^{2/3})$;
	because then the proof given above works.
	
	In the next section, we will go one step further to the een13 behavior
	that will be built on top of this section's en23 behavior,
	which was built on top of last section's ergodic behavior.
	The proof in the next section assumes the same format as that in \cref{lem:Z-een13}.

\section{Een13 Behavior}\label{sec:een13}

	The derivation of the een13 behavior from the en23 behavior uses
	the fact that $Ｚ(WＷj)$ is about $Ｚ(W)$ to the power of ${D_ZＷj}$.
	By that an ergodic kernel has nontrivial coset distance profile,
	at least one $D_ZＷj$ will square the $Ｚ(W)$.
	Now it is a matter of Hoeffding's inequality to control
	the frequency that a trajectory of $\{𝘡_n\}$ undergoes squaring.
	
	To facilitate the reasoning about coset distances,
	define the distance process $\{𝘋_n\}$ via $𝘋_n≔D_ZＷ{𝘑_n}$.
	Then, by FTPC$Z$,
	$𝘡_{n+1}≤q^{ℓ-𝘑_{n+1}}𝘡_n^{𝘋_{n+1}}≤q^ℓ𝘡_n^{𝘋_{n+1}}≈𝘡_n^{𝘋_{n+1}}$.
	And now we can telescope and talk about the product $𝘋₁𝘋₂\dotsm𝘋_n$.
	For instance, the product is $1$ with probability $𝘗\{𝘋₁=1\}^n$.
	This is the probability that $𝘡₀$--$𝘡_n$ do not undergo any squaring
	or higher powering, and we have no control on such $𝘡_n$
	(except the trivial one $𝘡_n≤q^{(ℓ-1)^n}𝘡₀$).
	Hence, I hope that $𝘗\{𝘋₁=1\}^n$ is dominated by
	the desired gap to capacity $ℓ^{-ϱn}$.
	Equivalently, I hope that $𝘗\{𝘋₁=1\}≤ℓ^{-ϱ}$.
	
	\begin{con}[Good $D$ implies good $ϱ$]
		For any kernel $G∈𝔽_q^{ℓ×ℓ}$,
		\[𝘗\{𝘋₁=1\}<ℓ^{-ϱ}.\]
	\end{con}
	
	Even if $𝘗\{𝘋₁=1\}≥ℓ^{-ϱ}$,
	we can still re-choose a lower $ϱ>0$ to satisfy $𝘗\{𝘋₁=1\}<ℓ^{-ϱ}$.
	This is one of the two new issues, in contrast to \cref{cha:origin},
	that we need to take care of in this section.
	The other new issue is that $\{𝘡_n\}$ is no longer a supermartingale.
	
	\begin{lem}[Artificial supermartingale]\label{lem:artifact}
		For any $ε>0$, there exist a smaller $ε>0$ and a small $δ>0$
		such that $\{𝘡_n^ε∧δ\}$ is a supermartingale.
		Here $𝘡_n^ε∧δ$ is a shorthand for $\min(𝘡_n^ε,δ)$.
	\end{lem}
	
	\begin{proof}
		It suffices to pick $ε$ and $δ$ such that $𝘌[𝘡₁^ε∧δ]≤𝘡₀^ε∧δ$ for any $W$.
		If $𝘡₀^ε≥δ$, there is nothing to prove;
		thus we may assume that $𝘡₀^ε<δ$.
		Start from $𝘗\{𝘋₁=1\}<ℓ^{-ϱ}<1$.
		Then either $q^{ℓε}𝘗\{𝘋₁=1\}<1$ or
		we can reselect a smaller $ε>0$ to make it true.
		We then choose a small $δ>0$ such that $q^{ℓε}\*𝘗\{𝘋₁=1\}+q^{ℓε}δ𝘗\{𝘋₁≥2\}≤1$.
		Now upper bound the conditional expectation by handling the two cases:
		\begin{align*}
			𝘌[𝘡₁^ε∧δ]
			&	≤(q^ℓ𝘡₀)^ε𝘗\{𝘋₁=1\}+(q^ℓ𝘡₀²)^ε𝘗\{𝘋₁≥2\}	\\
			&	=𝘡₀^ε(q^{ℓε}𝘗\{𝘋₁=1\}+q^{ℓε}𝘡₀^ε𝘗\{𝘋₁≥2\})≤𝘡₀^ε·1=𝘡₀^ε∧δ.
		\end{align*}
		This finishes the proof of $𝘌[𝘡₁^ε∧δ]≤𝘡₀$.
		By the tree structure of the process,
		$𝘌[𝘡_{n+1}｜𝘑₁𝘑₂\dotsm𝘑_n]$ versus $𝘡_n$ is just $𝘌[𝘡₁]$ vs $𝘡₀$ with $W←𝘞_n$.
		This completes the proof of the whole lemma.
	\end{proof}
	
	Remark:
	The choice of $δ$ in the lemma implies that $q^{ℓε}δ$ is less than
	$(1-q^{ℓε}𝘗\{𝘋₁=1\})/𝘗\{𝘋₁≥2\}<1$, or equivalently $q^{ℓε}≤δ^{-1}$.
	Thus $𝘡_{n+1}≤q^ℓ𝘡_n^{𝘋_{n+1}}≤δ^{-1/ε}𝘡_n^{𝘋_{n+1}}≤𝘡_n^{𝘋_{n+1}-ε}$
	whenever $𝘡_n^ε<δ$.
	In other words, when we look at $𝘡_n^ε$ in the “safe zone” $[0,δ]$,
	not only is $𝘡_n^ε$ a supermartingale, but FTPC$Z$ also takes a simpler form
	$𝘡_{n+1}≤𝘡_n^{𝘋_{n+1}-ε}≤𝘡_n^{𝘋_{n+1}(1-ε)}$ that facilitates telescoping.
	
	The main statement of the een13 behavior follows.
	
	\begin{lem}[From en23 to een13]\label{lem:dmc-een13}
		Given $𝘗\{𝘋₁=1\}<ℓ^{-ϱ}$ and \cref{lem:dmc-en23}, that is, given
		\[𝘗｛𝘡_n<e^{-n^{2/3}}｝>1-H(W)-ℓ^{-ϱn+o(n)},\]
		we have
		\[𝘗｛𝘡_n<\exp\(-e^{n^{1/3}}\)｝>1-H(W)-ℓ^{-ϱn+o(n)}.\label{ine:dmc-een13}\]
	\end{lem}
	
	\begin{proof}
		(Select constants.)
		Since $𝘗\{𝘋₁=1\}<ℓ^{-ϱ}$,
		there exists a large $λ>1$ such that $𝘌[𝘋₁^{-λ}]<ℓ^{-ϱ}$.
		Pick a small $ε>0$ such that $𝘌[𝘋₁^{-λ}]8^{ελ}<ℓ^{-ϱ}$.
		Invoke \cref{lem:artifact};
		pick a smaller $ε>0$ and a small $δ>0$
		such that $\{𝘡_n^ε∧δ\}$ is a supermartingale.
		According to the proof of \cref{lem:artifact} (and the remark underneath),
		this $δ$ is such that $𝘡_{n+1}≤𝘡^{𝘋_{n+1}-ε}$ whenever $Z_n^ε<δ$.
		
		(Define events.)
		Consider only perfect square $n$.
		(If it is not the case, follow the workaround in \cref{lem:Z-en23}.)
		Let $𝘌₀⁰$ be the empty event.
		For every $m=√n,2√n…n-√n$, we define five series of events
		$𝘈_m$, $𝘉_m$, $𝘊_m$, $𝘌_m$, and $𝘌₀^m$ inductively as below:
		Let $𝘈_m$ be $\{𝘡_m<\exp(-m^{2/3})\}、𝘌₀^{m-√n}$.
		Let $𝘉_m$ be a subevent of $𝘈_m$ where $𝘡_l^ε≥δ$ for some $l≥m$.
		Let $𝘊_m$ a subevent of $𝘈_m$ where
		\[𝘋_{m+1}𝘋_{m+2}\dotsm𝘋_{m+√n}≤8^{ε√n}.\label{ine:mgf-8en}\]
		Let $𝘌_m$ be $𝘈_m、(𝘉_m∪𝘊_m)$.
		Let $𝘌₀^m$ be $𝘌₀^{m-√n}∪𝘌_m$.
		Let $𝘢_m$, $𝘣_m$, $𝘤_m$, $𝘦_m$, and $𝘦₀^m$
		be the probability measures of the corresponding capital letter events.
		Moreover, let $𝘨_m$ be $1-H(W)-𝘦₀^m$.
		
		(Bound $𝘣_m/𝘢_m$ from above.)
		Conditioning on $𝘈_m$, I want to estimate
		the probability that $𝘡_l^ε≥δ$ for some $l≥m$.
		Recall that $\{𝘡_l^ε∧δ\}$ is made a supermartingale.
		By Ville's inequality,
		$𝘗\{𝘡_l^ε≥δ$ for some $l≥m｜𝘈_m\}≤𝘡_m^ε/δ<\exp(-m^{2/3}ε)/δ$.
		This is an upper bound on $𝘣_m/𝘢_m$
		and will be summoned in \cref{ine:core-een13}.
		
		(Bound $𝘤_m/𝘢_m$ from above.)
		I am to estimate how frequently \cref{ine:mgf-8en} happens.
		It is the probability of  $(𝘋_{m+1}𝘋_{m+2}\dotsm𝘋_{m+√n})^{-λ}≥8^{-ελ√n}$.
		This probability must not exceed
		$𝘌[(𝘋_{m+1}𝘋_{m+2}\dotsm𝘋_{m+√n})^{-λ}]\*8^{ελ√n}
			=𝘌[𝘋₁^{-λ}]^{√n}8^{ελ√n}=(𝘌[𝘋₁^{-λ}]8^{ελ})^{√n}<ℓ^{-ϱ√n}$
		by Markov's inequality.
		This is an upper bound on $𝘤_m/𝘢_m$
		and will be summoned in \cref{ine:core-een13}.
		
		(Bound $(𝘨_{m-√n}-𝘢_m)^+$ from above.)
		By definitions, $𝘨_{m-√n}-𝘢_m=1-H(W)-(𝘦₀^{m-√n}+𝘢_m)$.
		The definition of $𝘈_m$ forces it to be disjoint from $𝘌₀^{m-√n}$,
		thus $𝘦₀^{m-√n}+𝘢_m$ is the probability measure of $𝘌₀^{m-√n}∪𝘈_m$.
		This union event must contain the event
		$\{𝘡_m<\exp(-m^{2/3})\}$ by how $𝘈_m$ was defined.
		Recall the en23 behavior $𝘗\{𝘡_m<\exp(-m^{2/3})\}>1-H(W)-ℓ^{-ϱm+o(m)}$.
		Chaining all inequalities together, we deduce $𝘨_{m-√n}-𝘢_m<ℓ^{-ϱm+o(m)}$.
		Let $(𝘨_{m-√n}-𝘢_m)^+$ be $\max(0,𝘨_{m-√n}-𝘢_m)$
		so we can write $(𝘨_{m-√n}-𝘢_m)^+<ℓ^{-ϱm+o(m)}$.
		This upper bound will be summoned in \cref{ine:core-een13}.
		
		(Bound $𝘦₀^{n-√n}$ from below.)
		We start rewriting $𝘨_m$ with $𝘨_m^+$ being $\max(0,𝘨_m)$:
		\begin{align*}
			𝘨_m
			&	=1-H(W)-𝘦₀^m=1-H(W)-(𝘦₀^{m-√n}+𝘦_m)	\\
			&	=1-H(W)-𝘦₀^{m-√n}-𝘦_m=𝘨_{m-√n}-𝘦_m	\\
			&	=𝘨_{m-√n}（1-÷{𝘦_m}{𝘢_m}）+÷{𝘦_m}{𝘢_m}(𝘨_{m-√n}-𝘢_m)	\\
			&	≤𝘨_{m-√n}^+（1-÷{𝘦_m}{𝘢_m}）+÷{𝘦_m}{𝘢_m}(𝘨_{m-√n}-𝘢_m)^+	\\
			&	≤𝘨_{m-√n}^+（1-÷{𝘦_m}{𝘢_m}）+(𝘨_{m-√n}-𝘢_m)^+	\\
			&	≤𝘨_{m-√n}^+（÷{𝘣_m}{𝘢_m}+÷{𝘤_m}{𝘢_m}）+(𝘨_{m-√n}-𝘢_m)^+	\\
			&	<𝘨_{m-√n}^+（e^{-m^{2/3}ε}/δ+ℓ^{-ϱ√n}）+ℓ^{-ϱm+o(m)}.
				\label{ine:core-een13}
		\end{align*}
		The first four equalities are by the definitions of $𝘨_m$ and $𝘌₀^m$.
		The next equality is by elementary algebra.
		The next two inequalities are by $0≤𝘦_m/𝘢_m≤1$.
		The next inequality is by the definition of $𝘌_m$.
		The last inequality summons upper bounds derived in the last three paragraphs.
		The last line contains two terms in the big parentheses;
		between them, $ℓ^{-ϱ√n}$ dominates $\exp(-m^{2/3}ε)/δ$
		once $m$ is greater than $O(n^{3/4})$.
		Subsequently, we obtain this recurrence relation:
		\[\cas{
			𝘨_{O(n^{3/4})}≤1,	\\
			𝘨_m≤2𝘨_{m-√n}^+ℓ^{-ϱ√n}+ℓ^{-ϱm+o(m)}.
		}\]
		Solve it (cf.\ the master theorem);
		we get that $𝘨_{n-√n}<ℓ^{-ϱn+o(n)}$.
		By the definition of $𝘨_{n-√n}$,
		we immediately get $𝘦₀^{n-√n}>1-H(W)-ℓ^{-ϱn+o(n)}$.
		
		(Analyze $𝘌₀^{n-√n}$.)
		We want to estimate $𝘏_n$ when $𝘌₀^{n-√n}$ happens.
		To be precise, for each $m=√n,2√n…n-√n$,
		we attempt to bound $𝘡_{m+√n}$ when $𝘌_m$ happens.
		Fix an $m$.
		When $𝘌_m$ happens, its superevent $𝘈_m$ happens,
		so we know that $𝘡_m<\exp(-m^{2/3})$.
		But $𝘉_m$ does not happen, so $𝘡_l^ε<δ$ for all $l≥m$.
		This implies that $𝘡_{l+1}≤𝘡_l^{𝘋_{l+1}(1-ε)}$ for those $l$.
		Telescope;
		$𝘡_{m+√n}$ is less than or equal to $𝘡_m$ raised to
		the power of $𝘋_{m+1}𝘋_{m+2}\dotsm𝘋_{m+√n}(1-ε)^{√n}$.
		But $𝘊_m$ does not happen, so the product is greater than
		$8^{ε√n}(1-ε)^{√n}=(8√[ε]{1-ε})^{ε√n}$.
		The latter is greater than $2^{ε√n}$ granted that $ε<1/2$.
		Jointly we have $𝘡_{m+√n}≤𝘡_m^{2^{ε√n}}<\exp(-m^{2/3}2^{ε√n})$.
		Recall that $𝘡_{l+1}≤q^ℓ𝘡_l$ for all $l≥m+√n$.
		Then telescope again;
		$𝘡_n≤q^{ℓ(n-m-√n)}𝘡_{m+√n}<q^{ℓn}\exp(-m^{2/3}2^{ε√n})<\exp\(-e^{n^{1/3}}\)$
		provided that $n$ is sufficiently large.
		In other words, $𝘌₀^{n-√n}$ implies $𝘡_n<\exp\(-e^{n^{1/3}}\)$.
		
		(Summary.)
		Now we may conclude
		$𝘗｛𝘡_n<\exp\(-e^{n^{1/3}}\)｝≥𝘗(𝘌₀^{n-√n})=𝘦₀^n>1-H(W)-ℓ^{-ϱn+o(n)}$.
		And hence the proof of the een13 behavior is sound.
	\end{proof}
	
	In the next section, I will show the elpin behavior of $\{𝘡_n\}$.
	That will imply the MDP behavior of polar coding and
	the latter will specialize to the LDP and CLT behaviors of polar coding.
	In contract to this section, where we used
	$𝘗\{𝘋_n=1\}$ and $𝘗\{𝘋_n≥2\}$ to determine the behavior of $\{𝘡_n\}$,
	we will have to use the “full power” of $\{𝘋_n\}$ in what follows.

\section{Elpin Behavior}\label{sec:elpin}

	As the een13 behavior lowers $𝘡_m$ to the order of $\exp\(-e^{m^{1/3}}\)$,
	we can afford telescoping $𝘡_{l+1}≤𝘡_l^{𝘋_{l+1}-ε}$
	for more $l$'s without having to worry about $𝘡_l^ε≥δ$.
	So the trade-off between the block error probability ($≈∑𝘡_n$) and
	the code rate ($≈𝘗\{𝘡_n<⋯\}$) amounts to the behavior of the product
	$𝘋_{m+1}𝘋_{m+2}\dotsm𝘋_n$ or, as probability theorists may prefer,
	the behavior of the i.i.d.\ sum $㏒_ℓ𝘋_{m+1}+㏒_ℓ𝘋_{m+2}+\dotsb+㏒_ℓ𝘋_n$.
	
	As \cref{sec:regime} suggests,
	the limiting behavior of i.i.d.\ sums is well known.
	In particular, I will borrow the LDP result that concerns
	the limiting behavior of $𝘗\{¯X_n<x\}$ when $x$ is fixed and $n$ increases.
	First goes some definitions.
	
	\begin{dfn}
		The \emph{cumulant generating function} of $㏒_ℓ𝘋₁$ is defined to be
		the logarithm of its moment generating function or, more precisely,
		\[𝘒(t)≔㏒_ℓ(𝘌[𝘋₁^{t}])=÷1{㏑ℓ}㏑∑_{j=1}^ℓ÷{𝘋₁^{t}}{ℓ}.\]
	\end{dfn}
	
	\begin{dfn}
		The \emph{Cramér function} of $㏒_ℓ𝘋₁$ is defined to be
		the convex conjugate of the cumulant generating function or, more precisely,
		\[𝘓(s)≔\sup_{t≤0}st-𝘒(t).\]
	\end{dfn}
	
	$𝘓(s)$ controls the limiting law of
	the sum of $㏒_ℓ𝘋_l$ or the product of $𝘋_l$.
	An example (which we have seen) is that, when $𝘋_l∈\{1,2\}$,
	its logarithm is just a coin toss, and $1-h₂$ controls the limiting law.
	
	\begin{thm}[Cramér's theorem]\label{thm:cramer}
		For any $n$,
		\[𝘗\{𝘋₁𝘋₂\dotsm𝘋_n≤ℓ^{sn}\}≤ℓ^{-n𝘓(s)}.\]
		Furthermore, $𝘓(s)$ is the greatest value satisfying that, i.e.,
		\[\lim_{n→∞}÷1n㏒_ℓ𝘗\{𝘋₁𝘋₂\dotsm𝘋_n<ℓ^{sn}\}=-𝘓(s).\]
	\end{thm}
	
	\begin{proof}
		What I actually need in the sequel is the upper bound on $𝘗$.
		The fact that $𝘓(s)$ gives the best upper bound is irrelevant to
		the validity of the theorems in this chapter,
		but indicates that the theorems are somewhat optimal.
		The proof of the latter is thus omitted;
		see standard textbooks instead (e.g., \cite{DZ10}).
		
		To prove the upper bound, recall that $𝘓(s)$ is a supremum.
		This supremum, by some compactness argument, turns out to be a maximum.
		Let $t$ be the argument that maximizes $st-𝘒(t)$.
		Then $𝘋₁𝘋₂\dotsm𝘋_n≤ℓ^{sn}$ is equivalent to $(𝘋₁𝘋₂\dotsm𝘋_n)^t≥ℓ^{stn}$.
		Apply Markov's inequality, the probability that
		$𝘋₁^t𝘋₂^t\dotsm𝘋_n^t≥ℓ^{stn}$ is at most
		$𝘌[𝘋₁^t𝘋₂^t\dotsm𝘋_n^t]ℓ^{-stn}=(𝘌[𝘋₁^t]ℓ^{-st})^n=ℓ^{(𝘒(t)-st)n}=ℓ^{-n𝘓(s)}$
	\end{proof}
	
	Now we define the achievable region in terms of $𝘓(s)$;
	it ought to be the set of $(π,ρ)$ pairs such that
	$𝘗\{𝘡_n<\exp(-ℓ^{πn})\}>1-H(W)+ℓ^{-ρn}$.
	Consider $ϖ$ being the mean of $㏒_ℓ𝘋₁$.
	We know that, most of the times, i.i.d.\ averages of $㏒_ℓ𝘋_l$
	concentrate around the mean so we should not expect that $π>ϖ$.
	Similarly, we do not expect that $ρ>ϱ$.
	The remaining characterization of $(π,ρ)$ is in terms of $𝘓(s)$.
	
	\begin{dfn}
		Let $𝒪⊆[0,ϖ]×[0,ϱ]$ be an open region
		defined by the following criterion:
		for any $(π,ρ)∈𝒪$, the ray shooting from $(π,ρ)$ toward the opposite direction
		of $(0,ϱ)$ does not intersect the function graph of $ρ=𝘓(π)$.
	\end{dfn}
	
	Note that the criterion is equivalent to, geometrically,
	that $(π,ρ)$ lies to the left of the convex envelope of $(0,ϱ)$ and $𝘓(s)$.
	So is it equivalent to
	\[𝘓（÷{πn}{n-m}）>÷{ρn-ϱm}{n-m}\label{ine:Lray-raw}\]
	for all $0<n<m$.
	The main theorem of this chapter can now be stated and proved.
	
	\begin{thm}[From een13 to elpin]\label{thm:dmc-elpin}
		Fix a pair $(π,ρ)∈𝒪$.
		Given \cref{lem:dmc-een13}, that is, given
		\[𝘗｛𝘛_n<\exp\(-e^{n^{1/3}}\)｝>H(W)-ℓ^{-ϱn+o(n)},\]
		we have
		\[𝘗｛𝘛_n<e^{-ℓ^{πn}}｝>H(W)-ℓ^{-ρn+o(n)}.\]
	\end{thm}
	
	\begin{proof}
		(Select constants.)
		Since \cref{ine:Lray-raw} holds, there exists a small constant $ε>0$ such that
		\[𝘓（÷{πn}{n-m}+2ε）>÷{ρn-ϱm}{n-m}\label{ine:Lray-gap}\]
		by the compactness argument.
		Pass this $ε$ to \cref{lem:artifact}:
		There exists a smaller $ε>0$ and a small $δ>0$ such that $𝘡_n^ε∧δ$ is
		a supermartingale and $𝘡_{n+1}≤𝘡_n^{𝘋_{n+1}(1-ε)}$ whenever $𝘡_n<δ$.
		
		(Define events.)
		Let $n$ be a perfect square.
		(If it is not the case, see how \cref{thm:Z-e2pin} bypasses.)
		Let $𝘈₀⁰$ and $𝘌₀⁰$ be the empty event.
		For every $m=√n,2√n…n-√n$, we define six series of events
		$𝘈_m$, $𝘈₀^m$, $𝘉_m$, $𝘊_m$, $𝘌_m$, and $𝘌₀^m$
		inductively as follows:
		Let $𝘈_m$ be $｛𝘡_m<\exp(-e^{m^{1/3}})｝、𝘈₀^{m-√n}$.
		Let $𝘈₀^m$ be $𝘈₀^{m-√n}∪𝘈_m$.
		Let $𝘉_m$ be a subevent of $𝘈_m$ where $𝘡_l^ε≥δ$ for some $l≥m$.
		Let $𝘊_m$ a subevent of $𝘈_m$ where
		\[𝘋_{m+1}𝘋_{m+2}\dotsm𝘋_n≤ℓ^{πn+2ε(n-m)}.\label{ine:mgf-lpin}\]
		Let $𝘌_m$ be $𝘈_m、(𝘉_m∪𝘊_m)$.
		Let $𝘌₀^m$ be $𝘌₀^{m-√n}∪𝘌_m$.
		Let $𝘢_m$, $𝘢₀^m$, $𝘣_m$, $𝘤_m$, $𝘦_m$, and $𝘦₀^m$
		be the probability measures of the corresponding capital letter events.
		Moreover, let $𝘧_m$ be $1-H(W)-𝘢₀^m$ and let $𝘨_m$ be $1-H(W)-𝘦₀^m$.
		
		(Bound $𝘣_m/𝘢_m$ from above.)
		Conditioning on $𝘈_m$, we want to estimate
		the probability that $𝘡_l≥δ$ for some $l≥m$.
		Recall that $𝘡_l$ is a supermartingale.
		Hence by Ville's inequality,
		$𝘗\{𝘡_l^ε≥δ$ for some $l≥m｜𝘈_m\}≤𝘡_m^ε/δ<\exp\(-e^{m^{1/3}}ε\)/δ$.
		This is an upper bound on $𝘣_m/𝘢_m$
		and will be summoned in \cref{ine:core-elpin}.
		
		(Bound $𝘤_m/𝘢_m$ from above.)
		We want to estimate how often \cref{ine:mgf-lpin} happens.
		This is equivalent to asking how often do $n-m$
		fair coin tosses end up with $πn+2ε(n-m)$ heads.
		By \cref{thm:cramer}, this probability is less than $ℓ$ to the power of
		\[-(n-m)𝘓（÷{πn}{n-m}+2ε）.\]
		By \cref{ine:Lray-gap}, this exponent is less than $ϱm-ρn$.
		Thus, the probability is less than $ℓ^{ϱm-ρn}$.
		This is an upper bound on $𝘤_m/𝘢_m$
		and will be summoned in \cref{ine:core-elpin}.
		
		(Bound $𝘧_m^+$ from above.)
		The definition of $𝘧_m$ reads $1-H(W)-𝘢₀^m$.
		Here $𝘢₀^m$ is the probability measure of $𝘈₀^m$,
		and $𝘈₀^m$ is a superevent of $𝘈_m$ by how the former is defined.
		Event $𝘈₀^m$ must contain $｛𝘡_m<\exp\(-e^{m^{1/3}}\)｝$
		by how $𝘈_m$ was defined.
		By the een13 behavior, $𝘗｛𝘡_m<\exp\(-e^{m^{1/3}}\)｝>1-H(W)-ℓ^{-ϱm+o(m)}$.
		Chaining all inequalities together, we infer that $𝘧_m<ℓ^{-ϱm+o(m)}$.
		Let $𝘧_m^+$ be $\max(0,𝘧_m)$ so
		we can write $𝘧_m^+<ℓ^{-ϱm+o(m)}$.
		This upper bound will be summoned in \cref{ine:core-elpin}.
		
		(Bound $𝘦₀^{n-√n}$ from below.)
		We start rewriting $𝘨_m-𝘧_m^+$ with
		$(𝘧_{m-√n}-𝘢_m)^+$ being $\max(0,𝘧_{m-√n}-𝘢_m)$:
		\begin{align*}
			𝘨_m-𝘧_m^+
			&	=1-H(W)-𝘦₀^m-(1-H(W)-𝘢₀^m)^+	\\
			&	=1-H(W)-𝘦₀^{m-√n}-𝘦_m-(1-H(W)-𝘢₀^{m-√n}-𝘢_m)^+	\\
			&	=𝘨_{m-√n}-𝘦_m-(𝘧_{m-√n}-𝘢_m)^+	\\
			&	≤𝘨_{m-√n}-𝘦_m-÷{𝘦_m}{𝘢_m}(𝘧_{m-√n}-𝘢_m)^+	\\
			&	≤𝘨_{m-√n}-𝘦_m-÷{𝘦_m}{𝘢_m}(𝘧_{m-√n}^+-𝘢_m)	\\
			&	=𝘨_{m-√n}-𝘧_{m-√n}^++𝘧_{m-√n}^+（1-÷{𝘦_m}{𝘢_m}）	\\
			&	≤𝘨_{m-√n}-𝘧_{m-√n}^++𝘧_{m-√n}^+（÷{𝘣_m}{𝘢_m}+÷{𝘤_m}{𝘢_m}）	\\
			&	<𝘨_{m-√n}-𝘧_{m-√n}^++ℓ^{-ϱ(m-√n)+o(m-√n)}
				（\exp\(-e^{m^{1/3}}ε\)/δ+ℓ^{ϱm-ρn}）.\label{ine:core-elpin}
		\end{align*}
		The first three equalities are by the definitions of $𝘨_m$ and $𝘧_m$.
		The next inequality is by $0≤𝘦_m/𝘢_m≤1$.
		The next inequality is by $\max(0,f-a)=\max(a,f)-a≥\max(0,f)-a$.
		The next equality is by elementary algebra.
		The next inequality is by the definition of $𝘌_m$.
		The last inequality summons upper bounds derived in the last three paragraphs.
		Now the last line contains two terms in the big parentheses;
		between them, $ℓ^{ϱm-ρn}$ dominates $\exp\(-e^{m^{1/3}}ε\)/δ$ once $n→∞$.
		Subsequently, we obtain this recurrence relation
		\[\cas{
			𝘨₀-𝘧₀^+=0;	\\
			𝘨_m-𝘧_m^+≤𝘨_{m-√n}-𝘧_{m-√n}^++2ℓ^{-ρn+o(n)}.
		}\]
		Solve it (cf.\ the Cesàro summation);
		we get that $𝘨_{n-√n}-𝘧_{n-√n}^+<ℓ^{-ρn+o(n)}$.
		Once again we summon $𝘧_{n-√n}^+<ℓ^{-ϱ(n-√n)+o(n-√n)}<ℓ^{-ϱn+o(n)}$;
		therefore $𝘨_{n-√n}<ℓ^{-ρn+o(n)}$.
		Based on the definition of $𝘨_{n-√n}$
		we immediately get $𝘦₀^{n-√n}>1-H(W)-ℓ^{-ρn+o(n)}$.
		
		(Analyze $𝘌₀^{n-√n}$.)
		We want to estimate $𝘡_n$ when $𝘌₀^{n-√n}$ happens.
		To be precise, for each $m=√n,2√n…n-√n$,
		we attempt to bound $𝘡_n$ when $𝘌_m$ happens .
		Fix an $m$.
		When $𝘌_m$ happens, its superevent $𝘈_m$ happens,
		so we know that $𝘡_m<\exp\(-e^{m^{1/3}}\)$.
		But $𝘉_m$ does not happen, so $𝘡_l<δ$ for all $l≥m$.
		This implies $𝘡_{l+1}≤𝘡_l^{𝘋_{l+1}(1-ε)}$ for those $l$.
		Telescope;
		$𝘡_n$ is less than or equal to $𝘡_m$ raised to
		the power of $𝘋_{m+1}𝘋_{m+2}\dotsm𝘋_n(1-ε)^{n-m}$.
		But $𝘊_m$ does not happen, so the product
		is greater than $ℓ^{πn+2ε(n-m)}(1-ε)^{n-m}$,
		which is greater than $ℓ^{πn}$ granted that $ε<1/2$.
		Jointly we have
		$𝘡_n≤𝘡_m^{ℓ^{πn}}<\exp\(-e^{m^{1/3}}ℓ^{πn}\)<\exp(-ℓ^{πn})$.
		In other words, $𝘌₀^{n-√n}$ implies $𝘡_n<\exp(-ℓ^{πn})$.
		
		(Summary.)
		Now we may conclude
		$𝘗\{𝘡_n<\exp(-ℓ^{πn})\}≥𝘗(𝘌₀^{n-√n})=𝘦₀^n>1-H(W)-ℓ^{-ρn+o(n)}$.
		And hence the proof of the elpin behavior is sound.
	\end{proof}

\section{Consequences, Dual Picture Included}

	That $𝘡_n$ can be proved low implies that we have good codes over
	symmetric channels of prime-power input alphabet.
	To code over the most general DMCs, we need to talk about $\{𝘚_n\}$---it is
	the process of $Ｓ(𝘞_n)$ that controls the behavior of $𝘞_n$ when it becomes noisy.
	
	$Ｓ$ (and $\{𝘚_n\}$) really is the dual counterpart of $Ｚ$ (and $\{𝘡_n\}$).
	The Hölder tolls we pay to translate between $H$ and $Ｚ$
	is the same amount as the tolls we pay to translate between $1-H$ and $Ｓ$.
	Plus, FTPC$Z$ (\cref{thm:ftpcZ}) and FTPC$S$ (\cref{thm:ftpcS}) state
	almost the same phenomenon in terms of $Ｚ$ and $Ｓ$, respectively.
	By the process nonsense that was once used to show \cref{thm:T-e2pin}
	out of \cref{thm:Z-e2pin}, \cref{thm:dmc-elpin} assumes an $𝘚$-version.
	
	\begin{thm}[From eigen to elpin]
		Fix a pair $(π,ρ)$ lying to the left of the convex envelope of
		$(0,ϱ)$ and the Cramér function of $㏒_ℓD_SＷ{𝘑₁}$.
		Given the premise of \cref{lem:dmc-en23}, that is, given
		\[\sup_{0<H(W)<1}÷{h(H(WＷ1))+h(H(WＷ2))+\dotsb+h(H(WＷ{ℓ}))}{ℓh(H(W))}
			=ℓ^{-ϱn},\]
		then
		\[𝘗｛𝘚_n<e^{-ℓ^{πn}}｝>H(W)-ℓ^{-ρn+o(n)}.\]
	\end{thm}
	
	This theorem together with \cref{thm:dmc-elpin} implies good codes over all DMCs.
	This is because the block error probability is bounded from above
	by the sum of small $Ｐ(𝘞_n)$ and small $T(𝘞_n)$,
	and they are further bounded by $𝘡_n$ and $𝘚_n$.
	
	\begin{cor}[Good code for DMC]
		Fix a $q$-ary DMC.
		Fix an ergodic $G∈𝔽_q^{ℓ×ℓ}$ with a positive $ϱ>0$.
		Fix a pair $(π,ρ)$ lying to the left of
		(a)	the convex envelope of $(0,ϱ)$ and the Cramér function of $㏒_ℓD_ZＷ{𝘑₁}$ and
		(b)	the convex envelope of $(0,ϱ)$ and the Cramér function of $㏒_ℓD_SＷ{𝘑₁}$.
		Then polar coding with kernel $G$ enjoys block error probability $\exp(-ℓ^{πn})$
		and code rate $N^{-ρ}$ less than the channel capacity.
	\end{cor}
	
	The same can be stated for lossless and lossy compression.
	For lossless compression, if the source alphabet $𝒳$
	is not a prime power, add dummy symbols till it is.
	Then the block error probability is the sum of small $Ｐ(𝘞_n)$.
	For lossy compression, if the palette $𝒳$ in the distortion function
	$\dist：𝒳×𝒴→[0,1]$ is not prime power, add dummy symbols and
	penalize dummy symbols with distortion $1$ (i.e., $\dist(x,y)=1$ if $x$ is dummy).
	Then the excess of distortion is the average of small $T(𝘞_n)$.
	
	\begin{cor}[Good code for lossless compression]
		Fix a $q$-ary source $X$ and side information $Y$.
		Fix an ergodic $G∈𝔽_q^{ℓ×ℓ}$ (whose $ϱ$ is guaranteed to be positive).
		Fix a pair $(π,ρ)$ lying to the left of
		the convex envelope of $(0,ϱ)$ and the Cramér function of $㏒_ℓD_ZＷ{𝘑₁}$.
		Then polar coding with kernel $G$ enjoys block error probability $\exp(-ℓ^{πn})$
		and code rate $N^{-ρ}$ plus the conditional entropy.
	\end{cor}
	
	\begin{cor}[Good code for lossy compression]
		Fix a random source $Y∈𝒴$ and a distortion function $\dist：𝔽_q×𝒴→[0,1]$.
		Fix an ergodic $G∈𝔽_q^{ℓ×ℓ}$ (whose $ϱ$ is guaranteed to be positive).
		Fix a pair $(π,ρ)$ lying to the left of the convex envelope
		of $(0,ϱ)$ and the Cramér function of $㏒_ℓD_SＷ{𝘑₁}$.
		Then polar coding with kernel $G$ enjoys excess of distortion $\exp(-ℓ^{πn})$
		and code rate $N^{-ρ}$ plus the test channel capacity.
	\end{cor}
	
	There are questions to be answered.
	One, how can we predict the best $ϱ$ for a specific $G$?
	Two, what is the best $ϱ$ among all $G$?
	Can we achieve the optimal exponent $ϱ=1/2$?
	Three, Can we reduce the complexity?
	The first question is open.
	The next chapter answers the second question.
	The chapter after answers the third question.

\chapter{Random dynamic Kerneling}\label{cha:random}

	\dropcap
	Significant is the fact that random coding achieves capacity.
	Researchers have been using random coding
	to prove the first achievability bounds for various coding scenarios.
	This evinces that random coding is the very reason why capacity,
	or capacity region when there are more than two users, is what it is.
	This is the first moment of coding theory.
	In addition, random coding achieves capacity at an unbeatable pace---the
	block error probability decays exponentially fast in the block length,
	and the block length grows inverse-quadratically in the gap to capacity.
	This sets up a holy grail that is left to code designers to chase after,
	which is the second-moment paradigm of coding theory.
	
	The question is simple, Can polar coding touch the second-moment paradigm,
	i.e., achieve capacity at a pace that is comparable to random coding?
	
	This chapter gives an affirmative answer.
	But it requires one to rethink the polarizing kernel from the bottom up.
	As mentioned in  \cref{cha:general}, a general kernel $G$
	can be implemented as an EU--DU pair, and some copies of EU--DU pairs
	will wrap around the parent channels to synthesize the child channels.
	In doing so, not all channels need to be wrapped.
	As demonstrated in \cref{cha:prune}, A channel (a pair of pins)
	is left naked if an algorithm decides that the channel is
	sufficiently polarized and not worth more EU--DU pairs.
	The contribution made in \cref{cha:prune} is that this reduces complexity.
	Stretching this idea a bit, we see that a channel (a pair of pins) can also be
	wrapped by an EU--DU pair that corresponds to a different kernel \cite{YB15}.
	Once we accept the idea that $G$ can vary on a channel-by-channel basis,
	the problem becomes how to find the best $G$ that fits a given channel.
	
	So the simple question becomes, Can we find a good kernel for each and every channel
	to help polar coding touch the edge of the second-moment paradigm,
	i.e., to achieve capacity at a nearly-optimal pace?
	
	Now we rethink what exactly is needed here.
	Do we need an algorithm to generate, for each and every channel, a kernel 
	plus a certificate that this particular kernel is good at polarizing?
	Or, only do we need to prove that, for each and every channel,
	there exists at least one kernel that polarizes decently.
	The latter is considerably easier than the former because we do not
	specify \emph{which} kernel is good---only that it exists.
	And we know how this could be done: random coding.
	Use a random matrix $𝔾$ to polarize a channel and show that,
	on average, the channel is polarized to a satisfactory extent.
	Then we are done;
	the rest is the pigeonhole principle.
	
	I call selecting a kernel for each channel \emph{dynamic kerneling},
	This chapter in its entirety is a quantization of random dynamic kerneling.

\section{The Holy Grail}

	Fix a DMC $W$ and a capacity-achieving input distribution $Q$.
	The following optimality bound is borrowed.
	
	\begin{thm}[Optimal codes]\label{pro:optimal}
		\cite[Theorem~2]{AW14}, \cite[Theorem~6]{PV10}.
		Fix constants $π,ρ>0$ such that $π+2ρ>1$.
		Assume $V>0$, that is, assume
		\[V≔\operatorname{Var}「㏑÷{W(X｜Y)}{Q(X)}」>0.\]
		(Remark:
		This is an easy assumption to meet;
		it mostly excludes trivial channels such as those with $H=0,1$.)
		Then no block code assumes $Ｐ<\exp(-ℓ^{πn})$ and $R>C-N^{-ρ}$
		except for sufficiently small $N$.
	\end{thm}
	
	Fix exponents $π,ρ>0$ such that $π+2ρ<1$.
	I am to construct a series of codes with error probability $\exp(-N^π)$
	and code rate $N^{-ρ}$ less than the channel capacity.
	The construction, based on polar coding,
	will enjoy encoding and decoding complexity $O(N㏒N)$.

\section{Why Dynamic Kernels}

	Readers may prefer numerical evidence for why dynamic kerneling is superior.
	Since there is essentially one matrix of $2×2$ size,
	I will take $3×3$ for demonstration.
	Over $𝔽₂$, there are essentially two matrices of $3×3$ size:
	\[\GY≔\bma{1& & \\0&1& \\1&1&1}\qquad\GB≔\bma{1& & \\1&1& \\1&0&1}\]
	$\GY$ has coset distances $\{1,1,3\}$ and dual coset distances $\{1,2,2\}$.
	And $\GB$ is its dual:
	$\GB$ has coset distances $\{1,2,2\}$ and dual coset distance $\{1,1,3\}$.
	Both $\GY$ and $\GB$ have $ϱ=1/4.938$ over BECs.
	Their regions of realizable $(π,ρ)$-pairs over BECs are plotted in \cref{fig:rigid}.
	
	\begin{figure}
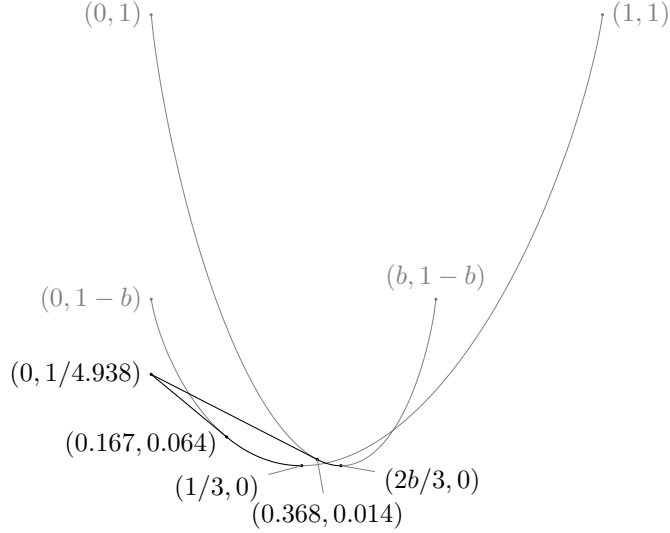

		\tikz[scale=6]{
			\draw[gray,very thin]
				(0,.3691)pic{dot}node[left]{$(0,1-b)$}--
				plot[domain=.00001:.1672,samples=34](\x,{h3(\x)})--(.1672,.06365)
				(1/3,0)--plot[domain=1/3:1,samples=133](\x,{h3(\x)})--
				(1,1)pic{dot}node[right]{$(1,1)$}
				(0,1)pic{dot}node[left]{$(0,1)$}--
				plot[domain=.00001:.3682,samples=75](\x,{h4(\x)})--(.3682,.01364)
				(.4206,0)--plot[domain=.4206:.6309,samples=43](\x,{h4(\x)})--
				(.6309,.3691)pic{dot}node[above]{$(b,1-b)$};
			\draw
				(0,1/4.938)pic{dot}node[left]{$(0,1/4.938)$}--
				(.1672,.06365)pic{dot}node[anchor=5]{$(0.167,0.064)$}--
				plot[domain=.1672:1/3,samples=133](\x,{h3(\x)})--
				(1/3,0)pic{dot}coordinate[pin=-165:{$(1/3,0)$}];
			\draw
				(0,1/4.938)--(.3682,.01364)pic{dot}coordinate[pin=-80:{$(0.368,0.014)$}]--
				plot[domain=.3682:.4206,samples=15](\x,{h4(\x)})--
				(.4206,0)pic{dot}coordinate[pin=-10:{$(2b/3,0)$}];
		}
		\caption{
			The two regions and Cramér functions of $\GY$ over BECs.
			Here $b=㏒₃2$.
			The curve passing $(0,1-b)$--$(0.167,0.064)$--$(1/3,0)$--$(1,1)$
			is the Cramér function of the uniform distribution on $\{0,0,1\}$
			and also the KL divergence of $p+(1-p)$ w.r.t.\ $1/3+2/3$.
			The curve passing $(0,1/4.938)$--$(0,1-b)$--$(0.167,0.064)$--$(1/3,0)$
			is the boundary of realizable $(π,ρ)$-pairs in the estimate of $𝘡_n$.
			The curve passing $(0,1)$--$(0.368,0.014)$--$(2b/3,0)$--$(b,1-b)$
			is the Cramér function of the uniform distribution on $\{b,b,0\}$
			and also the KL divergence of $p/b+(1-p/b)$ w.r.t.\ $2/3+1/3$.
			The curve passing $(0,1/4.938)$--$(0.368,0.014)$--$(2b/3,0)$--$(1,0)$
			is the boundary of realizable $(π,ρ)$-pairs in the estimate of $𝘚_n$.
			For $\GB$ over BECs, swap the two curves.
		}\label{fig:rigid}
	\end{figure}
	
	Now consider this variant of polar coding:
	When $H(W)≤1/2$, wrap the channel with $\GB$;
	when $H(W)>1/2$, wrap the channel with $\GY$.
	Let $𝒢$ denote an amoebic kernel that becomes $\GB$ when it sees a reliable channel
	and becomes $\GY$ when it sees a noisy one.
	Then two things change.
	Firstly, the estimate of $ϱ$ rises from $1/4.938$ to $1/4.183$.
	This is because $\GB$ is better at polarizing reliable channels
	and $\GY$ is better at polarizing noisy ones.
	This is the division of labor.
	Secondly, the coset distances are now $\{1,2,2\}$ on both reliable and noisy ends.
	Comparing to $\{1,1,3\}$, which reward you more ($3$)
	with a lower probability ($1/3$), distances $\{1,2,2\}$
	reward you less ($2$) but with a higher probability ($2/3$).
	And the random variable that rewards you less, but steadily, is preferred.
	Now $𝒢$'s region of $(π,ρ)$-pairs becomes \cref{fig:dynamic}.
	
	\begin{figure}
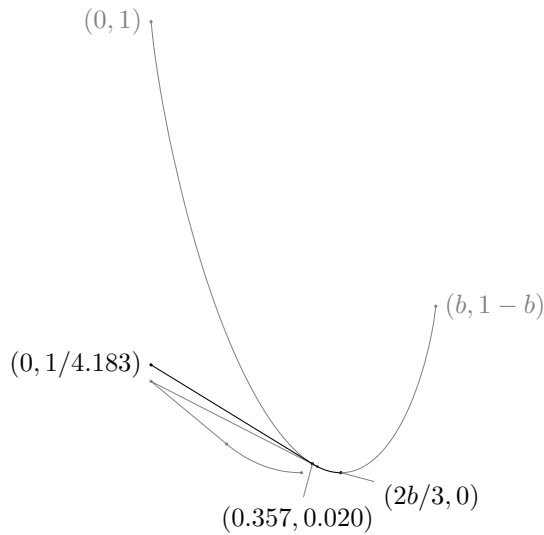

		\tikz[scale=6]{
			\draw[gray,very thin]
				(0,1)pic{dot}node[left]{$(0,1)$}--
				plot[domain=.00001:.3575,samples=72](\x,{h4(\x)})--(.3575,.01970)
				(.4206,0)--plot[domain=.4206:.6309,samples=43](\x,{h4(\x)})--
				(.6309,.3691)pic{dot}node[right]{$(b,1-b)$};
			\draw[gray,very thin]
				(0,1/4.938)pic{dot}--(.1672,.06365)pic{dot}--
				plot[domain=.1672:1/3,samples=133](\x,{h3(\x)})--(1/3,0)pic{dot}
				(0,1/4.938)--(.3682,.01364)pic{dot};
			\draw
				(0,1/4.183)pic{dot}node[left]{$(0,1/4.183)$}--
				(.3575,.01970)pic{dot}coordinate[pin=-105:{$(0.357,0.020)$}]
				plot[domain=.357:.4206,samples=15](\x,{h4(\x)})--
				(.4206,0)pic{dot}coordinate[pin=-15:{$(2b/3,0)$}];
			\draw[opacity=0]
				(1,1)node[right]{$(1,1)$};
		}
		\caption{
			The region and Cramér function of $𝒢$, the compound strategy, over BECs.
			The old regions and Cramér functions are in gray for the ease of comparison.
		}\label{fig:dynamic}
	\end{figure}
	
	For if we allow larger matrices, then it is not hard to imagine that
	the space of DMCs is partitioned into several territories, each is
	“ruled by” the kernel that is best at polarizing the channels living within.
	These kernels will be collectively called $𝒢$,
	and then we can talk about the $(π,ρ)$-pairs of $𝒢$.

\section{Why Random Kernels}

	Before the construction begins, let me elaborate on
	why random matrices are preferred.
	As the last chapter paved the way to
	the determination of the MDP region $𝒪$ of a kernel,
	we know $𝒢$ should meet the following three requirements:
	\begin{itemize}
		\item	LDP requirement:
				The coset distances of $𝒢$ should be large.
				Namely, the matrix that $𝒢$ becomes when it sees reliable channels
				should have large coset distances.
				In particular, the geometric mean of $D_ZＷj$ should be $ℓ-o(ℓ)$.
		\item	Dual LDP requirement:
				The geometric mean of $㏒_ℓD_SＷj$ should be $ℓ-o(ℓ)$.
				Namely, the matrix that $𝒢$ becomes when it sees noisy channels
				should have large coset distances.
		\item	CLT requirement:
				The performance of $𝒢$ measured by
				the eigenvalue $ℓ^{-ϱ}$ should be good.
				Namely, $ϱ=1/2-o(1)$ as $ℓ→∞$.
	\end{itemize}
	
	It is not easy to meet any of the requirements, let alone all of them at once.
	As a comparison, classical coding theory usually concerns
	how to construct a rectangular matrix (a code) with high Hamming weights.
	Less frequently it concerns finding nested matrices
	(aka.\ a flag of codes) with high Hamming weights.
	In this viewpoint, the LDP requirement is equivalent to finding a maximal chain
	of nested matrices (a full flag of codes) with high Hamming weight.
	The CLT requirement, on the other hand, does not seem to have classical counterpart.
	
	Polarizing $W$ is just side of the story.
	On the other side, we have to polarize $Q$, the input distribution
	seen as a channel, for asymmetric channels and lossy compression.
	Thus we have four requirements to meet---LDP, dual LDP, $W$'s CLT, and $Q$'s CLT.
	
	Random coding is the perfect remedy for this because it measures
	the density of the “good objects” that meet various conditions.
	Take the LDP requirement as an example.
	If $𝔾$ is a random matrix that is drawn uniformly,
	then any rectangular sub-matrix is uniformly random.
	By classical coding theory, we know how to compute
	the Hamming weights of a random rectangular matrix.
	To be more precise, we know how to compute the probability that the minimum distance
	of a random linear code is at least a certain threshold.
	We just have to repeat the argument for all rectangular sub-matrices,
	and then apply the union bound.
	
	Random matrices have more uses than that.
	If we draw $𝔾$ from the uniform ensemble of invertible matrices,
	then its Hamming distances are still under control.
	(In fact, excluding singular matrices avoids low Hamming distances
	in a helpful manner, so we are actually doing a favor
	for the random coding argument.)
	Now that the inverse $𝔾^{-1}$ follows the uniform distribution on
	invertible matrices, bounding the dual coset distances of $𝔾$ is as easy as
	bounding the primary coset distances of $𝔾$---just multiply the union bound by $2$.
	
	Below goes the actual plan to attack the problem:
	In \cref{sec:LDP}, I will show that a random matrix $𝔾$ has sufficient distances,
	and meet the LDP requirement (and its dual).
	In \cref{sec:CLT}, I will attempt to bound the eigenvalue of $𝔾$;
	this bound then comes down to four bounds:
	\begin{itemize}
		\item	a bound derived from the LDP requirement via Hölder tolls,
		\item	a bound derived from the dual LDP requirement via Hölder tolls,
		\item	a bound for noisy-channel coding via random linear codes, and
		\item	a bound for wiretap-channel coding via random linear codes.
	\end{itemize}
	\Cref{sec:prepare} prepares some materials for the last two $•$.
	In \cref{sec:gallager}, the noisy-channel coding bound will be proved.
	In \cref{sec:hayashi}, the wiretap-channel coding bound will be proved.
	
	Notice that the random matrix $𝔾$ is typeset in the blackboard bold font.
	So will the related notations (such as $ℙ,𝔼$)
	be typeset in the blackboard bold font.

\section{LDP Data of Random Matrix}\label{sec:LDP}

	This section is dedicated to showing that a random matrix $𝔾∈𝔽_q^{ℓ×ℓ}$
	drawn uniformly from the general linear group $\GL(ℓ,q)$
	has good distances $D_ZＷj$ anf $D_SＷj$.
	But before that, What is “good”?
	Let us consider $q=2$ as an example.
	The following hand-waving argument is borrowed from \cite{BF02}.
	
	To understand the typical behavior of $D_ZＷj$ when $𝔾$ varies,
	it is easier if we compute the average of $f_ZＷj(z)$ over $𝔾$,
	where $f_ZＷj(z)$ is the weight enumerator of the coset code
	$\{0₁^{j-1}1_ju_{j+1}^ℓ𝔾:u_{j+1}^ℓ∈𝔽_q^{ℓ-j}\}⊆𝔽_q^ℓ$.
	Since $𝔾$ is invertible and distributed uniformly,
	$0₁^{j-1}1_ju_{j+1}^ℓ𝔾$ is just a random nonzero vector drawn from $𝔽_q^ℓ$.
	So every message $u_{j+1}^ℓ$ contributed $((1+z)^ℓ-1)/(2^ℓ-1)$
	to the average of $f_ZＷj$, where $-1$ is to exclude the all-zero vector.
	Overall, the average of $f_ZＷj(z)$ over $𝔾$ is $((1+z)^ℓ-1)/2^j$.
	Forget about the $-1$;
	the $z^k$ coefficient of $(1+z)^ℓ/2^j$ is $\binom{ℓ}k/2^j$.
	By the large deviations theory, $[z^k]f_ZＷj=\binom{ℓ}k/2^j≈2^{ℓh₂(k/ℓ)-j}$.
	Roughly speaking, the coefficient would (very likely) be $0$
	for a realization of $𝔾$ if the average is less than $1$ (by a lot).
	So we conclude a rule of thumb:
	$D_ZＷj>k$ iff $[z^k]f_ZＷj<1$ iff $h₂(k/ℓ)<j/ℓ$.
	See \cref{fig:coeffi} for an illustration.
	
	\begin{figure}
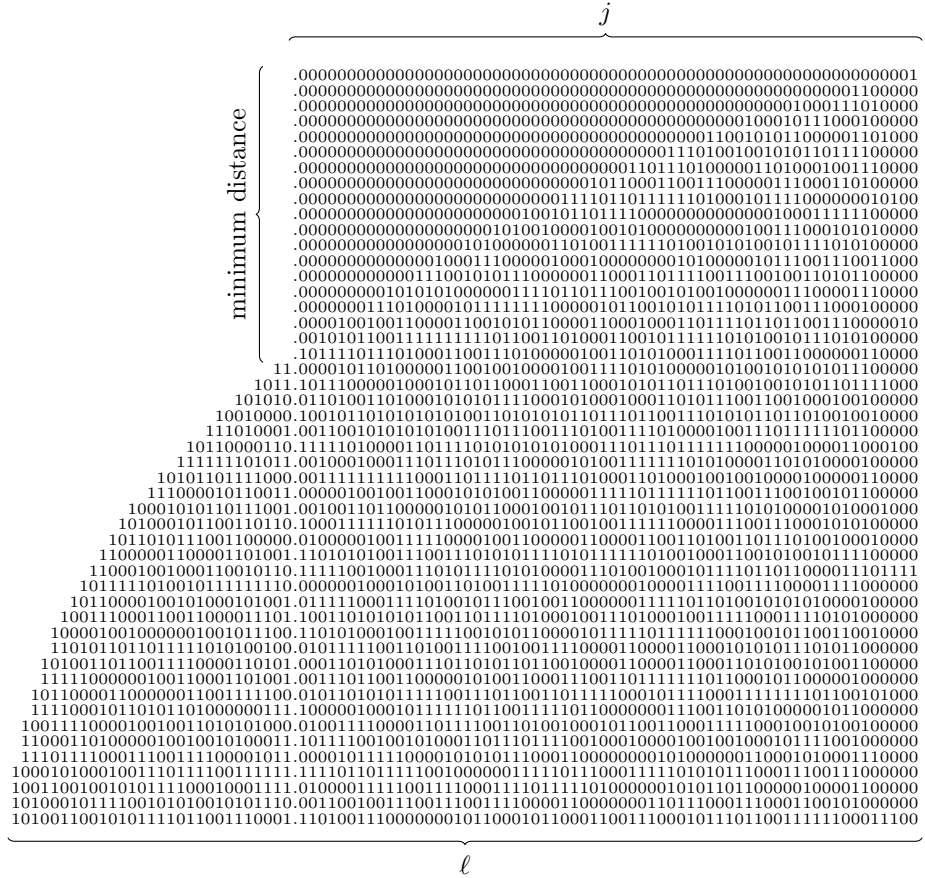

		\tikz{
			\matrix(h2analog)[matrix of math nodes,nodes={anchor=east,inner sep=1,font=\tiny}]{
		                             .0000000000000000000000000000000000000000000000000000000000000001 \\
		                             .0000000000000000000000000000000000000000000000000000000001100000 \\
		                             .0000000000000000000000000000000000000000000000000001000111010000 \\
		                             .0000000000000000000000000000000000000000000000100010111000100000 \\
		                             .0000000000000000000000000000000000000000001100101011000001101000 \\
		                             .0000000000000000000000000000000000000011101001001010110111100000 \\
		                             .0000000000000000000000000000000000110111010000011010001001110000 \\
		                             .0000000000000000000000000000001011000110011100000111000110100000 \\
		                             .0000000000000000000000000001111011011111101000101111000000010100 \\
		                             .0000000000000000000000010010110111100000000000000100011111100000 \\
		                             .0000000000000000000010100100001001010000000000100111000101010000 \\
		                             .0000000000000000010100000011010011111101001010100101111010100000 \\
		                             .0000000000000010001110000010001000000001010000010111001110011000 \\
		                             .0000000000001110010101110000001100011011110011100100110101100000 \\
		                             .0000000001010101000000111101101110010010100100000011100001110000 \\
		                             .0000000111010000101111111100000101100101011110101100111000100000 \\
		                             .0000100100110000110010101100001100010001101111011011001110000010 \\
		                             .0010101100111111111101100110100011001011111101010010111010100000 \\
		                             .1011110111010001100111010000010011010100011110110011000000110000 \\
		                           11.0000101101000001100100100001001111010100000101001010101011100000 \\
		                         1011.1011100000100010110110001100110001010110111010010010101101111000 \\
		                       101010.0110100110100010101011110001010001000110101110011001000100100000 \\
		                     10010000.1001011010101010100110101010110111011001110101011011010010010000 \\
		                    111010001.0011001010101010011101110011101001111010000100111011111101100000 \\
		                  10110000110.1111101000011011110101010101000111011101111111000001000011000100 \\
		                 111111101011.0010001000111011101011100000101001111111010100001101010000100000 \\
		               10101101111000.0011111111110001101111011011101000110100010010010000100000110000 \\
		              111000010110011.0000010010011000101010011000001111101111110110011100100101100000 \\
		            10001010110111001.0010011011000001010110001001011101101010011111010100001010001000 \\
		           101000101100110110.1000111111010111000001001011001001111110000111001110001010100000 \\
		          1011010111001100000.0100000100111110000100110000011000011001101001101110100100010000 \\
		         11000001100001101001.1101010100111001110101011110101111110100100011001010010111100000 \\
		        110001001000110010110.1111100100011101011110101000011101001000101111011011000011101111 \\
		       1011111010010111111110.0000001000101001101001111101000000010000111100111100001111000000 \\
		      10110000100101000101001.0111110001111010010111001001100000011111011010010101010000100000 \\
		     100111000110011000011101.1001101010101100110111101000100111010001001111100011110101000000 \\
		    1000010010000001001011100.1101010001001111100101011000010111110111111000100101100110010000 \\
		    1101011011011111010100100.0101111100110100111100100111100001100001100010101011101011000000 \\
		   10100110110011110000110101.0001101010001110110101101100100001100001100011010100101001100000 \\
		   11111000000100110001101001.0011101100110000010100110001110011011111110110001011000001000000 \\
		  101100001100000011001111100.0101101010111110011101100110111110001011110001111111101100101000 \\
		  111100010110101101000000111.1000001000101111110110011111011000000011100110101000001011000000 \\
		 1001111000010010011010101000.0100111100001101111001101001000101100110001111100010010100100000 \\
		 1100011010000010010010100011.1011110010010100011011101111001000100001001001000101111001000000 \\
		 1110111100011100111100001011.0000101111100001010101110001100000000101000000110001010001110000 \\
		10001010001001110111100111111.1111011011111001000000111110111000111110101011100011100111000000 \\
		10011001001010111100010001111.0100001111100111100011110111110100000010101101100000100001100000 \\
		10100010111100101010010101110.0011001001110011100111100001100000001101110001110001100101000000 \\
		10100110010101111011001110001.1101001110000000101100010110001100111000101110110011111100011100 \\
			};
			\draw[decoration=brace,decorate,transform canvas={xshift=-1em}]
				(h2analog-19-1.south west)--node[rotate=90,yshift=1em]{minimum distance}
				(h2analog-1-1.north west);
			\draw[decoration=brace,decorate]
				([yshift=1em]h2analog-1-1.north west)--node[yshift=1em]{$j$}
				([yshift=1em]h2analog-1-1.north east);
			\draw[decoration=brace,decorate]
				([yshift=-1em]h2analog-48-1.south east)--node[yshift=-1em]{$ℓ$}
				([yshift=-1em]h2analog-48-1.south west);
		}
		\caption{
			The binary representations of $\binom{96}k/2^{64}$
			for $k=0$ (top) to $k=48$ (bottom).
			Note how the leading digits depict the rotation of $h₂$.
			(Zoom out if it is not clear).
		}\label{fig:coeffi}
	\end{figure}
	
	Let us not forget that we are to bound the distances with high probability.
	Therefore, there is not a clear cut at $h₂(k/ℓ)<j/ℓ$;
	we have to leave a gap between $h₂(k/ℓ)$ and $j/ℓ$
	to facilitate Markov's inequality.
	Here, I will use the upper bound $√{ep}≥h₂(p)$ shown in \cref{fig:rootep},
	and will leave a gap by loosening the bound to $√{3p}≥h₂(p)$.
	The formal statement and proof is below.
	
	\begin{thm}[Typical LDP behavior]\label{thm:typical-LDP}
		Let $z≔Ｚ(W)$.
		Fix an $ℓ≥30$.
		Draw a $𝔾∈\GL(ℓ,q)$ uniformly at random.
		Fix a $j∈[ℓ]$.
		Then, with probability less than $3q^{-√ℓ/13}$, the following fails:
		\[f_ZＷj(z)≤ℓ(1+q'z)^{ℓ-1}(q'z)^{⌈j²/3ℓ⌉},\label{ine:enum-Z}\]
		where $q'≔q-1$.
		In particular, we have $D_ZＷj≥⌈j²/3ℓ⌉$.
	\end{thm}
	
	\begin{proof}
		Divide $j$ into two cases: $1≤j≤√{3ℓ}$ and $√{3ℓ}<j≤ℓ$.
		For $j=1,2…√{3ℓ}$, the exponent $⌈j²/3ℓ⌉$ is nothing but $1$,
		thus the inequality to be proved reads
		$f_ZＷj(z)≤ℓ{(1+q'z)^{ℓ-1}}q'z$, coefficient-wisely.
		The right-hand side overcounts all nonzero codewords by
		choosing a position ($ℓ$), assigning a nonzero symbol ($q'z$),
		and arbitrarily filling in the rest of $ℓ-1$ blanks ($(1+q'z)^{ℓ-1}$).
		On the left-hand side, $f_ZＷj(z)$ enumerates only codewords of the form
		$0₁^{j-1}1_ju_{j+1}^ℓ𝔾$, which are all nonzero as $𝔾$ is invertible.
		Hence \cref{ine:enum-Z} holds for $j≤√{3ℓ}$ and nonnegative $z$
		regardless of what kernel $𝔾$ is in effect.
		
		For $j=⌊√{3ℓ}⌋+1,⌊√{3ℓ}⌋+2…ℓ$, let $d≔j²/3ℓ$.
		To make \cref{ine:enum-Z} hold, we execute a two-phase procedure
		to avoid all codewords of weight less than $d$
		and to eliminate kernels with poor overall score.
		In further detail, we will reject a kernel $𝔾$ if there exists $u_{j+1}^ℓ$
		such that $\hwt(0₁^{j-1}1_ju_{j+1}^ℓ𝔾)<d$ and call it phase~I\@.
		Afterwards, among surviving kernels with only \emph{heavy}
		(high weight) codewords, we will reject a kernel if its overall score
		$f_ZＷj(z)$ is too low and call it phase~II\@.
		The failing probability $3q^{-√ℓ/13}$ is the price we pay for rejecting.
		Up to this point, two things remain to be analyzed:
		how much probability we pay for rejecting \emph{light}
		(low weight) codewords in phase~I (answer: $q^{-√ℓ/13}$),
		and what is the Markov cutoff that honors \cref{ine:enum-Z}
		in phase~II (answer: $2q^{-√ℓ/13}$).
		
		Phase~I analysis is as follows:
		Fix $u_{j+1}^ℓ$ and vary $𝔾∈\GL(ℓ,q)$;
		the codeword $𝕏₁^ℓ≔0₁^{j-1}1_ju_{j+1}^ℓ𝔾$ is
		a nonzero vector distributed uniformly on $𝔽_q^ℓ、\{0₁^ℓ\}$.
		This distribution is almost identical to the uniform distribution on $𝔽_q^ℓ$.
		Assume $𝕏₁^ℓ$ follows the latter;
		this makes $𝕏₁^ℓ$ lighter,
		which is compatible with the direction of the inequalities we want.
		Then the probability that $𝕏₁^ℓ$ has weight less than $d$ is the probability
		that $ℓ$ Bernoulli trials---each $𝕏_j$ is ``zero'' with probability $1/q$ and
		``nonzero'' with probability $q'/q$---result in less than $d$ ``nonzero''s.
		By the large deviations theory \cite[Exercise~2.2.23(b)]{DZ10},
		$\hwt(𝕏₁^ℓ)<d$ holds with probability less than
		\[\exp（-ℓ𝔻（÷d{ℓ}‖÷12））=2^{-ℓ(1-h₂(d/ℓ))}\]
		for the $q=2$ case, where $𝔻$ is the Kullback--Leibler divergence.
		For general $q$, similarly, $\hwt(𝕏₁^ℓ)<d$ holds with probability less than
		\[\exp（-ℓ𝔻（÷d{ℓ}‖1-÷1q））≤q^{-ℓ(1-h₂(d/ℓ))}.\]
		It is less than $q^{-ℓ(1-h₂(d/ℓ))}$ by the comparison made in \cref{fig:kldiv}
		(meaning that $q=2$ is the most difficult case).
		Now we obtain $h₂(d/ℓ)<√{ed/ℓ}=√{ej²/3ℓ²}=(√{e/3})j/ℓ<0.952j/ℓ$.
		Therefore, the single-word rejecting probability is less than
		$q^{-ℓ(1-h₂(d/ℓ))}<q^{-ℓ+0.952j}$.
		Take into account that there are $q^{ℓ-j}$ codewords, one for each $u_{j+1}^ℓ$.
		The union bound yields
		$q^{ℓ-j}q^{-ℓ+0.952j}=q^{-0.048j}<q^{-0.048√{3ℓ}}<q^{-√ℓ/13}$.
		Therefore, the total rejecting probability is $q^{-√ℓ/13}$.
		Phase~I ends here.
		
		Phase~II analysis is as follows:
		After we reject some $𝔾$ in phase~I, some codewords will disappear;
		particularly, this includes all light codewords.
		Therefore,
		the expectation of $f_ZＷj(z)$ is bounded by the weight enumerator
		of all heavy codewords rescaled by the number of codewords.
		In detail, start from
		\begin{align*}
			𝔼[f_ZＷj(z)｜𝔾† survives phase I†]
			&	=𝔼[f_ZＷj(z)·𝕀\{𝔾† survives†\}]/ℙ\{𝔾† survives†\}	\\
			&	≤𝔼[f_ZＷj(z)·𝕀\{𝔾† survives†\}]/(1-q^{-√ℓ/13}).\label{ine:phase}
		\end{align*}
		$𝕀$ is the indicator function.
		In the denominator, $1-q^{-√ℓ/13}>1/4$ as $ℓ≥30$.
		Put that aside and redefine $d≔⌈j²/3ℓ⌉$.
		The expected value part is bounded from above by
		\def\bn#1#2{\hbox{\large$\tbinom{#1}{#2}$}}
		\begin{align*}
			\qquad&\kern-2em
			𝔼[f_ZＷj(z)·𝕀\{𝔾† survives†\}]
				=𝔼「∑_{u_{j+1}^ℓ}z^{\hwt(0₁^{j-1}1_ju_{j+1}^ℓ𝔾)}·𝕀\{𝔾† survives†\}」	\\
			&	≤𝔼「∑_{u_{j+1}^ℓ}z^{\hwt(0₁^{j-1}1_ju_{j+1}^ℓ𝔾)}·
				𝕀\{\hwt(0₁^{j-1}1_ju_{j+1}^ℓ𝔾)≥d\}」	\\
			&	=∑_{u_{j+1}^ℓ}𝔼[z^{\hwt(0₁^{j-1}1_ju_{j+1}^ℓ𝔾)}·
				𝕀\{\hwt(0₁^{j-1}1_ju_{j+1}^ℓ𝔾)≥d\}]	\\
			&	≤q^{ℓ-j}𝔼[z^{\hwt(𝕏₁^ℓ)}·𝕀\{\hwt(𝕏₁^ℓ)≥d\}]
				=q^{ℓ-j}q^{-ℓ}∑_{x₁^ℓ}z^{\hwt(x₁^ℓ)}·𝕀\{\hwt(x₁^ℓ)≥d\}	\\
			&	=q^{-j}∑_{w≥d}\bn{ℓ}w(q'z)^w
				≤q^{-j}∑_{w≥d}\bn{ℓ}d\bn{ℓ-d}{w-d}(q'z)^w	\\
			&	=q^{-j}\bn{ℓ}d∑_{w≥d}\bn{ℓ-d}{w-d}(q'z)^{w-d}(q'z)^d
				=q^{-j}\bn{ℓ}d(1+q'z)^{ℓ-d}(q'z)^d	\\
			\shortintertext{(overestimate the scalar $q^{-j}\bn{ℓ}d$)}
			&≤	(q^{-√ℓ/13}ℓ/2)(1+q'z)^{ℓ-d}(q'z)^d.
		\end{align*}
		The first equality expands the definition.
		The next inequality replaces $𝔾$ surviving phase~I by a weaker condition.
		The next equality swaps $𝔼$ and $∑$.
		The next inequality replaces
		the ensemble of $0₁^{j-1}1_ju_{j+1}^ℓ𝔾$ by a uniform $𝕏₁^ℓ∈𝔽_q^ℓ$.
		The next equality expands the definition of the expectation over $𝕏₁^ℓ$.
		The next equality counts codewords.
		The next inequality selects $w$ positions
		by first selecting $d$ and then selecting $w-d$.
		The next two equalities factor and apply the binomial theorem.
		The rest is by a series of inequalities that overestimate the scalar:
		$q^{-j}\binom{ℓ}d=q^{-j}\binom{ℓ}{⌈j²/3ℓ⌉}<q^{-j}\binom{ℓ}{j²/3ℓ}ℓ/2
			<q^{-j}2^{ℓh₂(j²/3ℓ²)}ℓ/2≤q^{-j+ℓh₂(j²/3ℓ²)}ℓ/2$.
		Similar to the end of phase~I, the exponent part is
		$-j+ℓh₂(j²/3ℓ²)<-j+ℓ√{ej²/3ℓ²}=-j+j√{e/3}<-0.048j<-0.048√{3ℓ}<-√ℓ/13$.
		Hence the scalar part is less than $q^{-√ℓ/13}ℓ/2$.
		Put $1-q^{-√ℓ/13}>1/4$ back to the denominator as in \cref{ine:phase};
		$𝔼[f_ZＷj(z)｜𝔾† survives phase I†]$ has an upper bound of
		\[2q^{-√ℓ/13}ℓ(1+q'z)^{ℓ-d}(q'z)^d.\]
		By Markov's inequality,
		\cref{ine:enum-Z} holds with probability $1-2q^{-√ℓ/13}$,
		i.e., the rejecting probability is $2q^{-√ℓ/13}$.
		Phase~II ends here.
		
		The sum of the two rejecting probabilities is $3q^{-√ℓ/13}$ as claimed
		in the theorem statement, hence settles the proof of \cref{thm:typical-LDP}.
	\end{proof}
	
	\begin{figure}
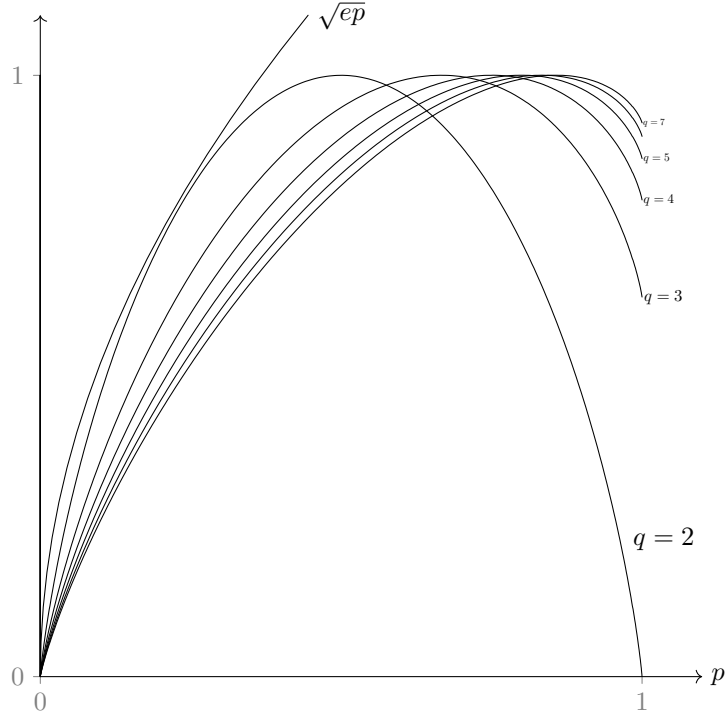

		\tikz[scale=8]{
			\draw
				(0,1)pic{y-tick=$1$}(0,0)pic{y-tick=$0$}
				(0,0)pic{x-tick=$0$}(1,0)pic{x-tick=$1$};
			\draw
				plot[domain=0:90,samples=360]({sin(\x)^2},{h2o(\x)})
				({sin(80)^2},{h2o(80)})node[above right]{$q=2$}
				plot[domain=0:1.1,samples=20](\x^2/e,\x)node[right]{$√{ep}$};
			\foreach\q in{3,4,5,6,7}{
				\PMS\lnqq{ln(\q-1)}\PMS\lnq{ln(\q)}\PMS\invq{1/log2(\q)}
				\PMD{hqo}1{%
					\PMS\sin{sin(##1)}\PMS\cos{cos(##1)}%
					\ifdim\sin pt=0pt%
						\PMP{0}%
					\else\ifdim\cos pt=0pt%
						\PMP{1-\lnqq/\lnq}%
					\else%
						\PMP{(\sin*(ln(\sin)*2+\lnq-\lnqq)*\sin
							+ \cos*(ln(\cos)*2+\lnq)*\cos)/\lnq}
					\fi\fi
				}
				\draw plot[domain=0:90,samples=360]({sin(\x)^2},{1-hqo(\x)});
			}
			\foreach\q in{3,4,5,7}{
				\PMS\lnqq{ln(\q-1)}\PMS\lnq{ln(\q)}\PMS\invq{1/log2(\q)}
				\draw(1,\lnqq/\lnq)node[right,inner sep=1,scale=\invq]{$q=\q$};
			}
			\draw[->](0,0)--(0,1.1);
			\draw[->](0,0)--(1.1,0)node[right]{$p$};
		}
		\caption{
			One minus the Kullback--Leibler divergences $1-𝔻(p\|1-1/q)/\ln(q)$
			for $q=2,3,5,7$ and an upper bound of $√{ep}$.
		}\label{fig:kldiv}
	\end{figure}
	
	The bound I just proved,
	\[f_ZＷj(z)≤ℓ(1+q'z)^{ℓ-1}(q'z)^{⌈j²/3ℓ⌉},\tagcopy{ine:enum-Z}\]
	implies its dual siblings because $𝔾^{-1}$ is uniform on $\GL(ℓ,q)$.
	
	\begin{cor}
		Let $s≔Ｓ(W)$.
		Then, with probability $1-3q^{-√ℓ/13}$, the following holds for each $j∈[ℓ]$:
		\[f_SＷ{ℓ-j+1}(s)≤ℓ(1+q'z)^{ℓ-1}(q's)^{⌈j²/3ℓ⌉}.\label{ine:enum-S}\]
	\end{cor}
	
	Together, we have that with failing probability at most $6ℓq^{-√ℓ/13}$,
	both \cref{ine:enum-Z,ine:enum-S} hold for all $j∈[ℓ]$.
	In particular, we have that $D_ZＷj≥⌈j²/3ℓ⌉$ and $D_ZＷ{ℓ-j+1}≥⌈j²/3ℓ⌉$.
	
	Does the preceding result allow the possibility that $π→1$?
	That is to say, Does \cref{ine:enum-Z} imply that
	the average of $㏒_ℓD_ZＷj$ is $1-o(1)$ as $ℓ→∞$?
	The answer is
	\[÷1{ℓ}∑_{j=1}^ℓ㏒_ℓD_ZＷj≥÷1{ℓ}∑_{j=√{3ℓ}}^ℓ㏒_ℓ÷{j²}{3ℓ}
		≈÷2{ℓ㏑ℓ}∫_{√{3ℓ}}^ℓ㏑j\,\diff j-㏒_ℓ3ℓ≈1-÷{3.1}{㏑ℓ}-÷{√3}{√ℓ}.\]
	Now readers can see why I insist on loosening $h₂$ to
	the square root---I do not want to integrate $㏑(h₂^{-1}(p))$.
	
	The next section estimates the eigenvalue, $ℓ^{-ϱ}$, of random kernels.

\section{CLT Data of Random Matrix}\label{sec:CLT}

	This section is devoted to proving the eigen behavior of random kernels.
	To that end, we first need a concave eigenfunction $h：[0,1]→[0,1]$.
	Let $α≔㏑(㏑ℓ)/㏑ℓ$ and define
	\[h(z)≕\min(z,1-z)^α.\]
	From our experience in \cref{lem:eigen-Z}, the eigenvalue for reliable channels
	is a supremum taken over an open condition $0<H(W)<δ$, and hence
	requires explicit bounds coming from FTPCs and Hölder tolls.
	The next lemma does that and shows that $h$ has a good eigenvalue if $H(W)<ℓ^{-2}$.
	
	\begin{lem}[Typical CLT---reliable case]\label{lem:typical-HZ}
		Let $ℓ≥\max(3^q,e⁴,q⁵)$.
		Assume that $0<H(W)<ℓ^{-2}$ and that
		$𝔾∈𝔽_q^{ℓ×ℓ}$ satisfies \cref{ine:enum-Z}, then
		\[÷{h(H(WＷ1))+h(H(WＷ2))+\dotsb+h(H(WＷ{ℓ}))}{ℓh(H(W))}<2ℓ^{-1/2+3α}.
			\label{ine:good}\]
	\end{lem}
	
	\begin{proof}
		In the proof of \cref{lem:eigen-Z}, we classify $h(H(WＷj))$
		into two cases---those with coset distance treated as $1$
		and those with nontrivial ($≥5$) coset distance.
		Now that we have a spectrum of coset distances ($⌈j²/3ℓ⌉$),
		there will be as many cases as there are coset distances.
		
		Consider first $j≤k≔⌊ℓ^{1/2+5α/2}⌋$.
		This is the case when $𝘏₁$ goes up but how far
		it can go is limited by FTPC$H$ (\cref{thm:ftpcH}).
		More precisely,
		\begin{align*}
			\qquad&\kern-2em
			h(H(WＷ1))+h(H(WＷ2))+\dotsb+h(H(WＷk))≤kh（÷{H(WＷ1)+\dotsb+H(WＷk)}k）	\\
			&	≤kh（÷{ℓH(W)}k）≤kℓ^αh(H(W))k^{-α}≤ℓ^α(ℓ^{1/2+5α/2})^{1-α}h(H(W))	\\
			&	≤ℓ^{α+1/2+5α/2-α/2-5α²/2}h(H(W))≤ℓ^{1/2+3α}h(H(W)).
		\end{align*}
		That is to say, the terms $h(H(WＷj))$ with $j≤k≔⌊ℓ^{1/2+5α/2}⌋$ contribute
		$ℓ^{1/2+3α}$ to the right-hand side of \cref{ine:good}.
		We now have another $ℓ^{1/2+3α}$ to spare for larger $j$.
		
		Consider next $j≥k≔⌈ℓ^{1/2+5α/2}⌉$.
		This is the case where we can invoke FTPC$Z$ (\cref{thm:ftpcZ})
		to show that $Ｚ(WＷj)$ is significantly smaller than $Ｚ(W)$.
		So we just sandwich FTPC$Z$ by Hölder tolls to translate it into control on $H$.
		In greater detail, with $z≔Ｚ(W)$,
		\begin{align*}
			\qquad&\kern-2em
			h(H(WＷk))+H(H(WＷ{k+1}))+\dotsb+h(H(WＷ{ℓ}))≤ℓ\max_{j≥k}H(WＷj)^α	\\
			&	≤ℓ\max_{j≥k}q^{3α}Ｚ(WＷj)^{α/2}
				≤ℓq^{3α}\max_{j≥k}f_ZＷj(z)^{α/2}\\
			&	≤ℓq^{3α}\max_{j≥k}\(ℓ(1+(q-1)z)^{ℓ-1}((q-1)z)^{⌈j²/3ℓ⌉}\)^{α/2}	\\
			&	≤ℓq^{3α}\max_{j≥k}ℓ^{α/2}e^{qzℓα/2}(qz)^{j²α/6ℓ}
				≤ℓq^{3α}ℓ^{α/2}e^{qzℓα/2}(qz)^{k²α/6ℓ}	\\
			&	≤ℓq^{3α}ℓ^{α/2}e^{qzℓα/2}(qz)^{㏑(ℓ)⁵α/6}
				≤ℓq^{3α}ℓ^{α/2}e^{q⁴α/2}(qz)^{㏑(ℓ)⁵α/6}	\\
			&	≤ℓq^{3α}ℓ^{α/2}e^{q⁴α/2}(q⁸H(W))^{㏑(ℓ)⁵α/12}.	\label{for:raise}
		\end{align*}
		Here, the first inequality uses the largest $H(WＷj)$ to bound the rest.
		The next inequality pays the Hölder toll.
		The next inequality applies FTPC$Z$.
		The next inequality applies the assumption that \cref{ine:enum-Z} holds.
		The next inequality simplifies $q-1$ and $ℓ-1$ and the ceiling function.
		The next inequality knows that the maximum happens at $j=k$.
		The next inequality uses $k²/ℓ=ℓ^{1-5α}/ℓ=ℓ^{5㏑(㏑ℓ)/㏑ℓ}=e^{5㏑(㏑ℓ)}=㏑(ℓ)⁵$.
		The next inequality pays the Hölder toll for the return trip
		to simplify $z≤q³√{H(W)}≤q³√{H(W)}≤q³/ℓ$ in the exponent.
		The next inequality pays the Hölder toll for the other $z$ in the base.
		
		We are half way to the goal.
		To show that \cref{for:raise}${}≤ℓ^{1/2+3α}h(H(W))$,
		raise them to the power of $12/α$ and take the quotient:
		\begin{align*}
			\qquad&\kern-2em
			\(ℓq^{3α}ℓ^{α/2}e^{q⁴α/2}(q⁸H(W))^{㏑(ℓ)⁵α/12}\)^{12/α}
				\div\(ℓ^{1/2+3α}h(H(W))\)^{12/α}	\\
			&	=ℓ^{12/α}q^{36}ℓ⁶e^{6q⁴}(q⁸H(W))^{㏑(ℓ)⁵}ℓ^{-6/α-36}H(W)^{-12}	\\
			&	=e^{6q⁴}ℓ^{6㏑ℓ-30}H(W)^{㏑(ℓ)⁵-12}q^{8㏑(ℓ)⁵+36}
				<e^{6q⁴}ℓ^{6㏑ℓ-30}ℓ^{-2㏑(ℓ)⁵+24}q^{8㏑(ℓ)⁵+36}	\\
			&	<e^{6q⁴}ℓ^{6㏑ℓ-2㏑(ℓ)⁵}q^{8㏑(ℓ)⁵+36}
				=e^{6q⁴}ℓ^{6㏑ℓ-0.4㏑(ℓ)⁵}ℓ^{-1.6㏑(ℓ)⁵}q^{8㏑(ℓ)⁵+36}	\\
			&	≤e^{6q⁴}ℓ^{6㏑ℓ-0.4㏑(ℓ)⁵}q^{-8㏑(ℓ)⁵}q^{8㏑(ℓ)⁵+36}
				=e^{6q⁴}ℓ^{6㏑ℓ-0.4㏑(ℓ)⁵}q^{36}	\\
			&	=e^{6q⁴}ℓ^{-0.3㏑(ℓ)⁵}ℓ^{6㏑ℓ-0.1㏑(ℓ)⁵}q^{36}
				<e^{6q⁴}e^{-0.3㏑(41)^2(q㏑3)⁴}ℓ^{6㏑ℓ-0.1㏑(ℓ)⁵}q^{36}	\\
			&	<e^{6q⁴}e^{-6.02q⁴}ℓ^{6㏑ℓ-0.1㏑(ℓ)⁵}q^{36}
				<ℓ^{6㏑ℓ-0.1㏑(ℓ)⁵}q^{36}	\\
			&	=ℓ^{6㏑ℓ-㏑(ℓ)⁵/15}ℓ^{-㏑(ℓ)⁵/30}q^{36}
				=ℓ^{6㏑ℓ-㏑(ℓ)⁵/15}q^{-5㏑(19)⁵/30}q^{36}	\\
			&	<ℓ^{6㏑ℓ-㏑(ℓ)⁵/15}q^{-36.8}q^{36}
				<ℓ^{6㏑ℓ-㏑(ℓ)⁵/15}≤ℓ^{6㏑ℓ-㏑(22)⁴㏑(ℓ)/15}	\\
			&	<ℓ^{6㏑ℓ-6.08}<ℓ⁰≤1.
		\end{align*}
		The inequality involving $1.6$ uses $ℓ≥q⁵$.
		The inequality involving $0.3$ uses $ℓ≥\max(41,3^q)$.
		The inequality involving $/30$ uses $ℓ≥\max(19,q⁵)$.
		The inequality involving $/15$ uses $ℓ≥22$.
		I have just showed that \cref{for:raise}$/ℓ^{-1/2+3α}h(𝘏_0)$
		is less than $1$, with and hence without the power of $12/α$.
		
		To summarize, we saw that $h(H(WＷk))+H(H(WＷ{k+1}))+\dotsb+h(H(WＷ{ℓ}))≤{}
		$\cref{for:raise}${}≤ℓ^{-1/2+3α}h(𝘏_0)$.
		Hence the summands $h(H(WＷj))$ with $j≥k≔⌈ℓ^{1/2+5α/2}⌉$ contribute
		$ℓ^{1/2+3α}$ to the right-hand side of \cref{ine:good}.
		Both small $j$ case and large $j$ case together
		contribute $2ℓ^{1/2+3α}$, as desired.
		This is the end of \cref{lem:typical-HZ}.
	\end{proof}
	
	The following is automatic by duality given \cref{lem:typical-HZ}.
	
	\begin{lem}[Typical CLT---noisy case]\label{lem:typical-HS}
		Assume $1-ℓ^{-2}<H(W)<1$ and that $𝔾∈𝔽_q^{ℓ×ℓ}$ satisfies \cref{ine:enum-S}.
		Then
		\[÷{ℓh(H(WＷ1))+ℓh(H(WＷ2))+\dotsb+ℓh(H(WＷ{ℓ}))}{ℓh(H(W))}<2ℓ^{-1/2+3α}.\]
	\end{lem}
	
	We are left with the case when $ℓ^{-2}<H(W)<1-ℓ^{-2}$.
	In this zone, neither of FTPC$Z$ and FTPC$S$ can help because the Hölder tolls
	do not yield any meaningful estimate on $Ｚ(W)$ or $Ｓ(W)$ to begin with.
	The contribution I make here is to reduce the estimate of $H(WＷj)$, with large $j$,
	to noisy-channel coding and to reduce the large $j$ cases to wiretap-channel coding.
	
	\begin{lem}[Typical CLT---mediocre case]\label{lem:typical-HH}
		Assume $ℓ^{-2}≤H(W)≤1-ℓ^{-2}$ and that
		$𝔾∈\GL(ℓ,q)$ is drawn uniformly at random.
		Then, with exceptional probability $2ℓ^{-㏑(ℓ)/20}$,
		\[÷{h(H(WＷ1))+h(H(WＷ2))+\dotsb+h(H(WＷ{ℓ}))}{ℓh(H(W))}<4ℓ^{-1/2+3α}.\]
	\end{lem}
	
	\begin{proof}
		The desired inequality is the sum of the following three inequalities:
		\begin{gather*}
			∑_{j=⌈H(W)ℓ+ℓ^{1/2+α}⌉+1}^ℓh(H(WＷj))<ℓ^{1/2+α},	\label{ine:bob}\\
			∑_{j=⌊H(W)ℓ-ℓ^{1/2+α}⌋+1}^{⌈H(W)ℓ+ℓ^{1/2+α}⌉}h(H(WＷj))<2ℓ^{1/2+α},	\\
			∑_{j=1}^{⌊H(W)ℓ-ℓ^{1/2+α}⌋}h(H(WＷj))<ℓ^{1/2+α}.	\label{ine:eve}
		\end{gather*}
		Here $ℓ^{-2}≤H(W)≤1-ℓ^{-2}$ is used to rewrite the denominator $h(H(W))≥ℓ^{2α}$.
		The middle one is trivial because $h≤2^{-α}$.
		\Cref{ine:bob} will be proved in \cref{sec:gallager}.
		\Cref{ine:eve} will be proved in \cref{sec:hayashi}.
	\end{proof}
	
	Let me comment on the heuristic behind those inequalities:
	To show that $h(H(WＷj))$ is in general small, we first
	have to identify at which end each $H(WＷj)$ will be.
	We learned from the $\loll$ case that
	smaller indices usually imply noisier synthetic channels.
	So we believe that for $j≪ℓH(W)$, the conditional entropy $H(WＷj)$ is high,
	and for $j≫ℓH(W)$, the conditional entropy $H(WＷj)$ is low.
	We also believe that there is a ambiguous zone $j≈H(WＷj)$
	where the conditional entropy can be anywhere.
	From the CLT regime of random coding, we believe that the width
	of the ambiguous zone should be on the order of $√ℓ$, hence the partition.
	
	Once we have \cref{lem:typical-HZ,lem:typical-HS} (proved above) and
	\cref{lem:typical-HH} (part of whose proof is in the next two sections
	and no later), we can conclude the eigen behavior of a random kernel $𝔾$.
	
	\begin{thm}[Typical CLT behavior]\label{thm:typical-CLT}
		Assume $W$ is a $q$-ary channel.
		Assume $ℓ≥\max(3^q,e⁴,q⁵)$.
		Assume $𝔾∈\GL(ℓ,q)$ is drawn uniformly at random.
		Then, with failing probability at most $2ℓ^{-㏑(ℓ)/20}+6ℓq^{-√ℓ/13}$,
		\[÷{h(H(WＷ1))+h(H(WＷ2))+\dotsb+h(H(WＷ{ℓ}))}{ℓh(H(W))}<4ℓ^{-1/2+3α}.\]
	\end{thm}
	
	This is equivalent to saying that $4ℓ^{-1/2+3α}$ is the eigenvalue, or $ϱ≈1/2-4α$.
	This exponent approaches $1/2$ as $ℓ→∞$.
	Hence we know, at least, that there is a chance to achieve $ρ+2π→1$.
	How exactly \cref{thm:typical-LDP,thm:typical-CLT} imply that all $ρ+2π<1$
	are possible will be explained in \cref{sec:signify}.

\section{Symmetrization and Universal Bound} \label{sec:prepare}

	Before the actual proof of \cref{ine:bob,ine:eve},
	There are two more tools to be developed.
	First is a reduction borrowed from \cite{MT14} that says, instead of considering
	general $q$-ary channel $W$, it suffices to consider a symmetric one $˜W$.
	A convenient consequence is that
	the uniform input distribution will achieve capacity.
	This helps simplify the inequalities further.
	
	\begin{lem}[Symmetrization]
		For any $q$-ary channel, there is a symmetric $q$-ary channel $˜W$ such that
		$H(˜W)=H(W)$ and $H(˜WＷj)=H(WＷj)$ for all $j∈[ℓ]$.
	\end{lem}
	
	\begin{proof}
		The strategy here is close to what we did in \cref{cha:dual}---one can establish
		an equivalence relation $W≅˜W$ on channels and show that channel parameters
		and channel transformations respect the equivalence relation.
		
		Please see \cite[Definition~6 and Lemmas 7 and~8]{MT14}
		plus the arguments in between for the formal treatment.
	\end{proof}
	
	The second tool, built on top of the first one, is an inequality
	concerning Gallager's $\Eo$ function.
	Let us start from the definition.
	
	\begin{dfn}
		Define \emph{Gallager's E-null function}
		\[\Eo(t)≔-㏑∑_{y∈𝒴}（∑_{x∈𝒳}Q(x)W(y｜x)^÷1{1+t}）^{1+t}\]
		and its complement
		\[¯\Eo(t)≔㏑∑_{y∈𝒴}（∑_{x∈𝒳}W(x,y)^÷1{1+t}）^{1+t}.\]
	\end{dfn}
	
	$¯\Eo$ is said to be the complement of $\Eo$ as
	\begin{align*}
		¯\Eo(t)
		&	=㏑∑_{y∈𝒴}（∑_{x∈𝒳}(q^{-1}W(y｜x))^÷1{1+t}）^{1+t}	\\
		&	=t㏑q+㏑∑_{y∈𝒴}（∑_{x∈𝒳}q^{-1}W(y｜x)^÷1{1+t}）^{1+t}=t㏑q-\Eo(t).
	\end{align*}
	Or $\Eo(t)+¯\Eo(t)=t㏑q$ in short.
	That $Q$ is uniform is used, otherwise
	a non-constant cannot penetrate the summations.
	The E-null function and its complement are
	deeply connected to the following family of measures.
	
	\begin{dfn}
		For any $t∈[-2/5,1]$, define the \emph{$t$-tilted probability mass function}
		$\Wt：𝒳×𝒴→[0,1]$ as in \cite[Definition~1]{CS07}:
		\[\Wt(x,y)≔÷{\(∑_{ξ∈𝒳}W(ξ,y)^÷1{1+t}\)^{1+t}}
			{∑_{η∈𝒴}\(∑_{ξ∈𝒳}W(ξ,η)^÷1{1+t}\)^{1+t}}
			×÷{W(x,y)^÷1{1+t}}{∑_{ξ∈𝒳}W(ξ,y)^÷1{1+t}}\]
	\end{dfn}
	
	When $t=0$, the tilted $\Wt(x,y)$ falls back to its italic origin $W(x,y)$.
	These measures can be interpreted as follows:
	$\Wt$ behaves like a channel with a dedicated input distribution.
	The first fraction in the definition
	specifies the output distribution $\Wt(y)$.
	The second fraction specifies the posterior distribution
	$\Wt(x｜y)$ when $y$ is known.
	As $\Wt$ is not an actual channel, it is not meaningful to
	alter the input distribution and ask for the corresponding output.
	Like the symmetrization technique, all that matters is that
	we can compute the conditional entropies as if they were real channels.
	Quantities we are interested in are listed below.
	
	\begin{dfn}
		Let $H_e$ be the base-$e$ entropy.
		Let $H_e(\Wt)$ be $H_e(\Xt｜\Yt)$ where
		$(\Xt,\Yt)$ is a random tuple that obeys $\Wt$.
		Let $H_e(\Xt↾y)$ be the entropy of
		the posterior distribution of $\Xt$ given $\Yt=y$;
		to be specific, $H_e(\Xt↾y)=∑_{x∈X}\Wt(x｜y)㏑\Wt(x｜y)$.
	\end{dfn}
	
	Then $¯\Eo$ and $\Wt$ are connected in the following manner---$\Wt$
	is a family of channels that “lives along the path” $¯\Eo(t)$.
	
	\begin{lem}[Second derivative]\label{lem:dtdt}
		\cite[Formula (13) and~(19)]{CS07}
		For $t∈[0,1]$,
		\def\dt{÷\diff{\diff t}}
		\def\dtdt{÷{\diff²}{\diff t²}}
		\[\dt¯\Eo(t)=¯\Eo'(t)=H_e(\Wt)\]
		and
		\begin{align*}
			\dtdt¯\Eo(t)
			&	=¯\Eo''(t)=\dt H_e(\Wt)
				=÷1{1+t}∑_{y∈𝒴}\Wt(y)∑_{x∈𝒳}\Wt(x｜y)㏑(\Wt(x｜y))^2	\\*
			&\kern10em +÷t{1+t}∑_{y∈𝒴}\Wt(y)H_e(\Xt↾y)^2-H_e(\Wt)^2.
				\label{equ:holomorph}
		\end{align*}
	\end{lem}
	
	Since $¯\Eo''(t)$ and every other term in \cref{equ:holomorph}
	is holomorphic in $t$, the equation holds in any region that assumes no poles.
	In particular, $-2/5≤t≤1$ is such a region.
	In that region, the next lemma helps bounding the terms in $¯\Eo''(t)$.
	
	\begin{lem}[Second moment]\label{lem:moment}
		If $w₁,w₂,…w_q$ are positive numbers that totals to $1$, then
		\[∑_iw_i㏑(w_i)^2≤\begin{Bmatrix}
			㏑(q)^2	&	†for †q≥3 \\
			0.563		&	†for †q=2
		\end{Bmatrix}≤1.2㏑(q)^2.\]
	\end{lem}
	
	\Cref{lem:dtdt} (with the holomorphic continuation to $t=-2/5$) and
	\cref{lem:moment} jointly imply the following universal quadratic bound.
	
	\begin{lem}[Universal quadratic bound]\label{lem:quadratic}
		\cite[Theorem~2]{CS07}
		Let $W$ be a $q$-ary channel.
		Assume the uniform input distribution.
		Then Gallager's E-null function satisfies
		\begin{gather*}
			\Eo(0)=0,	\\
			\Eo'(0)=I(W)㏑q,	\\
			\Eo''(t)≥-2㏑(q)^2
		\end{gather*}
		for all $t∈[-2/5,1]$.
		In particular, it satisfies
		\[\Eo(t)≥I(W)t㏑q-t^2㏑(q)^2.\]
	\end{lem}
	
	\begin{proof}
		With \cref{lem:moment},
		$∑_{x∈𝒳}\Wt(x｜y)(㏑\Wt(x｜y))^2≤1.2㏑(q)^2$ can be stated.
		Now \cref{equ:holomorph} becomes
		\begin{align*}
			¯\Eo''(t)
			&	≤÷1{1+t}∑_{y∈𝒴}\Wt(y)·1.2㏑(q)^2
				+÷{\max(0,t)}{1+t}∑_{y∈𝒴}\Wt(y)㏑(q)^2	\\
			&	≤÷1{1+t}·1.2㏑(q)^2+÷{\max(0,t)}{1+t}㏑(q)^2≤2㏑(q)^2
		\end{align*}
		for all $t∈[-2/5,1]$.
		
		Since $\Eo(t)$ is a linear function $t㏑q$ minus $¯\Eo(t)$,
		their first derivatives sum to $㏑q$ while
		their second derivatives are opposite.
		That means that $\Eo(t)=\Eo(0)+\Eo(0)'t+\Eo''(τ)t²/2$, for some $τ∈[-2/5,1]$,
		and therefore $≥0+I(W)t㏑q-t²㏑(q)²$.
	\end{proof}
	
	After seeing that it suffices to consider symmetric channels and that $E(t)$ has
	a universal quadratic bound, we are now ready to prove \cref{ine:bob,ine:eve}.

\section{Noisy-Channel Random Coding}\label{sec:gallager}

	This section and the next take advantage of the universal bound developed
	three lines ago and continues proving \cref{thm:typical-CLT}.
	This section deals with
	\[∑_{i=⌈H(W)ℓ+ℓ^{1/2+α}⌉+1}^ℓh(H(WＷi))<ℓ^{-1/2+α}\tagcopy{ine:bob}\]
	by passing it to an estimate that captures the performance of noisy-channel coding.
	
	\begin{proof}[Proof of \cref{ine:bob}]
		Owing to $h$'s concavity, the left-hand side
		of \cref{ine:bob} is first simplified into
		\[∑_{i=j+1}^ℓh(H(WＷk))≤(ℓ-j)h（÷1{ℓ-j}∑_{i=j+1}^ℓH(WＷi)）,
			\label{ine:bob-chain}\]
		where $j≔⌈H(W)ℓ+ℓ^{1/2+α}⌉$ for short.
		It suffices to prove that the right-hand side is less than $ℓ^{-1/2+α}$.
		But what lies inside the $h$ on the right-hand side is a sum of $H(WＷi)$,
		which is equal to, by the chain rule, $H(U_{j+1}^ℓ｜U₁^jY₁^ℓ)$.
		In order to prove \cref{ine:bob-chain}, I will then show
		\[(ℓ-j)h（÷1{ℓ-j}H(U_{j+1}^ℓ｜Y₁^ℓU₁^j)）<ℓ^{-1/2+α}.\label{ine:bob-block}\]
		
		But what is $H(U_{j+1}^ℓ｜Y₁^ℓU₁^j)$?
		It measures the equivocation at Bob's end when $U₁^j$ is known to Bob.
		In other words, we may as well pretend that
		\begin{itemize}
			\item	there is a random rectangular full-rank matrix $𝔾'$ with
					$ℓ$ columns and only $k≔ℓ-j=⌊I(W)ℓ-ℓ^{1/2+α}⌋$ rows,
			\item	Alice computes and sends $X₁^ℓ≔U_{j+1}^ℓ𝔾'$ to Bob, and
			\item	Bob attempts to decode $ˆU_{j+1}^ℓ$ upon receiving $Y₁^ℓ$ using
					the maximum a posteriori decoder.
		\end{itemize}
		The equivocation is thus, by Fano's inequality,
		bounded in terms of the probability that Bob fails to decode $U_{j+1}^ℓ$:
		\begin{align*}
			H(U_{j+1}^ℓ｜Y₁^ℓU₁^j)
			&	≤-Ｐ㏑_qＰ-(1-Ｐ)㏑_q(1-Ｐ)+Ｐ㏑_q(q^k-1)	\\
			&	≤-Ｐ㏑_qＰ+÷{Ｐ}{㏑q}+Ｐ=Ｐ·（÷{1-㏑Ｐ}{㏑q}+k）.	\label{ine:Fano}
		\end{align*}
		Here $Ｐ$ is the probability that Bob fails to decode, $ˆU_{j+1}^ℓ≠U_{j+1}^ℓ$.
		
		In what follows is how to compute Bob's block error probability.
		The generator matrix $𝔾'$ used by Alice is selected uniformly
		from the ensemble of full-rank $k$-by-$ℓ$ matrices.
		The difference of every pair of codewords
		distributes uniformly on $𝔽_q^ℓ、\{0₁^ℓ\}$.
		Over symmetric channels, the difference alone determines the difficulty
		of decoding because $W^ℓ(y₁^ℓ｜ξ₁^ℓ+x₁^ℓ)=W^ℓ(σ₁^ℓ(y₁^ℓ)｜x₁^ℓ)$ for
		some component-wise involution $σ₁^ℓ$ on $𝒴^ℓ$ depending on $ξ₁^ℓ$.
		Therefore, Gallager's bound applies.
		To elaborate, let $t∈[0,1]$.
		Then Bob's average error probability satisfies
		\cite[Inequalities (5.6.2) to~(5.6.14)]{Gallager68}
		\begin{align*}
			\qquad&\kern-2em
			𝔼Ｐ=𝔼EP\{†Bob fails to decode †U_{j+1}^ℓ† given †𝔾',Y₁^ℓ\} \\
			&=	𝔼∑_{u₁^k}÷1{q^k}∑_{y₁^ℓ}W^ℓ(y₁^ℓ｜u₁^k𝔾')
				𝕀\{†Bob has †ˆU_{j+1}^ℓ≠u₁^k† given †𝔾',y₁^ℓ｜U_{j+1}^ℓ=u₁^k\}	\\
			&=	𝔼∑_{y₁^ℓ}W^ℓ(y₁^ℓ｜0₁^ℓ)
				𝕀\{†Bob has †ˆU_{j+1}^ℓ≠0₁^k† given †𝔾',y₁^ℓ｜U_{j+1}^ℓ=0₁^k\}	\\
			&≤	𝔼∑_{y₁^ℓ}W^ℓ(y₁^ℓ｜0₁^ℓ)\biggl(∑_{v₁^k≠0₁^k}
				𝕀\{†Bob prefers †v₁^k† over †0₁^k† given †𝔾',y₁^ℓ\}\biggr)^t \\
			&≤	𝔼∑_{y₁^ℓ}W^ℓ(y₁^ℓ｜0₁^ℓ)（∑_{v₁^k≠0₁^k}
				÷{W^ℓ(y₁^ℓ｜v₁^k𝔾')^÷1{1+t}}{W^ℓ(y₁^ℓ｜0₁^ℓ)^÷1{1+t}}）^t \\
			&=	𝔼∑_{y₁^ℓ}W^ℓ(y₁^ℓ｜0₁^ℓ)^÷1{1+t}
				（∑_{v₁^k≠0₁^k}W^ℓ(y₁^ℓ｜v₁^k𝔾')^÷1{1+t}）^t \\
			&≤	∑_{y₁^ℓ}W^ℓ(y₁^ℓ｜0₁^ℓ)^÷1{1+t}
				（𝔼∑_{v₁^k≠0₁^k}W^ℓ(y₁^ℓ｜v₁^k𝔾')^÷1{1+t}）^t \\
			&=	∑_{y₁^ℓ}W^ℓ(y₁^ℓ｜0₁^ℓ)^÷1{1+t}（∑_{x₁^ℓ≠0₁^ℓ}
				÷{q^k-1}{q^ℓ-1}W^ℓ(y₁^ℓ｜x₁^ℓ)^÷1{1+t}）^t \\
			&≤	q^{kt}∑_{y₁^ℓ}W^ℓ(y₁^ℓ｜0₁^ℓ)^÷1{1+t}
				（∑_{x₁^ℓ≠0₁^ℓ}÷1{q^ℓ}W^ℓ(y₁^ℓ｜x₁^ℓ)^÷1{1+t}）^t \\
			&≤	q^{kt}∑_{y₁^ℓ}W^ℓ(y₁^ℓ｜0₁^ℓ)^÷1{1+t}
				（∑_{x₁^ℓ}÷1{q^ℓ}W^ℓ(y₁^ℓ｜x₁^ℓ)^÷1{1+t}）^t \\
			&=	q^{kt}∑_{y₁^ℓ}（∑_{x₁^ℓ}÷1{q^ℓ}W^ℓ(y₁^ℓ｜x₁^ℓ)^÷1{1+t}）
				（∑_{x₁^ℓ}÷1{q^ℓ}W^ℓ(y₁^ℓ｜x₁^ℓ)^÷1{1+t}）^{t} \\
			&=	q^{kt}∑_{y₁^ℓ}（∑_{x₁^ℓ}÷1{q^ℓ}W^ℓ(y₁^ℓ｜x₁^ℓ)^÷1{1+t}）^{1+t} \\
			&=	q^{kt}∑_{y₁^ℓ}（∑_{x₁^ℓ}Q^ℓ(x₁^ℓ)W^ℓ(y₁^ℓ｜x₁^ℓ)^÷1{1+t}）^{1+t} \\
			&=	\exp(kt㏑q-(\emph{the E-null function of }W^ℓ)(t)) \\
			&=	\exp(kt㏑q-ℓ\Eo(t)).
		\end{align*}
		
		In summary, the average block error probability
		$𝔼Ｐ=𝔼EP\{$Bob fails to decode $U_{j+1}^ℓ$ given $𝔾'\}$
		is no more than $\exp(kt㏑q-ℓ\Eo(t))$ whenever $0≤t≤1$.
		Recall the universal quadratic bound developed in \cref{lem:quadratic}:
		$\Eo(t)≥I(W)t㏑q-t^2㏑(q)^2$.
		We obtain that the exponent is
		\begin{align*}
			kt㏑q-ℓ\Eo(t)
			&	≤(I(W)ℓ-ℓ^{1/2+α})t㏑q-ℓ\Eo(t) \\
			&	≤(I(W)ℓ-ℓ^{1/2+α})t㏑q-ℓ(I(W)t㏑q-t^2㏑(q)^2) \\
			&	=(ℓt㏑q-ℓ^{1/2+α})t㏑q \\
			\shortintertext{(redeem the inequality at $t=ℓ^{-1/2+α}/2㏑q$)}
			&↦	(ℓℓ^{-1/2+α}/2-ℓ^{1/2+α})ℓ^{-1/2+α}/2 \\
			&=	-ℓ^{2α}/4 = -ℓ^{2㏑(㏑ℓ)/㏑ℓ}/4 = -㏑(ℓ)^2/4.
		\end{align*}
		So far, the average error probability $𝔼Ｐ$ is shown to be
		less than $\exp(-㏑(ℓ)²/4)=ℓ^{-㏑(ℓ)/4}$.
		
		Run Markov's inequality with cutoff $ℓ^{-㏑(ℓ)/20}$.
		To put it another way, we sample a random full-rank matrix $𝔾'∈𝔽_q^{k×ℓ}$
		and reject it if $P\{$Bob fails to decode $U_{j+1}^ℓ$ given $𝔾'\}≥ℓ^{-㏑(ℓ)/5}$.
		Then the rejecting probability is $ℓ^{-㏑(ℓ)/20}$ because $1/20+1/5=1/4$.
		An upper bound on Bob's error probability being $Ｐ<ℓ^{-㏑(ℓ)/5}$,
		an upper bound on Bob's equivocation is
		\[H(U_{j+1}^ℓ｜Y₁^ℓU₁^j)≤ℓ^{-㏑(ℓ)/5}（÷{1-㏑ℓ^{-㏑(ℓ)/5}}{㏑q}+k）
			=ℓ^{-㏑(ℓ)/5}（÷{1+㏑(ℓ)^2/5}{㏑q}+k）\]
		by \cref{ine:Fano}.
		Plugging the right-hand side into $kh(†this place†/k)$,
		we derive that the left-hand side of \cref{ine:bob-block} is less than
		\begin{align*}
			\qquad&\kern-2em
			kh（÷{ℓ^{-㏑(ℓ)/5}}k（÷{1+㏑(ℓ)^2/5}{㏑q}+k））
				=k·（ℓ^{-㏑(ℓ)/5}（÷{1+㏑(ℓ)^2/5}{k㏑q}+1））^α	\\
			&	=ℓ^{-α㏑(ℓ)/5}k·（÷{1+㏑(ℓ)^2/5}{k㏑q}+1）^α
				<ℓ^{-α㏑(ℓ)/5}ℓ·（÷{1+㏑(ℓ)^2/5}{ℓ㏑q}+1）^α	\\
			&	<ℓ^{-α㏑(ℓ)/5}·ℓ·2^α = 2^αℓ㏑(ℓ)^{-㏑(ℓ)/5}.
		\end{align*}
		The first inequality uses that the left-hand side increases
		monotonically in $k$ and $k$ is $ℓ-j=⌊I(W)ℓ-ℓ^{1/2+α}⌋<ℓ$.
		The second inequality uses the assumption $ℓ≥2$.
		The quantity at the end of the chain of inequalities
		decays to $0$ as $ℓ→∞$, so eventually it becomes less than $ℓ^{1/2+α}$,
		the right-hand side of \cref{ine:bob-block}.
		This proves that \cref{ine:bob,ine:bob-chain} hold with
		failing probability $ℓ^{-㏑(ℓ)/20}$ as soon as $ℓ$ is large enough. 
		
		The lower bound on $ℓ$ in the statement of
		\cref{lem:typical-HH} is large enough ($ℓ>20$).
		Hence \cref{ine:bob}, the first half of \cref{lem:typical-HH}, is settled.
	\end{proof}
	
	That random kernels make $h(H(WＷj))$ small for large $j≫ℓH(W)$ is the first half;
	the next section settles the second half of \cref{lem:typical-HH},
	making $h(H(WＷj))$ small for small $j≪ℓH(W)$.

\section{Wiretap-Channel Random Coding}\label{sec:hayashi}

	This subsection contains the very last ingredient
	of the proof of \cref{lem:typical-HH} and \cref{thm:typical-CLT}.
	We dealt with \cref{ine:bob} in the last subsection.
	We now deal with
	\[∑_{i=1}^{⌊H(W)ℓ-ℓ^{1/2+α}⌋}h(H(WＷi))<ℓ^{1/2+α}.\tagcopy{ine:eve}\]
	by passing it to an estimate that captures
	the performance of wiretap-channel coding.
	
	\begin{proof}[Proof of \cref{ine:eve}]
		Similar to how we motivated \cref{ine:bob-block},
		we hereby apply Jensen's inequality and the chain rule of
		conditional entropy to simplify \cref{ine:eve}.
		The left-hand side becomes $jh(H(U_1^j｜Y_1^ℓ)/j)$
		where $j≔⌊H(W)ℓ-ℓ^{1/2+α}⌋$ for short.
		(This is not the same $j$ as in the last subsection.)
		The input being uniform, the argument of $h$ is
		$H(U_1^j｜Y_1^ℓ)/j=1-I(U_1^j；Y_1^ℓ)/j$, which can be replaced by
		$I(U_1^j｜Y_1^ℓ)/j$ thanks to the symmetry $h(1-z)=h(z)$.
		We will show
		\[jh（÷1jI(U_1^j；Y_1^ℓ)）<ℓ^{1/2+α}.\label{ine:eve-block}\]
		
		But what is $I(U_1^j；Y_1^ℓ)$?
		It is the amount of information Eve learns from
		wiretapping $Y_1^ℓ$ if Eve knows that $U_{j+1}^ℓ$ are junk.
		In other words, we may pretend that
		\begin{itemize}
			\item	Alice transmits $X_1^ℓ≔U_1^jV_{j+1}^ℓ𝔾$, wherein $U_1^j$ are
					the confidential bits and $V_{j+1}^ℓ$ are the obfuscating bits,
			\item	Bob receives $X_1^ℓ$ in full, and
			\item	Eve learns $Y_1^ℓ$.
		\end{itemize}
		This context falls back to (a special case of) the traditional setup
		of wiretap channels \cite{Wyner75} where various bounds are studied,
		some in terms of Gallager's E-null function.
		
		Here are some preliminaries
		to control the information leaked to Eve.
		We follow the blueprint of how Hayashi derived
		the secrecy exponent in \cite[Inequality~(21)]{Hayashi06}.
		Consider the communication protocol depicted in \cref{fig:alice}:
		Karl fixes a kernel $𝔾∈\GL(ℓ,q)$ and everyone knows $𝔾$.
		Alice chooses the confidential message $U_1^ℓ$.
		Vincent chooses the obfuscating bits $V_{j+1}^ℓ$.
		Charlie generates $Y_1^ℓ$ by
		plugging $X_1^ℓ≔U_1^jV_{j+1}^ℓ𝔾$ into a simulator of $W^ℓ$.
		Eve learns $Y_1^ℓ$ and is interested in knowing $U_1^j$ alone.
		So the channel on topic is the composition of Vincent and Charlie.
		Notation:
		Running out of symbols, we all use $ℙ$ with proper subscripts
		to indicate the corresponding probability measures.
		That said, indices in the subscript will be omitted.
		As Eve is interested in the relation between $U_1^j$ and $Y_1^ℓ$,
		let $Y_1^ℓ↾Gu_1^j$ be the r.v.\ that follows
		the posterior distribution of $Y_1^ℓ$ given $𝔾=G$ and $U_1^j=u_1^j$.
		More formally, $ℙ_{Y↾Gu}(y_1^ℓ)
			=ℙ_{Y|𝔾U}(y_1^ℓ｜G,u_1^ℓ)=ℙ_{𝔾  UY}(G,u_1^j,y_1^ℓ)/ℙ_{𝔾U}(G,u_1^j)$.
		We could have defined $Y_1^ℓ↾G$ to be the
		posterior distribution of $Y_1^ℓ$ given $𝔾=G$;
		but it is simply the same distribution as $Y_1^ℓ$ since $U_1^jV_{j+1}^ℓG$
		traverses all inputs uniformly regardless of the choice of $G$.
		That is, $ℙ_{Y|𝔾}(y_1^ℓ｜G)=ℙ_Y(y_1^ℓ)$ for all $y_1^ℓ∈𝒴^ℓ$.
		
		\begin{figure}
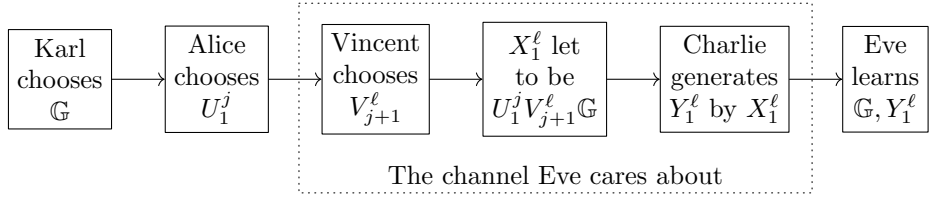

			\tikz{
				\draw[nodes={right,draw,align=center}]
					(.7,0)coordinate(X)
					            node(K){Karl\\chooses\\$𝔾$}
					(K.east)+(X)node(A){Alice\\chooses\\$U_1^j$}
					(A.east)+(X)node(V){Vincent\\chooses\\$V_{j+1}^ℓ$}
					(V.east)+(X)node(N)
								{$X_1^ℓ$ let\\to be\\$U_1^jV_{j+1}^ℓ𝔾$}
					(N.east)+(X)node(C){Charlie\\generates\\$Y_1^ℓ$ by $X_1^ℓ$}
					(C.east)+(X)node(E){Eve\\learns\\$𝔾,Y_1^ℓ$}
				;
				\graph[use existing nodes]{
					K -> A -> V -> N -> C -> E
				};
				\draw[dotted]
					(V.west|-C.north)+(-.3,.3)-|($(C.south east)+(.3,-.8)$)-|
					node[pos=.25,above]{The channel Eve cares about}cycle;
			}
			\caption{
				A finer setup for Hayashi's secrecy exponent.
				Charlie generates $Y_1^ℓ$ such that
				$X_1^ℓ≔U_1^jV_{j+1}^ℓ𝔾$ and $Y_1^ℓ$ follow $W^ℓ$.
				Despite of the seemingly sequential structure,
				Karl, Alice, and Vincent work independently.
			}\label{fig:alice}
		\end{figure}
		
		Fix $G$ as an instance of $𝔾$.
		Let $I_e$ be the base-$e$ mutual information.
		The channel Eve cares about leaks information of this amount:
		\begin{align*}
			\qquad&\kern-2em
			I_e(U_1^j；Y_1^ℓ｜G)=∑_{u_1^jy_1^ℓ}ℙ_{UY|𝔾}(u_1^j,y_1^ℓ｜G)
				㏑÷{ℙ_{Y|𝔾U}(y_1^ℓ｜G,u_1^j)}{ℙ_{Y|𝔾}(y_1^ℓ｜G)}	\\
			&	=∑_{u_1^j}ℙ_U(u_1^j)∑_{y_1^ℓ}ℙ_{Y|𝔾U}(y_1^ℓ｜G,u_1^j)
				㏑÷{ℙ_{Y|𝔾U}(y_1^ℓ｜G,u_1^j)}{ℙ_{Y|𝔾}(y_1^ℓ｜G)}	\\
			&	=∑_{u_1^j}ℙ_U(u_1^j)∑_{y_1^ℓ}ℙ_{Y↾Gu}(y_1^ℓ)
				㏑÷{ℙ_{Y↾Gu}(y_1^ℓ)}{ℙ_Y(y_1^ℓ)}
				=∑_{u_1^j}ℙ_U(u_1^j)𝔻(Y_1^ℓ↾Gu_1^j\|Y_1).	\label{for:fix-kld}
		\end{align*}
		$𝔻(Y_1^ℓ↾Gu_1^j\|Y_1^ℓ)$ is the Kullback--Leibler divergence
		from the posterior distribution of $Y_1^ℓ$ given $G,u_1^j$
		to the coarsest distribution $Y_1^ℓ$.
		We are to take expectation over $𝔾$ to find the average information leak
		since we are interested in Markov's inequality.
		\Cref{for:fix-kld} gives rise to
		\begin{align*}
			𝔼I_e(U_1^j；Y_1^ℓ｜𝔾)
			&	=∑_Gℙ_𝔾(G)I_e(U_1^j；Y_1^ℓ｜G)	\\
			&	=∑_Gℙ_𝔾(G)∑_{u_1^j}ℙ_U(u_1^j)𝔻(Y_1^ℓ↾Gu_1^j\|Y_1^ℓ).
				\label{for:exp-kld}
		\end{align*}
		We now discover that there are redundancies
		in traversing all $G$ and $u_1^ℓ$:
		After all, $X_1^j$ is $u_1^jV_{j+1}^ℓG=u_1^j0_{j+1}^ℓG+0_1^jV_{j+1}^ℓG$,
		which is a fixed linear combination of the first $j$ rows
		plus a random vector from the span of the bottom $ℓ-j$ rows.
		When $V_1^ℓ$ varies, the track of $X_1^ℓ$ forms an affine subspace of
		$𝔽_q^ℓ$, a \emph{coset code} as in the context of the fundamental theorems.
		So what matters is the distribution of this coset code.
		
		In the aforementioned manner, we replace the uniform ensemble of $(𝔾,U_1^j)$
		by the uniform ensemble of $𝕂$, a rank-$(ℓ-j)$ affine subspace of $𝔽_q^ℓ$,
		where $j≔⌊H(W)ℓ-ℓ^{1/2+α}⌋$.
		Karl and Alice together choose $𝕂$ uniformly.
		Vincent chooses $X_1^ℓ∈𝕂$ uniformly.
		Charlie generates $Y_1^ℓ$ by entering $X_1^ℓ$ into a simulator of $W^ℓ$.
		See \cref{fig:karl} for the depiction of the new scheme.
		Hence \cref{for:exp-kld} induces
		\[𝔼I_e(U_1^j；Y_1^ℓ｜𝔾)=∑_Kℙ_𝕂(K)𝔻(Y_1^ℓ↾K\|Y_1^ℓ)\]
		where $Y_1^ℓ↾K$ is the a posteriori distribution of $Y_1^ℓ$ given $𝕂=K$.
		Suddenly, the quantity $𝔼I_e(U_1^j；Y_1^ℓ｜𝔾)$ we are interested in
		turns into the mutual information $I_e(𝕂；Y_1^ℓ)$ between $𝕂$ and $Y_1^ℓ$
		as $𝕂$ replaces the role of $U_1^j$ in \cref{for:fix-kld}.
		Recall that in \cref{lem:quadratic} the mutual information
		is the derivative of Gallager's E-null function.
		We exploit this.
		Define the double-stroke E-null function for $(𝕂,Y_1^ℓ)$ as follows
		\[𝔼_0(t)≔-㏑∑_{y_1^ℓ}（∑_Kℙ_𝕂(K)ℙ_{Y|𝕂}(y_1^ℓ｜K)^÷1{1+t}）^{1+t}.\]
		Then $𝔼_0'(0)=I_e(𝕂；Y_1^ℓ)=𝔼I_e(U_1^j；Y_1^ℓ｜𝔾)$.
		Owing to the concavity of the E-null function,
		$𝔼_0'(0)≤𝔼_0(t)/t$ whenever $-2/5≤t<0$.
		Recap:
		To bound the average leaked information $𝔼I_e(U_1^j；Y_1^ℓ｜𝔾)$
		it suffices to bound $I_e(𝕂；Y_1^ℓ)$, which is then morphed
		to bound $𝔼_0'(0)$ from above and to bound $𝔼_0(t)$ from below.
		
		\begin{figure}
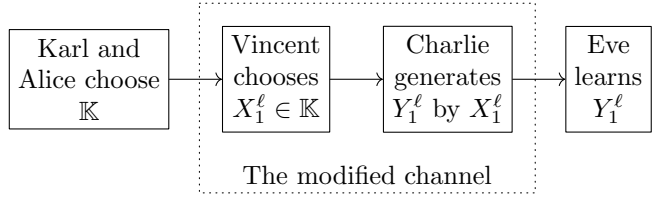

			\tikz{
				\draw[nodes={right,draw,align=center}]
					(.7,0)coordinate(X)
					            node(K){Karl and \\Alice choose\\$𝕂$}
					(K.east)+(X)node(V){Vincent\\chooses\\$X_1^ℓ∈𝕂$}
					(V.east)+(X)node(C){Charlie\\generates\\$Y_1^ℓ$ by $X_1^ℓ$}
					(C.east)+(X)node(E){Eve\\learns\\$Y_1^ℓ$};
				\graph[use existing nodes]{
					K -> V -> C -> E
				};
				\draw[dotted]
					(V.west|-C.north)+(-.3,.3)-|($(C.south east)+(.3,-.8)$)-|
					node[pos=.25,above]{The modified channel}cycle;
			}
			\caption{
				A simplified setup for Hayashi's secrecy exponent.
				Charlie generates $Y_1^ℓ$ such that $X_1^ℓ$ and $Y_1^ℓ$ follow $W^ℓ$.
			}\label{fig:karl}
		\end{figure}
		
		The double-stroke E-null function is bounded as below.
		Assume $-2/5≤t<0$.
		Let $s$ be $-t/(1+t)$; so $0<s≤2/3$ and $(1+s)(1+t)=1$.
		For any fixed $K$ and fixed $x_1^ℓ∈K$, the base of the $(1+t)$-th root
		in the definition of the double-stroke E-null function is
		\begin{align*}
			ℙ_{Y|𝕂}(y_1^ℓ｜K) 
			&	=∑_{ξ_1^ℓ∈K}ℙ_{X|𝕂}(ξ_1^ℓ｜K)ℙ_{Y|X}(y_1^ℓ｜ξ_1^ℓ)	\\
			&	=∑_{ξ_1^ℓ∈K}q^jℙ_X(ξ_1^ℓ)ℙ_{Y|X}(y_1^ℓ｜ξ_1^ℓ)
				=∑_{ξ_1^ℓ∈K}q^jℙ_{XY}(ξ_1^ℓ,y_1^ℓ)	\\
			&	=q^j（ℙ_{XY}(x_1^ℓ,y_1^ℓ)+∑_{x_1^ℓ≠ξ_1^ℓ∈K}ℙ_{XY}(ξ_1^ℓ,y_1^ℓ)）	\\
			&	=q^j（ℙ_{XY}(x_1^ℓ,y_1^ℓ)+ℙ_{XY}(K＼x_1^ℓ,y_1^ℓ)）.
		\end{align*}
		Here $ℙ_{XY}(K＼x_1^ℓ,y_1^ℓ)$ is a temporary shorthand for the summation
		of $ℙ_{XY}(ξ_1^ℓ,y_1^ℓ)$ over $ξ_1^ℓ∈K$ that excludes $x_1^ℓ$.
		Raise $ℙ_{Y|𝕂}(y_1^ℓ｜K)$ to the power of $s$;
		it becomes $q^{js}(ℙ_{XY}(x_1^ℓ,y_1^ℓ)+ℙ_{XY}(K＼x_1^ℓ,y_1^ℓ))^s
			≤q^{js}(ℙ_{XY}(x_1^ℓ,y_1^ℓ)^s+ℙ_{XY}(K＼x_1^ℓ,y_1^ℓ)^s)$ by sub-additivity.
		Put that aside; raise $ℙ_{Y|𝕂}(y_1^ℓ｜K)$ to the power of $1+s=1/(1+t)$:
		\begin{align*}
			\qquad&\kern-2em
			ℙ_{Y|𝕂}(y_1^ℓ｜K)^{1+s} = ℙ_{Y|𝕂}(y_1^ℓ｜K)ℙ_{Y|𝕂}(y_1^ℓ｜K)^s
				=∑_{x_1^ℓ∈K}q^jℙ_{XY}(x_1^ℓ,y_1^ℓ)ℙ_{Y|𝕂}(y_1^ℓ｜K)^s	\\
			&	≤∑_{x_1^ℓ∈K}q^jℙ_{XY}(x_1^ℓ,y_1^ℓ)
				q^{js}（ℙ_{XY}(x_1^ℓ,y_1^ℓ)^s+ℙ_{XY}(K＼x_1^ℓ,y_1^ℓ)^s）	\\
			&	=q^{j+js}（∑_{x_1^ℓ∈K}ℙ_{XY}(x_1^ℓ,y_1^ℓ)^{1+s}
				+∑_{x_1^ℓ∈K}ℙ_{XY}(x_1^ℓ,y_1^ℓ)ℙ_{XY}(K＼x_1^ℓ,y_1^ℓ)^s）.
		\end{align*}
		The inequality rewrites the $s$th power of $ℙ_{Y|𝕂}(y_1^ℓ｜K)$.
		Then the inner sum of the E-null function morphs as follows
		\begin{align*}
			∑_Kℙ_𝕂(K)ℙ_{Y|𝕂}(y_1^ℓ｜K)^{1+s}
			&	≤∑_Kℙ_𝕂(K)q^{j+js}（∑_{x_1^ℓ∈K}ℙ_{XY}(x_1^ℓ,y_1^ℓ)^{1+s}	\\*
			&\kern8em	+∑_{x_1^ℓ∈K}ℙ_{XY}(x_1^ℓ,y_1^ℓ)ℙ_{XY}(K＼x_1^ℓ,y_1^ℓ)^s）	\\
			&	=q^{j+js}∑_Kℙ_𝕂(K)∑_{x_1^ℓ∈K}
				ℙ_{XY}(x_1^ℓ,y_1^ℓ)^{1+s}	\tag{diagonal arc}\\
			&	+q^{j+js}∑_Kℙ_𝕂(K)∑_{x_1^ℓ∈K}
				ℙ_{XY}(x_1^ℓ,y_1^ℓ)ℙ_{XY}(K＼x_1^ℓ,y_1^ℓ)^s.	\tag{off-diagonal arc}
		\end{align*}
		The inequality rewrites the $(s+1)$th power of $ℙ_{Y|𝕂}(y_1^ℓ｜K)$.
		Divide and conquer---the inner sum of
		the double-stroke E-null function is split into two arcs as labeled.
		
		The diagonal arc is exactly
		\begin{align*}
			\qquad&\kern-2em
			q^{j+js}∑_Kℙ_𝕂(K)∑_{x_1^ℓ∈K}ℙ_{XY}(x_1^ℓ,y_1^ℓ)^{1+s}
				=q^{j+js}÷1{q^j}∑_{x_1^ℓ∈𝔽_q^ℓ}ℙ_{XY}(x_1^ℓ,y_1^ℓ)^{1+s}	\\
			&	=q^{js}\!∑_{x_1^ℓ∈𝔽_q^ℓ}ℙ_X(x_1^ℓ)^{1+s}ℙ_{Y|X}(y_1^ℓ｜x_1^ℓ)^{1+s}
				=q^{js-ℓs}\!∑_{x_1^ℓ∈𝔽_q^ℓ}ℙ_X(x_1^ℓ)ℙ_{Y|X}(y_1^ℓ｜x_1^ℓ)^{1+s}.
		\end{align*}
		The off-diagonal arc is
		\begin{align*}
			&	q^{j+js}∑_Kℙ_𝕂(K)∑_{x_1^ℓ∈K}
				ℙ_{XY}(x_1^ℓ,y_1^ℓ)ℙ_{XY}(K＼x_1^ℓ,y_1^ℓ)^s	\\
			&\quad	=q^{j+js}∑_{x_1^ℓ∈𝔽_q^ℓ}ℙ_{XY}(x_1^ℓ,y_1^ℓ)
				∑_{K∋x_1^ℓ}ℙ_𝕂(K)ℙ_{XY}(K＼x_1^ℓ,y_1^ℓ)^s.
		\end{align*}
		The inner sum is loosened to
		\begin{align*}
			\qquad&\kern-2em
			∑_{K∋x_1^ℓ}ℙ_𝕂(K)ℙ_{XY}(K＼x_1^ℓ,y_1^ℓ)^s
				=÷1{q^j}∑_{K∋x_1^ℓ}ℙ_{𝕂|X}(K｜x_1^ℓ)ℙ_{XY}(K＼x_1^ℓ,y_1^ℓ)^s	\\
			&	≤÷1{q^j}（∑_{K∋x_1^ℓ}ℙ_{𝕂|X}(K｜x_1^ℓ)ℙ_{XY}(K＼x_1^ℓ,y_1^ℓ)）^s	\\
			&	=÷1{q^j}（∑_{K∋x_1^ℓ}ℙ_{𝕂|X}(K｜x_1^ℓ)
				∑_{x_1^ℓ≠ξ_1^ℓ∈K}ℙ_{XY}(ξ_1^ℓ,y_1^ℓ)）^s	\\
			&	=÷1{q^j}（÷{q^{ℓ-j}-1}{q^ℓ-1}∑_{x_1^ℓ≠ξ_1^ℓ∈𝔽_q^ℓ}ℙ_{XY}(ξ_1^ℓ,y_1^ℓ)）^s
				≤÷1{q^{j+js}}（∑_{x_1^ℓ≠ξ_1^ℓ∈𝔽_q^ℓ}ℙ_{XY}(ξ_1^ℓ,y_1^ℓ)）^s
		\end{align*}
		The last equality counts the multiplicity of $ξ_1^ℓ$. 
		So the off-diagonal arc is loosened to
		\begin{align*}
			\qquad&\kern-2em
			q^{j+js}∑_{x_1^ℓ∈𝔽_q^ℓ}ℙ_{XY}(x_1^ℓ,y_1^ℓ)
				∑_{K∋x_1^ℓ}ℙ_𝕂(K)ℙ_{XY}(K＼x_1^ℓ,y_1^ℓ)^s	\\
			&	≤∑_{x_1^ℓ∈𝔽_q^ℓ}ℙ_{XY}(x_1^ℓ,y_1^ℓ)
				（∑_{x_1^ℓ≠ξ_1^ℓ∈𝔽_q^ℓ}ℙ_{XY}(ξ_1^ℓ,y_1^ℓ)）^s	\\
			&	≤∑_{x_1^ℓ∈𝔽_q^ℓ}ℙ_{XY}(x_1^ℓ,y_1^ℓ)
				（∑_{ξ_1^ℓ∈𝔽_q^ℓ}ℙ_{XY}(ξ_1^ℓ,y_1^ℓ)）^s	\\
			&	=∑_{x_1^ℓ∈𝔽_q^ℓ}ℙ_{XY}(x_1^ℓ,y_1^ℓ)ℙ_Y(y_1^ℓ)^s
				=ℙ_Y(y_1^ℓ)ℙ_Y(y_1^ℓ)^s = ℙ_Y(y_1^ℓ)^{1+s}.
		\end{align*}
		
		Both the diagonal and off-diagonal arcs being conquered,
		merge them and raise to the $(1+t)$-th power.
		The summand for any fixed $y_1^ℓ$ in
		the definition of the double-stroke E-null function is
		\begin{align*}
			\quad&\kern-2em
			（∑_Kℙ_𝕂(K)ℙ_{Y|𝕂}(y_1^ℓ｜K)^÷1{1+t}）^{1+t}
				=(†off-diagonal†+†diagonal†)^{1+t}	\\
			&	≤†off-diagonal†^{1+t}+†diagonal†^{1+t}
				≤（ℙ_Y(y_1^ℓ)^{1+s}）^{1+t}+†diagonal†^{1+t}	\\
			&	=ℙ_Y(y_1^ℓ)+†diagonal†^{1+t}
				≤ℙ_Y(y_1^ℓ)+（q^{js-ℓs}∑_{x_1^ℓ∈𝔽_q^ℓ}ℙ_X(x_1^ℓ)
				ℙ_{Y|X}(y_1^ℓ｜x_1^ℓ)^{1+s}）^{1+t}	\\
			&	=ℙ_Y(y_1^ℓ)+q^{ℓt-jt}（∑_{x_1^ℓ∈𝔽_q^ℓ}ℙ_X(x_1^ℓ)
				ℙ_{Y|X}(y_1^ℓ｜x_1^ℓ)^{1+s}）^{1+t}
		\end{align*}
		The first equality divides.
		The next inequality applies the sub-additivity
		of $(1+t)$th power (note that $t<0$).
		We can finally bound the double-stroke E-null function per se:
		\begin{align*}
			\exp(-𝔼_0(t))
			&	=∑_{y_1^ℓ}（∑_Kℙ_𝕂(K)ℙ_{Y|𝕂}(y_1^ℓ｜K)^÷1{1+t}）^{1+t}	\\
			&	≤∑_{y_1^ℓ}ℙ_Y(y_1^ℓ)+q^{ℓt-jt}（∑_{x_1^ℓ∈𝔽_q^ℓ}
				ℙ_X(x_1^ℓ)ℙ_{Y|X}(y_1^ℓ｜x_1^ℓ)^{1+s}）^{1+t}	\\
			&	=1+q^{ℓt-jt}∑_{y_1^ℓ}（∑_{x_1^ℓ∈𝔽_q^ℓ}
				ℙ_X(x_1^ℓ)ℙ_{Y|X}(y_1^ℓ｜x_1^ℓ)^{1+s}）^{1+t}	\\
			&	=1+q^{ℓt-jt}\exp(-(\emph{the E-null function of }W^ℓ)(t))	\\
			&	=1+q^{ℓt-jt}\exp(-ℓ\Eo(t)).
		\end{align*}
		
		All efforts we spent on bounding $I_e(U_1^j；Y_1^ℓ)$ are for three creeds:
		First, it demonstrates that Gallager's bounds via E-null functions
		(which behaves like cumulant generating functions)
		is a powerful tool that can be useful to the dual case.
		Second, it fits the paradigm that solving the primary (noisy channel)
		and the dual (wiretap channel) problems as a whole
		is easier than solving the primary problem alone.
		Third, the universal quadratic bound can be used
		to further bound the E-null function.
		
		We infer that
		\begin{align*}
			𝔼I_e(U_1^j；Y_1^ℓ｜𝔾)
			&	=I_e(𝕂；Y_1^ℓ)=𝔼_0'(0)≤÷1t𝔼_0(t) = ÷1{-t}㏑（\exp(-𝔼_0(t))）	\\
			&	≤÷1{-t}㏑（1+q^{ℓt-jt}\exp(-ℓ\Eo(t))）
				<÷1{-t}q^{ℓt-jt}\exp(-ℓ\Eo(t))	\\
			&	=\exp(-㏑(-t)+(ℓ-j)t㏑q-ℓ\Eo(t)).
		\end{align*}
		Recall the universal quadratic bound $\Eo(t)≥I(W)t㏑q-t^2㏑(q)^2$
		as stated in \cref{lem:quadratic} and used in the previous subsection.
		But this time $-2/5≤t<0$.
		We obtain that the exponent is
		\begin{align*}
			\qquad&\kern-2em
			-㏑(-t)+(ℓ-j)t㏑q-ℓ\Eo(t)	\\
			&	=-㏑(-t)+(ℓ-H(W)ℓ+ℓ^{1/2+α})t㏑q-ℓ\Eo(t)	\\
			&	=-㏑(-t)+(I(W)ℓ+ℓ^{1/2+α})t㏑q-ℓ\Eo(t)	\\
			&	≤-㏑(-t)+(I(W)ℓ+ℓ^{÷12+α})t㏑q-ℓ(I(W)t㏑q-t^2㏑(q)^2)	\\
			&	=-㏑(-t)+(ℓt㏑q+ℓ^{1/2+α})t㏑q	\\
			\shortintertext{(redeem the inequality at $t=-ℓ^{-1/2+α}/2㏑q$)}
			&	↦-㏑（÷{ℓ^{-1/2+α}}{2㏑q}）
				-（-÷{ℓℓ^{-1/2+α}}2+ℓ^{1/2+α}）÷{ℓ^{-1/2+α}}2	\\
			&	=÷{㏑ℓ}2-α㏑ℓ+㏑2+㏑(㏑q)-÷{ℓ^{2α}}4	\\
			&	=÷{㏑ℓ}2-㏑(㏑ℓ)+㏑2+㏑(㏑q)-÷{ℓ^{2㏑(㏑ℓ)/㏑ℓ}}4	\\
			&	<÷{㏑ℓ}2+㏑(㏑q)-÷{㏑(ℓ)²}4.
		\end{align*}
		The first inequality uses $ℓ-j=ℓ-H(W)ℓ+ℓ^{1/2+α}$.
		The last inequality uses the assumption $ℓ>e^2$.
		With the last line we conclude that $𝔼I_e(U_1^j；Y_1^ℓ｜𝔾)
			<\exp(㏑(ℓ)/2+㏑㏑q-㏑(ℓ)^2/4)=ℓ^{1/2-㏑(ℓ)/4}㏑q$.
		Switch back to the base-$q$ mutual information
		$𝔼I(U_1^j；Y_1^ℓ｜𝔾)<ℓ^{1/2-㏑(ℓ)/4}$.
		
		We now reject kernels $𝔾$ such that $I(U_1^j；Y_1^ℓ｜𝔾)≥ℓ^{1/2-㏑(ℓ)/5}$.
		By Markov's inequality, the opposite direction ($<$) holds
		with probability $1-ℓ^{-㏑(ℓ)/20}$ because $1/5+1/20=1/4$.
		Plug this upper bound into $h$.
		The left-hand side of \cref{ine:eve-block} is less than
		\begin{align*}
			jh（÷1jℓ^{1/2-㏑(ℓ)/5}）
			&	=jj^{-α}ℓ^{α/2-α㏑(ℓ)/5}<ℓ^{1-α}ℓ^{α/2-α㏑(ℓ)/5}	\\
			&	=ℓ^{1-α/2-α㏑(ℓ)/5}=ℓ㏑(ℓ)^{-1/2-㏑(ℓ)/5}.
		\end{align*}
		The inequality uses that the left-hand side
		increases monotonically in $j$ and $j≔H(W)ℓ-ℓ^{1/2+α}<ℓ$.
		The quantity at the end of the inequalities decays to $0$
		as $ℓ→∞$, so eventually it becomes less than $ℓ^{1/2+α}$,
		the right-hand side of \cref{ine:eve-block}.
		This proves that \cref{ine:eve} holds with
		failing probability $ℓ^{-㏑(ℓ)/20}$ as soon as $ℓ$ is large enough. 
		
		The lower bound on $ℓ$ in the statement of \cref{lem:typical-HH} is large
		enough, hence \cref{ine:eve}, the second half of \cref{lem:typical-HH} settled.
		That means the proof of the whole \cref{thm:typical-CLT} is complete.
	\end{proof}
	
	We just finished the last piece of the proof of \cref{thm:typical-CLT},
	which states that random kernels possesses good $ϱ$ with high probability.
	The next section combines this fact with the coset distance profile $⌈j²/ℓ⌉$
	to conclude that low-complexity nearly-optimal polar codes exist when $π+2ρ→1$.

\section{Chapter Conclusion}\label{sec:signify}

	\Cref{thm:typical-LDP} shows that, with high probability,
	a random kernel enjoys coset distance profile $D_ZＷj≥⌈j²/ℓ⌉$.
	And its dual $D_SＷ{ℓ-j+1}≥⌈j²/ℓ⌉$ is immediate.
	\Cref{thm:typical-CLT} shows that, with high probability, a random kernel enjoys
	eigenvalue $ℓ^{-1/2+4α}$, which means $ρ=1/2-4α$, where $α≔㏑ℓ/㏑(㏑ℓ)$.
	Now, Does any pair $(π,ρ)$ lying under the line $π+2ρ=1$
	lie to the left of the convex envelope of
	$(0,1/2-4α)$ and the Cramér function of $㏒_ℓ⌈𝘑₁²/ℓ⌉$ for some large $ℓ$?
	
	Let us not go all the way down to a derivation
	of the Cramér function of $㏒_ℓ⌈𝘑₁²/ℓ⌉$;
	a one-sided bound suffices.
	Consider the moment generating function evaluated at slope $-1/2$
	(that is the slope of $π+2ρ=1$):
	\begin{align*}
		\qquad&\kern-2em
		𝘌[𝘋_1^{-1/2}]
			=÷1{ℓ}∑_{j=1}^ℓ（\Bigl\lceil÷{j²}{3ℓ}\Bigr\rceil）^{-1/2}
			<÷1{ℓ}∑_{j=1}^{⌊√{3ℓ}⌋}1^{-1/2}
			+÷1{ℓ}∑_{j=⌊√{3ℓ}⌋+1}^ℓ（÷{j²}{3ℓ}）^{-1/2}	\\
		&	<÷1{ℓ}√{3ℓ}+÷1{ℓ}√{3ℓ}∫_{√{3ℓ}}^ℓ÷{\diff j}j
			=√{3ℓ}+÷{√3}{√ℓ}㏑j\Bigr\rvert_{√{3ℓ}}^ℓ<√{3ℓ}+√{3ℓ}㏑ℓ	\\
		&	=ℓ^{-1/2}+2ℓ^{-1/2+α}<4ℓ^{-1/2+α}<ℓ^{-1/2+2α}.
	\end{align*}
	So the cumulant generating function is bounded as
	\[𝘒(-1/2)=㏒_ℓ𝘌[𝘋_1^{-1/2}]<-1/2+2α.\]
	The Cramér function as a supremum is then bounded by
	\[𝘓(s)≥s·(-1/2)-𝘒(-1/2)≥-s/2+1/2-2α.\]
	
	The convex envelope of $(0,1/2-4α)$ and the segment $ρ=-π/2+1/2-2α$
	(note that only the part with $ρ≥0$ counts)
	is a straight line connecting $(0,1/2-4α)$ and $(1-4α,0)$.
	As $ℓ$ goes to infinity, $α$ goes to $0$, hence the convex envelope
	reveals the right triangle $(0,1/2)$--$(0,0)$--$(1,0)$.
	
	By \cref{cha:general}, any $(π,ρ)$ in this right triangle
	is realizable by some polar codes, with complexity $O(N㏒N)$.
	Let me state the full theorem here.
	
	\begin{thm}[Hypotenuse]
		Fix a $q$-ary channel $W$.
		Fix exponents $π+2ρ<1$.
		Then there exists a large $ℓ$ and an amoebic kernel
		(a strategy to assign kernels to synthetic channels) $𝒢$ such that
		\begin{gather*}
			𝘗\{𝘡_n<\exp(-ℓ^{πn})\}>1-H(W)-ℓ^{-ρn},	\\
			𝘗\{𝘚_n<\exp(-ℓ^{πn})\}>H(W)-ℓ^{-ρn}		
		\end{gather*}
		for large $n$.
	\end{thm}
	
	Recall that at the beginning of \cref{cha:general}, I have demonstrated
	how to reduce arbitrary input alphabet to prime-power input alphabet.
	Hence the theorem applies to all DMCs, and we have reached
	the holy grail at the beginning of the chapter.
	
	\begin{cor}[Random codes' durability, polar codes' simplicity]\label{cor:holy}
		Over any discrete memoryless channel, for any constants $π,ρ>0$ such that
		$\pi+2\rho<1$, there is a series of error correcting codes with
		block length $N$ approaching infinity,
		block error probability $\exp(-N^π)$,
		code rate $N^{-ρ}$ less than the channel capacity,
		and encoding and decoding complexity $O(N㏒N)$ per code block.
	\end{cor}
	
	This is the second-moment paradigm code that was promised in the abstract.
	\Cref{tab:refarray} compares this corollary to past works.
	The same can be stated for lossless compression and lossy compression.

\begin{table}
	\caption{
		The references to the various performances of polar codes over various channels.
		The starred behaviors are weak---only
		requiring $π$ or $ρ$ or both to be positive.
	}\label{tab:refarray}
	\pgfplotstableread{
		{ }		BEC			SBDMC		$p$-ary	$q$-ary	finite		BDMC	afinite	
		LLN		Arikan09	Arikan09	STA09	STA09	STA09		SRDR12	w		
		LDP★	AT09		AT09		STA09	MT10	Sasoglu11	HY13	w		
		CLT★	KMTU10		HAU14		BGNRS18	w		w			w		w		
		MDP★	GX15		GX15		BGS18	w		w			w		w		
		LDP		KSU10		KSU10		w		w		w			w		w		
		CLT		FHMV18		GRY20		w		w		w			w		w		
		MDP		w			w			w		w		w			w		w		
	}\tableRefarray
	\def\arraystretch{1.44}
	\def\citeform#1{#1}
	\advance\tabcolsep-1pt
	\smaller
	\def\bigstar{$^⋆$}
	\def\assigncontent#1{\pgfkeyssetvalue{/pgfplots/table/@cell content}{#1}}
	\def\decodecolorcontent#1#2\relax{
		\if\pgfplotstablecol0	\assigncontent{#1#2}
		\else\if#1w				\assigncontent{Cor \ref{cor:holy}}
		\else					\assigncontent{\cite{#1#2}}
		\fi\fi
	}
	\pgfplotstabletypeset[
		columns/afinite/.style={column name=finite},
		every head row/.style={
			before row=\toprule&\multicolumn5c{Symmetric}&\multicolumn2c{Asymmetric}\\,
			after row=\midrule},
		every last row/.style={after row=\bottomrule},
		assign cell content/.code={\decodecolorcontent#1\relax}
	]\tableRefarray
\end{table}

	\begin{cor}[Very good code for lossless compression]
		For any lossless compression problem, for any constants $π,ρ>0$ such that
		$\pi+2\rho<1$, there is a series of source codes with
		block length $N$ approaching infinity,
		block error probability $\exp(-N^π)$,
		code rate $N^{-ρ}$ plus the conditional entropy,
		and encoding and decoding complexity $O(N㏒N)$ per code block.
	\end{cor}
	
	\begin{cor}[Very good code for lossy compression]
		For any lossy compression problem, for any constants $π,ρ>0$ such that
		$\pi+2\rho<1$, there is a series of source codes with
		block length $N$ approaching infinity,
		block error probability $\exp(-N^π)$,
		code rate $N^{-ρ}$ plus the test channel capacity,
		and encoding and decoding complexity $O(N㏒N)$ per code block.
	\end{cor}
	
	The next chapter prunes.

\chapter{Joint Pruning and Kerneling}\label{cha:joint}

	\dropcap
	Combination of two techniques, when executed properly,
	gives results that inherit advantages from the two individual techniques.
	In this chapter, I want to combine pruning from \cref{cha:prune}
	and random dynamic kerneling from \cref{cha:random}, and will verify that
	the chimera codes have the optimal gap to capacity and log-logarithmic complexity.
	
	Let me elaborate.
	First, we do not expect the resulting codes to have elpin error,
	that is, block error probability $\exp(-ℓ^{πn})$.
	This is because we set the threshold $θ≔N^{-2}$
	and prune the channel tree whenever $𝘡_m$ reaches $θ$.
	Should we wait for $𝘡_m$ to become as small as $θ≔\exp(-ℓ^{πn})$,
	pruning will not take place at $O(㏒n)$ depth,
	and there will be little to no savings on EU--DU pairs.
	In conclusion, there is a conflict between elpin error and log-logarithmic
	complexity, and I will give up the optimal decay of error to favor low complexity.
	
	That being said, theer is no obvious conflict between
	complexity and code rate, and it is very likely that we can retain
	both from \cref{cha:prune} and \cref{cha:random}, respectively.
	Hence this constituents the goal for this section---Construct error correcting codes
	with gap to capacity close to $N^{-1/2}$, encoding and decoding complexity
	$O(N㏒(㏒N))$ per block, and block error probability as small as possible.

\section{Toolbox Checklist}

	From \cref{cha:random}, it is possible to construct a channel process $\{𝘞_m\}$
	such that, for any $π+2ρ<1$ and large $m$ (depending on $π,ρ$),
	\begin{gather*}
		𝘗｛𝘡_m<e^{-ℓ^{πm}}｝>1-H(W)-ℓ^{-ρm},	\label{ine:Z-elpin}\\
		𝘗｛𝘚_m<e^{-ℓ^{πm}}｝>H(W)-ℓ^{-ρm}.		\label{ine:S-elpin}
	\end{gather*}
	From \cref{cha:general}, we know that every ergodic kernel
	assumes a positive $ϱ$, so every kernel assumes at least
	a pair $π,ρ>0$ such that \cref{ine:Z-elpin,ine:S-elpin} hold.
	
	From \cref{cha:prune}, we can set a threshold $θ$, which serves two purposes:
	One, we prune the channel tree at the point where $𝘡_m$ or $𝘚_m$ falls below $θ$.
	This in turn defines the stopping time
	\[𝘴≔n∧\min\{m:\min(Ｚ(𝘞_m),Ｓ(𝘞_m))<θ\}.\]
	Two, we collect in $𝒥$ indices that point to synthetic channels $𝘞_m$ whose
	$Ｚ$-parameter reaches $θ$, which is the cause why $𝘴$ is set to this depth $m$.
	
	For the asymmetric case, $θ$ defines the stopping time
	\[𝘴≔n∧\min\left\{m:\mat{
		\min(Ｚ(𝘞_m),Ｓ(𝘞_m))<θ† and†	\\
		\min(Ｚ(𝘘_m),Ｓ(𝘘_m))<θ
	}\right\}.\]
	$𝒥$ will then collect indices pointing
	to $𝘞_m$ whose $Ｚ$-value reaches $θ$ and to $𝘘_m$ whose $Ｓ$-value reaches $θ$.
	To put it differently, $𝒥$ is the event where both $Ｚ(𝘞_𝘴),Ｓ(𝘘_𝘴)<θ$.
	
	$θ$ and $𝘴$ determine a code. 
	We have three lemmas that generalize
	\cref{lem:prune-C,lem:prune-R,lem:prune-P} to arbitrary matrix kernels.
	Their proofs are straightforward (perhaps tautological) and will be sketchy.
	
	\begin{lem}[Complexity in terms of $𝘴$]
		The encoding and decoding complexity is
		$O(𝘌[𝘴])$ per channel usage, or $O(N𝘌[𝘴])$ per code block.
	\end{lem}
	
	\begin{proof}
		Similar to \cref{lem:prune-C}, I claim without a proof that
		the encoding and decoding complexity is proportional to  the number of
		EU--DU devices in the circuit and to the number of synthetic channels
		(multiplicity included) that undergo channel transformation.
		
		Since a trajectory $𝘞₀,𝘞₁…𝘞_𝘴$ undergoes the transformation $𝘴$ times,
		the average number of transformations is $𝘌[𝘴]$,
		and hence the total number of transformations is $N𝘌[𝘴]$.
	\end{proof}
	
	\begin{lem}[$R$ in terms of $𝒥$]
		The code rate is $𝘗\{𝘑₁^𝘴∈𝒥\}$, or $𝘗(𝒥)$ for short.
	\end{lem}
	
	\begin{proof}
		Similar to \cref{lem:prune-R}, I claim without a proof that
		the code rate is the density of naked pins that are selected in $𝒥$.
		
		Every pair of naked pins possesses probability measure $1/N$
		because there are always $N$ pairs of naked pins.
		Every synthetic channel $𝘞_𝘴$ assumes $N/2^𝘴$ copies in the circuit,
		hence possesses probability measure $2^{-𝘴}$.
		All indices $𝘑₁^𝘴$ in $𝒥$ possesses probability measure $2^{-m}$,
		which coincides with the measure of pins.
		Thus the density of selected pins is $𝘗\{𝘑₁^𝘴∈𝒥\}$.
	\end{proof}
	
	\begin{lem}[$Ｐ$ in terms of $𝒥$]
		The block error probability is
		$N𝘌[Ｐ(𝘞_𝘴)·𝘐\{𝘑₁^𝘴∈𝒥\}]≤qNθ$ for the symmetric case.
		For the asymmetric case, the block error probability is
		$N𝘌[Ｐ(𝘞_𝘴)·𝘐\{𝘑₁^𝘴∈𝒥\}]≤qNθ+N𝘌[T(𝘞_𝘴)·𝘐\{𝘑₁^𝘴∈𝒥\}]$,
		witch is at most $3qNθ$ in total.
	\end{lem}
	
	\begin{proof}
		Similar to \cref{lem:prune-P}, I claim without a proof that,
		for the symmetric case, the block error probability is the sum of
		the $Ｐ$-values of the $𝘞_𝘴$ in $𝒥$ (multiplicity included).
		And for the symmetric case, the block error probability is the said sum
		plus the total of the $T$-values of the $𝘘_𝘴$ in $𝒥$ (multiplicity included).
		
		For the sum of $Ｐ(𝘞_𝘴)$, we learn from \cref{lem:PvsZ}
		that $Ｐ(𝘞_𝘴)≤qZ(𝘞_𝘴)/2≤qＺ(𝘞_𝘴)/2≤qθ/2$.
		Hence a sum of at most $N$ items, each at most $qθ/2$, is at most $qNθ$.
		
		For the sum of $T(𝘘_𝘴)$, we learn form \cref{lem:PvsT} that
		\begin{align*}
			T(𝘘_𝘴)
			&	≤÷{2(q-1)}q-÷2q（(q-1)qＰ(W)-(q-1)(q-2)）	\\
			&	=÷{2(q-1)}q（1-\(qＰ(𝘘_𝘴)-(q-2)\)）=÷{2(q-1)²}q（1-÷q{q-1}Ｐ(𝘘_𝘴)）
			\shortintertext{and then from \cref{lem:PvsS} that}
			&	≤÷{2(q-1)²}qS(𝘘_n)≤2qS(𝘘_n)≤2qＳ(𝘘_n)≤2qθ.
		\end{align*}
		Hence a sum of at most $N$ items, each at most $2qθ$, is at most $2qNθ$.
		The contributions of $Ｐ(𝘞_𝘴)$ and $T(𝘘_n)$ amount to $3qNθ$;
		that finishes the proof.
	\end{proof}
	
	Now I can state and prove the main theorem in this chapter.
	My contribution is twofold.
	First contribution:
	If you insist on using a certain matrix $G$ to construct polar codes,
	then either $G$ is not ergodic (in which case there is no polarization at all),
	or you can construct log-logarithm codes with some positive $ρ$.
	Second contribution:
	If you allow dynamic kerneling with very large $ℓ$, then you can construct
	log-logarithm codes whose $ρ$ is arbitrarily close to $1/2$, the optimal exponent.

\section{Log-logarithmic Codes}

	The first case I want to discuss here is when a fixed kernel $G$ is selected prior.
	In this case, the best $ρ$ we can hope for
	is the $ϱ$ in the eigen/en23/een13 behavior of $G$.
	(After all, pruning should not improve the extent of polarization.)
	
	\begin{thm}[Log-log for asymmetric $q$-ary]\label{thm:ll-R-P}
		Fix a $q$-ary channel $W$.
		Fix a kernel $G∈𝔽_q^{ℓ×ℓ}$ and a pair $(π,ρ)$
		that satisfy the two-sided elpin behavior
		\begin{gather*}
			𝘗｛𝘡_m<e^{-ℓ^{πm}}｝>1-H(W)-ℓ^{-ρm+o(m)},	\\
			𝘗｛𝘚_m<e^{-ℓ^{πm}}｝>H(W)-ℓ^{-ρm+o(m)}.	
		\end{gather*}
		Then pruning the channel tree with threshold $θ≔1/3qN²$ yields
		\begin{gather*}
			𝘌[𝘴]=O(㏒(㏒N)),							\\
			𝘗(𝒥)=I(W)-N^{-ρ+o(1)},					\\
			N𝘌[Ｐ(𝘞_𝘴)·𝘐(𝒥)]+N𝘌[T(𝘘_𝘴)·𝘐(𝒥)]≤1/N.	
		\end{gather*}
	\end{thm}
	
	\begin{proof}
		To estimate $𝘌[𝘴]=∑_{m=1}^{n-1}𝘗\{𝘴>m\}$, we mimic \cref{thm:actual-C}.
		Consider small $m$ and large $m$.
		Those with $\exp(-ℓ^{πm})≥θ$ are called small $m$.
		Those with $\exp(-ℓ^{πm})<θ$ are called large $m$.
		For small $m$, we do not expect decent polarization and assume $𝘴≥m$.
		That is, we upper bound $𝘗\{s>m\}≤1$ by a pessimistic value.
		For large $m$, we argue that
		\begin{align*}
			\quad&\kern-1em
			𝘗\{s>m\}
				≤𝘗\{Ｚ(𝘞_m)∨Ｓ(𝘞_m))≥θ† or †Ｚ(𝘘_m)∨Ｓ(𝘘_m)≥θ\}	\\
			&	≤𝘗\{\max(Ｚ(𝘞_m),Ｓ(𝘞_m))≥θ\}+𝘗\{\max(Ｚ(𝘘_m),Ｓ(𝘘_m))≥θ\}	\\
			&	≤𝘗｛Ｚ(𝘞_m)∨Ｚ(𝘞_m)≥e^{-ℓ^{πm}}｝
				+𝘗｛Ｓ(𝘘_m)∨Ｓ(𝘘_m)≥e^{-ℓ^{πm}}｝	\\
			&	≤2ℓ^{-ρm+o(m)}+2ℓ^{-ρm+o(m)}=ℓ^{-ρm+o(m)}.
		\end{align*}
		As a result, the complexity is
		\begin{align*}
			𝘌[𝘴]
			&	=∑_{m=0}^{n-1}𝘗\{𝘴>m\}=\#\{†small †m\}+∑_{†large m†}ℓ^{-ρm+o(m)}	\\
			&	=O(㏒n)+O(1)=O(㏒(㏒N)).
		\end{align*}
		Here we use the fact that the number of small $m$'s is the root of the equation
		$\exp(-ℓ^{πm})=θ=1/3qN²=1/3qℓ^{2n}$, which is $㏒(n)$ (note that $N=ℓ^n$).
		
		To estimate $R=𝘗(𝒥)$, we mimic \cref{thm:actual-R}.
		It is the same as estimating the frequency that
		$𝘴$ is set to $m$ due to $Ｚ(𝘞_m)<θ$ and $Ｓ(𝘘_m)<θ$.
		This frequency is
		\begin{align*}
			R
			&	=𝘗(𝒥)=𝘗\{Ｚ(𝘞_𝘴)≤θ† and †Ｓ(𝘘_𝘴)≤θ\}	\\
			&	≥𝘗\{Ｚ(𝘞_m)→0† and †Ｓ(𝘘_m)→0\}	\\*
			&\kern4em	-𝘗\{Ｚ(𝘞_m)→0† but †Ｚ(𝘞_𝘴)≥θ\}	\\*
			&\kern4em	-𝘗\{Ｓ(𝘘_m)→0† but †Ｓ(𝘘_𝘴)≥θ\}	\\
			&	≥I(W)-𝘗\{Ｚ(𝘞_m)→0† but †Ｓ(𝘞_𝘴)<θ\}	\\*
			&\kern4em	-𝘗\{Ｓ(𝘘_m)→0† but †Ｚ(𝘘_𝘴)<θ\}-𝘗\{𝘴=n\}.
		\end{align*}
		Here we use the fact that if $Ｚ(𝘞_𝘴)≥θ$, then $𝘴$ is set to its current value
		due to $Ｓ(𝘞_𝘴)<θ$ or, otherwise, due to hitting $n$.
		It remains to estimate the three minus terms on the right-hand side.
		
		To bound $𝘗\{Ｚ(𝘞_m)→0† but †Ｓ(𝘞_𝘴)<θ\}$, notice that
		$1-H(𝘞_𝘴)≤q³√{Ｓ(𝘞_𝘴)}<q³/N√{3q}≤N^{-1+o(1)}$.
		Now the probability that $H(𝘞_m)→0$ given $H(𝘞_𝘴)≥1-N^{-1+o(1)}$
		is $N^{-1+o(1)}$ by the martingale property.
		For a similar reason, $𝘗\{Ｓ(𝘘_m)→0† but †Ｚ(𝘘_𝘴)<θ\}$ is no greater than
		the probability that $H(𝘘_m)→1$ given $H(𝘘_𝘴)<N^{-1+o(1)}$,
		which is $N^{-1+o(1)}$ by the martingale property again.
		Last is to bound $𝘗\{𝘴=n\}$, but that is just $𝘗\{𝘴>n-1\}-𝘗\{𝘴>n\}$,
		and is thus $ℓ^{-ρ(n-1)-o(n-1)}=N^{-ρ+o(1)}$.
		We conclude
		\[R≥I(W)-N^{-1+o(1)}-N^{-1+o(1)}-N^{-ρ+o(1)}=I(W)-N^{-ρ+o(1)}.\]
		
		The bound $N𝘌[Ｐ(𝘞_𝘴)·𝘐(𝒥)]+N𝘌[T(𝘘_𝘴)·𝘐(𝒥)]≤1/N$ is tautological.
		And the proof ends here.
	\end{proof}
	
	When allowing dynamic kerneling, random or not, the proof of the last theorem
	does not change---we still set a threshold $θ$ and prune the channel tree by $θ$.
	The only difference is that $ρ$ can now be arbitrarily close to $1/2$.
	
	To put it another way, like I once commented under \cref{lem:tensor}, there is
	no point to use a large kernel without knowing that it has a good $(π,ρ)$ pair.
	So \cref{thm:ll-R-P} mainly has three use cases:
	(a)	To estimate the behavior of a minimalist kernel such as $\locl$,
	(b)	to estimate the behavior of a larger kernel with a bound on $ρ$, and
	(c)	to estimate the behavior of random dynamic kerenling.
	That leads to the following corollary.
	
	\begin{cor}[Log-log code for DMC]
		Given any $q$-ary DMC (presumably with virtual symbols) and any ergodic kernel
		$G∈𝔽_q^{ℓ×ℓ}$, pruned polar coding achieves
		encoding and decoding complexity $O(㏒(㏒N))$ per channel usage,
		block error probability $1/N$,
		and code rate $1/N^ρ$ less than the channel .
		Here, $ρ$ is a number guaranteed to be positive, lower bounded
		if you know more about the eigen/\AB en23/\AB een13/\AB eplin behavior of $G$,
		and very close to $1/2$ as $ℓ→∞$ if you allow dynamic kerneling.
	\end{cor}
	
	The same can be stated regarding lossless and lossy compression, and is omitted.
	
	The next section discusses a continuous trade-off between $Ｐ$ and complexity.

\section{Error--Complexity Trade-off}

	In the proof of \cref{thm:ll-R-P}, the complexity $𝘌[𝘴]$ and $θ=θ(n)$ are linked by
	the equation $\exp(-ℓ^{πm})=θ$ that determines how many $m$'s are small.
	The root thereof $m=m(θ)=m(θ(n))$ will be the complexity.
	We may as well alter $θ(n)$ and obtain a different error--complexity pair.
	The only restriction is $ℓ^{-2n}≪θ(n)≪\exp(-ℓ^{πn})$
	to avoid the necessity to deal with special/edge cases.
	
	\begin{cor}[Continuous error--complexity trade-off]\label{thm:tradeoff}
		Let $θ(n)$ be an asymptote lying between $ℓ^{-2n}$ and $\exp(-ℓ^{πn})$.
		Then pruned polar coding achieves complexity $O(㏒\abs{㏒θ(n)})$
		per channel usage and block error probability $ℓ^nθ(n).$
	\end{cor}
	
	In particular, if $θ(n)≔\exp(-ℓ^{πn})$, then the complexity is
	$O(㏒\abs{㏒θ(n)})=O(㏒\abs{ℓ^{πm}})=O(n)=O(㏒N)$.
	This restores the case where one insists on retaining the $\exp(-N^π)$ error
	while pruning, which only gives you constant-scalar improvement in complexity.
	In fact, we can almost prove this tight:
	The complexity, $𝘌[𝘴]$, is the (average) number
	of transformations a trajectory $𝘞_n$ undergoes.
	Since each transformation, at best, raises $Ｚ(𝘞_n)$ to the power of $ℓ$,
	you need $㏒_ℓ(㏒_{Ｚ(W)}θ)$ transformations to lower $Ｚ(𝘞_n)$ to $θ$.
	That implies that $𝘌[𝘴]≥㏒_ℓ(㏒_{Ｚ(W)}θ)$.
	
	On the other hand, if $θ(n)≔\exp(-n^τ)$ for some very large $τ>0$,
	then $𝘌[𝘴]=O(㏒\abs{㏒θ(n)})=O(㏒\abs{-n^τ})=O(τ㏒(㏒N))$.
	If $τ$ is a constant despite of being very large,
	then the complexity is still $O(㏒(㏒N))$.
	This is the complexity paradigm code that was promised in the abstract.
	\Cref{tab:complex} compare this result and \cref{cor:holy} to past works.

\begin{table}
	\caption{
		A comparison concerning the error--rate--complexity asymptotes
		of some well-known capacity-achieving codes.
	}\label{tab:complex}
	\pgfplotstableread{
		Code			Error		Gap		Complexity	Channel	
		random			e^{-N^π}	N^{-ρ}	\exp(N)		DMC		
		concatenation	e^{-N^π}	→0		\poly(N)	DMC		
		RM				→0			→0		O(N^2)		BEC		
		LDPC			→0			→0		†unclear†	SBDMC	
		RA~family		→0			→0		O(1)		BEC		
		MDP-polar		e^{-N^π}	N^{-ρ}	O(㏒N)		DMC		
		loglog-polar	e^{-n^τ}	N^{-ρ}	O(㏒(㏒N))	DMC		
	}\tableComplex
	\def\unclear/{\text{unclear}}
	\pgfplotstablemodifyeachcolumnelement{Error}\of\tableComplex\as\cell
		{\edef\cell{$\unexpanded\expandafter{\cell}$}}
	\pgfplotstablemodifyeachcolumnelement{Gap}\of\tableComplex\as\cell
		{\edef\cell{$\unexpanded\expandafter{\cell}$}}
	\pgfplotstablemodifyeachcolumnelement{Complexity}\of\tableComplex\as\cell
		{\edef\cell{$\unexpanded\expandafter{\cell}$}}
	\def\arraystretch{1.44}
	\centering\pgfplotstabletypeset[
		every head row/.style={before row=\toprule,after row=\midrule},
		every last row/.style={after row=\bottomrule},string type,
	]\tableComplex
\end{table}

	The same can be stated concerning lossless and lossy compression, and is omitted.
	
	The next chapter generalizes the second-moment paradigm
	and complexity paradigm to some network coding scenarios.

\chapter{Distributed Lossless Compression}\label{cha:dislession}

	\dropcap
	Network coding emerges as a generalization of one-to-one communication
	as there are, naturally, more than one party willing to participate.
	One of the easiest scenarios (in comparison with other network scenarios,
	not necessarily easy compared to one-to-one)
	is when there are more than two random sources to be compressed,
	each by a compressor that does not talk the other compressors, and
	a decompressor will gather all messages and reconstruct the original random sources.
	This scenario is referred to as \emph{distributed compression}.
	
	Distributed compression can be further divided into
	several sub-scenarios that are treated differently.
	Depending on whether a reconstruction needs to be faithful or can be fuzzy,
	there are distributed \emph{lossless} compression and its \emph{lossy} variants.
	Depending on whether a random source needs to be reconstructed
	or it provides side information to the other sources,
	the responsible compressor is called a \emph{sender} or a \emph{helper}.
	To summarize, there are three types of sources---%
	those that need to be reconstructed as is,
	those whose reconstruction can be less accurate,
	and those that need no reconstruction at all.
	And a distributed compression problem consists of
	any combination of theses three sources.
	
	Thanks to the infrastructures built in the past three chapters,
	if we manage to reduce a network coding problem to several one-to-one problems,
	then each of the one-to-one problems can be solved by polar coding
	that achieves capacity at a good pace and with low complexity.
	In this chapter, I will overview distributed lossless compression problems
	with two senders, one sender plus one helper, three senders,
	and finally many senders plus one helper.
	These are the problems whose rate region is known \cite{EK11},
	thus it makes sense to pursue the second-order behavior of the rate tuples.

\section{Slepian--Wolf: The Two Sender Problem}

	A \emph{Slepian--Wolf problem} is a distributed lossless compression problem
	with two senders, which is the first of the several cases we will consider.
	See \cref{fig:2sender} for the specification.
	In this and the other scenarios, the “array access” operator $[•]$
	will be used to distinguish random sources.
	The pair $(R[1],R[2])$ will be called the \emph{rate pair}.
	The $N$ is still the block length.
	And $Ｐ$ is the block error probability, the probability that
	either $ˆX[1]₁^N≠X[1]₁^N$ or $ˆX[2]₁^N≠X[2]₁^N$.
	
	\begin{figure}
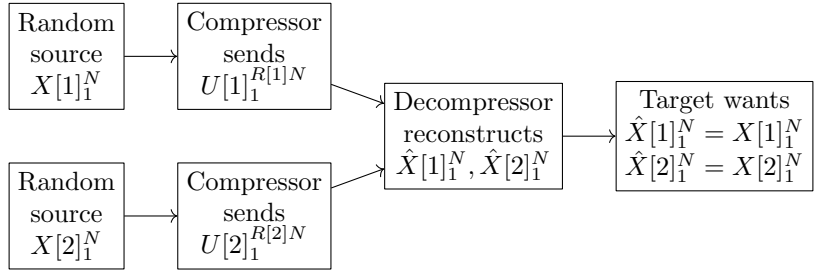

		\tikz{
			\draw[nodes={right,draw,align=center}]
				(.7,0)coordinate(X)(0,-.7)coordinate(Y)
				             node(A)[below]{Random\\source\\$X[1]₁^N$}
				(A.south)+(Y)node(B)[below]{Random\\source\\$X[2]₁^N$}
				(A.east) +(X)node(C){Compressor\\sends\\$U[1]₁^{R[1]N}$}
				(B.east) +(X)node(D){Compressor\\sends\\$U[2]₁^{R[2]N}$}
				($(C.east)!.5!(D.east)$)+(X)
				             node(T){Decompressor\\reconstructs\\$ˆX[1]₁^N,ˆX[2]₁^N$}
				(T.east) +(X)node(U){Target wants\\$ˆX[1]₁^N=X[1]₁^N$\\$ˆX[2]₁^N=X[2]₁^N$};
			\graph[use existing nodes]{
				A -> C -> T -> U,
				B -> D -> T
			};
		}
		\caption{
			A Slepian--Wolf problem.
		}\label{fig:2sender}
	\end{figure}
	
	The \emph{rate region} of a Slepian--Wolf problem
	is a region in $ℝ²$ of all achievable rate pairs $(R[1],R[2])$.
	There are clearly three pessimistic criteria:
	\begin{itemize}
		\item	About $H(X[1]｜X[2])$ bits of information is only available at
				source$[1]$, hence compressor$[1]$ should at least
				send out this much information, $R[1]≥H(X[1]｜X[2])$.
		\item	About $H(X[2]｜X[1])$ bits of information is only available at
				source$[2]$, hence compressor$[2]$ should at least
				send out this much information, $R[2]≥H(X[2]｜X[1])$.
		\item	About $H(X[1]X[2])$ bits of information are generated in total,
				hence the two compressors should at least
				send out this many bits in total, $R[1]+R[2]≥H(X[1]X[2])$.
	\end{itemize}
	As it turns out, these necessary criteria are sufficient.
	
	\begin{thm}[Slepian–Wolf Theorem]
		\cite{SW73}
		The rate region of a Slepian--Wolf problem
		consists of pairs $(R[1],R[2])∈ℝ²$ such that
		\begin{gather*}
			R[1]≥H(X[1]｜X[2]),		\\
			R[2]≥H(X[2]｜X[1]),		\\
			R[1]+R[2]≥H(X[1]X[2]),	
		\end{gather*}
		and is supported by the vertices
		\[\(\,H(X[1]｜X[2]),\,H(X[2])\,\)\quad†and†\quad
			\(\,H(X[1]),\,H(X[2]｜X[1])\,\).\]
		See also \cref{fig:pentagon}.
	\end{thm}
	
	\begin{figure}
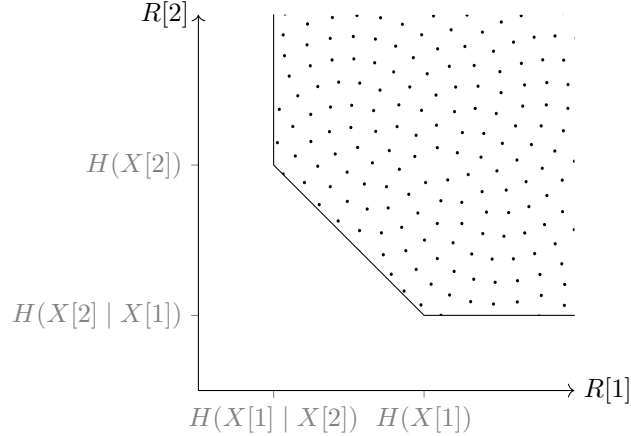

		\tikz{
			\draw
				(0,3)pic{y-tick=$H(X[2])$}
				(0,1)pic{y-tick=$H(X[2]｜X[1])$}
				(1,0)pic{x-tick=$H(X[1]｜X[2])$}
				(3,0)pic{x-tick=$H(X[1])$};
			\draw[->](0,0)--(0,5)node[left]{$R[2]$};
			\draw[->](0,0)--(5,0)node[right]{$R[1]$};
			\draw(1,5)--(1,3)--(3,1)--(5,1);
			\clip(1,5)--(1,3)--(3,1)--(5,1)|-cycle;
			\draw[shift={(3,3)},very thick,dash pattern=on0off99]
				foreach\i in{1,...,300}{
					\pgfextra{
						\PMS\t{mod(\i,2.618)*2.4}
						\PMS\r{sqrt(\i)/6}
						\PMS\x{\r*cos(\t r)}
						\PMS\y{\r*sin(\t r)}
					}
					(\x,\y)--+(0,1)
				};
		}
		\caption{
			Slepian--Wolf problem's rate region.
		}\label{fig:pentagon}
	\end{figure}
	
	Notation:
	When the context is clear, $H(1),H(2),H(12),H(1|2),H(2|1)$ denote the corresponding 
	(conditional) entropies with “$X[•]$” wrapping around every Arabic number.
	
	The problem here is, How fast can $R≔(R[1],R[2])$ approach the boundary
	that passes $(H(1|2),+∞)$--$(H(1|2),H(2))$--$(H(1),H(2|1))$--$(+∞,H(2|1))$?
	Random coding allowed, this was discussed in \cite{TK12} for the CLT regime.
	More elaborately, fixing a $Ｐ$, the (Euclidean) distance
	from $R$ to the boundary scales as $O(1/√N)$.
	Although there seems to be no references for the LDP and MDP regimes,
	we may presume that it obeys the same law as in the one-to-one case---namely,
	\[÷{-㏑Ｐ}{\dist(R,†boundary†)²}≈N.\]
	And we thus pose to ourselves a challenge about constructing
	polar codes (or any low-complexity codes)
	with $Ｐ≈\exp(-N^π)$ and $\dist≈N^{-ρ}$ whenever $π+2ρ<1$.
	
	An obvious strategy is to execute time-sharing, which is based on this simple idea:
	If we know how to achieve the point $(H(1),H(2|1))$ using a coding scheme
	and the point $(H(1|2),H(2))$ using another coding scheme,
	then we can alternate between two said schemes to approach any point lying on
	the \emph{sum-rate segment} $(H(1),H(2|1))$--$(H(1|2),H(2))$.
	And this time-sharing scheme, indeed, reaches the second moment goal in part.
	
	\begin{thm}[Timed Slepian--Wolf]
		Let $B$ be any point on the boundary of the rate region.
		Fix exponents $π+2ρ<1$.
		Then combining polar coding and time-sharing yields
		$Ｐ<\exp(-N^{π/(1+ρ/2)})$ and $\dist(R,B)<N^{-{ρ/(1+ρ/2)}}$
		at the cost of $O(㏒N)$ complexity per source observation.
		(Notice the penalty $(1+ρ/2)$.)
	\end{thm}
	
	\begin{proof}
		Let us first eliminate some trivial cases.
		Case one:
		If $B$ is on the vertical ray $(H(1|2),+∞)$--$(H(1|2),H(2))$,
		then the problem reduces to achieving the nontrivial vertex $(H(1|2),H(2))$.
		For this vertex, ask compressor$[2]$ to compress $X[2]$ and ask compressor$[1]$
		to compress $X[1]$ with side information $X[2]$, both using polar coding.
		Note that compressor$[1]$ need not, and cannot,
		access the side information $X[2]$.
		
		Case two:
		If $B$ is on the horizontal ray $(H(1),H(2|1))$--$(+∞,H(2|1))$,
		then the problem reduces to achieving $(H(1),H(2|1))$.
		For this vertex, ask compressor$[1]$ to compress $X[1]$
		and ask compressor$[2]$ to compress $X[2]$ with side information $X[1]$.
		For this case and case one, the problem reduces to
		one-to-one lossless compression and can be solved by polar coding
		within the specified gap to boundary/block error probability/complexity.
		
		Case three:
		Assume that $B$ is lying on the sum-rate segment
		$(H(1|2),H(2))$--$(H(1),H(2|1))$
		and is a rational combination of the two end points.
		That is, there exist positive integers $s,t$
		such that $(s+t)B=s(H(1|2),H(2))+t(H(1),H(2|1))$.
		Then this is what we do:
		For every $s+t$ code blocks, apply the coding scheme in case one for $s$ blocks
		and then apply the coding scheme in case two for $t$ blocks.
		Note that $s$ and $t$ are fixed constants, so the penalties imposed
		on $N$, $Ｐ$, $\dist(R,B)$, and the complexity are all constant scalars.
		
		Now we deal with the nontrivial case.
		Case four:
		Assume $B$ is on the sum-rate segment $(H(1|2),H(2))$--$(H(1),H(2|1))$
		and is an irrational combination of the two ends.
		Then this is what we do:
		Prepare the case-one scheme and case-two scheme with a large block length $M$.
		By the polar coding infrastructure in the past few chapters,
		the gap to entropy is $M^{-ρ}$ and the error is $\exp(-M^π)$.
		
		Here comes the punchline for case four:
		Pick large integers $s,t>0$ such that $s+t$ is about the size of $M^{ρ/2}$
		and $B$ is very close to (measured in the Euclidean distance)
		\[÷s{s+t}(H(1|2),H(2))+÷t{s+t}(H(1),H(2|1)).\label{app:rational}\]
		In other words, we use the denominator $s+t≈M^{ρ/2}$
		to approximate the irrational coefficients in the combination.
		Now, for every $s+t$ code blocks, we apply the case-one scheme for $s$ blocks
		and then apply the case-two scheme for $t$ blocks.
		Unlike the rational case, where the penalty is constant scalars,
		the penalty here scales as $M$ grows.
		Hence, in particular, the de facto block length is $N=(s+t)M≈M^{1+ρ/2}$.
		The error in terms of the de facto block length is
		$\exp(-M^π)=\exp(-N^{π/(1+ρ/2)})$, which
		suggests that the “de facto pi” is $π/(1+ρ/2)$.
		And the gap to boundary caused by the imperfect coding is
		$M^{-ρ}=N^{-ρ/(1+ρ/2)}$, which suggests a “de facto rho” of $ρ/(1+ρ/2)$.
		
		But coding is not the only cause of the gap.
		\Cref{app:rational} is very far away from $B$.
		In fact, by the Thue--Siegel--Roth theorem and its converse,
		the difference between an irrational number and its rational approximation
		is roughly the inverse square of the denominator,
		unless the irrational number lies in a measure-zero set.
		As a consequence, we almost alwayse have
		\[\dist(B,†\cref{app:rational}†)=Θ((s+t)^{-2})=Θ(M^{-ρ})=Θ(N^{-ρ/(1+ρ/2)}).\]
		This gap is comparable to the coding gap,
		so the overall gap is still $O(N^{-ρ/(1+ρ/2)})$.
		That finishes the proof.
	\end{proof}
	
	Remark: the de facto pi and rho satisfy
	$2·π/(1+ρ/2)+5·ρ/(1+ρ/2)=(2π+4ρ+ρ)/(1+ρ/2)<(2+ρ)/(1+ρ/2)=2$.
	Thus the region of de facto pi--rho pairs is strictly smaller than $π+2ρ<1$.
	
	Bibliographical remark:
	It is once suggested that a Slepian--Wolf problem can be solved by \emph{one}
	polar code via a technique called \emph{monotonic chain rule} \cite{Arikan12}.
	However, the CLT aspect of the monotonic chain rule is as capable as time-sharing;
	in fact, its CLT behavior is worse than my estimate here
	because the denominator therein can only be a power of $ℓ$.
	It would not help us cancel the $(1+ρ/2)$ penalty.
	
	In the next section, I borrow a technique that avoids
	approximating an irrational number using rational numbers.
	Intuitively speaking, this technique tunes the distribution of random variables
	(note that probabilities are real numbers, mostly irrational)
	to attain any necessary irrational number.

\section{Slepian--Wolf via Source-Splitting}

	\emph{Source-splitting}, in one sentence, divides the randomness
	carried by $X[2]$ into two random variables $X[2]⟨1⟩$ and $X[2]⟨2⟩$
	and then uses them to sandwich $X[1]$.
	By choosing a proper configuration of $X[2]⟨1⟩$ and $X[2]⟨2⟩$,
	we can attain any irrational combination on the sum-rate segment
	without referring to the time-sharing technique.
	
	In more detail, there will be $X[2]⟨1⟩$ and $X[2]⟨2⟩$ and
	a global “knob variable” $Q$ satisfying the axioms:
	\begin{itemize}
		\item	$Q$ is independent of $X[1]X[2]$, i.e.,
				the knob is a purely artificial variable;
		\item	$H(X[2]⟨1⟩X[2]⟨2⟩｜X[2]Q)=0$, i.e.,
				the knob fully controls how $X[2]$ is split;
		\item	$H(X[2]｜X[2]⟨1⟩X[2]⟨2⟩)=0$, i.e., piecing together
				the fragments of $X[2]$ yields the complete $X[2]$; and
		\item	$H(X[2]⟨1⟩｜Q)$, when $Q$ is tuned properly,
				varies from $0$ to $H(X[2])$, continuously and inclusively.
		\item	(Alternative to the fourth) $H(X[2]⟨2⟩｜Q)$
				varies from $0$ to $H(X[2])$, continuously and inclusively.
	\end{itemize}
	By the axioms, especially the fourth one,
	\begin{itemize}
		\item	$H(X[2]⟨1⟩｜X[1]X[2]⟨2⟩Q)+H(X[1]｜X[2]⟨2⟩Q)+H(X[2]⟨2⟩｜Q)
				=H(X[1]X[2]｜Q)$; and
		\item	$H(X[1]｜X[2]⟨2⟩Q)$ varies from $H(X[1])$ to $H(X[1]｜X[2])$,
				continuously and inclusively.
	\end{itemize}
	Now we let compressor$[1]$ compress $X[1]｜X[2]⟨2⟩Q$.
	Or equivalently, let it compress $X[1]$ given side information $X[2]⟨2⟩Q$.
	We also let compressor$[2]$ compress $X[2]⟨1⟩｜X[1]X[2]⟨2⟩Q$ and $X[2]⟨2⟩｜Q$.
	Or equivalently, let it compress $X[2]⟨1⟩$ given side information $X[1]X[2]⟨2⟩Q$,
	and then compress $X[2]⟨2⟩$ given side information $Q$.
	By that $B[1]≔H(X[1]｜X[2]⟨2⟩Q)$ varies from $H(X[1])$ to $H(X[1]｜X[2])$
	and that $B[2]≔($the sum of the other two conditional entropies$)$
	is $H(X[1]X[2])-B[1]$, we conclude that $B≔(B[1],B[2])$, as a function
	in the distribution of $Q$, can exhaust all points on the sum-rate segment.
	
	The formal statements are as follows.
	
	\begin{dfn}
		Let random variable $Q∈\{1,2\}$ be independent of $X[1]X[2]$.
		Let
		\[X[2]⟨1⟩≔\cas{
			X[2]	&	if $Q=1$,	\\
			♠		&	if $Q=2$,	
		}\]
		where $♠$ is a placeholder symbol that replaces $X[2]$
		when $Q$ decides that it should hide $X[2]$.
		Similarly,
		\[X[2]⟨2⟩≔\cas{
			♠		&	if $Q=1$,	\\
			X[2]	&	if $Q=2$.	
		}\]
		To abuse symbols, we can also say $X[2]⟨Q⟩=X[2]$ and $X[2]⟨3-Q⟩=♠$.
		This is called source-splitting or \emph{coded time-sharing}.
	\end{dfn}
	
	\begin{thm}[$Q$ed Slepian--Wolf]
		Let $B$ be any point on the boundary of the rate region.
		Fix exponents $π+2ρ<1$.
		Then combining polar coding and source-splitting (coded time-sharing)
		yields $Ｐ<\exp(-N^π)$ and $\dist(R,B)<N^{-ρ}$
		at the cost of $O(㏒N)$ complexity per source observation.
		(Notice the absence of penalty.)
	\end{thm}
	
	\begin{proof}
		Given the definition of $X[2]⟨1⟩$ and $X[2]⟨2⟩$, we give
		\begin{itemize}
			\item	compressor$[2]$ the task of compressing
					$X[2]⟨1⟩$ given $X[1]X[2]⟨2⟩Q$,
			\item	compressor$[1]$ the task of compressing
					$X[1]$ given $X[2]⟨2⟩Q$, and
			\item	compressor$[2]$ the task of compressing
					$X[2]⟨2⟩$ given $Q$.
		\end{itemize}
		By the polar coding infrastructure developed before,
		these tasks can be done with the specified error and gap to entropy.
		The only problem is, Does this coding scheme approach the correct entropy pair?
		
		To find out, let
		\begin{gather*}
			B[1]≔H(X[1]｜X[2]⟨2⟩Q),	\\
			B[2]≔H(X[2]⟨1⟩｜X[1]X[2]⟨2⟩Q)+H(X[2]⟨2⟩｜Q).
		\end{gather*}
		Then $B[1]+B[2]=H(X[2]⟨1⟩X[1]X[2]⟨2⟩｜Q)=H(12)$,
		so $(B[1],B[2])$ is always on the sum-rate segment.
		When $E[Q]=1$, that is, when $Q$ is always $1$,
		we see that $X[2]⟨2⟩$ is just a useless constant.
		In this case, $B[1]=H(X[1]｜Q)=H(1)$ and thus $B[2]=H(12)-H(1)=H(2|1)$.
		When $E[Q]=2$, that is, when $Q$ is always $2$,
		we see that $X[2]⟨2⟩=X[2]$.
		In this case, $B[1]=H(X[1]｜X[2]Q)=H(1|2)$ and thus $B[2]=H(12)-B[2]=H(2)$.
		
		All in all, when $E[Q]$ varies continuously from $1$ to $2$,
		the pair $(B[1],B[2])$ varies continuously from $(H(1),H(2|1))$ to
		$(H(1|2),H(2))$, and there must be a moment where $B=(B[1],B[2])$.
		Unless $B$ is on the vertical or horizontal ray,
		in which case the theorem is trivial.
		That ends the proof.
	\end{proof}
	
	That is how to generalize the second-moment paradigm to Slepian--Wolf problems.
	Similar statements can be made for $\exp(-n^τ)$ error and $㏒(㏒N)$ complexity.
	The proof will be essentially the same and is omitted.
	
	In the next section, we talk about a variant of Slepian--Wolf
	where the second source is not of interest, but
	compressing it helps the decompressor reconstruct the first.

\section{Compression with Helper}

	A \emph{lossless compression problem with a helper} is
	a distributed lossless compression problem with one sender and one helper.
	See \cref{fig:1helper} for the specification.
	The pair $R≔(R[1],R[∞])$ is still called the rate pair
	and the region of possible rate pairs is still called the rate region.
	The only difference is, this time, the block error probability $Ｐ$ 
	is the probability that $ˆX[1]₁^N≠X[1]₁^N$.
	
	\begin{figure}
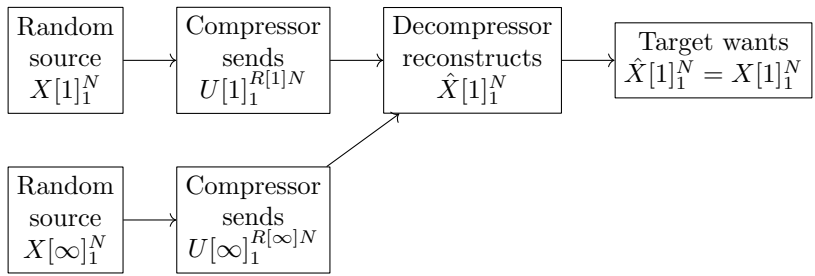

		\tikz{
			\draw[nodes={right,draw,align=center}]
				(.7,0)coordinate(X)(0,-.7)coordinate(Y)
				             node(A)[below]{Random\\source\\$X[1]₁^N$}
				(A.south)+(Y)node(B)[below]{Random\\source\\$X[∞]₁^N$}
				(A.east) +(X)node(C){Compressor\\sends\\$U[1]₁^{R[1]N}$}
				(B.east) +(X)node(D){Compressor\\sends\\$U[∞]₁^{R[∞]N}$}
				(C.east) +(X)node(T){Decompressor\\reconstructs\\$ˆX[1]₁^N$}
				(T.east) +(X)node(U){Target wants\\$ˆX[1]₁^N=X[1]₁^N$};
			\graph[use existing nodes]{
				A -> C -> T -> U,
				B -> D -> T
			};
		}
		\caption{
			Lossless compression with one helper.
		}\label{fig:1helper}
	\end{figure}
	
	Similar to the Slepian--Wolf case, the rate region of
	the one-helper problem can be characterized by an easy observation that
	“it should satisfy this” and a proof that confirms the observation.
	
	The easy observation is that, if we manage to find a random variable $U[∞]$
	that represents $X[∞]$ pretty well, then compressor$[1]$
	needs only to compress $X[1]$ given $U[∞]$ while compressor$[∞]$
	needs to assure that the decompressor receives $U[∞]$.
	For the latter, compressor$[∞]$ sends $I(X[∞]；U[∞])$ bits per source observation.
	
	\begin{thm}[Lossless compression with one helper]
		\cite{AK75}
		The rate region for lossless source coding of $X[1]$ with a helper source $X[∞]$
		consists of all pairs $(R[1],R[∞])∈ℝ²$ such that
		\begin{gather*}
			R[1]≥H(X[1]｜U[∞]),	\\
			R[∞]≥I(X[∞]；U[∞])
		\end{gather*}
		unioned over all random variables $U[∞]$
		that depend on $X[∞]$ but not on $X[1]$.
		Moreover, it is sufficient to consider the size of the alphabet of $U[∞]$
		that is one plus the size of the alphabet of $X[∞]$.
		See also \cref{fig:curvegon}.
	\end{thm}
	
	\begin{figure}
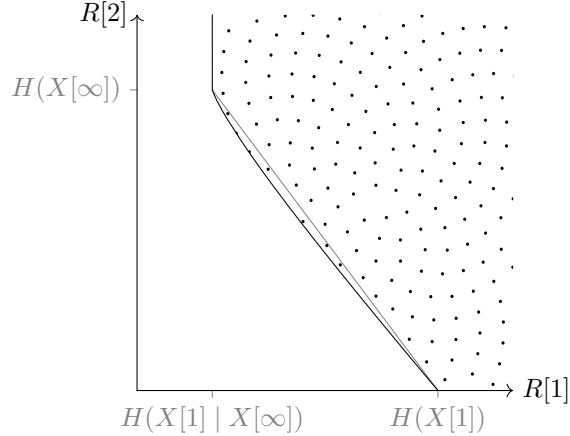

		\tikz{
			\draw[gray,very thin]
				(0,4)pic{y-tick=$H(X[∞])$}
				(1,0)pic{x-tick=$H(X[1]｜X[∞])$}
				(4,0)pic{x-tick=$H(X[1])$}
				(1,4)--(4,0);
			\draw[->](0,0)--(0,5)node[left]{$R[2]$};
			\draw[->](0,0)--(5,0)node[right]{$R[1]$};
			\draw(1,5)--
				plot[domain=0:45,samples=180]({4*h256(128+117.3267*cos(2*\x))},{4-4*h2o(\x)});
			\clip(1,5)--
				plot[domain=0:45,samples=180]({4*h256(128+117.3267*cos(2*\x))},{4-4*h2o(\x)})
				--(5,0)|-cycle;
			\draw[shift={(3,2.5)},very thick,dash pattern=on0off99]
				foreach\i in{1,...,360}{
					\pgfextra{
						\PMS\t{mod(\i,2.618)*2.4}
						\PMS\r{sqrt(\i)/6}
						\PMS\x{\r*cos(\t r)}
						\PMS\y{\r*sin(\t r)}
					}
					(\x,\y)--+(0,1)
				};
		}
		\caption{
			An example rate region of lossless compression problem with one helper.
			The gray segment is to contract the curved (unioned) boundary.
			Here, both $X[1]$ and $X[∞]$ are uniform binary sources
			with correlation $I(X[1]；X[∞])=3/4$.
		}\label{fig:curvegon}
	\end{figure}
	
	Remark:
	A subtle detail in the theorem statement is that the rate region is the union of
	all $(R[1],R[∞])$, not the \emph{convex hull} of the union of all $(R[1],R[∞])$.
	In other words, the rate region for this one-sender one-helper scenario
	does not need time-sharing to become convex---every point
	on the boundary is achievable by some clever choice of $U[∞]$.
	In fact, the variable $U[∞]$ can itself be the knob that controls, continuously,
	the time-sharing coefficient if there are really two schemes to be combined.
	
	Per the remark, we now have a very straightforward scheme to achieve
	the second moment behavior for this problem---for any $B$ on the boundary,
	pick a $U[∞]$ that achieves this $B$, use lossless compression to approach
	$H(X[1]｜U[∞])$, and use lossy compression to approach $I(X[∞]；U[∞])$.
	
	\begin{thm}[Polar coding with one helper]
		Let $B$ be any point on the boundary of the rate region.
		Fix exponents $π+2ρ<1$.
		Then polar coding alone yields $Ｐ<\exp(-N^π)$ and $\dist(R,B)<N^{-ρ}$
		at the cost of $O(㏒N)$ complexity per source observation.
	\end{thm}
	
	\begin{proof}
		Let $U[∞]$ be the auxiliary variable such that
		\[B=\(\,H(X[1]｜U[∞]),\,I(X[∞]；[∞])\,\).\]
		Tell compressor$[1]$ to compress $X[1]$
		losslessly with $U[∞]$ as the side information.
		Tell compressor$[∞]$ to do lossy compression
		with $U[∞]→X[∞]$ being the test channel.
		Both of them are covered by the infrastructure and
		can be done with the specified error and gap to entropy.
	\end{proof}
	
	Remark:
	The same can be stated with error $\exp(n^τ)$ and complexity $㏒(㏒N)$.
	The proof is exactly the same and thus omitted.
	
	It is pointless to have two helpers and no senders.
	So two senders and one-sender--one-helper are
	all we need to consider for two random sources.
	In the upcoming sections, we will see scenarios with more than two sources.
	The very next scenario we will go over is when there are three senders.

\section{Three-Sender Slepian--Wolf}

	Starting from three senders, a rate region of a distributed compression problem
	will be a subset in a higher-dimensional Euclidean space.
	Most importantly, the sum-rate segment will become
	a sum-rate polygon or even a sum-rate polyhedron.
	They are also called the \emph{dominant face}, where dominance refers to
	the fact that it is the set of minimal points under coordinate-wise comparison,
	and face refers to that it has co-dimension $1$ in the ambient space.
	
	To achieve the sum-rate polyhedron, we can always
	apply time-sharing and accept the penalty $(1+2ρ/3)$.
	(Note that it is even harder to approximate
	multiple irrational numbers using a common denominator,
	so $ρ/2$ will become $2ρ/3$, $3ρ/4$, etc.\ as the dimension increases.)
	We can also generalize source-splitting to multiple senders.
	In the latter case, the problem boils down to
	why source-splitting exhausts all points in the sum-rate polyhedron;
	and this is nontrivial.
	
	To demonstrate the non-triviality of source-splitting, consider three senders.
	See \cref{fig:3sender} for the specification.
	An easy observation is that, for any sender, its corresponding rate
	is at least the entropy of its source conditioned on the other two sources.
	That is,
	\[R[1]≥H(1|23)\quad R[2]≥H(2|13)\quad R[3]≥H(3|12).\label{ins:face}\]
	Similarly, any two senders should send out the entropy of
	their sources conditioned on the remaining one source.
	That is,
	\[R[1]+R[2]≥H(12|3)\quad R[1]+R[3]≥H(13|2)\quad R[2]+R[3]≥H(23|1).\label{ins:edge}\]
	And lastly, the sum-rate is no less than the overall entropy:
	\[R[1]+R[2]+R[3]≥H(123).\label{ine:vertex}\]
	Theses inequalities turn out the be the only inequalities
	that a feasible rate tuple needs to satisfy.
	
	\begin{figure}
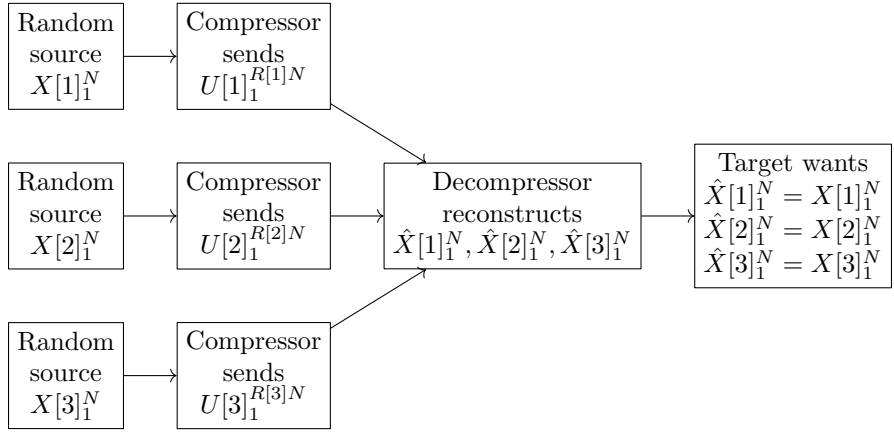

		\tikz{
			\draw[nodes={right,draw,align=center}]
				(.7,0)coordinate(X)(0,-.7)coordinate(Y)
				             node(A)[below]{Random\\source\\$X[1]₁^N$}
				(A.south)+(Y)node(B)[below]{Random\\source\\$X[2]₁^N$}
				(B.south)+(Y)node(C)[below]{Random\\source\\$X[3]₁^N$}
				(A.east) +(X)node(D){Compressor\\sends\\$U[1]₁^{R[1]N}$}
				(B.east) +(X)node(E){Compressor\\sends\\$U[2]₁^{R[2]N}$}
				(C.east) +(X)node(F){Compressor\\sends\\$U[3]₁^{R[3]N}$}
				(E.east) +(X)node(T){Decompressor\\reconstructs\\$ˆX[1]₁^N,ˆX[2]₁^N,ˆX[3]₁^N$}
				(T.east) +(X)node(U){Target wants\\$ˆX[1]₁^N=X[1]₁^N$\\
				                     $ˆX[2]₁^N=X[2]₁^N$\\$ˆX[3]₁^N=X[3]₁^N$};
			\graph[use existing nodes]{
				A -> D -> T -> U,
				B -> E -> T,
				C -> F -> T
			};
		}
		\caption{
			Distributed lossless compression with three senders.
		}\label{fig:3sender}
	\end{figure}
	
	With only \cref{ins:face}, the rate region looks like a cube
	(which actually extends to infinity).
	Now \cref{ins:edge} will chamfer the three edges that are closest to the axes.
	And finally, \cref{ine:vertex} will truncate
	the corner that is closest to the origin.
	See \cref{fig:truncate} for an illustration of the result.
	The \emph{rate hexagon} is the intersection of
	the rate region with \cref{ine:vertex} replaced by equality.
	
	\tdplotsetmaincoords{90-23.5}{120} 
	\def\cameradistance{10}
	\let\oldpointxyz\pgfpointxyz
	\def\pgfpointxyz#1#2#3{%
		\oldpointxyz{#1-2}{#2-2}{#3-2}
		\PMS\depth{\rcarot*\pgftemp@x+\rcbrot*\pgftemp@y+\rccrot*\pgftemp@z}%
		\PMS\depthrescale{\cameradistance/(\cameradistance-\depth)*2}%
		\pgf@x=\depthrescale\pgf@x%
		\pgf@y=\depthrescale\pgf@y%
	}
	\makeatletter
	\def\tikz@install@auto@anchor@perptip{\def\tikz@do@auto@anchor{%
		\pgfmathsetmacro\tikz@anchor{atan2(-\pgf@x,\pgf@y)}}}
	\tikzset{
		face/.style={preaction={fill=white,fill opacity=2/3}},
		mesh/.style={opacity=2/3,very thin}
	}
	\def\face#1:#2;{
		\fill[white,fill opacity=2/3]#1#2;
		\draw#1;
		\clip#1#2;
	}
	\def\mesh{\draw[opacity=2/3,very thin]}
	\begin{figure}
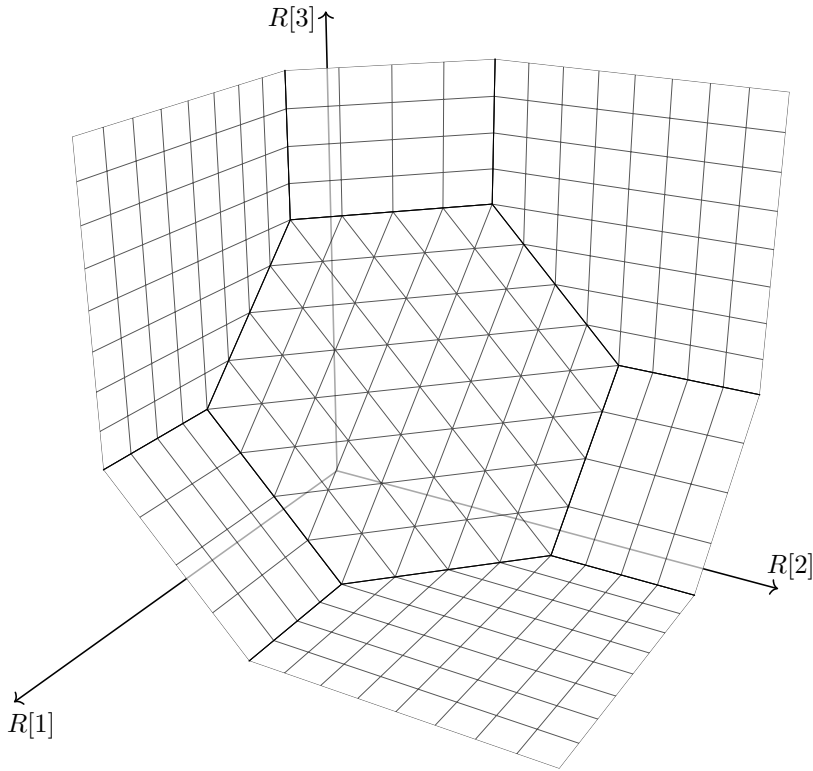

		\tikz[tdplot_main_coords]{
			\draw[semithick,nodes={pos=1,auto=perptip}](0,0,0)
				edge[->]node{$R[1]$}(4,0,0)
				edge[->]node{$R[2]$}(0,4,0)
				edge[->]node{$R[3]$}(0,0,4);
			\begin{scope}
				\face	(4,2,1)--(3,2,1)--(2,3,1)--(2,4,1):--(4,4,1)--cycle;
				\mesh	foreach\x in{0/4,.../4,8/4}{(2+\x,2,1)--(2+\x,4,1)}
						foreach\y in{0/4,.../4,8/4}{(2,2+\y,1)--(4,2+\y,1)};
			\end{scope}
			\begin{scope}
				\face	(2,1,4)--(2,1,3)--(3,1,2)--(4,1,2):--(4,1,4)--cycle;
				\mesh	foreach\z in{0/4,.../4,8/4}{(2,1,2+\z)--(4,1,2+\z)}
						foreach\x in{0/4,.../4,8/4}{(2+\x,1,2)--(2+\x,1,4)};
			\end{scope}
			\begin{scope}
				\face	(1,4,2)--(1,3,2)--(1,2,3)--(1,2,4):--(1,4,4)--cycle;
				\mesh	foreach\y in{0/4,.../4,8/4}{(1,2+\y,2)--(1,2+\y,4)}
						foreach\z in{0/4,.../4,8/4}{(1,2,2+\z)--(1,4,2+\z)};
			\end{scope}
			\begin{scope}
				\face	(2,1,4)--(2,1,3)--(1,2,3)--(1,2,4):--cycle;
				\mesh	foreach\z in{0/4,.../4,4/4}{(1+\z,2-\z,3)--(1+\z,2-\z,4)}
						foreach\z in{0/4,.../4,4/4}{(1,2,3+\z)--(2,1,3+\z)};
			\end{scope}
			\begin{scope}
				\face	(1,4,2)--(1,3,2)--(2,3,1)--(2,4,1):--cycle;
				\mesh	foreach\y in{0/4,.../4,4/4}{(2-\y,3,1+\y)--(2-\y,4,1+\y)}
						foreach\y in{0/4,.../4,4/4}{(2,3+\y,1)--(1,3+\y,2)};
			\end{scope}
			\begin{scope}
				\face	(4,2,1)--(3,2,1)--(3,1,2)--(4,1,2):--cycle;
				\mesh	foreach\x in{0/4,.../4,4/4}{(3,1+\x,2-\x)--(4,1+\x,2-\x)}
						foreach\x in{0/4,.../4,4/4}{(3+\x,1,2)--(3+\x,2,1)};
			\end{scope}
			\begin{scope}
				\face	(1,2,3)--(1,3,2)--(2,3,1)--(3,2,1)--(3,1,2)--(2,1,3)--cycle:;
				\mesh	foreach\z in{0/4,.../4,8/4}{(4-\z,1,1+\z)--(1,4-\z,1+\z)}
						foreach\y in{0/4,.../4,8/4}{(1,1+\y,4-\y)--(4-\y,1+\y,1)}
						foreach\x in{0/4,.../4,8/4}{(1+\x,4-\x,1)--(1+\x,1,4-\x)};
			\end{scope}
		}
		\caption{
			The rate region for if there are three senders.
			The camera is inside the rate region, looking at the origin.
			A pentagon with square mesh is when
			one of \cref{ins:face} is forced to be an equality.
			A rectangle with rectangular mesh is when
			one of \cref{ins:edge} is forced to be an equality.
			The hexagon is when the sum-rate equals the total entropy.
		}\label{fig:truncate}
	\end{figure}
	
	Now consider the following source-splitting setup.
	\begin{itemize}
		\item	$X[1]$ will not be split, but will be denoted by
				$X[1]⟨1⟩$ for notational compatibility;
		\item	$X[2]$ will be split into $X[2]⟨1⟩$ and $X[2]⟨2⟩$; and
		\item	$X[3]$ will be split into
				$X[3]⟨1⟩$, $X[3]⟨2⟩$, $X[3]⟨3⟩$, and $X[3]⟨4⟩$.
	\end{itemize}
	And now we use the fragments of $X[3]$ to sandwich
	“the sandwich made out of $X[1]$ and $X[2]$”.
	More precisely, we order them as
	\[X[3]⟨1⟩,X[2]⟨1⟩,X[3]⟨2⟩,X[1]⟨1⟩,X[3]⟨3⟩,X[2]⟨2⟩,X[3]⟨4⟩.\label{seq:hanoi}\]
	In general, the fragment $X[m]⟨l⟩$ will be placed at $(2l-1)/2^m$
	on the number line, and then we read off the fragments from left to right.
	
	Let $Q$ be a random variable that outputs a permutation of $\{1,2,3\}$.
	That is to say, $Q∈S₃≔\{123,132,312,321,231,213\}$.
	Depending on $Q$, we want to assign each fragment
	a true value or a placeholder symbol.
	For a fixed $m$, all $X[m]⟨l⟩$ will be $♠$ except that one will be $X[m]$.
	The fragments that get the true values are such that $X[Q(1)]$ will appear first
	on the number line, followed by $X[Q(2)]$, and finishing with $X[Q(3)]$.
	For example, if $Q=231$, then $X[2]⟨1⟩$ gets the true value of $X[2]$
	and $X[3][2]$ gets the true value of $X[3]$.
	Now \cref{seq:hanoi} becomes
	\[♠,\,X[2]⟨1⟩=X[2],\,X[3]⟨2⟩=X[3],\,X[1]⟨1⟩=X[1],\,♠,\,♠,\,♠.\]
	See \cref{tab:permute} for the other $Q$'s.
	The assignment is not necessarily unique (e.g., when $Q=213$ or $Q=312$);
	we will get back to this soon.
	
	\begin{table}
		\caption{
			Splitting three sources per the permutation $Q$.
			Note that there are eight rows because
			two permutations ($312$ and $213$) assume two solutions.
		}\label{tab:permute}
		\pgfplotstableread[header=false]{
			Q		X[3]⟨1⟩	X[2]⟨1⟩	X[3]⟨2⟩	X[1]⟨1⟩	X[3]⟨3⟩	X[2]⟨2⟩	X[3]⟨4⟩	
			123		♠		♠		♠		X[1]	♠		X[2]	X[3]	
			132		♠		♠		♠		X[1]	X[3]	X[2]	♠		
			312		♠		♠		X[3]	X[1]	♠		X[2]	♠		
			312		X[3]	♠		♠		X[1]	♠		X[2]	♠		
			321		X[3]	X[2]	♠		X[1]	♠		♠		♠		
			231		♠		X[2]	X[3]	X[1]	♠		♠		♠		
			213		♠		X[2]	♠		X[1]	X[3]	♠		♠		
			213		♠		X[2]	♠		X[1]	♠		♠		X[3]	
	}\tabelPermute
		\def\arraystretch{1.44}
		\centering\pgfplotstabletypeset[
			every head row/.style={output empty row,after row=\toprule},
			every row no 0/.style={after row=\midrule},
			every row no 1/.style={before row=\rowcolor{lightgray}},
			every row no 3/.style={before row=\rowcolor{lightgray}},
			every row no 4/.style={before row=\rowcolor{lightgray}},
			every row no 6/.style={before row=\rowcolor{lightgray}},
			every last row/.style={after row=\bottomrule},
			assign cell content/.style={@cell content/.initial=$#1$}
		]\tabelPermute
	\end{table}
	
	With the fragments defined, I will specify the coding scheme:
	For each $X[m]⟨l⟩$, it will be compressed by compressor$[m]$
	given all fragments to the right and $Q$.
	More precisely,
	\begin{itemize}
		\item	compressor$[3]$ will compress $X[3]⟨1⟩$ given
				$X[2]⟨1⟩\AB X[3]⟨2⟩\AB X[1]⟨1⟩\AB X[3]⟨3⟩\AB X[2]⟨2⟩\AB X[3]⟨4⟩\AB Q$,
		\item	compressor$[2]$ will compress $X[2]⟨1⟩$ given
				$X[3]⟨2⟩\AB X[1]⟨1⟩\AB X[3]⟨3⟩\AB X[2]⟨2⟩\AB X[3]⟨4⟩\AB Q$,
		\item	compressor$[3]$ will compress $X[3]⟨2⟩$ given
				$X[1]⟨1⟩\AB X[3]⟨3⟩\AB X[2]⟨2⟩\AB X[3]⟨4⟩\AB Q$,
		\item	compressor$[1]$ will compress $X[1]⟨1⟩$ given
				$X[3]⟨3⟩\AB X[2]⟨2⟩\AB X[3]⟨4⟩\AB Q$,
		\item	compressor$[3]$ will compress $X[3]⟨3⟩$ given
				$X[2]⟨2⟩\AB X[3]⟨4⟩\AB Q$,
		\item	compressor$[2]$ will compress $X[2]⟨2⟩$ given
				$X[3]⟨4⟩\AB Q$, and finally
		\item	compressor$[3]$ will compress $X[3]⟨4⟩$ given
				$\AB Q$.
	\end{itemize}
	Let $B[m]$ be the sum of the conditional entropies of the duties of compressor$[m]$.
	It is clear that the sum of these conditional entropies
	is $H(X[1]X[2]X[3])$ by the chain rule;
	so $B≔(B[1],B[2],B[3)$ lies on the sum-rate hexagon.
	The claim is that, by equipping $Q$ with the correct distribution,
	$B$ achieves any point on the sum-rate hexagon.
	
	\begin{thm}[Onto the hexagon]
		$B≔(B[1],B[2],B[3])$ exhausts all points on the sum-rate hexagon
		as $Q$ varies over all distributions on $S₃$.
	\end{thm}
	
	\begin{proof}
		The sum-rate hexagon has vertices
		\begin{gather*}
					(H(1|3),H(2|13),H(3)),\quad(H(1|23),H(2|3),H(3)),			\\
			(H(1),H(2|31),H(3|1)),		\kern12em		(H(1|32),H(2),H(3|2)),	\\
					(H(1),H(2|1),H(3|21)),\quad(H(1|2),H(2),H(3|12)).			
		\end{gather*}
		Each permutation on $\{1,2,3\}$ corresponds to a vertex by specifying
		which sources are compressed in full and which are conditioned on the others.
		For instance, the permutation $123$ corresponds to the vertex on the top right,
		and the permutation $321$ corresponds to the vertex on the bottom left.
		See \cref{fig:hexagon} for more on this correspondence.
		(Note that $(H(3|21),H(2|1),H(1))$ does not make sense,
		because compressor$[1]$ cannot access $X[3]$.)
		
		Our strategy is as follows:
		We first show that every vertex is achievable.
		We then show that every edge is achievable.
		We lastly show that the entire hexagon is achievable.
		
		The first goal is straightforward.
		If we want to achieve, for instance, the vertex $(H(1),H(2|31),H(3|1))$
		corresponding to the permutation $231$, then let $Q=231$ with probability $1$.
		This means that $X[2]⟨1⟩=X[2]$ and $X[3]⟨2⟩=X[3]$ constantly,
		and the other fragments are all $♠$ constantly.
		As a result, compressor$[2]$ will have to compress $X[2]｜X[3]X[1]$,
		compressor$[3]$ will have to compress $X[3]｜X[1]$,
		and compressor$[1]$ will have to compress $X[1]$.
		Now $B≔(B[1],B[2],B[3])$ becomes $(H(1),H(2|31),H(3|1))$, as desired.
		For any other vertex, the argument is similar and thus omitted.
		
		To achieve the second goal, take the edge $312$--$321$ as an example.
		Now we let $Q=312$ with probability $1-t$ and let $Q=321$ with probability $t$.
		As $t$ goes from $0$ to $1$, the knob $Q$ varies from
		constantly $312$ to constantly $321$.
		This means that $B$ moves from the vertex $312$ to the vertex $321$.
		Along this process, $X[3]$ is always compressed conditioned on the other two,
		so $B[3]$ is always $H(3|12)=H(3|21)$.
		This means that $B(t)$, as a function in $t$,
		maps surjectively onto the edge $312$--$321$ of the hexagon.
		For any other edge, the argument is similar and thus omitted.
		
		It remains to show that $B$, as a function in the distribution of $Q$,
		maps surjectively onto the hexagon.
		To this end, consider the following “Tour de France” definition of $Q_t$
		\[Q_t=\cas{
			†constantly †123	&	when $t=0$,	\\
			†constantly †132	&	when $t=1$,	\\
			†constantly †312	&	when $t=2$,	\\
			†constantly †321	&	when $t=3$,	\\
			†constantly †231	&	when $t=4$,	\\
			†constantly †213	&	when $t=5$,	\\
			†constantly †123	&	when $t=6$,
		}\]
		and filling in the non-integer $t$ by interpolation
		\[Q_t=\cas{
			Q_{⌈t⌉}	&	 with probability $⌈t⌉-t$,	\\
			Q_{⌊t⌋}	&	 with probability $t-⌊t⌋$.
		}\]
		By the previous paragraph, $B(t)$ will travel through each edge of the hexagon
		(although we have no idea the velocity it travels)
		in the order given in \cref{tab:permute}.
		
		Now let me borrow some algebraic topology nonsense:
		The space of the distributions on $S₃$
		is contractible to the uniform distribution.
		Thus the Tour de France $Q_t$ is a cycle
		(mapped to $0$ by the co-differential operator $∂$)
		that happens to be the boundary (the image under $∂$) of some disk $D$.
		This disk $D$ has its boundary mapped to the boundary of the hexagon with
		winding number $1$ (degree $1$), so $D$ will map surjectively onto the hexagon.
		That completes the proof.
	\end{proof}
	
	\begin{figure}
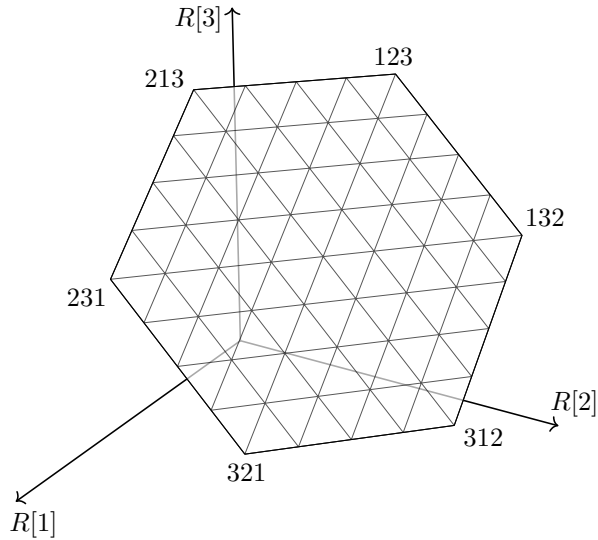

		\tikz[tdplot_main_coords]{
			\draw[semithick,nodes={pos=1,auto=perptip}](0,0,0)
				edge[->]node{$R[1]$}(3,0,0)
				edge[->]node{$R[2]$}(0,3,0)
				edge[->]node{$R[3]$}(0,0,3);
			\draw[nodes={pos=0,auto=perptip}]
				(1,2,3)--node{$123$}
				(1,3,2)--node{$132$}
				(2,3,1)--node{$312$}
				(3,2,1)--node{$321$}
				(3,1,2)--node{$231$}
				(2,1,3)--node{$213$}
				cycle;
			\face	(1,2,3)--(1,3,2)--(2,3,1)--(3,2,1)--(3,1,2)--(2,1,3)--cycle:;
			\mesh	foreach\z in{0/4,.../4,8/4}{(4-\z,1,1+\z)--(1,4-\z,1+\z)}
					foreach\y in{0/4,.../4,8/4}{(1,1+\y,4-\y)--(4-\y,1+\y,1)}
					foreach\x in{0/4,.../4,8/4}{(1+\x,4-\x,1)--(1+\x,1,4-\x)};
		}
		\caption{
			The sum-rate hexagon with vertices indexed by permutations.
		}\label{fig:hexagon}
	\end{figure}
	
	The proof is inspired by \cite{GRUW01},
	where multiple access channels are considered.
	Once we know how to use $Q$ to attain any point on the sum-rate hexagon,
	use the infrastructure to code.
	
	\begin{cor}[Polar coding for three senders]
		For three-sender distributed lossless compression,
		polar coding coupled with source-splitting attains every point
		on the boundary of the rate region with $N^{-ρ}$ gap.
	\end{cor}
	
	The next section combines everything we learned so far in this chapter
	to attack the problem with more than three senders and one helper.

\section{Many Senders with One Helper}

	Below \cite[Theorem~10.4]{EK11}, the authors commented that
	the optimal rate region is unknown when there are two helpers.
	That is to say, the most general case with known rate region
	is when there are multiple senders and one helper.
	We stick to the cases with known rate region because
	we want to state theorems about the pace of convergence;
	only when the aimed limit is optimal is this pace meaningful.
	
	For a many-sender one helper problem, the description of the rate region
	is a combination of inequalities of the form $R[1]+R[3]+R[5]≥H(135|246)$
	and $R[∞]≥$ the capacity of a proper test channel.
	
	\begin{thm}[Distributed lossless compression with a helper]
		The optimal rate region for lossless source coding of
		$X[1]…X[M]$ with helper source $X[∞]$ is described by
		\[∑_{m∈𝒮}R[m]≤H(X[𝒮]｜U[∞],X[𝒮^∁])\label{ins:marriage}\]
		for all subsets $𝒮⊆\{1…M\}$ and
		\[R[∞]≥I(X[∞]；U[∞])\]
		unioned over all random variables $U[∞]$.
		Here, $X[𝒮]$ is the tuple $(X[m]:m∈𝒮)$ and
		$X[𝒮^∁]$ is what is left $(X[m]:m∉𝒮)$.
	\end{thm}
	
	Note that the theorem implicitly uses $U[∞]$ as a coded time-sharing knob
	so there is not need to take the convex hull of the union.
	Now fix a point $B$ on the boundary of the rate region.
	Fix a $U[∞]$ that achieve this point in the rate region.
	Then \cref{ins:marriage} is a family of inequalities
	parametrized by the subset $𝒮$ of $\{1…m\}$.
	The right-hand side of the inequalities,
	\[H(X[𝒮]｜U[∞],X[𝒮^∁])=H\((X[m]:m∈𝒮)｜U[∞],(X[m]:m∉𝒮)\),\]
	is a supermodular function in $𝒮$.
	To verify this, it suffices to check the three variable case.
	
	\begin{lem}[Supermodularity]
		For any random variables $X,Y,Z$,
		\[H(XY｜Z)+H(YZ｜X)≤H(Y｜XZ)+H(XYZ).\]
	\end{lem}
	
	\begin{proof}
		Subtract $2H(XYZ)$ from both sides;
		the desired inequality is equivalent to $H(X)+H(Z)≥H(XZ)$.
	\end{proof}
	
	A supermodular function comes with a contra-polymatroid
	defined by the sum of coordinates over a subset $𝒮$ being
	greater than or equal to the function evaluation at $𝒮$.
	In other words, the rate region with a fixed $U[∞]$ is a contra-polymatroid.
	This is the dual case of a polymatroid defined by a submodular function.
	The latter is seen when one considers
	the capacity region of a multiple access channel \cite{GRUW01}.
	
	By the duality between (submodular function, polymatroid)
	and (supermodular function, contra-polymatroid),
	the proof given in \cite{GRUW01} applies here.
	The proof therein says that the \emph{rate}-splitting technique exhausts
	all points on the sum-rate polyhedron of a multiple access channel.
	And here, we conclude that the source-splitting technique exhausts
	all points on the sum-rate polyhedron of a distributed lossless compression.
	
	\begin{cor}[Polar coding for many-sender one-helper]
		Polar coding coupled with source-splitting attains every point
		on the boundary of the rate region with $N^{-ρ}$ gap.
	\end{cor}
	
	\begin{proof}[Sketch of the proof]
		Given the polar coding infrastructure, it suffices to
		construct a scheme to split sources and show that it exhausts
		all point on the sum-rate polyhedron (aka.\ the dominant face).
		
		The splitting scheme will look like the following:
		For each $m$, the $m$th source $X[m]$ will be split into $2^{m-1}$ fragments
		and the latter are named $X[m]⟨1⟩…X[m]⟨2^{m-1}⟩$.
		Each fragment $X[m]⟨l⟩$ will be placed at $(2l-1)/2^m$ on the number line.
		The $m$th compressor will compress $X[m]⟨l⟩$ given everything to its right.
		The $m$th duty entropy, $B[m]$, will be the sum of
		$H(X[m]⟨l⟩｜†fragments to its right†)$ over all $l$,
		which will also be the limit of $R[m]$ as the block length goes to infinity.
		
		Now apply induction to show that every $d$-dimensional facet
		of the sum-rate polyhedron is achievable by source-splitting.
		\begin{itemize}
			\item	Show that every vertex of the sum-rate polyhedron
					corresponds to a permutation of $\{1…M\}$;
					and show that by reordering the sources in all possible ways,
					the duty tuple $(B[1]…B[M])$ attains all vertices.
			\item	Show that every edge of the sum-rate polyhedron corresponds to
					a smooth transition between two permutations that differ by a swap.
			\item	Show that every face ($2$-dimensional facet) is mapped surjectively
					because there is a Tour de France $Q_t$ whose image
					goes around the boundary once while the domain is contractible.
			\item	Show the similar argument that if a map from a topological ball
					maps the boundary to the boundary of a facet of the polygon,
					plus the induced map on the top homology groups is $1$
					(multiplying by one), then the map is surjective.
		\end{itemize}
		For more details, see \cite{GRUW01}.
	\end{proof}
	
	A similar statement can be made with $㏒(㏒N)$ complexity.
	The proof is essentially the same except that when we invoke the infrastructure,
	the log-log code is used instead of the second-moment code.
	
	I will end the dissertation with a remark on the dual of distributed compression.

\section{On Multiple Access Channels}

	Multiple access channels are noisy channels that take multiple inputs---%
	each from a different encoder---and output to a unified decoder.
	The assumption that the encoders cannot talk to each other
	and that the decoder has all information in hand to make decisions
	make this problem a proper dual of distributed compression.
	See \cref{fig:mac} for an example specification of a multiple access channel
	with three senders and compare it with \cref{fig:3sender}.
	
	\begin{figure}
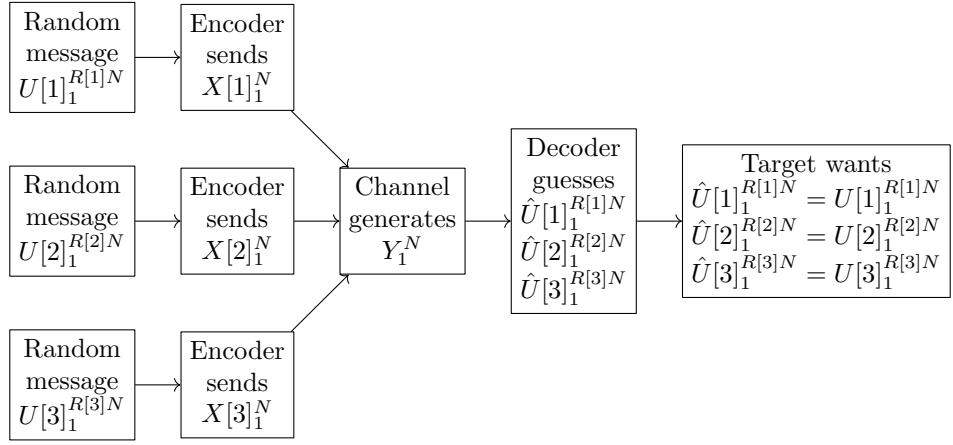

		\tikz{
			\draw[nodes={right,draw,align=center}]
				(.6,0)coordinate(X)(0,-.7)coordinate(Y)
				             node(A)[below]{Random\\message\\$U[1]₁^{R[1]N}$}
				(A.south)+(Y)node(B)[below]{Random\\message\\$U[2]₁^{R[2]N}$}
				(B.south)+(Y)node(C)[below]{Random\\message\\$U[3]₁^{R[3]N}$}
				(A.east) +(X)node(D){Encoder\\sends\\$X[1]₁^N$}
				(B.east) +(X)node(E){Encoder\\sends\\$X[2]₁^N$}
				(C.east) +(X)node(F){Encoder\\sends\\$X[3]₁^N$}
				(E.east) +(X)node(S){Channel\\generates\\$Y₁^N$}
				(S.east) +(X)node(T){Decoder\\guesses\\$ˆU[1]₁^{R[1]N}$\\
				                     $ˆU[2]₁^{R[2]N}$\\$ˆU[3]₁^{R[3]N}$}
				(T.east) +(X)node(U){Target wants\\
				                     $ˆU[1]₁^{R[1]N}=U[1]₁^{R[1]N}$\\
				                     $ˆU[2]₁^{R[2]N}=U[2]₁^{R[2]N}$\\
				                     $ˆU[3]₁^{R[3]N}=U[3]₁^{R[3]N}$};
			\graph[use existing nodes]{
				A -> D -> S -> T -> U,
				B -> E -> S,
				C -> F -> S
			};
		}
		\caption{
			Multiple access channel with three senders.
		}\label{fig:mac}
	\end{figure}
	
	The capacity region of a multiple access channel is defined
	similarly to the rate region of distributed compression.
	For instance for three senders, the capacity region is the set of
	$(R[1],R[2],R[3])$ such that a reliable communication can be carried.
	Then the pessimistic criteria are, for instance
	\begin{itemize}
		\item	sender$[1]$ can send out at most $I(X[1]；Y｜X[2]X[3])$ bits reliably,
		\item	senders $[1]$ and $[2]$ together can send out
				at most $I(X[1]X[2]；Y｜X[3])$ bits reliably, and
		\item	all three senders together can send out
				at most $I(X[1]X[2]X[3]；Y)$ bits reliably.
	\end{itemize}
	Hence the following characterization of the capacity region.
	See also \cref{fig:polymatroid}.
	
	\begin{figure}
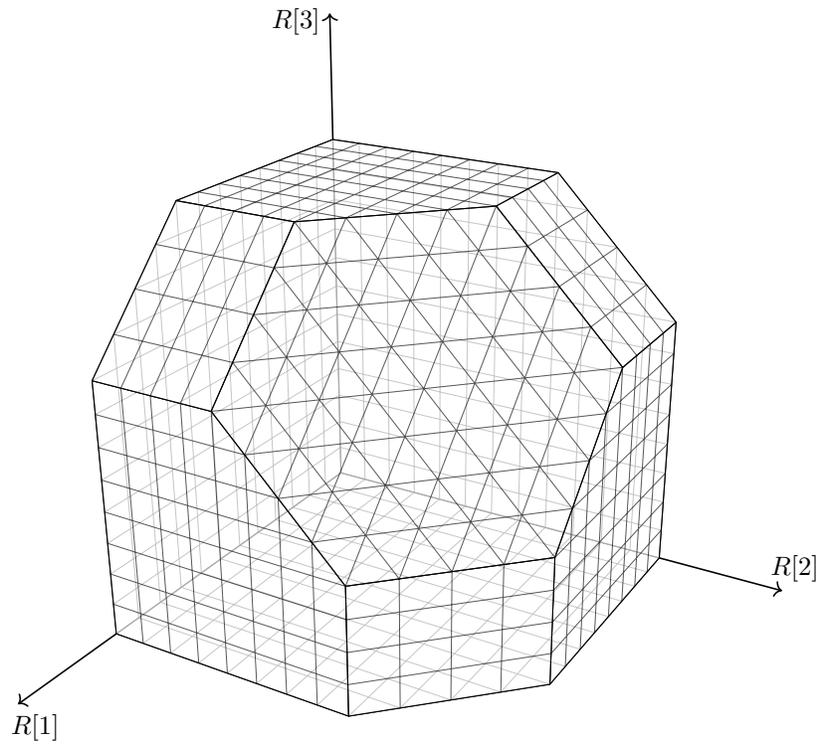

		\tikz[tdplot_main_coords]{
			\draw[semithick,nodes={pos=1,auto=perptip}](0,0,0)
				edge[->]node{$R[1]$}(4,0,0)
				edge[->]node{$R[2]$}(0,4,0)
				edge[->]node{$R[3]$}(0,0,4);
			\begin{scope}
				\face	(3,0,0)--(3,2,0)--(2,3,0)--(0,3,0)--(0,0,0)--cycle:;
				\mesh	foreach\x in{1/4,.../4,12/4}{(\x,0,0)--(\x,3,0)}
						foreach\y in{1/4,.../4,12/4}{(0,\y,0)--(3,\y,0)};
			\end{scope}
			\begin{scope}
				\face	(0,0,3)--(2,0,3)--(3,0,2)--(3,0,0)--(0,0,0)--cycle:;
				\mesh	foreach\z in{1/4,.../4,12/4}{(0,0,\z)--(3,0,\z)}
						foreach\x in{1/4,.../4,12/4}{(\x,0,0)--(\x,0,3)};
			\end{scope}
			\begin{scope}
				\face	(0,3,0)--(0,3,2)--(0,2,3)--(0,0,3)--(0,0,0)--cycle:;
				\mesh	foreach\y in{1/4,.../4,12/4}{(0,\y,0)--(0,\y,3)}
						foreach\z in{1/4,.../4,12/4}{(0,0,\z)--(0,3,\z)};
			\end{scope}
			\begin{scope}
				\face	(0,2,3)--(1,2,3)--(2,1,3)--(2,0,3)--(0,0,3)--cycle:;
				\mesh	foreach\x in{1/4,.../4,8/4}{(\x,0,3)--(\x,2,3)}
						foreach\y in{1/4,.../4,8/4}{(0,\y,3)--(2,\y,3)};
			\end{scope}
			\begin{scope}
				\face	(2,3,0)--(2,3,1)--(1,3,2)--(0,3,2)--(0,3,0)--cycle:;
				\mesh	foreach\z in{1/4,.../4,8/4}{(0,3,\z)--(2,3,\z)}
						foreach\x in{1/4,.../4,8/4}{(\x,3,0)--(\x,3,2)};
			\end{scope}
			\begin{scope}
				\face	(3,0,2)--(3,1,2)--(3,2,1)--(3,2,0)--(3,0,0)--cycle:;
				\mesh	foreach\y in{1/4,.../4,8/4}{(3,\y,0)--(3,\y,2)}
						foreach\z in{1/4,.../4,8/4}{(3,0,\z)--(3,2,\z)};
			\end{scope}
			\begin{scope}
				\face	(3,2,0)--(3,2,1)--(2,3,1)--(2,3,0)--cycle:;
				\mesh	foreach\z in{1/4,.../4,8/4}{(2+\z,3-\z,0)--(2+\z,3-\z,1)}
						foreach\z in{1/4,.../4,8/4}{(3,2,\z)--(2,3,\z)};
			\end{scope}
			\begin{scope}
				\face	(2,0,3)--(2,1,3)--(3,1,2)--(3,0,2)--cycle:;
				\mesh	foreach\y in{1/4,.../4,8/4}{(3-\y,0,2+\y)--(3-\y,1,2+\y)}
						foreach\y in{1/4,.../4,8/4}{(2,\y,3)--(3,\y,2)};
			\end{scope}
			\begin{scope}
				\face	(0,3,2)--(1,3,2)--(1,2,3)--(0,2,3)--cycle:;
				\mesh	foreach\x in{1/4,.../4,8/4}{(0,2+\x,3-\x)--(1,2+\x,3-\x)}
						foreach\x in{1/4,.../4,8/4}{(\x,3,2)--(\x,2,3)};
			\end{scope}
			\begin{scope}
				\face	(1,2,3)--(1,3,2)--(2,3,1)--(3,2,1)--(3,1,2)--(2,1,3)--cycle:;
				\mesh	foreach\z in{0/4,.../4,8/4}{(4-\z,1,1+\z)--(1,4-\z,1+\z)}
						foreach\y in{0/4,.../4,8/4}{(1,1+\y,4-\y)--(4-\y,1+\y,1)}
						foreach\x in{0/4,.../4,8/4}{(1+\x,4-\x,1)--(1+\x,1,4-\x)};
			\end{scope}
		}
		\caption{
			The capacity region of a three-sender multiple access channel.
		}\label{fig:polymatroid}
	\end{figure}
	
	\begin{thm}[Rate region of multiple access channel]
		For any multiple access channel with $M$ senders, the capacity region is
		the set of points $(R[1]…R[M])$ such that, for all subsets $𝒮⊆\{1…M\}$,
		\[∑_{m∈𝒮}R[m]≤I(X[𝒮]；Y｜S[𝒮^∁],Q),\]
		unioned over all possible distributions
		of the inputs $X[1]…X[M]$ and the knob variable $Q$.
	\end{thm}
	
	As commented before, \cite{GRUW01} showed that one can
	split a multiple access channel into several one-to-one DMCs.
	Once that is done, we can apply the infrastructure.
	
	\begin{cor}[Polar coding for multiple access channel]
		Polar coding coupled with rate-splitting attains every point
		on the boundary of the capacity region with $N^{-ρ}$ gap.
	\end{cor}
	
	\begin{proof}
		\cite{GRUW01} reduces this problem into at most $2M$
		(twice the number of senders) DMCs.
		For each DMC, apply the polar coding infrastructure.
	\end{proof}
	
	A similar statement can be made with $㏒(㏒N)$ complexity.
	The proof is essentially the same except that when we invoke the infrastructure,
	the log-log code is used instead of the second-moment code.

\hbadness999\g@addto@macro\sloppy{\advance\baselineskip0ptplus1ptminus1pt}
\bibliographystyle{alphaurl}
\bibliography{Chilly-4}

\end{document}